\documentclass{lmcs} %
\pdfoutput=1

\usepackage{lastpage}
\lmcsdoi{20}{4}{17}
\lmcsheading{}{\pageref{LastPage}}{}{}%
{Apr.~24,~2023}{Nov.~22,~2024}{}

\usepackage[utf8]{inputenc}

\keywords{algebraic effects, scoped effects, calculus, operational semantics,
  type- and effect system}

\usepackage{hyperref}
\theoremstyle{plain} %

\usepackage{thmtools}
\usepackage{thm-restate}

\usepackage{mathpartir}

\newcommand{\mathscalefactor}{1}
\usepackage{tikz}
\usepackage{amsmath} %
\usepackage{cleveref}
\usepackage{mathtools}
\usetikzlibrary{fit}

\usepackage{changebar}

\usepackage{multirow}

\makeatletter
\@ifundefined{lhs2tex.lhs2tex.sty.read}%
  {\@namedef{lhs2tex.lhs2tex.sty.read}{}%
   \newcommand\SkipToFmtEnd{}%
   \newcommand\EndFmtInput{}%
   \long\def\SkipToFmtEnd#1\EndFmtInput{}%
  }\SkipToFmtEnd

\newcommand\ReadOnlyOnce[1]{\@ifundefined{#1}{\@namedef{#1}{}}\SkipToFmtEnd}
\usepackage{amstext}
\usepackage{amssymb}
\usepackage{stmaryrd}
\DeclareFontFamily{OT1}{cmtex}{}
\DeclareFontShape{OT1}{cmtex}{m}{n}
  {<5><6><7><8>cmtex8
   <9>cmtex9
   <10><10.95><12><14.4><17.28><20.74><24.88>cmtex10}{}
\DeclareFontShape{OT1}{cmtex}{m}{it}
  {<-> ssub * cmtt/m/it}{}

\DeclareFontShape{OT1}{cmtt}{bx}{n}
  {<5><6><7><8>cmtt8
   <9>cmbtt9
   <10><10.95><12><14.4><17.28><20.74><24.88>cmbtt10}{}
\DeclareFontShape{OT1}{cmtex}{bx}{n}
  {<-> ssub * cmtt/bx/n}{}

\newcommand{\Conid}[1]{\mathit{#1}}
\newcommand{\Varid}[1]{\mathit{#1}}
\newcommand{\anonymous}{\kern0.06em \vbox{\hrule\@width.5em}}

\newcommand{\bind}{\mathbin{>\!\!\!>\mkern-6.7mu=}}

\renewcommand{\leq}{\leqslant}

\usepackage{polytable}

\@ifundefined{mathindent}%
  {\newdimen\mathindent\mathindent\leftmargini}%
  {}%

\def\resethooks{%
  \global\let\SaveRestoreHook\empty
  \global\let\ColumnHook\empty}
\newcommand*{\savecolumns}[1][default]%
  {\g@addto@macro\SaveRestoreHook{\savecolumns[#1]}}
\newcommand*{\restorecolumns}[1][default]%
  {\g@addto@macro\SaveRestoreHook{\restorecolumns[#1]}}
\newcommand*{\aligncolumn}[2]%
  {\g@addto@macro\ColumnHook{\column{#1}{#2}}}

\resethooks

\newcommand{\onelinecommentchars}{\quad-{}- }
\newcommand{\commentbeginchars}{\enskip\{-}
\newcommand{\commentendchars}{-\}\enskip}

\newcommand{\visiblecomments}{%
  \let\onelinecomment=\onelinecommentchars
  \let\commentbegin=\commentbeginchars
  \let\commentend=\commentendchars}

\newcommand{\invisiblecomments}{%
  \let\onelinecomment=\empty
  \let\commentbegin=\empty
  \let\commentend=\empty}

\visiblecomments

\newlength{\blanklineskip}
\newcommand{\hsindent}[1]{\quad}%
\let\hspre\empty
\let\hspost\empty

\EndFmtInput
\makeatother
\ReadOnlyOnce{polycode.fmt}%
\makeatletter

\newcommand{\hsnewpar}[1]%
  {{\parskip=0pt\parindent=0pt\par\vskip #1\noindent}}

\newcommand{\hscodestyle}{}

\newcommand{\sethscode}[1]%
  {\expandafter\let\expandafter\hscode\csname #1\endcsname
   \expandafter\let\expandafter\endhscode\csname end#1\endcsname}

  {\par\noindent
   \advance\leftskip\mathindent
   \hscodestyle
   \let\\=\@normalcr
   \let\hspre\(\let\hspost\)%
   \pboxed}%
  {\endpboxed\)%
   \par\noindent
   \ignorespacesafterend}

  {\hsnewpar\abovedisplayskip
   \advance\leftskip\mathindent
   \hscodestyle
   \let\hspre\(\let\hspost\)%
   \pboxed}%
  {\endpboxed%
   \hsnewpar\belowdisplayskip
   \ignorespacesafterend}

  {\hsnewpar\abovedisplayskip
   \advance\leftskip\mathindent
   \hscodestyle
   \let\\=\@normalcr
   \(\pboxed}%
  {\endpboxed\)%
   \hsnewpar\belowdisplayskip
   \ignorespacesafterend}

\newcommand{\plainhs}{\sethscode{plainhscode}}

\plainhs

  {\hsnewpar\abovedisplayskip
   \advance\leftskip\mathindent
   \hscodestyle
   \let\\=\@normalcr
   \(\parray}%
  {\endparray\)%
   \hsnewpar\belowdisplayskip
   \ignorespacesafterend}

  {\parray}{\endparray}

  {\(\parray}{\endparray\)}

\def\codeframewidth{\arrayrulewidth}
\RequirePackage{calc}

  {\parskip=\abovedisplayskip\par\noindent
   \hscodestyle
   \arrayrulewidth=\codeframewidth
   \tabular{@{}|p{\linewidth-2\arraycolsep-2\arrayrulewidth-2pt}|@{}}%
   \hline\framedhslinecorrect\\{-1.5ex}%
   \let\endoflinesave=\\
   \let\\=\@normalcr
   \(\pboxed}%
  {\endpboxed\)%
   \framedhslinecorrect\endoflinesave{.5ex}\hline
   \endtabular
   \parskip=\belowdisplayskip\par\noindent
   \ignorespacesafterend}

\newcommand{\framedhslinecorrect}[2]%
  {#1[#2]}

  {\(\def\column##1##2{}%
   \let\>\undefined\let\<\undefined\let\\\undefined
   \newcommand\>[1][]{}\newcommand\<[1][]{}\newcommand\\[1][]{}%
   \def\fromto##1##2##3{##3}%
   }{\) }%

  {\let\orighscode=\hscode
   \let\origendhscode=\endhscode
   \def\endhscode{\def\hscode{\endgroup\def\@currenvir{hscode}\\}\begingroup}
   \orighscode\def\hscode{\endgroup\def\@currenvir{hscode}}}%
  {\origendhscode
   \global\let\hscode=\orighscode
   \global\let\endhscode=\origendhscode}%

\makeatother
\EndFmtInput

\usepackage{enumitem}
\newlist{qwq}{itemize}{1}
\setlist[qwq]{label={}, nosep, leftmargin=0.5cm}
\newcommand{\indentbegin}{\begin{qwq} \item}
  \newcommand{\indentend}{\end{qwq}}

\newcommand{\grammarspacing}{0.1cm}
\usepackage{pifont}
\newcommand{\xmark}{\ding{55}}%

\newcommand{\ruleform}[1]{\fbox{$#1$}}

\crefname{lem}{lemma}{lemmas}
\Crefname{lem}{Lemma}{Lemmas}
\crefname{thm}{theorem}{theorems}
\Crefname{thm}{Theorem}{Theorems}

\usepackage{placeins}

\begin{document}

\title{A Calculus for Scoped Effects \& Handlers}

\author[R.~Bosman]{Roger Bosman$^*$\lmcsorcid{0000-0002-6693-4653}}[a]
\author[B.~van den Berg]{Birthe van den Berg$^*$\lmcsorcid{0000-0002-0088-9546}}[a]
\author[W.~Tang]{Wenhao Tang$^*$\lmcsorcid{0009-0000-6589-3821}}[b]
\author[T.~Schrijvers]{Tom Schrijvers\lmcsorcid{0000-0001-8771-5559}}[a]

\address{KU Leuven, Celestijnenlaan 200A, Belgium}
\email{\{roger.bosman/birthe.vandenberg/tom.schrijvers\}@kuleuven.be}

\address{The University of Edinburgh, 10 Crichton Street, United Kingdom}
\email{wenhao.tang@ed.ac.uk}

\begin{abstract}
Algebraic effects \& handlers have become a standard approach for 
side-effects in functional programming.
Their modular composition between different effects and clean separation of syntax and
semantics make them attractive to a wide audience.
However, not all effects can be classified as algebraic; some need a more
sophisticated handling.
In particular, effects that have or create a delimited scope need special care,
as their continuation consists of two parts---in and out of
the scope---and their modular composition introduces additional complexity.
These effects are called \emph{scoped} and have gained attention by their
growing applicability and adoption in popular libraries.
While calculi have been designed with algebraic effects \& handlers 
built in to facilitate their use, a
calculus that supports scoped effects \& handlers in a similar manner 
does not yet exist.
This work fills this gap: we present \ensuremath{\lambda_{\mathit{sc}}}, a calculus with native support
for both algebraic and scoped effects \& handlers.
The most novel part of \ensuremath{\lambda_{\mathit{sc}}} is the support for modular composition
of different scoped effects \& handlers by extending handlers with
forwarding clauses, which make \ensuremath{\lambda_{\mathit{sc}}} much more expressive than
existing calculi with algebraic effects \& handlers.
Our calculus is based on Eff, an existing calculus for algebraic effects, extended
with Koka-style row polymorphism, and
consists of a formal grammar, operational semantics, and a (type-safe) type-and-effect system.
We give a prototype implementation of \ensuremath{\lambda_{\mathit{sc}}} with type inference and
demonstrate \ensuremath{\lambda_{\mathit{sc}}} on a range of examples.
\end{abstract}

\maketitle
\def\thefootnote{*}\footnotetext{These authors contributed equally to this work.}\def\thefootnote{\arabic{footnote}}

\section{Introduction}
\label{sec:introductions}

While monads \cite{Moggi89,Moggi95,DBLP:conf/afp/Wadler95} have long been the
go-to approach for modelling effects, \emph{algebraic effects \& handlers}
\cite{Plotkin03,Plotkin09} are steadily gaining more traction. They offer a more
structured and modular approach to composing effects, based on an algebraic
model. The approach consists of two parts: effects denote the syntax of
operations, and handlers interpret them by means of structural recursion.
By composing handlers that each interpret only a part of the syntax in the
desired order, one can modularly build an interpretation for the entire program.
Algebraic
effects \& handlers  have been adopted
in several libraries (e.g., fused-effects \cite{fused-effects},
extensible-effects \cite{extensible-effects}, Eff in OCaml \cite{eff-ocaml})
and languages (e.g., Eff~\cite{pretnar15}, Links \cite{Hillerstrom16}, Koka \cite{Leijen17}, Effekt
\cite{effekt}).

Although the modular approach of algebraic effects \& handlers is desirable for
every effectful program, it is not always applicable.
In particular, those effects that have or introduce a delimited scope (e.g.,
exceptions, concurrency, local state) are not algebraic.
Essentially, these so-called \emph{scoped effects} \cite{wu14} split the program
into two: a part \emph{in} the scope of the effect (called the scoped
computation) and a part \emph{out} of the scope (called the continuation).

This separation breaks algebraicity, which states that operations commute with
sequencing.
Modeling scoped effects as handlers \cite{Plotkin03} has been proposed as a way
of encoding scoped effects in an algebraic framework.
However, this comes at the cost of modularity \cite{DBLP:conf/esop/YangPWBS22}
(and thus expressiveness).
Wu et al.\ proposes to treat scoped effects as separate, built-in operations,
and present their solution in a denotational setting.
However, their solution of modular composition of scoped handlers is rather
ad-hoc, and not a generalization of the algebraic case.
Therefore, a treatment of scoped effects \& handlers in an operational setting
with full support for modular composition is desirable.
The growing interest in scoped effects \& handlers, evidenced by their adoption
at GitHub~\cite{DBLP:journals/pacmpl/ThomsonRWS22} and in Haskell libraries
(e.g., eff~\cite{eff}, polysemy~\cite{polysemy},
fused-effects~\cite{fused-effects}), motivates the need for such a calculus, all
the more because they use the same ad-hoc approach of modular composition as
suggested by Wu et al.

This paper aims to fill this gap in the literature: we cover scoped effects \&
handlers, which in previous work has only been covered in a denotational
setting, in an operational setting by presenting \ensuremath{\lambda_{\mathit{sc}}}, a calculus with scoped
effects \& handlers.
Our main source of inspiration is Eff \cite{DBLP:conf/calco/BauerP13,bauer15,pretnar15}, %
a calculus for algebraic effects \& handlers, effectively easing programming with those features.
Although Eff is an appropriate starting point, the extension to support scoped effects
\& handlers is non-trivial, for two reasons.
First, scoped effects require polymorphic handlers, which we
support by let-polymorphism and \ensuremath{F_{\omega}}-style type operators. %
Second, we need to be able to forward unknown operations in order to keep the desired modularity.
Whereas algebraic effects \& handlers have a uniform (and implicit) forwarding
mechanism, scoped effects \& handlers need explicit forwarding clauses in
order to allow sufficient freedom in their implementation.

In what follows,
after introducing the appropriate background (\Cref{sec:background-motivation})
and informally motivating the challenges and design choices of our calculus
(\Cref{sec:design-decisions}), we formalize \ensuremath{\lambda_{\mathit{sc}}}.
We make the following contributions:
\begin{itemize}
\item We present our solution for the modular composition of scoped handlers,
  called forwarding (\Cref{sec:scoped-result-type}), and generalize it later
  (\Cref{sec:genfwd}).
\item We design formal syntax for \ensuremath{\lambda_{\mathit{sc}}} terms, types and contexts
(\Cref{sec:syntax}).
\item We provide an operational semantics (\Cref{sec:operational-semantics}).
\item We define a type-and-effect system for \ensuremath{\lambda_{\mathit{sc}}} (\Cref{sec:typing}).
\item We formulate and prove \ensuremath{\lambda_{\mathit{sc}}}'s metatheoretical properties (\Cref{sec:metatheory}).
\item We show the usability of our calculus on a range of examples (\Cref{sec:examples}). %
\item We provide an interpreter of our calculus with type inference in which we implement all our examples (supplementary material).
\end{itemize}
\section{Background \& Motivation}
\label{sec:background-motivation}

This section provides the necessary background and motivates our goal. We review \emph{algebraic effects \& handlers} as a
modular approach to composing side-effects in effectful programs.
Next, we present \emph{scoped effects \& handlers}: effects that have or create
a delimited scope (such as \ensuremath{\mathtt{once}} for nondeterminism \cite{pirog18,wu14}), and
motivate the need for a calculus with built-in support for these scoped effects.

\subsection{Algebraic effects \& handlers}

Algebraic effects \& handlers consist of \emph{operations}, denoting
their syntax, and \emph{handlers}, denoting their semantics.
This separation gives us modular composition, which has
intrinsic value \emph{and} allows controlling effect
interaction.

\subsubsection{Algebraic Operations}
Effects are denoted by a name (or \emph{label}) and characterized by a \emph
{signature} \ensuremath{\Conid{A}\rightarrowtriangle\Conid{B}}, taking a value of type \ensuremath{\Conid{A}} and producing a value of type
\ensuremath{\Conid{B}}.
For example, \ensuremath{\mathtt{choose}\mathbin{:}()\rightarrowtriangle\mathsf{Bool}} takes a unit value and produces a boolean
(e.g., nondeterministically).
\emph{Operations} invoke effects, combining the \ensuremath{{\mathbf{op}}} keyword, an effect to
invoke, a \emph{parameter} passed to the effect, and a \emph{continuation},
containing the rest of the program.
For brevity of presentation, we follow the syntax of
Eff~\cite{pretnar15} to have an explicit continuation for every
operation.
\label{eq:cnd1}
\indentbegin \begin{hscode}\SaveRestoreHook
\column{B}{@{}>{\hspre}l<{\hspost}@{}}%
\column{3}{@{}>{\hspre}l<{\hspost}@{}}%
\column{11}{@{}>{\hspre}l<{\hspost}@{}}%
\column{36}{@{}>{\hspre}l<{\hspost}@{}}%
\column{42}{@{}>{\hspre}l<{\hspost}@{}}%
\column{57}{@{}>{\hspre}l<{\hspost}@{}}%
\column{E}{@{}>{\hspre}l<{\hspost}@{}}%
\>[3]{}\Varid{c_{\mathsf{ND1}}}\mathrel{=}{}\<[11]%
\>[11]{}{\mathbf{op}}\;\mathtt{choose}\;\mathsf{()}\;(\Varid{b}\, .\,\mathbf{if}\;\Varid{b}\;{}\<[36]%
\>[36]{}\mathbf{then}\;{}\<[42]%
\>[42]{}{\mathbf{return}}\;\mathrm{1}\;\mathbf{else}\;{}\<[57]%
\>[57]{}{\mathbf{return}}\;\mathrm{2}){}\<[E]%
\ColumnHook
\end{hscode}\resethooks
\indentend %
In accordance with its signature, \ensuremath{\mathtt{choose}} is passed \ensuremath{()}, and in the supplied
contination \ensuremath{\Varid{b}} has type \ensuremath{\mathsf{Bool}}.
As a result, \ensuremath{\Varid{c_{\mathsf{ND1}}}} is a computation that returns either 1 or 2.

Some operations commute with sequencing.
For example:\indentbegin \begin{hscode}\SaveRestoreHook
\column{B}{@{}>{\hspre}l<{\hspost}@{}}%
\column{3}{@{}>{\hspre}l<{\hspost}@{}}%
\column{10}{@{}>{\hspre}l<{\hspost}@{}}%
\column{E}{@{}>{\hspre}l<{\hspost}@{}}%
\>[10]{}\mathbf{do}\;\Varid{x}\leftarrow {\mathbf{op}}\;\mathtt{choose}\;\mathsf{()}\;(\Varid{b}\, .\,\mathbf{if}\;\Varid{b}\;\mathbf{then}\;{\mathbf{return}}\;\mathrm{1}\;\mathbf{else}\;{\mathbf{return}}\;\mathrm{2})\, ;{\mathbf{return}}\;\Varid{x}^{\mathrm{2}}{}\<[E]%
\\
\>[3]{}\equiv\;{}\<[10]%
\>[10]{}{\mathbf{op}}\;\mathtt{choose}\;\mathsf{()}\;(\Varid{b}\, .\,\mathbf{do}\;\Varid{x}\leftarrow \mathbf{if}\;\Varid{b}\;\mathbf{then}\;{\mathbf{return}}\;\mathrm{1}\;\mathbf{else}\;{\mathbf{return}}\;\mathrm{2}\, ;{\mathbf{return}}\;\Varid{x}^{\mathrm{2}}){}\<[E]%
\ColumnHook
\end{hscode}\resethooks
\indentend This equivalence is an instance of the \emph{algebraicity property}, and
operations are \emph{algebraic} if they satisfy this property.
Algebraicity states that the sequencing
of a computation \ensuremath{\Varid{c}_{2}} after an operation \ensuremath{{\mathbf{op}}\;\ell\;\Varid{v}\;(\Varid{y}\, .\,\Varid{c}_{1})} is equivalent to
sequencing the same computation after the \emph{continuation} of this operation:
\indentbegin \begin{hscode}\SaveRestoreHook
\column{B}{@{}>{\hspre}l<{\hspost}@{}}%
\column{3}{@{}>{\hspre}l<{\hspost}@{}}%
\column{E}{@{}>{\hspre}l<{\hspost}@{}}%
\>[3]{}\mathbf{do}\;\Varid{x}\leftarrow {\mathbf{op}}\;\ell\;\Varid{v}\;(\Varid{y}\, .\,\Varid{c}_{1})\, ;\Varid{c}_{2}\;\equiv\;{\mathbf{op}}\;\ell\;\Varid{v}\;(\Varid{y}\, .\,\mathbf{do}\;\Varid{x}\leftarrow \Varid{c}_{1}\, ;\Varid{c}_{2}{}\<[E]%
\ColumnHook
\end{hscode}\resethooks
\indentend \subsubsection{Handlers}
\emph{Handlers} give meaning to operations. For example, handler \ensuremath{\Varid{h_{\mathsf{ND}}}}
interprets \ensuremath{\mathtt{choose}} by collecting all the choices in a list:
\indentbegin \begin{hscode}\SaveRestoreHook
\column{B}{@{}>{\hspre}l<{\hspost}@{}}%
\column{3}{@{}>{\hspre}l<{\hspost}@{}}%
\column{13}{@{}>{\hspre}l<{\hspost}@{}}%
\column{22}{@{}>{\hspre}c<{\hspost}@{}}%
\column{22E}{@{}l@{}}%
\column{25}{@{}>{\hspre}l<{\hspost}@{}}%
\column{41}{@{}>{\hspre}l<{\hspost}@{}}%
\column{92}{@{}>{\hspre}c<{\hspost}@{}}%
\column{92E}{@{}l@{}}%
\column{E}{@{}>{\hspre}l<{\hspost}@{}}%
\>[3]{}\Varid{h_{\mathsf{ND}}}\mathrel{=}{}\<[13]%
\>[13]{}{\mathbf{handler}}\;{}\<[22]%
\>[22]{}\{\mskip1.5mu {}\<[22E]%
\>[25]{}{\mathbf{return}}\;\Varid{x}{}\<[41]%
\>[41]{}\mapsto{\mathbf{return}}\;[\mskip1.5mu \Varid{x}\mskip1.5mu]{}\<[E]%
\\
\>[22]{},{}\<[22E]%
\>[25]{}{\mathbf{op}}\;\mathtt{choose}\;\anonymous \;\Varid{k}{}\<[41]%
\>[41]{}\mapsto\mathbf{do}\;\Varid{xs}\leftarrow \Varid{k}\;\mathsf{true}\, ;\mathbf{do}\;\Varid{ys}\leftarrow \Varid{k}\;\mathsf{false}\, ;\Varid{xs}+\hspace{-0.4em}+\Varid{ys}{}\<[92]%
\>[92]{}\mskip1.5mu\}{}\<[92E]%
\ColumnHook
\end{hscode}\resethooks
\indentend This handler has two clauses. The first clause returns a singleton list in case
a value \ensuremath{\Varid{x}} is returned. The second clause, which interprets \ensuremath{\mathtt{choose}},
executes both branches by applying the continuation \ensuremath{\Varid{k}} to both
\ensuremath{\mathsf{true}} and \ensuremath{\mathsf{false}}, and concatenates their resulting lists with the \ensuremath{(+\hspace{-0.4em}+)}-operator.
We apply \ensuremath{\Varid{h_{\mathsf{ND}}}} to \ensuremath{\Varid{c_{\mathsf{ND1}}}} to obtain both of its results:
\indentbegin \begin{hscode}\SaveRestoreHook
\column{B}{@{}>{\hspre}l<{\hspost}@{}}%
\column{3}{@{}>{\hspre}c<{\hspost}@{}}%
\column{3E}{@{}l@{}}%
\column{8}{@{}>{\hspre}l<{\hspost}@{}}%
\column{E}{@{}>{\hspre}l<{\hspost}@{}}%
\>[8]{}{\mathbf{with}}\;\Varid{h_{\mathsf{ND}}}\;{\mathbf{handle}}\;\Varid{c_{\mathsf{ND1}}}{}\<[E]%
\\
\>[3]{}\leadsto^{\ast}{}\<[3E]%
\>[8]{}[\mskip1.5mu \mathrm{1},\mathrm{2}\mskip1.5mu]{}\<[E]%
\ColumnHook
\end{hscode}\resethooks
\indentend %
The algebraicity property plays a critical role in giving semantics to
algebraic effects and handlers.
For example, consider the following program which uses \ensuremath{\Varid{h_{\mathsf{ND}}}} to handle
a computation that sequences two \ensuremath{\mathtt{choose}} operations using the \ensuremath{\mathbf{do}}
syntax.
\indentbegin \begin{hscode}\SaveRestoreHook
\column{B}{@{}>{\hspre}l<{\hspost}@{}}%
\column{3}{@{}>{\hspre}l<{\hspost}@{}}%
\column{11}{@{}>{\hspre}l<{\hspost}@{}}%
\column{E}{@{}>{\hspre}l<{\hspost}@{}}%
\>[3]{}\Varid{c_{\mathsf{ND2}}}\mathrel{=}{}\<[11]%
\>[11]{}\mathbf{do}\;\Varid{p}\leftarrow {\mathbf{op}}\;\mathtt{choose}\;\mathsf{()}\;(\Varid{b}\, .\,{\mathbf{return}}\;\Varid{b})\, ;{\mathbf{op}}\;\mathtt{choose}\;\mathsf{()}\;(\Varid{q}\, .\,{\mathbf{return}}\;(\Varid{p},\Varid{q})){}\<[E]%
\ColumnHook
\end{hscode}\resethooks
\indentend %
In order to handle the first \ensuremath{\mathtt{choose}} operation, the handler \ensuremath{\Varid{h_{\mathsf{ND}}}}
needs to get access to the continuation of it. Thus, we use the
algebraicity property of \ensuremath{\mathtt{choose}} to put the second \ensuremath{\mathtt{choose}} operation
in the explicit continuation of the first \ensuremath{\mathtt{choose}} operation.
\indentbegin \begin{hscode}\SaveRestoreHook
\column{B}{@{}>{\hspre}l<{\hspost}@{}}%
\column{7}{@{}>{\hspre}l<{\hspost}@{}}%
\column{9}{@{}>{\hspre}l<{\hspost}@{}}%
\column{E}{@{}>{\hspre}l<{\hspost}@{}}%
\>[7]{}{\mathbf{with}}\;\Varid{h_{\mathsf{ND}}}\;{\mathbf{handle}}\;\Varid{c_{\mathsf{ND2}}}{}\<[E]%
\\
\>[B]{}\leadsto{}\<[7]%
\>[7]{}{\mathbf{with}}\;\Varid{h_{\mathsf{ND}}}\;{\mathbf{handle}}\;{\mathbf{op}}\;\mathtt{choose}\;\mathsf{()}\;(\Varid{b}\, .\,{}\<[E]%
\\
\>[7]{}\hsindent{2}{}\<[9]%
\>[9]{}\mathbf{do}\;\Varid{p}\leftarrow {\mathbf{return}}\;\Varid{b}\, ;{\mathbf{op}}\;\mathtt{choose}\;\mathsf{()}\;(\Varid{q}\, .\,{\mathbf{return}}\;(\Varid{p},\Varid{q}))){}\<[E]%
\\
\>[B]{}\leadsto^{\ast}{}\<[7]%
\>[7]{}[\mskip1.5mu (\mathsf{true},\mathsf{true}),(\mathsf{true},\mathsf{false}),(\mathsf{false},\mathsf{true}),(\mathsf{false},\mathsf{false})\mskip1.5mu]{}\<[E]%
\ColumnHook
\end{hscode}\resethooks
\indentend %
\vspace{0cm}
\vspace{0cm}
\vspace{0cm}
\vspace{0cm}
\vspace{0cm}

Algebraic effects \& handlers bring several interesting advantages. Most
interestingly, their separation of syntax and semantics allows a modular
composition of different effects, which in turn allows for altering the meaning
of a program by different effect interactions.

\label{sec:denotational-handlers-are-fold-intro}
For those familiar with the category-theoretical backings of folds, handler
application can be seen as a fold over an abstract syntax tree, where operations
are nodes in that tree, and handlers are the fold algebra.
The result type of the fold is the carrier of the algebra.

\subsubsection{Modular Composition}
\label{sec:modularcomposition}
Effects can be composed by combining different primitive operations.
For example, computation \ensuremath{\Varid{c_{\mathsf{c},\mathsf{g}}}} below uses \ensuremath{\mathtt{get}\mathbin{:}\mathsf{()}\rightarrowtriangle\mathsf{String}} in addition to \ensuremath{\mathtt{choose}}.
\indentbegin \begin{hscode}\SaveRestoreHook
\column{B}{@{}>{\hspre}l<{\hspost}@{}}%
\column{3}{@{}>{\hspre}l<{\hspost}@{}}%
\column{10}{@{}>{\hspre}l<{\hspost}@{}}%
\column{E}{@{}>{\hspre}l<{\hspost}@{}}%
\>[3]{}\Varid{c_{\mathsf{c},\mathsf{g}}}\mathrel{=}{}\<[10]%
\>[10]{}{\mathbf{op}}\;\mathtt{choose}\;\mathsf{()}\;(\Varid{b}\, .\,\mathbf{if}\;\Varid{b}\;\mathbf{then}\;{\mathbf{return}}\;\mathrm{1}\;\mathbf{else}\;{\mathbf{op}}\;\mathtt{get}\;\mathsf{()}\;(\Varid{x}\, .\,{\mathbf{return}}\;\Varid{x})){}\<[E]%
\ColumnHook
\end{hscode}\resethooks
\indentend \noindent
Instead of having to write a handler for each combination of effects,
algebraic effects \& handlers allow us to write a handler specific for the effect \ensuremath{\mathtt{get}},
and to compose it with the existing handler \ensuremath{\Varid{h_{\mathsf{ND}}}}.
\indentbegin \begin{hscode}\SaveRestoreHook
\column{B}{@{}>{\hspre}l<{\hspost}@{}}%
\column{3}{@{}>{\hspre}l<{\hspost}@{}}%
\column{14}{@{}>{\hspre}l<{\hspost}@{}}%
\column{23}{@{}>{\hspre}c<{\hspost}@{}}%
\column{23E}{@{}l@{}}%
\column{26}{@{}>{\hspre}l<{\hspost}@{}}%
\column{41}{@{}>{\hspre}l<{\hspost}@{}}%
\column{56}{@{}>{\hspre}l<{\hspost}@{}}%
\column{71}{@{}>{\hspre}l<{\hspost}@{}}%
\column{E}{@{}>{\hspre}l<{\hspost}@{}}%
\>[3]{}\Varid{h_{\mathsf{get}}}\mathrel{=}{}\<[14]%
\>[14]{}{\mathbf{handler}}\;{}\<[23]%
\>[23]{}\{\mskip1.5mu {}\<[23E]%
\>[26]{}{\mathbf{return}}\;\Varid{x}{}\<[41]%
\>[41]{}\mapsto{\mathbf{return}}\;\Varid{x},{}\<[56]%
\>[56]{}{\mathbf{op}}\;\mathtt{get}\;\anonymous \;\Varid{k}{}\<[71]%
\>[71]{}\mapsto\Varid{k}\;\mathrm{2}\mskip1.5mu\}{}\<[E]%
\ColumnHook
\end{hscode}\resethooks
\indentend When composing handlers \ensuremath{{\mathbf{with}}\;\Varid{h_{\mathsf{ND}}}\;{\mathbf{handle}}\;({\mathbf{with}}\;\Varid{h_{\mathsf{get}}}\;{\mathbf{handle}}\;\Varid{c_{\mathsf{c},\mathsf{g}}})}, \ensuremath{\Varid{h_{\mathsf{get}}}} is
applied first, and handles \ensuremath{\mathtt{get}}.
Since \ensuremath{\Varid{h_{\mathsf{get}}}} does not contain a clause for \ensuremath{\mathtt{choose}}, it leaves (we say
\emph{``forwards''}, see below) the \ensuremath{\mathtt{choose}} operation to be handled by another
handler.
This forwarding behavior is key to the modular reuse and composition of
handlers.
Handler \ensuremath{\Varid{h_{\mathsf{ND}}}} then takes care of the remaining effects.
\indentbegin \begin{hscode}\SaveRestoreHook
\column{B}{@{}>{\hspre}l<{\hspost}@{}}%
\column{7}{@{}>{\hspre}l<{\hspost}@{}}%
\column{E}{@{}>{\hspre}l<{\hspost}@{}}%
\>[7]{}{\mathbf{with}}\;\Varid{h_{\mathsf{ND}}}\;{\mathbf{handle}}\;({\mathbf{with}}\;\Varid{h_{\mathsf{get}}}\;{\mathbf{handle}}\;\Varid{c_{\mathsf{c},\mathsf{g}}}){}\<[E]%
\\
\>[B]{}\leadsto^{\ast}{}\<[7]%
\>[7]{}{\mathbf{with}}\;\Varid{h_{\mathsf{ND}}}\;{\mathbf{handle}}\;\Varid{c_{\mathsf{ND1}}}{}\<[E]%
\\
\>[B]{}\leadsto^{\ast}{}\<[7]%
\>[7]{}[\mskip1.5mu \mathrm{1},\mathrm{2}\mskip1.5mu]{}\<[E]%
\ColumnHook
\end{hscode}\resethooks
\indentend \subsubsection{Forwarding}
\label{sec:algebraic-forwarding}
  Critical to modular composition is the possibility to apply partial handlers,
  i.e.\ to apply handlers to computations where the handlers may only handle a
  subset of the effects present in the computations.
  As discussed in \Cref{sec:denotational-handlers-are-fold-intro}, handler
  application can be seen as a fold, where handlers form the folding algebra.
  Therefore, the partial algebra formed by partial handlers must be supplemented
  by a \emph{forwarding} algebra that interprets the effects not handled by the
  handler such that they can be handled by another handler.

  Generally, calculi for algebraic effects require that handler
  clauses leave computations as their output.
  The forwarding algebra must ``reinterpret'' computations into
  computations, which is essentially the identity function.
  Therefore, for these calculi, forwarding can be and is done uniformly, which
  means handlers do not need to give specialised forwarding clauses for effects they do not
  handle.
  For example, \ensuremath{\mathtt{choose}} is forwarded by \ensuremath{\Varid{h_{\mathsf{get}}}} by simply leaving the effect with
  the handled computation.
\indentbegin \begin{hscode}\SaveRestoreHook
\column{B}{@{}>{\hspre}l<{\hspost}@{}}%
\column{3}{@{}>{\hspre}c<{\hspost}@{}}%
\column{3E}{@{}l@{}}%
\column{7}{@{}>{\hspre}l<{\hspost}@{}}%
\column{E}{@{}>{\hspre}l<{\hspost}@{}}%
\>[7]{}{\mathbf{with}}\;\Varid{h_{\mathsf{get}}}\;{\mathbf{handle}}\;{\mathbf{op}}\;\mathtt{choose}\;()\;(\Varid{b}\, .\,{\mathbf{return}}\;\Varid{b}){}\<[E]%
\\
\>[3]{}\leadsto{}\<[3E]%
\>[7]{}{\mathbf{op}}\;\mathtt{choose}\;()\;(\Varid{b}\, .\,{\mathbf{with}}\;\Varid{h_{\mathsf{get}}}\;{\mathbf{handle}}\;{\mathbf{return}}\;\Varid{b}){}\<[E]%
\ColumnHook
\end{hscode}\resethooks
\indentend %
In the calculus we will present we cannot forward generically, and instead
utilize explicit forwarding clauses.
This is not novel: it already shows up in the algebraic case when
generalizing handler clauses when they may return types that are not
computations.
Since the return type of handlers corresponds to the carrier of the fold, when
generalizing handler clauses like this we get a fold with a carrier type
that is not a computation.
In this case, we can no longer uniformly forward the unhandled operations.
Instead, we must specify, per carrier, how to forward unhandled operations by
means of an explicit forwarding clause \cite{DBLP:conf/haskell/SchrijversPWJ19}.
Requirement of explicit forwarding algebras for algebraic effects and
handlers often happens in libraries, but not in languages,
because languages with primitive support for algebraic effects
and handlers always require the carriers of handlers to be
computations.

As stated, in the case of scoped effects, we are required to specify such an
explicit forwarding clause even when the carrier is a computation.
For now, we would like to observe that this forwarding clause is \emph{not} a
consequence of implementing scoped effects, but merely a consequence of moving
away from a (very specific) instance of algebraic effects.

\subsubsection{Effect Interaction}
One of the valuable features of the modular composition of algebraic effects \&
handlers is that effects can %
interact differently by applying their handlers in a different order.
Consider the effect \ensuremath{\mathtt{inc}\mathbin{:}\mathsf{()}\rightarrowtriangle\mathsf{Int}}, which produces an (incremented) integer.
\label{sec:hinc-intro}
The handler \ensuremath{\Varid{h_{\mathsf{inc}}}} turns computations into state-passing functions.\indentbegin \begin{hscode}\SaveRestoreHook
\column{B}{@{}>{\hspre}l<{\hspost}@{}}%
\column{3}{@{}>{\hspre}l<{\hspost}@{}}%
\column{9}{@{}>{\hspre}c<{\hspost}@{}}%
\column{9E}{@{}l@{}}%
\column{16}{@{}>{\hspre}l<{\hspost}@{}}%
\column{25}{@{}>{\hspre}c<{\hspost}@{}}%
\column{25E}{@{}l@{}}%
\column{28}{@{}>{\hspre}l<{\hspost}@{}}%
\column{42}{@{}>{\hspre}c<{\hspost}@{}}%
\column{42E}{@{}l@{}}%
\column{47}{@{}>{\hspre}l<{\hspost}@{}}%
\column{62}{@{}>{\hspre}l<{\hspost}@{}}%
\column{70}{@{}>{\hspre}c<{\hspost}@{}}%
\column{70E}{@{}l@{}}%
\column{E}{@{}>{\hspre}l<{\hspost}@{}}%
\>[3]{}\Varid{h_{\mathsf{inc}}}{}\<[9]%
\>[9]{}\mathrel{=}{}\<[9E]%
\>[16]{}{\mathbf{handler}}\;{}\<[25]%
\>[25]{}\{\mskip1.5mu {}\<[25E]%
\>[28]{}{\mathbf{return}}\;\Varid{x}{}\<[42]%
\>[42]{}\mapsto{}\<[42E]%
\>[47]{}{\mathbf{return}}\;(\boldsymbol{\lambda}\Varid{s}\, .\,{}\<[62]%
\>[62]{}{\mathbf{return}}\;(\Varid{x},\Varid{s})){}\<[E]%
\\
\>[25]{},{}\<[25E]%
\>[28]{}{\mathbf{op}}\;\mathtt{inc}\;\anonymous \;\Varid{k}{}\<[42]%
\>[42]{}\mapsto{}\<[42E]%
\>[47]{}{\mathbf{return}}\;(\boldsymbol{\lambda}\Varid{s}\, .\,{}\<[62]%
\>[62]{}\mathbf{do}\;\Varid{s'}\leftarrow \Varid{s}\mathbin{+}\mathrm{1}\, ;{}\<[E]%
\\
\>[62]{}\mathbf{do}\;\Varid{k'}\leftarrow \Varid{k}\;\Varid{s'}\, ;{}\<[E]%
\\
\>[62]{}\Varid{k'}\;\Varid{s'}){}\<[70]%
\>[70]{}\mskip1.5mu\}{}\<[70E]%
\ColumnHook
\end{hscode}\resethooks
\indentend %
The state \ensuremath{\Varid{s}} represents the current counter value. On every occurrence of
\ensuremath{\mathtt{inc}}, the incremented value is passed to the continuation twice: (1) for
updating the counter value and (2) for returning the result of the operation.
The latter is for the continuation and the former for serving the next \ensuremath{\mathtt{inc}}
operation. We use the folowing syntactic
sugar:
\indentbegin \begin{hscode}\SaveRestoreHook
\column{B}{@{}>{\hspre}l<{\hspost}@{}}%
\column{3}{@{}>{\hspre}l<{\hspost}@{}}%
\column{E}{@{}>{\hspre}l<{\hspost}@{}}%
\>[3]{}\Varid{run_{\mathsf{inc}}}\;\Varid{s}\;\Varid{c}\;\equiv\;\mathbf{do}\;\Varid{c'}\leftarrow {\mathbf{with}}\;\Varid{h_{\mathsf{inc}}}\;{\mathbf{handle}}\;\Varid{c}\, ;\Varid{c'}\;\Varid{s}{}\<[E]%
\ColumnHook
\end{hscode}\resethooks
\indentend %
\label{sec:example-cInc}
The computation \ensuremath{\Varid{c_{\mathsf{inc}}}} combines \ensuremath{\mathtt{choose}} and \ensuremath{\mathtt{inc}}:
\indentbegin \begin{hscode}\SaveRestoreHook
\column{B}{@{}>{\hspre}l<{\hspost}@{}}%
\column{3}{@{}>{\hspre}l<{\hspost}@{}}%
\column{9}{@{}>{\hspre}l<{\hspost}@{}}%
\column{35}{@{}>{\hspre}l<{\hspost}@{}}%
\column{41}{@{}>{\hspre}l<{\hspost}@{}}%
\column{69}{@{}>{\hspre}l<{\hspost}@{}}%
\column{E}{@{}>{\hspre}l<{\hspost}@{}}%
\>[3]{}\Varid{c_{\mathsf{inc}}}{}\<[9]%
\>[9]{}\mathrel{=}{\mathbf{op}}\;\mathtt{choose}\;()\;(\Varid{b}\, .\,\mathbf{if}\;\Varid{b}\;{}\<[35]%
\>[35]{}\mathbf{then}\;{}\<[41]%
\>[41]{}{\mathbf{op}}\;\mathtt{inc}\;()\;(\Varid{x}\, .\,\Varid{x}\mathbin{+}\mathrm{5})\;\mathbf{else}\;{}\<[69]%
\>[69]{}{\mathbf{op}}\;\mathtt{inc}\;()\;(\Varid{y}\, .\,\Varid{y}\mathbin{+}\mathrm{2})){}\<[E]%
\ColumnHook
\end{hscode}\resethooks
\indentend %
When handling \ensuremath{\mathtt{inc}} first, each \ensuremath{\mathtt{choose}} branch gets the same initial counter
value \ensuremath{\mathrm{0}}.
\label{eq:hdntoruninc}

\indentbegin \begin{hscode}\SaveRestoreHook
\column{B}{@{}>{\hspre}l<{\hspost}@{}}%
\column{3}{@{}>{\hspre}c<{\hspost}@{}}%
\column{3E}{@{}l@{}}%
\column{8}{@{}>{\hspre}l<{\hspost}@{}}%
\column{10}{@{}>{\hspre}l<{\hspost}@{}}%
\column{20}{@{}>{\hspre}l<{\hspost}@{}}%
\column{26}{@{}>{\hspre}l<{\hspost}@{}}%
\column{E}{@{}>{\hspre}l<{\hspost}@{}}%
\>[8]{}{\mathbf{with}}\;\Varid{h_{\mathsf{ND}}}\;{\mathbf{handle}}\;\Varid{run_{\mathsf{inc}}}\;\mathrm{0}\;\Varid{c_{\mathsf{inc}}}{}\<[E]%
\\
\>[3]{}\leadsto^{\ast}{}\<[3E]%
\>[8]{}{\mathbf{with}}\;\Varid{h_{\mathsf{ND}}}\;{\mathbf{handle}}\;{\mathbf{op}}\;\mathtt{choose}\;()\;(\Varid{b}\, .\,{}\<[E]%
\\
\>[8]{}\hsindent{2}{}\<[10]%
\>[10]{}\mathbf{do}\;\Varid{p'}\leftarrow {}\<[20]%
\>[20]{}\mathbf{if}\;\Varid{b}\;{}\<[26]%
\>[26]{}\mathbf{then}\;{\mathbf{with}}\;\Varid{h_{\mathsf{inc}}}\;{\mathbf{handle}}\;({\mathbf{op}}\;\mathtt{inc}\;()\;(\Varid{x}\, .\,\Varid{x}\mathbin{+}\mathrm{5})){}\<[E]%
\\
\>[26]{}\mathbf{else}\;{\mathbf{with}}\;\Varid{h_{\mathsf{inc}}}\;{\mathbf{handle}}\;({\mathbf{op}}\;\mathtt{inc}\;()\;(\Varid{y}\, .\,\Varid{y}\mathbin{+}\mathrm{2}))\, ;{}\<[E]%
\\
\>[8]{}\hsindent{2}{}\<[10]%
\>[10]{}\Varid{p'}\;\mathrm{0}){}\<[E]%
\\
\>[3]{}\leadsto^{\ast}{}\<[3E]%
\>[8]{}[\mskip1.5mu (\mathrm{6},\mathrm{1}),(\mathrm{3},\mathrm{1})\mskip1.5mu]{}\<[E]%
\ColumnHook
\end{hscode}\resethooks
\indentend %
In contrast, when handling \ensuremath{\mathtt{choose}} first, the counter value is threaded through
the successive branches, showing that with algebraic effects \& handlers, the
manner in which effects interact can be controlled by the order in which
handlers are applied.
\indentbegin \begin{hscode}\SaveRestoreHook
\column{B}{@{}>{\hspre}l<{\hspost}@{}}%
\column{3}{@{}>{\hspre}c<{\hspost}@{}}%
\column{3E}{@{}l@{}}%
\column{8}{@{}>{\hspre}l<{\hspost}@{}}%
\column{20}{@{}>{\hspre}l<{\hspost}@{}}%
\column{24}{@{}>{\hspre}l<{\hspost}@{}}%
\column{E}{@{}>{\hspre}l<{\hspost}@{}}%
\>[8]{}\Varid{run_{\mathsf{inc}}}\;\mathrm{0}\;({\mathbf{with}}\;\Varid{h_{\mathsf{ND}}}\;{\mathbf{handle}}\;\Varid{c_{\mathsf{inc}}}){}\<[E]%
\\
\>[3]{}\leadsto^{\ast}{}\<[3E]%
\>[8]{}\Varid{run_{\mathsf{inc}}}\;\mathrm{0}\;({}\<[20]%
\>[20]{}\mathbf{do}\;{}\<[24]%
\>[24]{}\Varid{xs}\leftarrow {\mathbf{with}}\;\Varid{h_{\mathsf{ND}}}\;{\mathbf{handle}}\;{\mathbf{op}}\;\mathtt{inc}\;()\;(\Varid{x}\, .\,\Varid{x}\mathbin{+}\mathrm{5})\, ;{}\<[E]%
\\
\>[20]{}\mathbf{do}\;{}\<[24]%
\>[24]{}\Varid{ys}\leftarrow {\mathbf{with}}\;\Varid{h_{\mathsf{ND}}}\;{\mathbf{handle}}\;{\mathbf{op}}\;\mathtt{inc}\;()\;(\Varid{y}\, .\,\Varid{y}\mathbin{+}\mathrm{2})\, ;{}\<[E]%
\\
\>[20]{}{\mathbf{return}}\;\Varid{xs}+\hspace{-0.4em}+\Varid{ys}){}\<[E]%
\\
\>[3]{}\leadsto^{\ast}{}\<[3E]%
\>[8]{}([\mskip1.5mu \mathrm{6},\mathrm{4}\mskip1.5mu],\mathrm{2}){}\<[E]%
\ColumnHook
\end{hscode}\resethooks
\indentend 

\subsection{Scoped effects}
\label{sec:scoped-eff}
Not all effects are algebraic.
There is a wide class of higher-order effects \cite{VANDENBERG2024103086} which do not satisfy
the algebraicity property including scoped effects \cite{wu14}, latent effects
\cite{DBLP:conf/aplas/BergSPW21}, and parallel effects \cite{parallel-effects}.
In this paper, we focus on scoped effects.
Consider extending nondeterminism with an effect that takes a computation that
contains nondeterministic choice and returns only its first result.
This is known as the \ensuremath{\mathtt{once}} operation~\cite{pirog18}.
Since effect operations syntactically already receive the continuation as and
argument, which looks like a computation, we might be tempted to use it for this
purpose, i.e.\ to use it to denote the scope of \ensuremath{\mathtt{once}}.
Subscripting \xmark{} to indicate an erroneous example, we could attempt to
syntactically write this as the algebraic operation \ensuremath{\mathtt{once}_{\text{\xmark}}\mathbin{:}()\rightarrowtriangle()} by
extending \ensuremath{\Varid{h_{\mathsf{ND}}}} with a clause for \ensuremath{\mathtt{once}_{\text{\xmark}}} that executes the continuation,
and then takes the head of the result if possible. Otherwise it just
returns the empty list.
We then try to prune the first \ensuremath{\mathtt{choose}} of \ensuremath{\Varid{c_{\mathsf{ND2}}}} by using \ensuremath{\mathtt{once}}.

\indentbegin \begin{hscode}\SaveRestoreHook
\column{B}{@{}>{\hspre}l<{\hspost}@{}}%
\column{3}{@{}>{\hspre}l<{\hspost}@{}}%
\column{53}{@{}>{\hspre}l<{\hspost}@{}}%
\column{E}{@{}>{\hspre}l<{\hspost}@{}}%
\>[3]{}\Varid{h_{\mathsf{once}}}_{\text{\xmark}}\mathrel{=}{\mathbf{handler}}\;\{\mskip1.5mu \ldots,{\mathbf{op}}\;\mathtt{once}_{\text{\xmark}}\;\anonymous \;\Varid{k}\mapsto{}\<[53]%
\>[53]{}\mathbf{do}\;\Varid{ts}\leftarrow \Varid{k}\;()\, ;{}\<[E]%
\\
\>[53]{}\mathbf{do}\;\Varid{b}\leftarrow \Varid{ts}=[\mskip1.5mu \mskip1.5mu]\, ;{}\<[E]%
\\
\>[53]{}\mathbf{if}\;\Varid{b}\;\mathbf{then}\;{\mathbf{return}}\;[\mskip1.5mu \mskip1.5mu]\;\mathbf{else}\;\mathsf{head}\;\Varid{ts}\mskip1.5mu\}{}\<[E]%
\ColumnHook
\end{hscode}\resethooks
\indentend \indentbegin \begin{hscode}\SaveRestoreHook
\column{B}{@{}>{\hspre}l<{\hspost}@{}}%
\column{3}{@{}>{\hspre}l<{\hspost}@{}}%
\column{17}{@{}>{\hspre}l<{\hspost}@{}}%
\column{21}{@{}>{\hspre}l<{\hspost}@{}}%
\column{40}{@{}>{\hspre}l<{\hspost}@{}}%
\column{44}{@{}>{\hspre}l<{\hspost}@{}}%
\column{E}{@{}>{\hspre}l<{\hspost}@{}}%
\>[3]{}\Varid{c_{\mathsf{once}}}_{\text{\xmark}}\mathrel{=}{}\<[17]%
\>[17]{}\mathbf{do}\;{}\<[21]%
\>[21]{}\Varid{p}\leftarrow {\mathbf{op}}\;\mathtt{once}_{\text{\xmark}}\;{}\<[40]%
\>[40]{}()\;{}\<[44]%
\>[44]{}(\anonymous \, .\,{\mathbf{op}}\;\mathtt{choose}\;()\;(\Varid{b}\, .\,{\mathbf{return}}\;\Varid{b}))\, ;{}\<[E]%
\\
\>[17]{}{\mathbf{op}}\;\mathtt{choose}\;()\;(\Varid{q}\, .\,{\mathbf{return}}\;(\Varid{p},\Varid{q})){}\<[E]%
\ColumnHook
\end{hscode}\resethooks
\indentend We intend for \ensuremath{{\mathbf{with}}\;\Varid{h_{\mathsf{once}}}_{\text{\xmark}}\;{\mathbf{handle}}\;\Varid{c_{\mathsf{once}}}_{\text{\xmark}}} to return \ensuremath{[\mskip1.5mu (\mathsf{true},\mathsf{true}),(\mathsf{true},\mathsf{false})\mskip1.5mu]} as the first \ensuremath{\mathtt{choose}} is pruned by \ensuremath{\mathtt{once}_{\text{\xmark}}} to only return
the first alternative.
The second \ensuremath{\mathtt{choose}} is out of scope of \ensuremath{\mathtt{once}_{\text{\xmark}}}, so should still return both
results.
However, algebraicity pulls the second \ensuremath{\mathtt{choose}} inside the scope of \ensuremath{\mathtt{once}_{\text{\xmark}}}:\indentbegin \begin{hscode}\SaveRestoreHook
\column{B}{@{}>{\hspre}l<{\hspost}@{}}%
\column{9}{@{}>{\hspre}c<{\hspost}@{}}%
\column{9E}{@{}l@{}}%
\column{14}{@{}>{\hspre}l<{\hspost}@{}}%
\column{40}{@{}>{\hspre}l<{\hspost}@{}}%
\column{44}{@{}>{\hspre}l<{\hspost}@{}}%
\column{55}{@{}>{\hspre}l<{\hspost}@{}}%
\column{63}{@{}>{\hspre}l<{\hspost}@{}}%
\column{67}{@{}>{\hspre}l<{\hspost}@{}}%
\column{E}{@{}>{\hspre}l<{\hspost}@{}}%
\>[14]{}{\mathbf{with}}\;\Varid{h_{\mathsf{once}}}_{\text{\xmark}}\;{\mathbf{handle}}\;({}\<[40]%
\>[40]{}\mathbf{do}\;{}\<[44]%
\>[44]{}\Varid{p}\leftarrow {\mathbf{op}}\;\mathtt{once}_{\text{\xmark}}\;{}\<[63]%
\>[63]{}()\;{}\<[67]%
\>[67]{}(\anonymous \, .\,{\mathbf{op}}\;\mathtt{choose}\;()\;(\Varid{b}\, .\,{\mathbf{return}}\;\Varid{b}))\, ;{}\<[E]%
\\
\>[40]{}{\mathbf{op}}\;\mathtt{choose}\;()\;(\Varid{q}\, .\,{\mathbf{return}}\;(\Varid{p},\Varid{q}))){}\<[E]%
\\
\>[9]{}\makebox[\widthof{\ensuremath{\leadsto^{\ast}}}][l]{\ensuremath{\leadsto}}{}\<[9E]%
\>[14]{}{\mathbf{with}}\;\Varid{h_{\mathsf{once}}}_{\text{\xmark}}\;{\mathbf{handle}}\;\mathtt{once}_{\text{\xmark}}\;()\;(\anonymous \, .\,{}\<[55]%
\>[55]{}\mathbf{do}\;{\mathbf{op}}\;\mathtt{choose}\;()\;(\Varid{b}\, .\,{\mathbf{return}}\;\Varid{b})\, ;{}\<[E]%
\\
\>[55]{}{\mathbf{op}}\;\mathtt{choose}\;()\;(\Varid{q}\, .\,{\mathbf{return}}\;(\Varid{p},\Varid{q}))){}\<[E]%
\\
\>[9]{}\leadsto^{\ast}{}\<[9E]%
\>[14]{}[\mskip1.5mu (\mathsf{true},\mathsf{true})\mskip1.5mu]{}\<[E]%
\ColumnHook
\end{hscode}\resethooks
\indentend %
There are many more examples of operations that have a scope; some of which we
present in \Cref{sec:examples}:
\begin{itemize}
\item \ensuremath{\mathtt{catch}} for catching exceptions that are raised during program execution;
\item \ensuremath{\mathtt{local}} for creating local variables (local state);
\item \ensuremath{\mathtt{call}} for creating a scope in a nondeterministic program, where branches
  can be cut using the algebraic \ensuremath{\mathtt{cut}} operation;
\item \ensuremath{\mathtt{depth}} for bounding the depth in the depth-bounded search strategy;
\end{itemize}
Following Wu et al. \cite{wu14}, we call them \emph{scoped operations}.
Plotkin and Power \cite{Plotkin03} have already realised that
algebraic effects are unable to represent so-called effect deconstructors (e.g.\
scoped effects) and propose
to model them as handlers.
Although this solution is used (for example by Thompson et al.\
\cite{DBLP:journals/pacmpl/ThomsonRWS22}), it merges
syntax and semantics at the cost of modularity
\cite{wu14,DBLP:conf/esop/YangPWBS22}. In \Cref{sec:examples} we revisit this
issue, showing attempts at encoding scoped effects as handlers, their problems,
and how \ensuremath{\lambda_{\mathit{sc}}} remedies the situation.

The goal of this work
is to implement scoped effects while maintaining a separation between
syntax and semantics, and thus preserve modular composition and control over
effect interaction.
It follows a line of research~\cite{pirog18,wu14,DBLP:conf/esop/YangPWBS22} that
has developed denotational semantic domains, backed by categorical
models.
However, they did not consider the modular composition of scoped
effects and handlers in their denotional semantics.
What is lacking from the literature is a calculus that allows
programming with both algebraic and scoped operations and their
handlers as well as supports the modular composition between them.

\section{Design Decisions \& Challenges}
\label{sec:design-decisions}
This section informally discusses the design of \ensuremath{\lambda_{\mathit{sc}}}, a novel calculus with
support for scoped effects \& handlers as built-in features.
We present the main challenges and design choices, and address how to forward
unknown operations.

\subsection{Scoped Effects as Built-in Operations}
\ensuremath{\lambda_{\mathit{sc}}} supports scoped effects as built-in operations, signalled by the \ensuremath{{\mathbf{sc}}}
keyword.
\indentbegin \begin{hscode}\SaveRestoreHook
\column{B}{@{}>{\hspre}l<{\hspost}@{}}%
\column{3}{@{}>{\hspre}l<{\hspost}@{}}%
\column{E}{@{}>{\hspre}l<{\hspost}@{}}%
\>[3]{}{\mathbf{sc}}\;\mathtt{once}\;()\;(\Varid{y}\, .\,\Varid{c}_{1})\;(\Varid{z}\, .\,\Varid{c}_{2}){}\<[E]%
\ColumnHook
\end{hscode}\resethooks
\indentend %
Like algebraic operations, scoped operations feature a label \ensuremath{\mathtt{once}} to identify
the effect, a parameter, in this case \ensuremath{()}, and a continuation (\ensuremath{\Varid{z}\, .\,\Varid{c}_{2}}).
However, unlike algebraic operations, scoped operations feature an additional
\emph{scoped computation} (\ensuremath{\Varid{y}\, .\,\Varid{c}_{1}}).
For \ensuremath{\mathtt{once}}, the scoped computation entails the computation to be
restricted to the first result (i.e. the computation in scope).
The result from the scoped computation \ensuremath{\Varid{c}_{1}} is passed to the continuation \ensuremath{\Varid{c}_{2}}
by means of argument \ensuremath{\Varid{z}}.

Scoped handler clauses feature an additional argument when
compared to algebraic handler clauses, corresponding to the scoped computation.
For example, below we define a handler for \ensuremath{\mathtt{once}}, which extends \ensuremath{\Varid{h_{\mathsf{ND}}}}.
\indentbegin \begin{hscode}\SaveRestoreHook
\column{B}{@{}>{\hspre}l<{\hspost}@{}}%
\column{3}{@{}>{\hspre}l<{\hspost}@{}}%
\column{12}{@{}>{\hspre}l<{\hspost}@{}}%
\column{21}{@{}>{\hspre}c<{\hspost}@{}}%
\column{21E}{@{}l@{}}%
\column{24}{@{}>{\hspre}l<{\hspost}@{}}%
\column{31}{@{}>{\hspre}l<{\hspost}@{}}%
\column{50}{@{}>{\hspre}l<{\hspost}@{}}%
\column{E}{@{}>{\hspre}l<{\hspost}@{}}%
\>[3]{}\Varid{h_{\mathsf{once}}}\mathrel{=}{}\<[12]%
\>[12]{}{\mathbf{handler}}\;{}\<[21]%
\>[21]{}\{\mskip1.5mu {}\<[21E]%
\>[24]{}\ldots,{}\<[31]%
\>[31]{}{\mathbf{sc}}\;\mathtt{once}\;\anonymous \;\Varid{p}\;\Varid{k}\mapsto{}\<[50]%
\>[50]{}\mathbf{do}\;\Varid{ts}\leftarrow \Varid{p}\;()\, ;{}\<[E]%
\\
\>[50]{}\mathbf{do}\;\Varid{b}\leftarrow \Varid{ts}=[\mskip1.5mu \mskip1.5mu]\, ;{}\<[E]%
\\
\>[50]{}\mathbf{if}\;\Varid{b}\;\mathbf{then}\;{\mathbf{return}}\;[\mskip1.5mu \mskip1.5mu]\;\mathbf{else}\;\mathbf{do}\;\Varid{t}\leftarrow \mathsf{head}\;\Varid{ts}\, ;\Varid{k}\;\Varid{t}\mskip1.5mu\}{}\<[E]%
\ColumnHook
\end{hscode}\resethooks
\indentend %
Adding scoped operations gives rise to a variant of the algebraicity property,
which models the desired behavior of sequencing for scoped operations:
scoped operations commute with sequencing in the continuation, but leave the
scoped computation intact.
\label{sec:scopedalgebraicity}
\indentbegin \begin{hscode}\SaveRestoreHook
\column{B}{@{}>{\hspre}l<{\hspost}@{}}%
\column{3}{@{}>{\hspre}l<{\hspost}@{}}%
\column{E}{@{}>{\hspre}l<{\hspost}@{}}%
\>[3]{}\mathbf{do}\;\Varid{x}\leftarrow {\mathbf{sc}}\;\ell^{\textsf{sc}}\;\Varid{v}\;(\Varid{y}\, .\,\Varid{c}_{1})\;(\Varid{z}\, .\,\Varid{c}_{2})\, ;\Varid{c}_{3}\;\equiv\;{\mathbf{sc}}\;\ell^{\textsf{sc}}\;\Varid{v}\;(\Varid{y}\, .\,\Varid{c}_{1})\;(\Varid{z}\, .\,\mathbf{do}\;\Varid{x}\leftarrow \Varid{c}_{2}\, ;\Varid{c}_{3}){}\<[E]%
\ColumnHook
\end{hscode}\resethooks
\indentend %
Using \ensuremath{\mathtt{once}} as a scoped operation correctly restrict only the first \ensuremath{\mathtt{choose}},
as can be seen when applying \ensuremath{\Varid{h_{\mathsf{once}}}} to \ensuremath{\Varid{c_{\mathsf{once}}}}, which is like
\ensuremath{\Varid{c_{\mathsf{once}}}_{\text{\xmark}}} defined in \Cref{sec:scoped-eff} but with the scoped
effect \ensuremath{\mathtt{once}} instead of the algebraic effect \ensuremath{\mathtt{once}_{\text{\xmark}}}.
\label{eq:honcetoconce}\indentbegin \begin{hscode}\SaveRestoreHook
\column{B}{@{}>{\hspre}l<{\hspost}@{}}%
\column{3}{@{}>{\hspre}l<{\hspost}@{}}%
\column{12}{@{}>{\hspre}l<{\hspost}@{}}%
\column{24}{@{}>{\hspre}l<{\hspost}@{}}%
\column{29}{@{}>{\hspre}l<{\hspost}@{}}%
\column{E}{@{}>{\hspre}l<{\hspost}@{}}%
\>[3]{}\Varid{c_{\mathsf{once}}}\mathrel{=}{}\<[12]%
\>[12]{}{\mathbf{sc}}\;\mathtt{once}\;()\;{}\<[24]%
\>[24]{}(\anonymous \, .\,{}\<[29]%
\>[29]{}{\mathbf{op}}\;\mathtt{choose}\;()\;(\Varid{b}\, .\,{\mathbf{return}}\;\Varid{b}))\;{}\<[E]%
\\
\>[24]{}(\Varid{p}\, .\,{}\<[29]%
\>[29]{}\mathbf{do}\;\Varid{q}\leftarrow {\mathbf{op}}\;\mathtt{choose}\;()\;(\Varid{b}\, .\,{\mathbf{return}}\;\Varid{b})\, ;{\mathbf{return}}\;(\Varid{p},\Varid{q})){}\<[E]%
\ColumnHook
\end{hscode}\resethooks
\indentend %
\indentbegin \begin{hscode}\SaveRestoreHook
\column{B}{@{}>{\hspre}l<{\hspost}@{}}%
\column{E}{@{}>{\hspre}l<{\hspost}@{}}%
\>[B]{}{\mathbf{with}}\;\Varid{h_{\mathsf{once}}}\;{\mathbf{handle}}\;\Varid{c_{\mathsf{once}}}\leadsto^{\ast}[\mskip1.5mu (\mathsf{true},\mathsf{true}),(\mathsf{true},\mathsf{false})\mskip1.5mu]{}\<[E]%
\ColumnHook
\end{hscode}\resethooks
\indentend %
\paragraph{Signatures}
Scoped effects, like algebraic effects, have signatures \ensuremath{\Conid{A}\rightarrowtriangle\Conid{B}}.
Like algebraic effects, the left-hand side of this signature refers to the type
of the value passed to the operation (\ensuremath{\mathsf{()}} in the case of \ensuremath{\mathtt{once}}).
However, where the right-hand side of the signature in the case of algebraic
effects refers to the argument of the \emph{continuation}, in the case of scoped
effects it refers to the \emph{scoped computation}.
Whilst this may seem strange from this point of view, it makes sense from the
point of view of writing handler clauses: in both cases, a value of type \ensuremath{\Conid{A}} is
given, and a value of type \ensuremath{\Conid{B}} is produced by the handler clause.
The difference comes from what computation this produced value of type \ensuremath{\Conid{B}} is
passed to: algebraic effect handler clauses pass this value to the continuation,
whereas scoped effect handler clauses pass this value to the scoped computation
(whose result is passed to the continuation).

Therefore, the signature of \ensuremath{\mathtt{once}} is \ensuremath{()\rightarrowtriangle()} since the handler clause calls
the scoped computation with \ensuremath{()}.
The scoped result type (\Cref{sec:scoped-result-type}) varies depending on the
computation \ensuremath{\mathtt{once}} is applied to, and is therefore not a part of the signature.

\subsection{\ensuremath{\mathsf{Eff}} with Row Typing}
Our calculus is based on Eff~\cite{pretnar15,DBLP:conf/calco/BauerP13,bauer15},
an existing calculus for algebraic effects \& handlers.
However, instead of using the subtyping-based effect system of Eff, we
use a row-based effect system in the style of Koka~\cite{Leijen17} for
simplicity.
Therefore, computations have types of shape \ensuremath{\Conid{A}\mathbin{!}\langle\Conid{E}\rangle}, where \ensuremath{\Conid{A}} is the type
of the value returned by the computation, and \ensuremath{\Conid{E}} is a collection (row) of
effects that \emph{may} occur during its evaluation.
For example, \ensuremath{(\mathsf{Bool},\mathsf{Bool})\hspace{0.1em}!\hspace{0.1em}\langle\mathtt{choose},\mathtt{once}\rangle} is the type of \ensuremath{\Varid{c_{\mathsf{once}}}}.

Handler types, of shape \ensuremath{\Conid{A}\mathbin{!}\langle\Conid{E}\rangle\Rightarrow \Conid{B}\mathbin{!}\langle\Conid{F}\rangle}, reflect that handlers turn
one computation (of type \ensuremath{\Conid{A}\mathbin{!}\langle\Conid{E}\rangle}) into another (of type \ensuremath{\Conid{B}\mathbin{!}\langle\Conid{F}\rangle}).
For example, \ensuremath{\Varid{h_{\mathsf{once}}}} handles \ensuremath{\mathtt{choose}} and \ensuremath{\mathtt{once}}.
\indentbegin \begin{hscode}\SaveRestoreHook
\column{B}{@{}>{\hspre}l<{\hspost}@{}}%
\column{3}{@{}>{\hspre}l<{\hspost}@{}}%
\column{E}{@{}>{\hspre}l<{\hspost}@{}}%
\>[3]{}\Varid{h_{\mathsf{once}}}\mathbin{:}(\mathsf{Bool},\mathsf{Bool})\mathbin{!}\langle\mathtt{choose},\mathtt{once}\rangle\Rightarrow \mathsf{List}\;(\mathsf{Bool},\mathsf{Bool})\mathbin{!}\langle\rangle{}\<[E]%
\ColumnHook
\end{hscode}\resethooks
\indentend %
\subsection{Scoped result type}
\label{sec:scoped-result-type}
As we have seen before, the two computations taken by a scoped
operation \ensuremath{{\mathbf{sc}}\;\ell^{\textsf{sc}}\;\Varid{v}\;(\Varid{y}\, .\,\Varid{c}_{1})\;(\Varid{z}\, .\,\Varid{c}_{2})} are not just two irrelevant terms;
they are connected in the sense that the result of the scoped
computation \ensuremath{\Varid{y}\, .\,\Varid{c}_{1}} is passed to the continuation \ensuremath{\Varid{z}\, .\,\Varid{c}_{2}} and exactly
bound as \ensuremath{\Varid{z}} in \ensuremath{\Varid{c}_{2}}.
Therefore, they must agree on the type of this result, which we name the
\emph{scoped result type}.
For example, consider \ensuremath{\Varid{c_{\mathsf{once}}}} with overall type \ensuremath{(\mathsf{Bool},\mathsf{Bool})\hspace{0.1em}!\hspace{0.1em}\langle\mathtt{once}\, ;\mathtt{choose}\rangle}.
Its scoped result type is (a singular) \ensuremath{\mathsf{Bool}}: it is the type that is
\emph{produced} by the scoped computation, and \emph{consumed} by the
continuation.
\indentbegin \begin{hscode}\SaveRestoreHook
\column{B}{@{}>{\hspre}l<{\hspost}@{}}%
\column{3}{@{}>{\hspre}l<{\hspost}@{}}%
\column{E}{@{}>{\hspre}l<{\hspost}@{}}%
\>[3]{}{\mathbf{sc}}\;\mathtt{once}\;()\;(\underbrace{\anonymous \, .\,{\mathbf{op}}\;\mathtt{choose}\;()\;(\Varid{b}\, .\,{\mathbf{return}}\;\Varid{b})}_{{\scriptsize ()\to \colorbox{lightgray}{$\mathsf{Bool}$}\hspace{0.1em}!\hspace{0.1em}\langle\mathtt{once}\, ;\mathtt{choose}\rangle}})\;(\underbrace{\Varid{p}\, .\,\mathbf{do}\;\Varid{q}\leftarrow {\mathbf{op}}\ldots\vphantom{{\mathbf{op}}}\;{\mathbf{return}}\;(\Varid{p},\Varid{q})}_{{\scriptsize \colorbox{lightgray}{$\mathsf{Bool}$}\to (\mathsf{Bool},\mathsf{Bool})\hspace{0.1em}!\hspace{0.1em}\langle\mathtt{once}\, ;\mathtt{choose}\rangle}}){}\<[E]%
\ColumnHook
\end{hscode}\resethooks
\indentend %
Dealing with the presence of this scoped result type is the most
non-trivial part for designing a calculus with composable scoped
effects. There are two main complications. First, since the type does
not occur in the computation's overall type, \emph{polymorphic
handlers} are required to handle scoped effects. Second, the scoped
result type describes a dependency between the scoped computation and
continuation: if one changes its type, the other must match this. This
makes generic \emph{forwarding} impossible: it alters the type of the
scoped computation, but does not make up for it in the continuation.

  The rest of this section covers these two complications. While understanding
  these issues is important for fully understanding the semantics presented in
  \Cref{sec:operational-semantics}, it may be hard to fully comprehend the
  remainder section without having seen the type system presented in
  \Cref{sec:typing}.
  We recommend the reader to skim over parts of the remainder of this section
  that are unclear, and revisit them after having seen the type system.

\subsubsection{Polymorphic Handlers}
\label{sec:motivation-polymorphic-handers}
Applying a handler to a computation involves recursively
applying the handler to the computation's subcomputations as well.
In the case of algebraic effects, these subcomputations always have the same
type as the operation itself, as witnessed by the algebraicity property.
This means that %
calculi that only support algebraic effects \& handlers, such as Eff,
allow typing handlers monomorphically, without limitations in
expressivity.

However, typing scoped effect handlers monomorphically \emph{does} limit their
implementation freedom: it only allows scoped operations of which the scoped
result type matches the operation's overall type.
For example, consider the type \ensuremath{(\mathsf{Bool},\mathsf{Bool})\hspace{0.1em}!\hspace{0.1em}\langle\mathtt{choose}\, ;\mathtt{once}\rangle\Rightarrow \mathsf{List}\;(\mathsf{Bool},\mathsf{Bool})\hspace{0.1em}!\hspace{0.1em}\langle\rangle} we previously assigned to \ensuremath{\Varid{h_{\mathsf{once}}}}.
This monomorphic type requires the scoped result type to be \ensuremath{(\mathsf{Bool},\mathsf{Bool})} as
well, as it is the only type of computation monomorphic \ensuremath{\Varid{h_{\mathsf{once}}}} can handle.
This is not the case for \ensuremath{\Varid{c_{\mathsf{once}}}}: as established, its scoped result type is \ensuremath{\mathsf{Bool}}.
Therefore, scoped computations such as \ensuremath{\Varid{c_{\mathsf{once}}}}, cannot be handled by monomorphic
handlers.
The solution is to let handlers abstract over the value type of computations,
allowing for the handling of scoped operations with \emph{any} scoped result
type.
This way, \ensuremath{\Varid{h_{\mathsf{once}}}} can be typed as follows:
\indentbegin \begin{hscode}\SaveRestoreHook
\column{B}{@{}>{\hspre}l<{\hspost}@{}}%
\column{3}{@{}>{\hspre}l<{\hspost}@{}}%
\column{E}{@{}>{\hspre}l<{\hspost}@{}}%
\>[3]{}\Varid{h_{\mathsf{once}}}\mathbin{:}\forall\;\alpha\, .\,\alpha\hspace{0.1em}!\hspace{0.1em}\langle\mathtt{choose}\, ;\mathtt{once}\rangle\Rightarrow \mathsf{List}\;\alpha\hspace{0.1em}!\hspace{0.1em}\langle\rangle{}\<[E]%
\ColumnHook
\end{hscode}\resethooks
\indentend %
With this polymorphic typing in place, \ensuremath{{\mathbf{with}}\;\Varid{h_{\mathsf{once}}}\;{\mathbf{handle}}\;\Varid{c_{\mathsf{once}}}} can now be
evaluated by \emph{polymorphic recursion}.
To support this, \ensuremath{\lambda_{\mathit{sc}}} features \ensuremath{\mathbf{let}}-polymorphism, \ensuremath{F_{\omega}}-style
type operators, and requires all handlers for scoped effects to be polymorphic.

\subsubsection{Forwarding Unknown Operations}
\label{sec:motivation-forwarding}
In order to retain the modularity of composing different effects, as discussed
in \Cref{sec:algebraic-forwarding}, we allow for partial handlers (i.e.\
handlers that handle only a subset of the effects they are applied to) to obtain
dedicated handlers that interpret only their part of the syntax.
This requires that all remaining operations be \emph{forwarded} to other
handlers.
As established, this forwarding can, in the case for algebraic effects and in
some specific situations, be done in a canonical way.

One might hope to forward scoped effects in a similar way, where the inner
computations are handled, and the effect is left as-is for future handlers to
handle.
For example, consider the following computation which uses a scoped
operation \ensuremath{\mathtt{catch}\mathbin{:}\mathsf{String}\rightarrowtriangle\mathsf{Bool}} for catching exceptions.
\indentbegin \begin{hscode}\SaveRestoreHook
\column{B}{@{}>{\hspre}l<{\hspost}@{}}%
\column{3}{@{}>{\hspre}l<{\hspost}@{}}%
\column{15}{@{}>{\hspre}l<{\hspost}@{}}%
\column{E}{@{}>{\hspre}l<{\hspost}@{}}%
\>[3]{}\Varid{c_{\mathsf{catch}}}\mathrel{=}{\mathbf{sc}}\;\mathtt{catch}\;\text{\ttfamily \char34 err\char34}\;(\Varid{b}\, .\,\mathbf{if}\;\Varid{b}\;\mathbf{then}\;{\mathbf{return}}\;\mathrm{1}\;\mathbf{else}\;{\mathbf{return}}\;\mathrm{2})\;(\Varid{x}\, .\,{\mathbf{return}}\;\Varid{x})\;{}\<[E]%
\\
\>[3]{}\hsindent{12}{}\<[15]%
\>[15]{}{\mathbf{with}}\;\Varid{h_{\mathsf{once}}}\;{\mathbf{handle}}\;\Varid{c_{\mathsf{catch}}}{}\<[E]%
\ColumnHook
\end{hscode}\resethooks
\indentend %
Applying the handler \ensuremath{\Varid{h_{\mathsf{once}}}} to it gives us\indentbegin \begin{hscode}\SaveRestoreHook
\column{B}{@{}>{\hspre}l<{\hspost}@{}}%
\column{3}{@{}>{\hspre}l<{\hspost}@{}}%
\column{15}{@{}>{\hspre}l<{\hspost}@{}}%
\column{31}{@{}>{\hspre}l<{\hspost}@{}}%
\column{E}{@{}>{\hspre}l<{\hspost}@{}}%
\>[3]{}\leadsto_{\text{\xmark}}\;{}\<[15]%
\>[15]{}{\mathbf{sc}}\;\mathtt{catch}\;\text{\ttfamily \char34 err\char34}\;{}\<[31]%
\>[31]{}(\Varid{b}\, .\,{\mathbf{with}}\;\Varid{h_{\mathsf{once}}}\;{\mathbf{handle}}\;\mathbf{if}\;\Varid{b}\ldots)\;(\Varid{x}\, .\,{\mathbf{with}}\;\Varid{h_{\mathsf{once}}}\;{\mathbf{handle}}\;{\mathbf{return}}\;\Varid{x}){}\<[E]%
\ColumnHook
\end{hscode}\resethooks
\indentend %
Unfortunately, this does not work: again, the hurdle is in the scoped result
type, in combination with the fact that applying a handler to a computation
changes its type.
In our example, \ensuremath{\Varid{h_{\mathsf{once}}}} applies the type operator \ensuremath{\mathsf{List}} when handling a
computation, resulting in a handled scoped computation of type \ensuremath{\mathsf{Bool}\to \mathsf{List}\;\mathsf{Int}\hspace{0.1em}!\hspace{0.1em}\langle\mathtt{catch}\rangle}, and a handled continuation of type \ensuremath{\mathsf{Int}\to \mathsf{List}\;\mathsf{Int}\hspace{0.1em}!\hspace{0.1em}\langle\mathtt{catch}\rangle}.
\indentbegin \begin{hscode}\SaveRestoreHook
\column{B}{@{}>{\hspre}l<{\hspost}@{}}%
\column{3}{@{}>{\hspre}l<{\hspost}@{}}%
\column{E}{@{}>{\hspre}l<{\hspost}@{}}%
\>[3]{}{\mathbf{sc}}\;\mathtt{catch}\;\text{\ttfamily \char34 err\char34}\;(\underbrace{\Varid{b}\, .\,{\mathbf{with}}\;\Varid{h_{\mathsf{once}}}\;{\mathbf{handle}}\;\mathbf{if}\;\Varid{b}\;\mathbf{then}\ldots}_{\mathsf{Bool}\to {\scriptsize \colorbox{lightgray}{$\mathsf{List}\;\mathsf{Int}$}}\hspace{0.1em}!\hspace{0.1em}\langle\mathtt{catch}\rangle})\;(\underbrace{\Varid{x}\, .\,{\mathbf{with}}\;\Varid{h_{\mathsf{once}}}\;{\mathbf{handle}}\;{\mathbf{return}}\;\Varid{x}}_{{\scriptsize \colorbox{lightgray}{$\mathsf{Int}$}}\to \mathsf{List}\;\mathsf{Int}\hspace{0.1em}!\hspace{0.1em}\langle\mathtt{catch}\rangle}){}\<[E]%
\ColumnHook
\end{hscode}\resethooks
\indentend %
This introduced a type mismatch, as the return type of the scoped computation
has changed (\ensuremath{\mathsf{List}\;\mathsf{Int}}), whereas the continuation still expects the original
type (\ensuremath{\mathsf{Int}}).
In other words: the scoped computation and continuation no longer agree on a
scoped result type.
Thus, scoped effects cannot be forwarded uniformly.
Therefore, we require that every handler is equipped with an explicit forwarding
clause for unknown scoped operations, specific to the handler's carrier (i.e.\
the type constructor is applied to a computation when handled).
This forwarding clause is used to resolve the type mismatch between
the result type of the handled scoped computation and the parameter
type of continuation.

\label{sec:bind-intro}
Considering the above example, the goal of applying a function of type
\ensuremath{\mathsf{Int}\to \mathsf{List}\;\mathsf{Int}\hspace{0.1em}!\hspace{0.1em}\langle\Conid{E}\rangle} to a value of type \ensuremath{\mathsf{List}\;\mathsf{Int}} naturally
reminds us of monadic binding.
Indeed, for most carriers, this forwarding clause is very similar to
that of the monadic bind \ensuremath{\bind }.
In fact, for now we will assume this to be the case for all effects, and
introduce \ensuremath{\lambda_{\mathit{sc}}} based on these bind-like forwarding clauses.
In \Cref{sec:genfwd} we will revisit this firstly by giving examples
where bind-like clauses are not expressive enough, secondly by generalizing the
forwarding clause, and lastly by showing how to desugar bind-like clauses to
generalized forwarding clauses.
Therefore, forwarding clauses are---for now---of shape \ensuremath{{\mathbf{bind}}\;\Varid{x}\;\Varid{k}}, where \ensuremath{\Varid{x}}
represents the result of the scoped computation, and \ensuremath{\Varid{k}} represents the
continuation, yielding a type similar to \ensuremath{\bind }.

What should we write for a forwarding clause?
The answer depends on what we expect when we want to connect the results of a
handler to outside.
For example, consider the handler \ensuremath{\Varid{h}\;\Conid{Once}}.
The \ensuremath{\Varid{x}} here is the list of results we get from the nondeterministic
program in the scope after handled by \ensuremath{\Varid{h_{\mathsf{once}}}}. The \ensuremath{\Varid{k}} here is what
remains to do for each result.
The most intuitive way to connect \ensuremath{\Varid{x}} and \ensuremath{\Varid{k}} is to run \ensuremath{\Varid{k}} on
every element in \ensuremath{\Varid{x}} and connect their results together.
This is exactly what the \ensuremath{\mathsf{concatMap}} function in Haskell does.
\indentbegin \begin{hscode}\SaveRestoreHook
\column{B}{@{}>{\hspre}l<{\hspost}@{}}%
\column{3}{@{}>{\hspre}l<{\hspost}@{}}%
\column{20}{@{}>{\hspre}c<{\hspost}@{}}%
\column{20E}{@{}l@{}}%
\column{23}{@{}>{\hspre}c<{\hspost}@{}}%
\column{23E}{@{}l@{}}%
\column{28}{@{}>{\hspre}c<{\hspost}@{}}%
\column{28E}{@{}l@{}}%
\column{31}{@{}>{\hspre}l<{\hspost}@{}}%
\column{E}{@{}>{\hspre}l<{\hspost}@{}}%
\>[3]{}\Varid{h_{\mathsf{once}}}\mathrel{=}{\mathbf{handler}}\;{}\<[20]%
\>[20]{}\{\mskip1.5mu {}\<[20E]%
\>[23]{}\ldots{}\<[23E]%
\>[28]{},{}\<[28E]%
\>[31]{}{\mathbf{bind}}\;\Varid{x}\;\Varid{k}\mapsto\mathsf{concatMap}\;\Varid{x}\;\Varid{k}\mskip1.5mu\}{}\<[E]%
\ColumnHook
\end{hscode}\resethooks
\indentend %
We define \ensuremath{\mathsf{concatMap}} in \ensuremath{\lambda_{\mathit{sc}}} as follows.
\indentbegin \begin{hscode}\SaveRestoreHook
\column{B}{@{}>{\hspre}l<{\hspost}@{}}%
\column{12}{@{}>{\hspre}l<{\hspost}@{}}%
\column{22}{@{}>{\hspre}l<{\hspost}@{}}%
\column{25}{@{}>{\hspre}c<{\hspost}@{}}%
\column{25E}{@{}l@{}}%
\column{28}{@{}>{\hspre}l<{\hspost}@{}}%
\column{32}{@{}>{\hspre}l<{\hspost}@{}}%
\column{E}{@{}>{\hspre}l<{\hspost}@{}}%
\>[B]{}\mathsf{concatMap}{}\<[12]%
\>[12]{}\mathbin{:}\forall\;\alpha\;\beta\;\mu\, .\,\mathsf{List}\;\beta\;{\rightarrow}^{\mu}\;(\beta\;{\rightarrow}^{\mu}\;\mathsf{List}\;\alpha)\;{\rightarrow}^{\mu}\;\mathsf{List}\;\alpha{}\<[E]%
\\
\>[B]{}\mathsf{concatMap}\;{}\<[12]%
\>[12]{}[\mskip1.5mu \mskip1.5mu]\;{}\<[22]%
\>[22]{}\Varid{f}{}\<[25]%
\>[25]{}\mathrel{=}{}\<[25E]%
\>[28]{}{\mathbf{return}}\;[\mskip1.5mu \mskip1.5mu]{}\<[E]%
\\
\>[B]{}\mathsf{concatMap}\;{}\<[12]%
\>[12]{}(\Varid{b}\mathbin{:}\Varid{bs})\;{}\<[22]%
\>[22]{}\Varid{f}{}\<[25]%
\>[25]{}\mathrel{=}{}\<[25E]%
\>[28]{}\mathbf{do}\;{}\<[32]%
\>[32]{}\Varid{as}\leftarrow \Varid{f}\;\Varid{b}\, ;\Varid{as'}\leftarrow \mathsf{concatMap}\;\Varid{bs}\;\Varid{f}\, ;\Varid{as}+\hspace{-0.4em}+\Varid{as'}{}\<[E]%
\ColumnHook
\end{hscode}\resethooks
\indentend %
Note that there is no unique correct forwarding clause for every handler; it is
up to the programmers to decide which kind of forwarding behaviours they want.
For example, for \ensuremath{\Varid{h_{\mathsf{once}}}} we can also use a reversed version of
\ensuremath{\mathsf{concatMap}} which does the same thing except for traversing the list
from right to left.
Alternatively, we can only pass the first element of \ensuremath{\Varid{x}} to \ensuremath{\Varid{k}}, or
even invoke other effects like printing some information when
forwarding.
Comparing these different styles of forwarding, we prefer the
original version using a left-to-right \ensuremath{\mathsf{concatMap}} since we usually do
not want other scoped effects to unexpectedly change the results of
\ensuremath{\Varid{h_{\mathsf{once}}}} by discarding some results or reversing the order of results.
From our experience of implementing a range of examples of scoped effects in
\Cref{sec:examples}, we found that a simple and intuitive forwarding clause
usually yields the expect semantics.

In what follows, we put our calculus on formal footing, discussing its syntax,
operational semantics and type-and-effect system.
\section{Term Syntax}
\label{sec:syntax}
\begin{figure}[t]
\begin{center}
\renewcommand\arraystretch{1}
\begin{tabular}[c]{rrlr}
  values \ensuremath{\Varid{v}} &\ensuremath{::=} &\ensuremath{()\ \mid\ (\Varid{v}_{1},\Varid{v}_{2})\ \mid\ \Varid{x}\ \mid\ \boldsymbol{\lambda}\Varid{x}\, .\,\Varid{c}\ \mid\ \Varid{h}}\\[\grammarspacing]

  handlers \ensuremath{\Varid{h}} &\ensuremath{::=} &\ensuremath{{\mathbf{handler}}} \:\ensuremath{\{\mskip1.5mu {\mathbf{return}}\;\Varid{x}\mapsto\Varid{c}_{\Varid{r}}} &return clause\\
      & & \phantom{\ensuremath{{\mathbf{handler}}}} \: \ensuremath{,\Varid{oprs}} &effect clauses\\
      & & \phantom{\ensuremath{{\mathbf{handler}}}} \: \ensuremath{,{\mathbf{bind}}\;\Varid{x}\;\Varid{k}\mapsto\Varid{c}_{\Varid{f}}\mskip1.5mu\}} &forwarding clause\\[\grammarspacing]

  operation clauses \ensuremath{\Varid{oprs}} &\ensuremath{::=} &\ensuremath{ \cdot} &empty \ensuremath{\Varid{oprs}}\\
        &$\mid$ &\ensuremath{{\mathbf{op}}\;\ell^{\textsf{op}}\;\Varid{x}\;\Varid{k}\mapsto\Varid{c},\Varid{oprs}}   &algebraic effect clauses\\
        &$\mid$ &\ensuremath{{\mathbf{sc}}\;\ell^{\textsf{sc}}\;\Varid{x}\;\Varid{p}\;\Varid{k}\mapsto\Varid{c},\Varid{oprs}} &scoped effect clauses\\[\grammarspacing]

  computations \ensuremath{\Varid{c}} &\ensuremath{::=} &\ensuremath{{\mathbf{return}}\;\Varid{v}} &return value\\
      &$\mid$ &\ensuremath{{\mathbf{op}}\;\ell^{\textsf{op}}\;\Varid{v}\;(\Varid{y}\, .\,\Varid{c})} &algebraic operation\\
      &$\mid$ &\ensuremath{{\mathbf{sc}}\;\ell^{\textsf{sc}}\;\Varid{v}\;(\Varid{y}\, .\,\Varid{c}_{1})\;(\Varid{z}\, .\,\Varid{c}_{2})} &scoped operation\\
      &$\mid$ &\ensuremath{{\mathbf{with}}\;\Varid{v}\;{\mathbf{handle}}\;\Varid{c}} &handle\\
      &$\mid$ &\ensuremath{\mathbf{do}\;\Varid{x}\leftarrow \Varid{c}_{1}\, ;\Varid{c}_{2}} &do-statement\\
      &$\mid$ &\ensuremath{\Varid{v}_{1}\;\Varid{v}_{2}} &application\\
      &$\mid$ &\ensuremath{\mathbf{let}\;\Varid{x}\mathrel{=}\Varid{v}\;\mathbf{in}\;\Varid{c}} &let \\[\grammarspacing]

\end{tabular}
\end{center}
\caption{Terms of \ensuremath{\lambda_{\mathit{sc}}}.}
\label{fig:term-syntax}
\end{figure}
As stated, \ensuremath{\lambda_{\mathit{sc}}} is based on Eff \cite{DBLP:conf/calco/BauerP13,bauer15}.
Before adding support for scoped effects, we have altered Eff from its
presentation in \cite{DBLP:conf/calco/BauerP13,bauer15} in two ways.
Firstly, we have made a number of cosmetic changes that arguably improve the
readability of the calculus.
Secondly, we adopt row-based typing in the style of Koka \cite{Leijen17}.
The row types will be introduced in \Cref{sec:typing}.

Like Eff we implement fine-grained call-by-value semantics
\cite{DBLP:journals/iandc/LevyPT03}.
Therefore, terms are split into inert values and computations that can be
reduced.

\subsection{Computations}
For computations, \ensuremath{{\mathbf{return}}} can be used to return values.
Handlers can be applied to values with the \ensuremath{{\mathbf{with}}} and \ensuremath{{\mathbf{handle}}} keywords.
As seen before, computations may be sequenced by means of do-statements (\ensuremath{\mathbf{do}\;\Varid{x}\leftarrow \Varid{c}_{1}\, ;\Varid{c}_{2}}).
Applications reduce, so they are computations.
As discussed in \Cref{sec:motivation-polymorphic-handers}, to support
polymorphic handlers we support let-polymorphism and thus let-bindings.
Finally, a computation may be the invocation of an effect by means of an
operation.

To be able to differentiate between algebraic and scoped effects, we add the
effect keyword \ensuremath{{\mathbf{sc}}} to model scoped effects.
Consequently, \ensuremath{{\mathbf{op}}} now ranges over algebraic effects only.
Furthermore, we annotate labels with either \ensuremath{{\textsf{op}}} or \ensuremath{{\textsf{sc}}} to indicate if they
are the label of an algebraic or scoped effect, respectively.
We implicitly assume any label \ensuremath{\ell} occurs either as an algebraic or scoped
effect label.
Like their algebraic counterparts, scoped effect operations feature a label
\ensuremath{\ell^{\textsf{sc}}}, argument \ensuremath{\Varid{v}} and continuation \ensuremath{(\Varid{z}\, .\,\Varid{c}_{2})}.
In addition, scoped effect operations feature a scoped computation \ensuremath{(\Varid{y}\, .\,\Varid{c}_{1})}.
This way, the scope of effect \ensuremath{\ell^{\textsf{sc}}} is delimited: (scoped) computation \ensuremath{(\Varid{y}\, .\,\Varid{c}_{1})}
is in scope, continuation \ensuremath{(\Varid{z}\, .\,\Varid{c}_{2})} is not.

\subsection{Values}
Values consist of the unit value \ensuremath{()}, value pairs \ensuremath{(\Varid{v}_{1},\Varid{v}_{2})}, variables \ensuremath{\Varid{x}},
functions \ensuremath{\boldsymbol{\lambda}\Varid{x}\, .\,\Varid{c}} and handlers \ensuremath{\Varid{h}}.
Handlers \ensuremath{\Varid{h}} have three kinds of clauses: one return clause, zero or more
operation clauses, and a forwarding clause.

The return clause \ensuremath{{\mathbf{return}}\;\Varid{x}\mapsto\Varid{c}} denotes that the result
\ensuremath{\Varid{x}} of a computation is processed by computation \ensuremath{\Varid{c}}.

Algebraic operation clauses \ensuremath{{\mathbf{op}}\;\ell^{\textsf{op}}\;\Varid{x}\;\Varid{k}\mapsto\Varid{c}} specify that handling an effect
with label \ensuremath{\ell^{\textsf{op}}}, parameter \ensuremath{\Varid{x}} and continuation \ensuremath{\Varid{k}} is processed by computation
\ensuremath{\Varid{c}} (e.g., \ensuremath{\Varid{h_{\mathsf{ND}}}}, \ensuremath{\Varid{h_{\mathsf{get}}}}, \ensuremath{\Varid{h_{\mathsf{inc}}}}).
For scoped effect clauses the extension is analogous to the operation case: we
take the algebraic clause and add support for a scoped computation, which in the
case for the clause has the form of parameter \ensuremath{\Varid{p}}.

  Finally, as motivated in \Cref{sec:motivation-forwarding}, we have forwarding
  clauses of shape \ensuremath{{\mathbf{bind}}\;\Varid{x}\;\Varid{k}\mapsto\Varid{c}} that deal with forwarding unknown scoped
  operations with some label \ensuremath{\ell^{\textsf{sc}}}. We will generalize these in
  \Cref{sec:genfwd}.
\color{black}
\section{Operational Semantics}
\label{sec:operational-semantics}
\Cref{fig:operational-semantics} displays the small-step operational semantics
of \ensuremath{\lambda_{\mathit{sc}}}.
Here, relation \ensuremath{\Varid{c}\leadsto\Varid{c'}} denotes that computation \ensuremath{\Varid{c}} steps to
computation \ensuremath{\Varid{c'}}, with \ensuremath{\leadsto^{\ast}} its reflexive, transitive closure.
The highlighted rules deal with the extensions that support scoped effects.
The following discussion of the semantics is exemplified by snippets of
derivations of computations\footnotemark{} used in
\Cref{sec:background-motivation}.
We refer to \Cref{app:opsem-derivations} for the full version of these
derivations.

\footnotetext{Following convention, these examples may contain elements not present in our
calculus, such as integers and if-then-else statements. These may be viewed as
syntactic sugar for their Church encodings.}

\begin{figure}[t]
\ruleform{\ensuremath{\Varid{c}\leadsto\Varid{c'}}} \quad\text{Computation reduction} \hspace*{\fill}

\begin{center}
\begin{mathpar}
  \inferrule*[right=E-AppAbs]
  { }
  { \ensuremath{(\boldsymbol{\lambda}\Varid{x}\, .\,\Varid{c})\;\Varid{v}\leadsto\Varid{c}\;[\mskip1.5mu \Varid{v}\mathbin{/}\Varid{x}\mskip1.5mu]} }

  \inferrule*[right=E-Let]
  { }
  { \ensuremath{\mathbf{let}\;\Varid{x}\mathrel{=}\Varid{v}\;\mathbf{in}\;\Varid{c}\leadsto\Varid{c}\;[\mskip1.5mu \Varid{v}\mathbin{/}\Varid{x}\mskip1.5mu]} }

  \inferrule*[right=E-Do]
  { \ensuremath{\Varid{c}_{1}\leadsto\Varid{c}_{1}'}
  }
  { \ensuremath{\mathbf{do}\;\Varid{x}\leftarrow \Varid{c}_{1}\, ;\Varid{c}_{2}\leadsto\mathbf{do}\;\Varid{x}\leftarrow \Varid{c}_{1}'\, ;\Varid{c}_{2}} }

  \inferrule*[right=E-DoRet]
  { }
  { \ensuremath{\mathbf{do}\;\Varid{x}\leftarrow {\mathbf{return}}\;\Varid{v}\, ;\Varid{c}_{2}\leadsto\Varid{c}_{2}\;[\mskip1.5mu \Varid{v}\mathbin{/}\Varid{x}\mskip1.5mu]} }

  \inferrule*[right=E-DoOp]
  { }
  { \ensuremath{\mathbf{do}\;\Varid{x}\leftarrow {\mathbf{op}}\;\ell^{\textsf{op}}\;\Varid{v}\;(\Varid{y}\, .\,\Varid{c}_{1})\, ;\Varid{c}_{2}\leadsto{\mathbf{op}}\;\ell^{\textsf{op}}\;\Varid{v}\;(\Varid{y}\, .\,\mathbf{do}\;\Varid{x}\leftarrow \Varid{c}_{1}\, ;\Varid{c}_{2})} }

  \inferrule*[right=E-DoSc]
  { }
  { \ensuremath{\mathbf{do}\;\Varid{x}\leftarrow {\mathbf{sc}}\;\ell^{\textsf{sc}}\;\Varid{v}\;(\Varid{y}\, .\,\Varid{c}_{1})\;(\Varid{z}\, .\,\Varid{c}_{2})\, ;\Varid{c}_{3}\leadsto{\mathbf{sc}}\;\ell^{\textsf{sc}}\;\Varid{v}\;(\Varid{y}\, .\,\Varid{c}_{1})\;(\Varid{z}\, .\,\mathbf{do}\;\Varid{x}\leftarrow \Varid{c}_{2}\, ;\Varid{c}_{3})} }

  \inferrule*[right=E-Hand]
  { \ensuremath{\Varid{c}\leadsto\Varid{c'}}
  }
  { \ensuremath{{\mathbf{with}}\;\Varid{h}\;{\mathbf{handle}}\;\Varid{c}\leadsto{\mathbf{with}}\;\Varid{h}\;{\mathbf{handle}}\;\Varid{c'}} }

  \inferrule*[right=E-HandRet]
  { \ensuremath{({\mathbf{return}}\;\Varid{x}\mapsto\Varid{c}_{\Varid{r}})\;\in\;\Varid{h}}
  }
  { \ensuremath{{\mathbf{with}}\;\Varid{h}\;{\mathbf{handle}}\;{\mathbf{return}}\;\Varid{v}\leadsto\Varid{c}_{\Varid{r}}\;[\mskip1.5mu \Varid{v}\mathbin{/}\Varid{x}\mskip1.5mu]} }

  \inferrule*[right=E-HandOp]
  { \ensuremath{({\mathbf{op}}\;\ell^{\textsf{op}}\;\Varid{x}\;\Varid{k}\mapsto\Varid{c})\;\in\;\Varid{h}}
  }
  { \ensuremath{{\mathbf{with}}\;\Varid{h}\;{\mathbf{handle}}\;{\mathbf{op}}\;\ell^{\textsf{op}}\;\Varid{v}\;(\Varid{y}\, .\,\Varid{c}_{1})\leadsto\Varid{c}\;[\mskip1.5mu \Varid{v}\mathbin{/}\Varid{x},(\boldsymbol{\lambda}\Varid{y}\, .\,{\mathbf{with}}\;\Varid{h}\;{\mathbf{handle}}\;\Varid{c}_{1})\mathbin{/}\Varid{k}\mskip1.5mu]} }

  \inferrule*[right=E-FwdOp]
  { \ensuremath{({\mathbf{op}}\;\ell^{\textsf{op}}\;\anonymous \;\anonymous )\;\notin\;\Varid{h}}
  }
  { \ensuremath{{\mathbf{with}}\;\Varid{h}\;{\mathbf{handle}}\;{\mathbf{op}}\;\ell^{\textsf{op}}\;\Varid{v}\;(\Varid{y}\, .\,\Varid{c}_{1})\leadsto{\mathbf{op}}\;\ell^{\textsf{op}}\;\Varid{v}\;(\Varid{y}\, .\,{\mathbf{with}}\;\Varid{h}\;{\mathbf{handle}}\;\Varid{c}_{1})} }

  \inferrule*[right=E-HandSc]
  { \ensuremath{({\mathbf{sc}}\;\ell^{\textsf{sc}}\;\Varid{x}\;\Varid{p}\;\Varid{k}\mapsto\Varid{c})\;\in\;\Varid{h}}
  }
  { \ensuremath{{\mathbf{with}}\;\Varid{h}\;{\mathbf{handle}}\;{\mathbf{sc}}\;\ell^{\textsf{sc}}\;\Varid{v}\;(\Varid{y}\, .\,\Varid{c}_{1})\;(\Varid{z}\, .\,\Varid{c}_{2})\leadsto} \ensuremath{\begin{array}{l} \ensuremath{\Varid{c}\;[\mskip1.5mu \Varid{v}\mathbin{/}\Varid{x},(\boldsymbol{\lambda}\Varid{y}\, .\,{\mathbf{with}}\;\Varid{h}\;{\mathbf{handle}}\;\Varid{c}_{1})\mathbin{/}\Varid{p},} \\
                                                             \qquad \ensuremath{(\boldsymbol{\lambda}\Varid{z}\, .\,{\mathbf{with}}\;\Varid{h}\;{\mathbf{handle}}\;\Varid{c}_{2})\mathbin{/}\Varid{k}\mskip1.5mu]} \end{array}} }

  \inferrule*[right=E-Bind]
  { \ensuremath{({\mathbf{sc}}\;\ell^{\textsf{sc}}\;\anonymous \;\anonymous \;\anonymous )\;\notin\;\Varid{h}}
    \\ \ensuremath{({\mathbf{bind}}\;\Varid{x}\;\Varid{k}\mapsto\Varid{c})\;\in\;\Varid{h}}
  }
  { \ensuremath{{\mathbf{with}}\;\Varid{h}\;{\mathbf{handle}}\;{\mathbf{sc}}\;\ell^{\textsf{sc}}\;\Varid{v}\;(\Varid{y}\, .\,\Varid{c}_{1})\;(\Varid{z}\, .\,\Varid{c}_{2})\leadsto} \ensuremath{\begin{array}{l} \ensuremath{{\mathbf{sc}}\;\ell^{\textsf{sc}}\;\Varid{v}\;(\Varid{y}\, .\,{\mathbf{with}}\;\Varid{h}\;{\mathbf{handle}}\;\Varid{c}_{1})\;(\Varid{z}\, .\,\Varid{c}\;[\mskip1.5mu \Varid{z}\mathbin{/}\Varid{x},} \\
                  \qquad  \ensuremath{(\boldsymbol{\lambda}\Varid{z}\, .\,{\mathbf{with}}\;\Varid{h}\;{\mathbf{handle}}\;\Varid{c}_{2})\mathbin{/}\Varid{k}\mskip1.5mu])} \end{array} } }
\end{mathpar}
\end{center}
\caption{Operational semantics of \ensuremath{\lambda_{\mathit{sc}}}.}
\label{fig:operational-semantics}
\end{figure}

Rules \textsc{E-AppAbs} and \textsc{E-Let} deal with function application and
let-binding, respectively, and are standard. The rest of the rules consist of
two parts: sequencing and handling.

\paragraph{Sequencing}
For sequencing computations \ensuremath{\mathbf{do}\;\Varid{x}\leftarrow \Varid{c}_{1}\, ;\Varid{c}_{2}}, we distinguish between the situation
where \ensuremath{\Varid{c}_{1}} can take a step (\textsc{E-Do}), and where \ensuremath{\Varid{c}_{1}} is in normal form (\ensuremath{{\mathbf{return}}}, \ensuremath{{\mathbf{op}}}, or \ensuremath{{\mathbf{sc}}}).
First, if \ensuremath{\Varid{c}_{1}} returns a value \ensuremath{\Varid{v}}, we substitute \ensuremath{\Varid{v}} for \ensuremath{\Varid{x}} in \ensuremath{\Varid{c}_{2}}
(\textsc{E-DoRet}).
Second, if \ensuremath{\Varid{c}_{1}} is an algebraic operation, we rewrite the computation using the
algebraicity property (\textsc{E-DoOp}), bubbling up the algebraic operation
to the front of the computation.
Third, the new case, where \ensuremath{\Varid{c}_{1}} is a scoped operation, is analogous: the
generalization of the algebraicity property (\Cref{sec:scopedalgebraicity}) is
used to rewrite the computation (\textsc{E-DoSc}).

\paragraph{Handling}
For handling computations with a handler of the form \ensuremath{{\mathbf{with}}\;\Varid{h}\;{\mathbf{handle}}\;\Varid{c}}, we
distinguish six situations.
First, if
possible, \ensuremath{\Varid{c}} takes a step (\textsc{E-Hand}); in the other cases, \ensuremath{\Varid{c}} is in normal form.
If \ensuremath{\Varid{c}} returns a value \ensuremath{\Varid{v}}, we use the handler's return clause \ensuremath{{\mathbf{return}}\;\Varid{x}\mapsto\Varid{c}_{\Varid{r}}},
switching evaluation to \ensuremath{\Varid{c}_{\Varid{r}}} with \ensuremath{\Varid{x}} replaced by \ensuremath{\Varid{v}} (\textsc{E-HandRet}).

If computation \ensuremath{\Varid{c}} is an algebraic operation \ensuremath{{\mathbf{op}}\;\ell^{\textsf{op}}\;\Varid{v}\;(\Varid{y}\, .\,\Varid{c}_{1})}, its label is
looked up in the handler \ensuremath{\Varid{h}}.
If the handler contains an algebraic clause with this label, evaluation switches
to the clause's computation \ensuremath{\Varid{c}} (\textsc{E-HandOp}), with \ensuremath{\Varid{v}} substituted for
parameter \ensuremath{\Varid{x}} and continuation \ensuremath{\Varid{k}} replaced by a function that, given the
original argument \ensuremath{\Varid{y}}, contains the already-handled\footnote{Like eff, \ensuremath{\lambda_{\mathit{sc}}}
  implements \emph{deep} handlers, as opposed to \emph{shallow}
  handlers.\cite{DBLP:conf/aplas/HillerstromL18}} continuation.
For example, \ensuremath{{\mathbf{with}}\;\Varid{h_{\mathsf{ND}}}\;{\mathbf{handle}}\;\Varid{c_{\mathsf{ND1}}}} (p. \pageref{eq:cnd1}) reduces as follows.

\indentbegin \begin{hscode}\SaveRestoreHook
\column{B}{@{}>{\hspre}l<{\hspost}@{}}%
\column{3}{@{}>{\hspre}c<{\hspost}@{}}%
\column{3E}{@{}l@{}}%
\column{8}{@{}>{\hspre}l<{\hspost}@{}}%
\column{12}{@{}>{\hspre}l<{\hspost}@{}}%
\column{45}{@{}>{\hspre}l<{\hspost}@{}}%
\column{51}{@{}>{\hspre}l<{\hspost}@{}}%
\column{66}{@{}>{\hspre}l<{\hspost}@{}}%
\column{E}{@{}>{\hspre}l<{\hspost}@{}}%
\>[8]{}{\mathbf{with}}\;\Varid{h_{\mathsf{ND}}}\;{\mathbf{handle}}\;\Varid{c_{\mathsf{ND1}}}{}\<[E]%
\\
\>[3]{}\makebox[\widthof{\ensuremath{\leadsto^{\ast}}}][l]{\ensuremath{\leadsto}}{}\<[3E]%
\>[8]{}\mathbf{do}\;{}\<[12]%
\>[12]{}\Varid{xs}\leftarrow (\boldsymbol{\lambda}\Varid{b}\, .\,{\mathbf{with}}\;\Varid{h_{\mathsf{ND}}}\;{\mathbf{handle}}\;\mathbf{if}\;\Varid{b}\;{}\<[45]%
\>[45]{}\mathbf{then}\;{}\<[51]%
\>[51]{}{\mathbf{return}}\;\mathrm{1}\;\mathbf{else}\;{}\<[66]%
\>[66]{}{\mathbf{return}}\;\mathrm{2})\;\mathsf{true}{}\<[E]%
\\
\>[8]{}\mathbf{do}\;{}\<[12]%
\>[12]{}\Varid{ys}\leftarrow (\boldsymbol{\lambda}\Varid{b}\, .\,{\mathbf{with}}\;\Varid{h_{\mathsf{ND}}}\;{\mathbf{handle}}\;\mathbf{if}\;\Varid{b}\;{}\<[45]%
\>[45]{}\mathbf{then}\;{}\<[51]%
\>[51]{}{\mathbf{return}}\;\mathrm{1}\;\mathbf{else}\;{}\<[66]%
\>[66]{}{\mathbf{return}}\;\mathrm{2})\;\mathsf{false}{}\<[E]%
\\
\>[12]{}\Varid{xs}+\hspace{-0.4em}+\Varid{ys}{}\<[E]%
\\
\>[3]{}\leadsto^{\ast}{}\<[3E]%
\>[8]{}{\mathbf{return}}\;[\mskip1.5mu \mathrm{1},\mathrm{2}\mskip1.5mu]{}\<[E]%
\ColumnHook
\end{hscode}\resethooks
\indentend In case \ensuremath{\Varid{h}} does not contain a clause
for label \ensuremath{\ell^{\textsf{op}}}, the effect is forwarded (\textsc{E-FwdOp}). %
Algebraic effects can be forwarded
generically: we re-invoke the operation and recursively apply the handler to
continuation \ensuremath{\Varid{c}_{1}}.
For example, during the application of \ensuremath{\Varid{run_{\mathsf{inc}}}} in \ensuremath{{\mathbf{with}}\;\Varid{h_{\mathsf{ND}}}\;{\mathbf{handle}}\;\Varid{run_{\mathsf{inc}}}\;\mathrm{0}\;\Varid{c_{\mathsf{inc}}}} (p. \pageref{eq:hdntoruninc}), \ensuremath{\mathtt{choose}} is forwarded:

\indentbegin \begin{hscode}\SaveRestoreHook
\column{B}{@{}>{\hspre}l<{\hspost}@{}}%
\column{3}{@{}>{\hspre}c<{\hspost}@{}}%
\column{3E}{@{}l@{}}%
\column{10}{@{}>{\hspre}l<{\hspost}@{}}%
\column{E}{@{}>{\hspre}l<{\hspost}@{}}%
\>[10]{}{\mathbf{with}}\;\Varid{h_{\mathsf{ND}}}\;{\mathbf{handle}}\;\Varid{run_{\mathsf{inc}}}\;\mathrm{0}\;\Varid{c_{\mathsf{inc}}}{}\<[E]%
\\
\>[3]{}\makebox[\widthof{\ensuremath{\leadsto^{\ast}}}][l]{\ensuremath{\leadsto}}{}\<[3E]%
\>[10]{}{\mathbf{with}}\;\Varid{h_{\mathsf{ND}}}\;{\mathbf{handle}}\;\mathbf{do}\;\Varid{p'}\leftarrow {\mathbf{op}}\;\mathtt{choose}\;()\;(\Varid{b}\, .\,{\mathbf{with}}\;\Varid{h_{\mathsf{inc}}}\;{\mathbf{handle}}\;\mathbf{if}\;\Varid{b}\;\mathbf{then}\;{\mathbf{op}}\;\mathtt{inc}\;(){}\<[E]%
\\
\>[10]{}\quad\;(\Varid{x}\, .\,\Varid{x}\mathbin{+}\mathrm{5})\;\mathbf{else}\;{\mathbf{op}}\;\mathtt{inc}\;()\;(\Varid{y}\, .\,\Varid{y}\mathbin{+}\mathrm{2}))\, ;\Varid{p'}\;\mathrm{0}{}\<[E]%
\\
\>[3]{}\leadsto^{\ast}{}\<[3E]%
\>[10]{}{\mathbf{return}}\;[\mskip1.5mu (\mathrm{6},\mathrm{1}),(\mathrm{3},\mathrm{1})\mskip1.5mu]{}\<[E]%
\ColumnHook
\end{hscode}\resethooks
\indentend %

If computation \ensuremath{\Varid{c}} is a scoped operation \ensuremath{{\mathbf{sc}}\;\ell^{\textsf{sc}}\;\Varid{v}\;(\Varid{y}\, .\,\Varid{c}_{1})\;(\Varid{z}\, .\,\Varid{c}_{2})}, we again distinguish
two situations: the case where \ensuremath{\Varid{h}} contains a clause for \ensuremath{\ell^{\textsf{sc}}}, and where it does not.
If \ensuremath{\Varid{h}} contains a clause for label \ensuremath{\ell^{\textsf{sc}}}, evaluation switches to the clause's
computation \ensuremath{\Varid{c}} (\textsc{E-HandSc}), with \ensuremath{\Varid{v}} substituted for parameter \ensuremath{\Varid{x}}.
Both the scoped computation and the continuation are replaced by a function that
contains the already-handled computations \ensuremath{\Varid{c}_{1}} and \ensuremath{\Varid{c}_{2}}. For example, this
happens for the scoped operation \ensuremath{\mathtt{once}} in \ensuremath{{\mathbf{with}}\;\Varid{h_{\mathsf{once}}}\;{\mathbf{handle}}\;\Varid{c_{\mathsf{once}}}} (p.
\pageref{eq:honcetoconce}).
\indentbegin \begin{hscode}\SaveRestoreHook
\column{B}{@{}>{\hspre}l<{\hspost}@{}}%
\column{3}{@{}>{\hspre}c<{\hspost}@{}}%
\column{3E}{@{}l@{}}%
\column{8}{@{}>{\hspre}l<{\hspost}@{}}%
\column{12}{@{}>{\hspre}l<{\hspost}@{}}%
\column{16}{@{}>{\hspre}l<{\hspost}@{}}%
\column{E}{@{}>{\hspre}l<{\hspost}@{}}%
\>[8]{}{\mathbf{with}}\;\Varid{h_{\mathsf{once}}}\;{\mathbf{handle}}\;\Varid{c_{\mathsf{once}}}{}\<[E]%
\\
\>[3]{}\makebox[\widthof{\ensuremath{\leadsto^{\ast}}}][l]{\ensuremath{\leadsto}}{}\<[3E]%
\>[8]{}\mathbf{do}\;{}\<[12]%
\>[12]{}\Varid{ts}{}\<[16]%
\>[16]{}\leftarrow (\boldsymbol{\lambda}\anonymous \, .\,{\mathbf{with}}\;\Varid{h_{\mathsf{once}}}\;{\mathbf{handle}}\;{\mathbf{op}}\;\mathtt{choose}\;(\Varid{b}\, .\,{\mathbf{return}}\;\Varid{b}))\;(){}\<[E]%
\\
\>[8]{}\mathbf{do}\;{}\<[12]%
\>[12]{}\Varid{t}{}\<[16]%
\>[16]{}\leftarrow \mathsf{head}\;\Varid{ts}{}\<[E]%
\\
\>[8]{}\quad\;(\boldsymbol{\lambda}\Varid{p}\, .\,{\mathbf{with}}\;\Varid{h_{\mathsf{once}}}\;{\mathbf{handle}}\;(\mathbf{do}\;\Varid{q}\leftarrow {\mathbf{op}}\;\mathtt{choose}\;(\Varid{b}\, .\,{\mathbf{return}}\;\Varid{b})\, ;{\mathbf{return}}\;(\Varid{p},\Varid{q})))\;\Varid{t}{}\<[E]%
\\
\>[3]{}\leadsto^{\ast}{}\<[3E]%
\>[8]{}{\mathbf{return}}\;[\mskip1.5mu (\mathsf{true},\mathsf{true}),(\mathsf{true},\mathsf{false})\mskip1.5mu]{}\<[E]%
\ColumnHook
\end{hscode}\resethooks
\indentend %

When \ensuremath{\Varid{h}} does not contain a clause for label \ensuremath{\ell^{\textsf{sc}}}, we must forward the effect.
As we showed in \Cref{sec:motivation-forwarding}, forwarding scoped effects
cannot happen canonically, but (for now) rather proceeds via the handler's
forwarding clause \ensuremath{{\mathbf{bind}}\;\Varid{x}\;\Varid{k}\mapsto\Varid{c}}.
Similar to the forwarding of algebraic effects, we recursively handle the
subcomputations.
However, in the case of \ensuremath{{\mathbf{bind}}}, we set the continuation of the to-be-forwarded
effect to the forwarding clause's body \ensuremath{\Varid{c}}, substituted with the relevant
arguments.

For a computation, consider again the example described in
\Cref{sec:motivation-forwarding}.
Even though we have not given any semantics to \ensuremath{\mathtt{catch}} yet (we will do so in
\Cref{eg:exceptions}), we know \ensuremath{\Varid{h_{\mathsf{once}}}} does not contain a clause for \ensuremath{\mathtt{catch}}.
As described, \ensuremath{\mathtt{catch}} must be forwarded, addressing the type mismatch between
the handled scoped computation and handled continuation with \ensuremath{\mathsf{concatMap}}.
\indentbegin \begin{hscode}\SaveRestoreHook
\column{B}{@{}>{\hspre}l<{\hspost}@{}}%
\column{6}{@{}>{\hspre}l<{\hspost}@{}}%
\column{23}{@{}>{\hspre}l<{\hspost}@{}}%
\column{E}{@{}>{\hspre}l<{\hspost}@{}}%
\>[6]{}{\mathbf{with}}\;\Varid{h_{\mathsf{once}}}\;{\mathbf{handle}}\;\Varid{c_{\mathsf{catch}}}{}\<[E]%
\\
\>[B]{}\leadsto{}\<[6]%
\>[6]{}{\mathbf{sc}}\;\mathtt{catch}\;\text{\ttfamily \char34 err\char34}\;{}\<[23]%
\>[23]{}(\Varid{b}\, .\,{\mathbf{with}}\;\Varid{h_{\mathsf{once}}}\;{\mathbf{handle}}\;\mathbf{if}\;\Varid{b}\;\mathbf{then}\;{\mathbf{return}}\;\mathrm{1}\;\mathbf{else}\;{\mathbf{return}}\;\mathrm{2})\;{}\<[E]%
\\
\>[23]{}(\Varid{x}\, .\,\mathsf{concatMap}\;\Varid{x}\;(\boldsymbol{\lambda}\Varid{z}\, .\,{\mathbf{with}}\;\Varid{h_{\mathsf{once}}}\;{\mathbf{handle}}\;{\mathbf{return}}\;\Varid{z})){}\<[E]%
\ColumnHook
\end{hscode}\resethooks
\indentend %
\section{Type-and-Effect System}
\label{sec:typing}

This section presents the type-and-effect system of \ensuremath{\lambda_{\mathit{sc}}}. As before, we
distinguish between values, computations and handlers.

\subsection{Grammar}
\begin{figure}[t]
\begin{center}
\renewcommand\arraystretch{1}
\begin{tabular}[c]{rrl@{\hskip -1.5cm}r}
  value types \ensuremath{\Conid{A}}, \ensuremath{\Conid{B}}, \ensuremath{\Conid{M}} &\ensuremath{::=} &\ensuremath{\mathsf{()}\ \mid\ (\Conid{A},\Conid{B})\ \mid\ \Conid{A}\to  \underline{C} \ \mid\  \underline{C} \Rightarrow  \underline{D} } &\\
      &$\mid$ &\ensuremath{\alpha} &type variable\\
      &$\mid$ &\ensuremath{\lambda\;\alpha\, .\,\Conid{A}} &type operator abstraction\\
      &$\mid$ &\ensuremath{\Conid{M}\;\Conid{A}} &type application\\[\grammarspacing]

  type schemes \ensuremath{\sigma} &\ensuremath{::=} &\ensuremath{\Conid{A}\ \mid\ \forall\;\mu\, .\,\sigma\ \mid\ \forall\;\alpha\, .\,\sigma} & \\ %
  computation types \ensuremath{ \underline{C} }, \ensuremath{ \underline{D} } &\ensuremath{::=} &\ensuremath{\Conid{A}\hspace{0.1em}!\hspace{0.1em}\langle\Conid{E}\rangle} & \\
  effect rows \ensuremath{\Conid{E}}, \ensuremath{\Conid{F}} &\ensuremath{::=} &\ensuremath{ \cdot\ \mid\ \mu\ \mid\ \ell\, ;\Conid{E}}\\[\grammarspacing]

  signature contexts \ensuremath{\Sigma} &\ensuremath{::=} &\ensuremath{ \cdot\ \mid\ \Sigma,\ell^{\textsf{op}}\mathbin{:}\Conid{A}\rightarrowtriangle\Conid{B}\ \mid\ \Sigma,\ell^{\textsf{sc}}\mathbin{:}\Conid{A}\rightarrowtriangle\Conid{B}} & \\
  type contexts \ensuremath{\Gamma} &\ensuremath{::=} &\ensuremath{ \cdot\ \mid\ \Gamma,\Varid{x}\mathbin{:}\sigma\ \mid\ \Gamma,\mu\ \mid\ \Gamma,\alpha} &
\end{tabular}
\end{center}
\caption{Types of \ensuremath{\lambda_{\mathit{sc}}}.}
\label{fig:type-syntax}
\end{figure}
\Cref{fig:type-syntax} displays the grammar of the types of \ensuremath{\lambda_{\mathit{sc}}}.
Like terms, types are split: values have value types \ensuremath{\Conid{A},\Conid{B}}, computations have
computation types \ensuremath{ \underline{C} , \underline{D} }.
Value types consist of the unit type \ensuremath{()}, pair types \ensuremath{(\Conid{A},\Conid{B})}, function types
\ensuremath{\Conid{A}\to  \underline{C} }, handler types \ensuremath{ \underline{C} \Rightarrow  \underline{D} }, type variables \ensuremath{\alpha} and abstraction over
them \ensuremath{\boldsymbol{\lambda}\alpha\, .\,\Conid{A}}, and type application \ensuremath{\Conid{M}\;\Conid{A}}.
Following convention and to allow for meaningful examples, we may add base value
types to the calculus, such as \ensuremath{\mathsf{String}}, \ensuremath{\mathsf{Int}} and \ensuremath{\mathsf{Bool}}.
Functions take a value of type \ensuremath{\Conid{A}} as argument and return a computation of type
\ensuremath{ \underline{C} }; handlers take a computation of type \ensuremath{ \underline{C} } as argument and return a
computation of type \ensuremath{ \underline{D} }.

A computation type \ensuremath{\Conid{A}\hspace{0.1em}!\hspace{0.1em}\langle\Conid{E}\rangle} consists of a value type \ensuremath{\Conid{A}}, representing the
type of the value the computation evaluates to, and an effect type \ensuremath{\Conid{E}},
representing the effects that \emph{may} be called during this evaluation.
Different from Eff, we implement effect types as effect rows using row
polymorphism \cite{DBLP:conf/sfp/Leijen05} in the style of Koka \cite{Leijen17}.
Therefore, rows \ensuremath{\Conid{E}} are represented as collections of the previously discussed
atomic labels \ensuremath{\ell^{\textsf{op}}} and \ensuremath{\ell^{\textsf{sc}}}, optionally terminated by a row variable \ensuremath{\mu}.
Finally, we can abstract over both type and row variables, giving rise to type
schemes \ensuremath{\sigma}.

\subsection{Value typing}
\label{sec:value-typing}
\Cref{fig:value-typing} displays the typing rules for values. Rules
\textsc{T-Var}, \textsc{T-Unit}, \textsc{T-Pair} and \textsc{T-Abs} type
variables, units, pairs and term abstractions, respectively, and are standard.
\begin{figure}[t]
\scalebox{\mathscalefactor}{\ruleform{\ensuremath{\Gamma\vdash\Varid{v}\mathbin{:}\sigma}} \quad\text{Value Typing} \hspace*{\fill}}

\begin{center}
\begin{mathpar}
  \inferrule*[right=T-Var]
  { \ensuremath{(\Varid{x}\mathbin{:}\sigma)\;\in\;\Gamma}
  }
  { \ensuremath{\Gamma\vdash\Varid{x}\mathbin{:}\sigma} }

  \inferrule*[right=T-Unit]
  { }
  { \ensuremath{\Gamma\vdash\mathsf{()}\mathbin{:}\mathsf{()}} }

  \inferrule*[right=T-Pair]
  { \ensuremath{\Gamma\vdash\Varid{v}_{1}\mathbin{:}\Conid{A}}
    \\ \ensuremath{\Gamma\vdash\Varid{v}_{2}\mathbin{:}\Conid{B}}
  }
  { \ensuremath{\Gamma\vdash(\Varid{v}_{1},\Varid{v}_{2})\mathbin{:}(\Conid{A},\Conid{B})} }

  \inferrule*[right=T-Abs]
  {
    \ensuremath{\Gamma,\Varid{x}\mathbin{:}\Conid{A}\vdash\Varid{c}\mathbin{:} \underline{C} }
  }
  { \ensuremath{\Gamma\vdash\boldsymbol{\lambda}\Varid{x}\, .\,\Varid{c}\mathbin{:}\Conid{A}\to  \underline{C} } }

  \inferrule*[right=T-EqV]
  { \ensuremath{\Gamma\vdash\Varid{v}\mathbin{:}\Conid{A}}
  \\ \ensuremath{\Conid{A}\;\equiv\;\Conid{B}}
  }
  { \ensuremath{\Gamma\vdash\Varid{v}\mathbin{:}\Conid{B}} }

  \inferrule*[right=T-Inst]
  { \ensuremath{\Gamma\vdash\Varid{v}\mathbin{:}\forall\;\alpha\, .\,\sigma}
    \\ \ensuremath{\Gamma\vdash\Conid{A}}
  }
  { \ensuremath{\Gamma\vdash\Varid{v}\mathbin{:}[\mskip1.5mu \Conid{A}\mathbin{/}\alpha\mskip1.5mu]\;\sigma} }

  \inferrule*[right=T-Gen]
  { \ensuremath{\Gamma,\alpha\vdash\Varid{v}\mathbin{:}\sigma}
  \\ \ensuremath{\alpha\;\notin\;\Gamma}
  }
  { \ensuremath{\Gamma\vdash\Varid{v}\mathbin{:}\forall\;\alpha\, .\,\sigma} }

  \inferrule*[right=T-InstEff]
  { \ensuremath{\Gamma\vdash\Varid{v}\mathbin{:}\forall\;\mu\, .\,\sigma}
    \\ \ensuremath{\Gamma\vdash\Conid{E}}
  }
  { \ensuremath{\Gamma\vdash\Varid{v}\mathbin{:}[\mskip1.5mu \Conid{E}\mathbin{/}\mu\mskip1.5mu]\;\sigma} }

  \inferrule*[right=T-GenEff]
  { \ensuremath{\Gamma,\mu\vdash\Varid{v}\mathbin{:}\sigma}
  \\ \ensuremath{\mu\;\notin\;\Gamma}
  }
  { \ensuremath{\Gamma\vdash\Varid{v}\mathbin{:}\forall\;\mu\, .\,\sigma} }
\end{mathpar}
\end{center}
\caption{Value typing.}
\label{fig:value-typing}
\end{figure}

Rule \textsc{T-EqV} expresses that typing holds up to equivalence of types. The
full type equivalence relation (\ensuremath{\Conid{A}\;\equiv\;\Conid{B}}), which also uses the equivalence of rows (\ensuremath{\Conid{E}\;\equiv_{\langle\rangle}\;\Conid{F}}),
is included in \Cref{app:type-equiv}.
However, these relations can be described as the congruence closure of the following two rules.
\begin{mathpar}
  \inferrule*[right=Q-AppAbs]
  { }
  { \ensuremath{(\lambda\;\alpha\, .\,\Conid{A})\;\Conid{B}\;\equiv\;\Conid{A}\;[\mskip1.5mu \Conid{B}\mathbin{/}\alpha\mskip1.5mu]} }

  \inferrule*[right=R-Swap]
  { \ensuremath{\ell_{1}\;\not =\;\ell_{2}}
  }
  { \ensuremath{\ell_{1}\, ;\ell_{2}\, ;\Conid{E}\;\equiv_{\langle\rangle}\;\ell_{2}\, ;\ell_{1}\, ;\Conid{E}} }
\end{mathpar}
Rule \textsc{Q-AppAbs} captures type application, following \ensuremath{F_{\omega}}, and \textsc{R-Swap} captures
the insignificance of the order in effect rows, following Koka's row typing
approach.

The final four value typing rules deal with generalization and instantiation of
type variables and row variables. Rule
\textsc{T-Inst} instantiates the type variables \ensuremath{\alpha} in a type scheme with a value type \ensuremath{\Conid{A}}.
Rule \textsc{T-Gen} is its dual, abstracting over a type variable.
The rules for row variables are similar:
\textsc{T-InstEff} instantiates row variable with an effect row \ensuremath{\Conid{E}};
\textsc{T-GenEff} abstracts over a row variable.
The definition of well-scopedness for types \ensuremath{\Gamma\vdash\sigma} and effect rows \ensuremath{\Gamma\vdash\Conid{E}} is straightforward (\Cref{app:well-scopedness}).

\subsection{Computation typing}
\label{sec:computation-typing}

\begin{figure}[t]
\ruleform{\ensuremath{\Gamma\vdash\Varid{c}\mathbin{:} \underline{C} }} \quad\text{Computation Typing} \hspace*{\fill}
\begin{center}
\begin{mathpar}
  \inferrule*[right=T-App]
  { \ensuremath{\Gamma\vdash\Varid{v}_{1}\mathbin{:}\Conid{A}\to  \underline{C} }
    \\ \ensuremath{\Gamma\vdash\Varid{v}_{2}\mathbin{:}\Conid{A}}
  }
  { \ensuremath{\Gamma\vdash\Varid{v}_{1}\;\Varid{v}_{2}\mathbin{:} \underline{C} } }

  \inferrule*[right=T-Do]
  { \ensuremath{\Gamma\vdash\Varid{c}_{1}\mathbin{:}\Conid{A}\hspace{0.1em}!\hspace{0.1em}\langle\Conid{E}\rangle}
    \\ \ensuremath{\Gamma,\Varid{x}\mathbin{:}\Conid{A}\vdash\Varid{c}_{2}\mathbin{:}\Conid{B}\hspace{0.1em}!\hspace{0.1em}\langle\Conid{E}\rangle}
  }
  { \ensuremath{\Gamma\vdash\mathbf{do}\;\Varid{x}\leftarrow \Varid{c}_{1}\, ;\Varid{c}_{2}\mathbin{:}\Conid{B}\hspace{0.1em}!\hspace{0.1em}\langle\Conid{E}\rangle} }

  \inferrule*[right=T-EqC]
  { \ensuremath{\Gamma\vdash\Varid{c}\mathbin{:} \underline{C} }
    \\ \ensuremath{ \underline{C} \;\equiv\; \underline{D} }
  }{ \ensuremath{\Gamma\vdash\Varid{c}\mathbin{:} \underline{D} } }

  \inferrule*[right=T-Let]
  { \ensuremath{\Gamma\vdash\Varid{v}\mathbin{:}\sigma}
    \\ \ensuremath{\Gamma,\Varid{x}\mathbin{:}\sigma\vdash\Varid{c}\mathbin{:} \underline{C} }
  }
  { \ensuremath{\Gamma\vdash\mathbf{let}\;\Varid{x}\mathrel{=}\Varid{v}\;\mathbf{in}\;\Varid{c}\mathbin{:} \underline{C} } }

  \inferrule*[right=T-Ret]
  { \ensuremath{\Gamma\vdash\Varid{v}\mathbin{:}\Conid{A}}
    \\ \ensuremath{\Gamma\vdash\Conid{E}}
  }
  { \ensuremath{\Gamma\vdash{\mathbf{return}}\;\Varid{v}\mathbin{:}\Conid{A}\hspace{0.1em}!\hspace{0.1em}\langle\Conid{E}\rangle} }

  \inferrule*[right=T-Hand]
  { \ensuremath{\Gamma\vdash\Varid{v}\mathbin{:} \underline{C} \Rightarrow  \underline{D} }
    \\ \ensuremath{\Gamma\vdash\Varid{c}\mathbin{:} \underline{C} }
  }
  { \ensuremath{\Gamma\vdash{\mathbf{with}}\;\Varid{v}\;{\mathbf{handle}}\;\Varid{c}\mathbin{:} \underline{D} } }

  \inferrule*[right=T-Op]
  { \ensuremath{(\ell^{\textsf{op}}\mathbin{:}\Conid{A}_{1}\rightarrowtriangle\Conid{A}_{2})\;\in\;\Sigma}
    \\ \ensuremath{\Gamma\vdash\Varid{v}\mathbin{:}\Conid{A}_{1}}
    \\ \ensuremath{\Gamma,\Varid{y}\mathbin{:}\Conid{A}_{2}\vdash\Varid{c}\mathbin{:}\Conid{A}_{3}\hspace{0.1em}!\hspace{0.1em}\langle\Conid{E}\rangle}
    \\ \ensuremath{\ell^{\textsf{op}}\;\in\;\Conid{E}}
  }
  { \ensuremath{\Gamma\vdash{\mathbf{op}}\;\ell^{\textsf{op}}\;\Varid{v}\;(\Varid{y}\, .\,\Varid{c})\mathbin{:}\Conid{A}_{3}\hspace{0.1em}!\hspace{0.1em}\langle\Conid{E}\rangle} }

  \inferrule*[right=T-Sc]
  { \ensuremath{(\ell^{\textsf{sc}}\mathbin{:}\Conid{A}_{1}\rightarrowtriangle\Conid{A}_{2})\;\in\;\Sigma}
    \\ \ensuremath{\Gamma\vdash\Varid{v}\mathbin{:}\Conid{A}_{1}}
    \\ \ensuremath{\Gamma,\Varid{y}\mathbin{:}\Conid{A}_{2}\vdash\Varid{c}_{1}\mathbin{:}\Conid{A}_{3}\hspace{0.1em}!\hspace{0.1em}\langle\Conid{E}\rangle}
    \\ \ensuremath{\Gamma,\Varid{z}\mathbin{:}\Conid{A}_{3}\vdash\Varid{c}_{2}\mathbin{:}\Conid{A}_{4}\hspace{0.1em}!\hspace{0.1em}\langle\Conid{E}\rangle}
    \\ \ensuremath{\ell^{\textsf{sc}}\;\in\;\Conid{E}}
  }
  { \ensuremath{\Gamma\vdash{\mathbf{sc}}\;\ell^{\textsf{sc}}\;\Varid{v}\;(\Varid{y}\, .\,\Varid{c}_{1})\;(\Varid{z}\, .\,\Varid{c}_{2})\mathbin{:}\Conid{A}_{4}\hspace{0.1em}!\hspace{0.1em}\langle\Conid{E}\rangle} }
\end{mathpar}
\end{center}
\caption{Computation typing.}
\label{fig:computation-typing}
\end{figure}

\Cref{fig:computation-typing} shows the rules for computation typing.
Rules \textsc{T-App} and \textsc{T-Do} capture application and sequencing, and are
standard.
Like value typing, typing of computations holds up to equivalence of types
(\textsc{T-EqC}).
Rule \textsc{T-Let} is part of our extension of Eff,
as scoped effects require introducing let-polymorphism.

Rule \textsc{T-Ret} assigns a computation type to a return statement. This type
consists of the value \ensuremath{\Varid{v}} in the return, together with a effect row \ensuremath{\Conid{E}}. Notice
that, as in Koka, this row can be freely chosen.

Rule \textsc{T-Hand} types handler application. The typing rules for handlers
and their clauses are discussed in \Cref{sec:handler-typing}. A handler of type
\ensuremath{\Conid{C}\Rightarrow \Conid{D}} denotes a handler that transforms computations of type \ensuremath{\Conid{C}} to a
computation of type \ensuremath{\Conid{D}}.

Rule \textsc{T-Op} types algebraic effects. Looking up label \ensuremath{\ell^{\textsf{op}}} in \ensuremath{\Sigma} yields a signature
\ensuremath{\Conid{A}_{1}\rightarrowtriangle\Conid{A}_{2}}, where \ensuremath{\Conid{A}_{1}} is the type of the operation's parameter \ensuremath{\Varid{v}},
and \ensuremath{\Conid{A}_{2}} is the type of argument \ensuremath{\Varid{y}} of the continuation.
The resulting effect row includes \ensuremath{\ell^{\textsf{op}}}.
Indeed, \ensuremath{\ell\;\in\;\Conid{E}} means that there is some \ensuremath{\Conid{E'}} such that \ensuremath{\Conid{E}\;\equiv_{\langle\rangle}\;\ell\, ;\Conid{E'}}.
Finally, the operation's type equals that of continuation \ensuremath{\Varid{c}}.

\label{sec:sc-clause-typing}
Similarly, rule \textsc{T-Sc} types scoped effects. Again, looking up label \ensuremath{\ell^{\textsf{sc}}} in
\ensuremath{\Sigma} yields signature \ensuremath{{A_{\textsf{sc}}}\rightarrowtriangle{B_{\textsf{sc}}}} where \ensuremath{{A_{\textsf{sc}}}} corresponds to the type of the operation's
parameter \ensuremath{\Varid{v}}.
However, where \ensuremath{{B_{\textsf{op}}}} in the algebraic case refers to the \emph{continuation's}
argument, \ensuremath{{B_{\textsf{sc}}}} now describes the \emph{scoped computation's} argument. This
leaves the the scoped result type undescribed by the signature, but as discussed
in \Cref{sec:motivation-polymorphic-handers}, this freedom is exactly what we
want.
As for the effect rows, \textsc{T-Sc} requires the rows of the scoped
computation to match.

\subsection{Handler typing}
\label{sec:handler-typing}

The typing rules for handlers and handler clauses are shown in
\Cref{fig:handler-typing}.
They consist of four judgments.
Judgment \ensuremath{\Gamma\vdash{\mathbf{return}}\;\Varid{x}\mapsto\Varid{c}_{\Varid{r}}\mathbin{:}\Conid{M}\;\Conid{A}\hspace{0.1em}!\hspace{0.1em}\langle\Conid{E}\rangle} types return clauses,
\ensuremath{\Gamma\vdash\Varid{oprs}\mathbin{:}\Conid{M}\;\Conid{A}\hspace{0.1em}!\hspace{0.1em}\langle\Conid{E}\rangle} types operation clauses, and \ensuremath{\Gamma\vdash{\mathbf{bind}}\;\Varid{x}\;\Varid{k}\mapsto\Varid{c}\mathbin{:}\Conid{M}\;\Conid{A}\hspace{0.1em}!\hspace{0.1em}\langle\Conid{E}\rangle} types forwarding clauses.
Finally, \ensuremath{\Gamma\vdash\Varid{h}\mathbin{:}\forall\;\alpha\, .\,\alpha\hspace{0.1em}!\hspace{0.1em}\langle\Conid{F}\rangle\Rightarrow \Conid{M}\;\alpha\hspace{0.1em}!\hspace{0.1em}\langle\Conid{E}\rangle} types
handlers, using the first three judgments.

\begin{figure}[t]
  \raggedright
  \ruleform{\ensuremath{\Gamma\vdash{\mathbf{return}}\;\Varid{x}\mapsto\Varid{c}_{\Varid{r}}\mathbin{:} \underline{C} }} \ruleform{\ensuremath{\Gamma\vdash\Varid{oprs}\mathbin{:} \underline{C} }} %
  \ruleform{\ensuremath{\Gamma\vdash{\mathbf{bind}}\;\Varid{x}\;\Varid{k}\mapsto\Varid{c}_{\Varid{f}}\mathbin{:} \underline{C} }}\\[0.07cm]
  \text{Handler clause typing} \hspace*{\fill}

  \begin{center}
  \begin{mathpar}
    \inferrule*[right=T-Return]
    { \ensuremath{\Gamma,\Varid{x}\mathbin{:}\Conid{A}\mapsto\Varid{c}_{\Varid{r}}\mathbin{:} \underline{C} }
    }
    { \ensuremath{\Gamma\vdash{\mathbf{return}}\;\Varid{x}\mapsto\Varid{c}_{\Varid{r}}\mathbin{:} \underline{C} } }

    \inferrule*[right=T-Empty]
    {
    }
    { \ensuremath{\Gamma\vdash \cdot\mathbin{:} \underline{C} } }

    \inferrule*[right=T-OprOp]
    { \ensuremath{\Gamma\vdash\Varid{oprs}\mathbin{:} \underline{C} }
      \\ \ensuremath{(\ell^{\textsf{op}}\mathbin{:}\Conid{A}_{1}\rightarrowtriangle\Conid{A}_{2})\;\in\;\Sigma}
      \\\\ \ensuremath{\Gamma,\Varid{x}\mathbin{:}\Conid{A}_{1},\Varid{k}\mathbin{:}\Conid{A}_{2}\to  \underline{C} \vdash\Varid{c}\mathbin{:} \underline{C} }
    }
    { \ensuremath{\Gamma\vdash{\mathbf{op}}\;\ell^{\textsf{op}}\;\Varid{x}\;\Varid{k}\mapsto\Varid{c},\Varid{oprs}\mathbin{:} \underline{C} } }

    \inferrule*[right=T-OprSc]
      { \ensuremath{\Gamma\vdash\Varid{oprs}\mathbin{:}\Conid{M}\;\Conid{A}_{1}\hspace{0.1em}!\hspace{0.1em}\langle\Conid{E}\rangle}
        \\ \ensuremath{(\ell^{\textsf{sc}}\mathbin{:}\Conid{A}_{2}\rightarrowtriangle\Conid{A}_{3})\;\in\;\Sigma}
        \\ \ensuremath{\Gamma,\beta,\Varid{x}\mathbin{:}\Conid{A}_{2},\Varid{p}\mathbin{:}\Conid{A}_{3}\to \Conid{M}\;\beta\hspace{0.1em}!\hspace{0.1em}\langle\Conid{E}\rangle,\Varid{k}\mathbin{:}\beta\to \Conid{M}\;\Conid{A}_{1}\hspace{0.1em}!\hspace{0.1em}\langle\Conid{E}\rangle\vdash\Varid{c}\mathbin{:}\Conid{M}\;\Conid{A}_{1}\hspace{0.1em}!\hspace{0.1em}\langle\Conid{E}\rangle}
      }
      { \ensuremath{\Gamma\vdash{\mathbf{sc}}\;\ell^{\textsf{sc}}\;\Varid{x}\;\Varid{p}\;\Varid{k}\mapsto\Varid{c},\Varid{oprs}\mathbin{:}\Conid{M}\;\Conid{A}_{1}\hspace{0.1em}!\hspace{0.1em}\langle\Conid{E}\rangle} }

    \inferrule*[right=T-Bind]
    {
      \ensuremath{\Gamma,\alpha,\Varid{x}\mathbin{:}\Conid{M}\;\alpha,\Varid{k}\mathbin{:}\alpha\to \Conid{M}\;\Conid{A}\hspace{0.1em}!\hspace{0.1em}\langle\Conid{E}\rangle\vdash\Varid{c}_{\Varid{f}}\mathbin{:}\Conid{M}\;\Conid{A}\hspace{0.1em}!\hspace{0.1em}\langle\Conid{E}\rangle}
    }
    { \ensuremath{\Gamma\vdash{\mathbf{bind}}\;\Varid{x}\;\Varid{k}\mapsto\Varid{c}_{\Varid{f}}\mathbin{:}\Conid{M}\;\Conid{A}\hspace{0.1em}!\hspace{0.1em}\langle\Conid{E}\rangle} }
  \end{mathpar}
  \end{center}

  \bigskip
  \raggedright
  \ruleform{\ensuremath{\Gamma\vdash\Varid{h}\mathbin{:}\forall\;\alpha\, .\,\alpha\hspace{0.1em}!\hspace{0.1em}\langle\Conid{E}\rangle\Rightarrow \Conid{M}\;\alpha\hspace{0.1em}!\hspace{0.1em}\langle\Conid{F}\rangle}} \quad \text{Handler typing} \hspace*{\fill}\\

  \begin{center}
  \begin{mathpar}

    \inferrule[T-\ensuremath{\text{Handler}_{\text{Bind}}}]
    { \ensuremath{\langle\Conid{F}\rangle\equiv_{\langle\rangle}\langle\mathit{labels\!}\;(\Varid{oprs})\, ;\Conid{E}\rangle}
      \\ \ensuremath{\Gamma,\alpha\vdash{\mathbf{return}}\;\Varid{x}\mapsto\Varid{c}_{\Varid{r}}\mathbin{:}\Conid{M}\;\alpha\hspace{0.1em}!\hspace{0.1em}\langle\Conid{E}\rangle}
      \\ \ensuremath{\Gamma,\alpha\vdash\Varid{oprs}\mathbin{:}\Conid{M}\;\alpha\hspace{0.1em}!\hspace{0.1em}\langle\Conid{E}\rangle}
      \\ \ensuremath{\Gamma,\alpha\vdash{\mathbf{bind}}\;\Varid{x}\;\Varid{k}\mapsto\Varid{c}_{\Varid{f}}\mathbin{:}\Conid{M}\;\alpha\hspace{0.1em}!\hspace{0.1em}\langle\Conid{E}\rangle}
    }
    { \ensuremath{\Gamma\vdash{\mathbf{handler}}\;\{\mskip1.5mu {\mathbf{return}}\;\Varid{x}\mapsto\Varid{c}_{\Varid{r}},\Varid{oprs},{\mathbf{bind}}\;\Varid{x}\;\Varid{k}\mapsto\Varid{c}_{\Varid{f}}\mskip1.5mu\}\mathbin{:}\forall\;\alpha\, .\,\alpha\hspace{0.1em}!\hspace{0.1em}\langle\Conid{F}\rangle\Rightarrow \Conid{M}\;\alpha\hspace{0.1em}!\hspace{0.1em}\langle\Conid{E}\rangle} }
  \end{mathpar}
  \end{center}
  \caption{Handler typing.}
  \label{fig:handler-typing}
\end{figure}

\paragraph{Return Clauses}
Rule \textsc{T-Return} types return clauses of the form \ensuremath{{\mathbf{return}}\;\Varid{x}\mapsto\Varid{c}_{\Varid{r}}}.
It binds variable \ensuremath{\Varid{x}} to type \ensuremath{\Conid{A}}, adds it to the context,
and returns the type \ensuremath{\Conid{M}\;\Conid{A}\hspace{0.1em}!\hspace{0.1em}\langle\Conid{E}\rangle} of \ensuremath{\Varid{c}_{\Varid{r}}} as the type of the return clause.

\paragraph{Operation Clauses}
The judgment \ensuremath{\Gamma\vdash\Varid{oprs}\mathbin{:}\Conid{M}\;\Conid{A}\hspace{0.1em}!\hspace{0.1em}\langle\Conid{E}\rangle} denotes that all operations in the
sequence of operations \ensuremath{\Varid{oprs}} have type \ensuremath{\Conid{M}\;\Conid{A}\hspace{0.1em}!\hspace{0.1em}\langle\Conid{E}\rangle}.
The base case \textsc{T-Empty} types the empty sequence.
The other two cases %
require the head of the sequence (either \ensuremath{{\mathbf{op}}} or \ensuremath{{\mathbf{sc}}}) to have the same
type as the tail.

Rule \textsc{T-OprOp} types algebraic operation clauses \ensuremath{{\mathbf{op}}\;\ell^{\textsf{op}}\;\Varid{x}\;\Varid{k}\mapsto\Varid{c}}.
Looking up label \ensuremath{\ell^{\textsf{op}}} in \ensuremath{\Sigma} yields signature \ensuremath{{A_{\textsf{op}}}\rightarrowtriangle{B_{\textsf{op}}}}, where \ensuremath{{A_{\textsf{op}}}}
describes the type of parameter \ensuremath{\Varid{x}}, and \ensuremath{{B_{\textsf{op}}}} the type of the argument of
continuation \ensuremath{\Varid{k}}.
In order for an operation \ensuremath{{\mathbf{op}}} to have type \ensuremath{\Conid{M}\;\Conid{A}\hspace{0.1em}!\hspace{0.1em}\langle\Conid{E}\rangle}, \ensuremath{\Varid{c}} should %
have the same type. %

Once again, the case for typing a scoped clause \ensuremath{{\mathbf{sc}}\;\ell^{\textsf{sc}}\;\Varid{x}\;\Varid{p}\;\Varid{k}\mapsto\Varid{c}}
(\textsc{T-OprSc}) is similar to its algebraic equivalent, extended to include
the scoped computation.
Notice the type of \ensuremath{\Varid{p}} and \ensuremath{\Varid{k}} when typing \ensuremath{\Varid{c}}.
First, as \ensuremath{\lambda_{\mathit{sc}}} allows freedom in the scoped result type, the type variable
\ensuremath{\beta} is used for this type.
Second, as shown in the operational semantics (rule \textsc{E-HandSc}), for a
clause \ensuremath{{\mathbf{sc}}\;\ell^{\textsf{sc}}\;\Varid{x}\;\Varid{p}\;\Varid{k}\mapsto\Varid{c}}, computation \ensuremath{\Varid{c}} uses the
\emph{already-handled} subcomputations \ensuremath{\Varid{p}} and \ensuremath{\Varid{k}}.
Therefore, type operator \ensuremath{\Conid{M}} occurs in the scoped result type as well as in the
continuation's result type:

\begin{equation*}
  \ensuremath{\Varid{p}\mathbin{:}{B_{\textsf{sc}}}\to \Conid{M}\;\beta\hspace{0.1em}!\hspace{0.1em}\langle\Conid{E}\rangle} \qquad \qquad \ensuremath{\Varid{k}\mathbin{:}\beta\to \Conid{M}\;\Conid{A}\hspace{0.1em}!\hspace{0.1em}\langle\Conid{E}\rangle}
\end{equation*}

This means that, even though our focus on mitigating the type mismatch between
scoped computation and continuation so far has been on forwarding \emph{unknown}
scoped effects, the same applies when handling \emph{known} scoped effects, where
computation \ensuremath{\Varid{c}} accounts for this discrepancy.

\paragraph{Forwarding Clause}
In accordance with \Cref{sec:bind-intro}, when an application of a handler to a
scoped effect the handler does not contain a clause for is encountered, the
effect must be forwarded.
Like the algebraic case, the effect is left in-place with handled
sub-computations.
However, as described in \Cref{sec:bind-intro}, simply handling the
sub-computations would result in a type mismatch.
It is clear from the typing rule that \ensuremath{{\mathbf{bind}}\;\Varid{x}\;\Varid{k}\mapsto\Varid{c}_{\Varid{f}}} basically
deals with this type mismatch between the type \ensuremath{\Conid{M}\;\alpha} of \ensuremath{\Varid{x}} and the
parameter type \ensuremath{\alpha} of \ensuremath{\Varid{k}}.

\paragraph{Handler}
Rule \textsc{T-Handler} types handlers with a polymorphic type of the form
\ensuremath{\forall\;\alpha\, .\,\alpha\mathbin{!}\langle\Conid{F}\rangle\Rightarrow \Conid{M}\;\alpha\mathbin{!}\langle\Conid{E}\rangle}.
A handler consists of a \ensuremath{{\mathbf{return}}}-clause (\textsc{T-Return}),
zero or more operation clauses (\textsc{T-Empty}, \textsc{T-OprOp} and \textsc{T-OprSc}),
and a forwarding clause (\textsc{T-Fwd}).
All clauses should agree on their result type \ensuremath{\Conid{M}\;\alpha\mathbin{!}\langle\Conid{E}\rangle}.
Notice that \ensuremath{\Conid{E}} denotes a collection with at least the labels of the present algebraic
and scoped operation clauses in the handler (computed by the \ensuremath{\mathit{labels\!}}\,-function).

\section{Generalized Forwarding}
\label{sec:genfwd}
In \Cref{sec:bind-intro} we showed how \ensuremath{{\mathbf{bind}}} can be used to forward certain
scoped effects.
Here we also noted that not all effects can be forwarded by \ensuremath{{\mathbf{bind}}}, and that
some require a more general forwarding construct.
In this section we revisit this issue.
First, we show how in some situations, we need a more expressive forwarding
clause.
We will then present a generalized forwarding clause, and show how the
\ensuremath{{\mathbf{bind}}}-based forwarding clauses can be straightforwardly desugared into
generalized forwarding clauses.
Finally, we will amend the calculus with rules for its grammar, semantics and
typing.

\subsection{\ensuremath{\Conid{Inc}}, revisited}
We introduced \ensuremath{\Varid{h_{\mathsf{inc}}}} in an algebraic setting, and have not returned to it.
We have not yet given a forwarding clause for \ensuremath{\Varid{h_{\mathsf{inc}}}} (with carrier
\ensuremath{(\mathsf{Int}\to (\alpha,\mathsf{Int})\hspace{0.1em}!\hspace{0.1em}\langle\Conid{E}\rangle)\hspace{0.1em}!\hspace{0.1em}\langle\Conid{E}\rangle} for some effect type \ensuremath{\Conid{E}}).
We could implement the forwarding clause of it using \ensuremath{{\mathbf{bind}}} as follows.
\indentbegin \begin{hscode}\SaveRestoreHook
\column{B}{@{}>{\hspre}l<{\hspost}@{}}%
\column{3}{@{}>{\hspre}l<{\hspost}@{}}%
\column{14}{@{}>{\hspre}l<{\hspost}@{}}%
\column{25}{@{}>{\hspre}c<{\hspost}@{}}%
\column{25E}{@{}l@{}}%
\column{28}{@{}>{\hspre}l<{\hspost}@{}}%
\column{42}{@{}>{\hspre}c<{\hspost}@{}}%
\column{42E}{@{}l@{}}%
\column{47}{@{}>{\hspre}l<{\hspost}@{}}%
\column{E}{@{}>{\hspre}l<{\hspost}@{}}%
\>[3]{}\Varid{h_{\mathsf{inc}}}_{\text{\xmark}}{}\<[14]%
\>[14]{}\mathrel{=}{\mathbf{handler}}\;{}\<[25]%
\>[25]{}\{\mskip1.5mu {}\<[25E]%
\>[28]{}{\mathbf{return}}\;\Varid{x}{}\<[42]%
\>[42]{}\mapsto{}\<[42E]%
\>[47]{}{\mathbf{return}}\;(\boldsymbol{\lambda}\Varid{s}\, .\,{\mathbf{return}}\;(\Varid{x},\Varid{s})){}\<[E]%
\\
\>[25]{},{}\<[25E]%
\>[28]{}{\mathbf{op}}\;\mathtt{inc}\;\anonymous \;\Varid{k}{}\<[42]%
\>[42]{}\mapsto{}\<[42E]%
\>[47]{}{\mathbf{return}}\;(\boldsymbol{\lambda}\Varid{s}\, .\,\mathbf{do}\;\Varid{s'}\leftarrow \Varid{s}\mathbin{+}\mathrm{1}\, ;\mathbf{do}\;\Varid{k'}\leftarrow \Varid{k}\;\Varid{s'}\, ;\Varid{k'}\;\Varid{s'}){}\<[E]%
\\
\>[25]{},{}\<[25E]%
\>[28]{}{\mathbf{bind}}\;\Varid{x}\;\Varid{k}{}\<[42]%
\>[42]{}\mapsto{}\<[42E]%
\>[47]{}{\mathbf{return}}\;(\boldsymbol{\lambda}\Varid{s}\, .\,\mathbf{do}\;(\Varid{y},\Varid{s'})\leftarrow \Varid{x}\;\Varid{s}\, ;\Varid{k'}\leftarrow \Varid{k}\;\Varid{y}\, ;\Varid{k'}\;\Varid{s'})\mskip1.5mu\}{}\<[E]%
\ColumnHook
\end{hscode}\resethooks
\indentend %
However, this implementation of forwarding would actually result in
unexpected behaviours.
Consider the derivation displayed in \Cref{fig:fwd-deriv-bind}.
The chosen example has an occurrence of \ensuremath{\mathtt{choose}} inside an occurrence of \ensuremath{\mathtt{inc}},
both of them inside an occurrence of \ensuremath{\mathtt{once}}.
Since \ensuremath{\mathtt{choose}} is in the scope of \ensuremath{\mathtt{once}}, we would expect that the
\ensuremath{\mathtt{choose}} is pruned by the \ensuremath{\mathtt{once}} to only keep the first branch.
However, the derivation shows that by using the \ensuremath{{\mathbf{bind}}}-based
forwarding clause of \ensuremath{\Varid{h_{\mathsf{inc}}}_{\text{\xmark}}} the \ensuremath{\mathtt{choose}} escapes the scope of
\ensuremath{\mathtt{once}}.
As can be seen, when \ensuremath{\mathsf{head}} (i.e.\ the embodiment of the \ensuremath{\mathtt{once}}
effect) is evaluated, \ensuremath{\mathtt{choose}} has not even been handled yet; it is
captured by a function in the result list.

For the brevity of presentation, we omitted the unimportant part of terms in the
derivation, and we inline the guard that checks for the empty list in
the if statement.
We underlined parts that will reduce in the next step(s).

\begin{figure}\indentbegin \begin{hscode}\SaveRestoreHook
\column{B}{@{}>{\hspre}l<{\hspost}@{}}%
\column{7}{@{}>{\hspre}l<{\hspost}@{}}%
\column{9}{@{}>{\hspre}l<{\hspost}@{}}%
\column{11}{@{}>{\hspre}l<{\hspost}@{}}%
\column{13}{@{}>{\hspre}l<{\hspost}@{}}%
\column{17}{@{}>{\hspre}l<{\hspost}@{}}%
\column{19}{@{}>{\hspre}l<{\hspost}@{}}%
\column{21}{@{}>{\hspre}l<{\hspost}@{}}%
\column{23}{@{}>{\hspre}l<{\hspost}@{}}%
\column{27}{@{}>{\hspre}l<{\hspost}@{}}%
\column{29}{@{}>{\hspre}l<{\hspost}@{}}%
\column{33}{@{}>{\hspre}l<{\hspost}@{}}%
\column{34}{@{}>{\hspre}l<{\hspost}@{}}%
\column{40}{@{}>{\hspre}l<{\hspost}@{}}%
\column{43}{@{}>{\hspre}l<{\hspost}@{}}%
\column{E}{@{}>{\hspre}l<{\hspost}@{}}%
\>[7]{}{\mathbf{with}}\;\Varid{h_{\mathsf{once}}}\;{\mathbf{handle}}{}\<[E]%
\\
\>[7]{}\hsindent{2}{}\<[9]%
\>[9]{}\mathbf{do}\;\Varid{f}\leftarrow \underline{{\mathbf{with}}\;\Varid{h_{\mathsf{inc}}}_{\text{\xmark}}\;{\mathbf{handle}}}{}\<[E]%
\\
\>[9]{}\hsindent{2}{}\<[11]%
\>[11]{}{\mathbf{sc}}\;\underline{\mathtt{once}}\;()\;(\anonymous \, .\,{\mathbf{op}}\;\mathtt{inc}\;()\;(\anonymous \, .\,{\mathbf{op}}\;\mathtt{choose}\;()\;(\Varid{b}\, .\,{\mathbf{return}}\;\Varid{b})))\;(\Varid{z}\, .\,{\mathbf{return}}\;\Varid{z})\, ;{}\<[E]%
\\
\>[7]{}\hsindent{2}{}\<[9]%
\>[9]{}\Varid{f}\;\mathrm{0}{}\<[E]%
\\
\>[B]{}\leadsto^{\ast}{}\<[7]%
\>[7]{}{\mathbf{with}}\;\Varid{h_{\mathsf{once}}}\;{\mathbf{handle}}\;{}\<[E]%
\\
\>[7]{}\hsindent{2}{}\<[9]%
\>[9]{}\underline{\;\mathbf{do}\;\Varid{f}}\leftarrow {}\<[E]%
\\
\>[9]{}\hsindent{2}{}\<[11]%
\>[11]{}{\mathbf{sc}}\;\mathtt{once}\;()\;{}\<[23]%
\>[23]{}(\anonymous \, .\,{\mathbf{with}}\;\Varid{h_{\mathsf{inc}}}_{\text{\xmark}}\;{\mathbf{handle}}\;({\mathbf{op}}\;\mathtt{inc}\;()\;(\anonymous \, .\,{\mathbf{op}}\;\mathtt{choose}\;()\;(\Varid{b}\, .\,{\mathbf{return}}\;\Varid{b}))))\;{}\<[E]%
\\
\>[23]{}(\Varid{z}\, .\,\ldots){}\<[E]%
\\
\>[7]{}\hsindent{2}{}\<[9]%
\>[9]{}\underline{\Varid{f}\;\mathrm{0}}{}\<[E]%
\\
\>[B]{}\leadsto{}\<[7]%
\>[7]{}\underline{{\mathbf{with}}\;\Varid{h_{\mathsf{once}}}\;{\mathbf{handle}}}\;{}\<[E]%
\\
\>[7]{}\hsindent{2}{}\<[9]%
\>[9]{}{\mathbf{sc}}\;\underline{\mathtt{once}}\;()\;{}\<[27]%
\>[27]{}(\anonymous \, .\,{}\<[33]%
\>[33]{}{\mathbf{with}}\;\Varid{h_{\mathsf{inc}}}_{\text{\xmark}}\;{\mathbf{handle}}\;({\mathbf{op}}\;\mathtt{inc}\;()\;(\anonymous \, .\,{\mathbf{op}}\;\mathtt{choose}\;()\;(\Varid{b}\, .\,{\mathbf{return}}\;\Varid{b})))))\;{}\<[E]%
\\
\>[9]{}\hsindent{12}{}\<[21]%
\>[21]{}(\Varid{z}\, .\,{}\<[27]%
\>[27]{}\ldots){}\<[E]%
\\
\>[B]{}\leadsto^{\ast}{}\<[7]%
\>[7]{}\mathbf{do}\;\Varid{ts}\leftarrow {}\<[17]%
\>[17]{}{\mathbf{with}}\;\Varid{h_{\mathsf{once}}}\;{\mathbf{handle}}\;{}\<[E]%
\\
\>[17]{}\hsindent{2}{}\<[19]%
\>[19]{}\underline{{\mathbf{with}}\;\Varid{h_{\mathsf{inc}}}_{\text{\xmark}}\;{\mathbf{handle}}}\;({\mathbf{op}}\;\underline{\mathtt{inc}}\;()\;(\anonymous \, .\,{\mathbf{op}}\;\mathtt{choose}\;()\;(\Varid{b}\, .\,{\mathbf{return}}\;\Varid{b}))){}\<[E]%
\\
\>[7]{}\mathbf{if}\;(\Varid{ts}=[\mskip1.5mu \mskip1.5mu])\;\mathbf{then}\;{\mathbf{return}}\;[\mskip1.5mu \mskip1.5mu]\;\mathbf{else}\;\mathbf{do}\;\Varid{t}\leftarrow \mathsf{head}\;\Varid{ts}\, ;\ldots{}\<[E]%
\\
\>[B]{}\leadsto^{\ast}{}\<[7]%
\>[7]{}\mathbf{do}\;\Varid{ts}\leftarrow {}\<[17]%
\>[17]{}\underline{{\mathbf{with}}\;\Varid{h_{\mathsf{once}}}\;{\mathbf{handle}}}\;{}\<[E]%
\\
\>[17]{}\hsindent{2}{}\<[19]%
\>[19]{}\underline{{\mathbf{return}}}\;(\boldsymbol{\lambda}\Varid{s}\, .\,{}\<[40]%
\>[40]{}\mathbf{do}\;\Varid{s'}\leftarrow \Varid{s}\mathbin{+}\mathrm{1}\, ;{}\<[E]%
\\
\>[19]{}\hsindent{15}{}\<[34]%
\>[34]{}\mathbf{do}\;\Varid{k'}\leftarrow {\mathbf{with}}\;\Varid{h_{\mathsf{inc}}}_{\text{\xmark}}\;{\mathbf{handle}}\;{\mathbf{op}}\;\mathtt{choose}\;()\;(\Varid{b}\, .\,{\mathbf{return}}\;\Varid{b})\, ;{}\<[E]%
\\
\>[19]{}\hsindent{15}{}\<[34]%
\>[34]{}\Varid{k'}\;\Varid{s'}){}\<[E]%
\\
\>[7]{}\mathbf{if}\;(\Varid{ts}=[\mskip1.5mu \mskip1.5mu])\;\mathbf{then}\;{\mathbf{return}}\;[\mskip1.5mu \mskip1.5mu]\;\mathbf{else}\;\mathbf{do}\;\Varid{t}\leftarrow \mathsf{head}\;\Varid{ts}\, ;\ldots{}\<[E]%
\\
\>[B]{}\leadsto^{\ast}{}\<[7]%
\>[7]{}\mathbf{do}\;\underline{\Varid{ts}}\leftarrow \underline{{\mathbf{return}}}\;[\mskip1.5mu \boldsymbol{\lambda}\Varid{s}\, .\,{}\<[43]%
\>[43]{}\mathbf{do}\;\Varid{s'}\leftarrow \Varid{s}\mathbin{+}\mathrm{1}\, ;{}\<[E]%
\\
\>[43]{}\mathbf{do}\;\Varid{k'}\leftarrow {\mathbf{with}}\;\Varid{h_{\mathsf{inc}}}_{\text{\xmark}}\;{\mathbf{handle}}\;{\mathbf{op}}\;\mathtt{choose}\;()\;(\Varid{b}\, .\,{\mathbf{return}}\;\Varid{b})\, ;{}\<[E]%
\\
\>[43]{}\Varid{k'}\;\Varid{s'}\mskip1.5mu]\, ;{}\<[E]%
\\
\>[7]{}\underline{\;\mathbf{if}\;}\;(\Varid{ts}=[\mskip1.5mu \mskip1.5mu])\;\underline{\;\mathbf{then}\;}\;{\mathbf{return}}\;[\mskip1.5mu \mskip1.5mu]\;\underline{\;\mathbf{else}\;}\;\mathbf{do}\;\Varid{t}\leftarrow \mathsf{head}\;\Varid{ts}\, ;\ldots{}\<[E]%
\\
\>[B]{}\leadsto^{\ast}{}\<[7]%
\>[7]{}\mathbf{do}\;\Varid{t}{}\<[13]%
\>[13]{}\leftarrow \mathsf{head}\;[\mskip1.5mu \boldsymbol{\lambda}\Varid{s}\, .\,{}\<[29]%
\>[29]{}\mathbf{do}\;\Varid{s'}\leftarrow \Varid{s}\mathbin{+}\mathrm{1}\, ;{}\<[E]%
\\
\>[29]{}\mathbf{do}\;\Varid{k'}\leftarrow {\mathbf{with}}\;\Varid{h_{\mathsf{inc}}}_{\text{\xmark}}\;{\mathbf{handle}}\;{\mathbf{op}}\;\mathtt{choose}\;()\;(\Varid{b}\, .\,{\mathbf{return}}\;\Varid{b})\, ;{}\<[E]%
\\
\>[29]{}\Varid{k'}\;\Varid{s'}\mskip1.5mu]\, ;{}\<[E]%
\\
\>[7]{}\ldots{}\<[E]%
\\
\>[B]{}\leadsto^{\ast}{}\<[7]%
\>[7]{}[\mskip1.5mu (\mathsf{true},\mathrm{1}),(\mathsf{false},\mathrm{1})\mskip1.5mu]{}\<[E]%
\ColumnHook
\end{hscode}\resethooks
\indentend \caption{\ensuremath{\mathtt{once}} escaping its scope when forwarded by \ensuremath{{\mathbf{bind}}}}
\label{fig:fwd-deriv-bind}
\end{figure}

The essence of this problem comes from the fact that the carrier of
\ensuremath{\Varid{h_{\mathsf{inc}}}} contain computations. As a result, other operations in scope
might be captured in these computations and then be carried out of the
scope.
The \ensuremath{{\mathbf{bind}}}-based forwarding of \ensuremath{\Varid{h_{\mathsf{inc}}}} is not expressive enough to
avoid this kind of unexpected behaviour, because it has no control
over what is returned by the scoped computation of scoped operations.
For the example above, to correctly put \ensuremath{\mathtt{choose}} in the scope, we need
to evaluate the result function of the scoped computation of \ensuremath{\mathtt{once}}
that is handled by \ensuremath{\Varid{h_{\mathsf{inc}}}}, instead of directly returning this function
which unexpectedly captures \ensuremath{\mathtt{choose}}.
We will show how to solve this problem using a generalized version of
forwarding clauses in the remaining of this section.

In general, this problem shows up for any handler whose carrier suspends
computations, such as the carrier of \ensuremath{\Varid{h_{\mathsf{inc}}}} which is a function.
The effects captured in the suspend computation can escape their scope without
being handled.
The generalized forwarding clause fixes this problem by providing us the
flexibility to early execute the suspend computation before leaving the scope.

\subsection{Adding generalized forwarding clauses}
Whereas in the case of \ensuremath{{\mathbf{bind}}}-like forwarding clauses the remittal of the
unknown, to-be-forwarded scoped effect is dealt with by the semantics (i.e., the
reinvocation of scoped effect is embedded in the rule \textsc{E-Bind}), in the
case of generalized forwarding clauses, denoted by \ensuremath{{\mathbf{fwd}}}, this responsibility is
deferred to the forwarding clause.
This allows generalized forwarding clauses to execute logic \emph{before}
remitting the forwarded effect, and to alter not only the continuation but also
the scoped computation of the forwarded effect.
An intuitive implementation of \ensuremath{{\mathbf{fwd}}} may look like this:
\begin{mathpar}
  \inferrule*[right=E-\ensuremath{\text{FwdSc'}}]
  { \ensuremath{({\mathbf{sc}}\;\ell^{\textsf{sc}}\;\anonymous \;\anonymous \;\anonymous )\;\notin\;\Varid{h}}
    \\ \ensuremath{({\mathbf{fwd}}'\;\ell^{\textsf{sc}\prime}\;\Varid{x}\;\Varid{p}\;\Varid{k}\mapsto\Varid{c}_{\Varid{f}})\;\in\;\Varid{h}}
  }
  { \ensuremath{{\mathbf{with}}\;\Varid{h}\;{\mathbf{handle}}\;{\mathbf{sc}}\;\ell^{\textsf{sc}}\;\Varid{v}\;(\Varid{y}\, .\,\Varid{c}_{1})\;(\Varid{z}\, .\,\Varid{c}_{2})\leadsto} \ensuremath{\begin{array}{l}\ensuremath{\Varid{c}_{\Varid{f}}\;[\mskip1.5mu \ell^{\textsf{sc}}\mathbin{/}\ell^{\textsf{sc}\prime},\Varid{v}\mathbin{/}\Varid{x},}
                                                            \\ \qquad \ensuremath{(\boldsymbol{\lambda}\Varid{y}\, .\,{\mathbf{with}}\;\Varid{h}\;{\mathbf{handle}}\;\Varid{c}_{1})\mathbin{/}\Varid{p},}
                                                            \\ \qquad \ensuremath{(\boldsymbol{\lambda}\Varid{z}\, .\,{\mathbf{with}}\;\Varid{h}\;{\mathbf{handle}}\;\Varid{c}_{2})\mathbin{/}\Varid{k}\mskip1.5mu]}\end{array}} }
\end{mathpar}
The forwarding clause has form \ensuremath{{\mathbf{fwd}}'\;\ell^{\textsf{sc}\prime}\;\Varid{x}\;\Varid{p}\;\Varid{k}\mapsto\Varid{c}_{\Varid{f}}}, where
\ensuremath{\Varid{x},\Varid{p},\Varid{k}} represents for the parameter, (deeply handled) scoped
computation, and (deeply handled) continuation, respectively.
While this would give us our desired expressivity, it requires us to
extend the calculus with the ability to abstract and substitute
operation labels.

To avoid unnecessary complication, we opt for another design choice. Because of
parametricity, there is only one thing a forwarding clause can do with an
unknown label: invoke it with some (possibly altered) scoped computation and
continuation.
Below we utilize this property by, instead of supplying the forwarding clause
with the label, we supply it with what is essentially a partial application of
the \ensuremath{{\mathbf{sc}}} operator to the unknown label.
This partial application can be used by the forwarding clause in exactly the
same way as the original label was. Therefore, this simplifies the type system,
but does not affect the expressivity of forwarding clauses.

\begin{mathpar}
\inferrule*[right=E-FwdSc]
  { \ensuremath{({\mathbf{sc}}\;\ell^{\textsf{sc}}\;\anonymous \;\anonymous \;\anonymous )\;\notin\;\Varid{h}}
  \\ \ensuremath{({\mathbf{fwd}}\;\Varid{f}\;\Varid{p}\;\Varid{k}\mapsto\Varid{c}_{\Varid{f}})\;\in\;\Varid{h}}
  }
  { \ensuremath{\begin{array}{l} \ensuremath{{\mathbf{with}}\;\Varid{h}\;{\mathbf{handle}}} \\\quad \ensuremath{{\mathbf{sc}}\;\ell^{\textsf{sc}}\;\Varid{v}\;(\Varid{y}\, .\,\Varid{c}_{1})\;(\Varid{z}\, .\,\Varid{c}_{2})}\end{array}} \ensuremath{\leadsto}
  \ensuremath{\begin{array}{l} \ensuremath{\Varid{c}_{\Varid{f}}\;[\mskip1.5mu (\boldsymbol{\lambda}\Varid{y}\, .\,{\mathbf{with}}\;\Varid{h}\;{\mathbf{handle}}\;\Varid{c}_{1})\mathbin{/}\Varid{p},}
                                                            \\ \qquad \ensuremath{(\boldsymbol{\lambda}\Varid{z}\, .\,{\mathbf{with}}\;\Varid{h}\;{\mathbf{handle}}\;\Varid{c}_{2})\mathbin{/}\Varid{k},}
                                                            \\ \qquad \ensuremath{(\boldsymbol{\lambda}(\Varid{p'},\Varid{k'})\, .\,{\mathbf{sc}}\;\ell^{\textsf{sc}}\;\Varid{v}\;(\Varid{y}\, .\,\Varid{p'}\;\Varid{y})\;(\Varid{z}\, .\,\Varid{k'}\;\Varid{z}))\mathbin{/}\Varid{f}\mskip1.5mu]} \end{array}} }
\end{mathpar}

\subsection{Syntax \& typing}
We replace the \ensuremath{{\mathbf{bind}}} syntax in handlers with \ensuremath{{\mathbf{fwd}}}, and implement
\ensuremath{{\mathbf{bind}}} as syntactic sugar for \ensuremath{{\mathbf{fwd}}}.
\begin{gather*}
  \begin{tabular}[c]{rrlr}
    handlers \ensuremath{\Varid{h}} &\ensuremath{::=} &\ensuremath{{\mathbf{handler}}} \:\ensuremath{\{\mskip1.5mu }\dots \\
                 & & \phantom{\ensuremath{{\mathbf{handler}}}} \: \ensuremath{,{\mathbf{fwd}}\;\Varid{f}\;\Varid{p}\;\Varid{k}\mapsto\Varid{c}_{\Varid{f}}\mskip1.5mu\}} &generalized forwarding clause \\
    \ensuremath{{\mathbf{bind}}\;\Varid{x}\;\Varid{k}\mapsto\Varid{c}} & \ensuremath{\equiv} & \ensuremath{{\mathbf{fwd}}\;\Varid{f}\;\Varid{p}\;\Varid{k}\mapsto\Varid{f}\;(\Varid{p},(\boldsymbol{\lambda}\Varid{x}\, .\,\Varid{c}))}
  \end{tabular}
\end{gather*}
Furthermore, we add a typing rule for typing \ensuremath{{\mathbf{fwd}}} clauses, and a new rule for
typing handlers containing \ensuremath{{\mathbf{fwd}}} clauses.
\begin{mathpar}
\inferrule*[right=T-Fwd]
  {    \ensuremath{\Conid{A}_{\Varid{p}}\mathrel{=}\alpha\to \Conid{M}\;\beta\hspace{0.1em}!\hspace{0.1em}\langle\Conid{E}\rangle}
    \\ \ensuremath{\Conid{A}_{\Varid{p}}'\mathrel{=}\alpha\to \gamma\hspace{0.1em}!\hspace{0.1em}\langle\Conid{E}\rangle}
    \\ \ensuremath{\Conid{A}_{\Varid{k}}\mathrel{=}\beta\to \Conid{M}\;\Conid{A}\hspace{0.1em}!\hspace{0.1em}\langle\Conid{E}\rangle}
    \\ \ensuremath{\Conid{A}_{\Varid{k}}'\mathrel{=}\gamma\to \delta\hspace{0.1em}!\hspace{0.1em}\langle\Conid{E}\rangle}
    \\ \ensuremath{\Gamma,\alpha,\beta,\Varid{p}\mathbin{:}\Conid{A}_{\Varid{p}},\Varid{k}\mathbin{:}\Conid{A}_{\Varid{k}},\Varid{f}\mathbin{:}\forall\;\gamma\;\delta\, .\,(\Conid{A}_{\Varid{p}}',\Conid{A}_{\Varid{k}}')\to \delta\hspace{0.1em}!\hspace{0.1em}\langle\Conid{E}\rangle\vdash\Varid{c}_{\Varid{f}}\mathbin{:}\Conid{M}\;\Conid{A}\hspace{0.1em}!\hspace{0.1em}\langle\Conid{E}\rangle}
  }
  { \ensuremath{\Gamma\vdash{\mathbf{fwd}}\;\Varid{f}\;\Varid{p}\;\Varid{k}\mapsto\Varid{c}_{\Varid{f}}\mathbin{:}\Conid{M}\;\Conid{A}\hspace{0.1em}!\hspace{0.1em}\langle\Conid{E}\rangle} }

\inferrule[T-Handler]
  {    \ensuremath{\langle\Conid{F}\rangle\equiv_{\langle\rangle}\langle\mathit{labels\!}\;(\Varid{oprs})\, ;\Conid{E}\rangle}
    \\ \ensuremath{\Gamma,\alpha\vdash{\mathbf{return}}\;\Varid{x}\mapsto\Varid{c}_{\Varid{r}}\mathbin{:}\Conid{M}\;\alpha\hspace{0.1em}!\hspace{0.1em}\langle\Conid{E}\rangle}
    \\ \ensuremath{\Gamma,\alpha\vdash\Varid{oprs}\mathbin{:}\Conid{M}\;\alpha\hspace{0.1em}!\hspace{0.1em}\langle\Conid{E}\rangle}
    \\ \ensuremath{\Gamma,\alpha\vdash{\mathbf{fwd}}\;\Varid{f}\;\Varid{p}\;\Varid{k}\mapsto\Varid{c}_{\Varid{f}}\mathbin{:}\Conid{M}\;\alpha\hspace{0.1em}!\hspace{0.1em}\langle\Conid{E}\rangle}
  }
  { \ensuremath{\Gamma\vdash{\mathbf{handler}}\;\{\mskip1.5mu {\mathbf{return}}\;\Varid{x}\mapsto\Varid{c}_{\Varid{r}},\Varid{oprs},{\mathbf{fwd}}\;\Varid{f}\;\Varid{p}\;\Varid{k}\mapsto\Varid{c}_{\Varid{f}}\mskip1.5mu\}\mathbin{:}\forall\;\alpha\, .\,\alpha\hspace{0.1em}!\hspace{0.1em}\langle\Conid{F}\rangle\Rightarrow \Conid{M}\;\alpha\hspace{0.1em}!\hspace{0.1em}\langle\Conid{E}\rangle} }
\end{mathpar}
Rule \textsc{E-Bind} and \textsc{$\text{T-Handler}_{\text{Bind}}$} can be
derived (and are therefore superseded by) \textsc{E-Fwd} and \textsc{T-Handler}.

\subsection{\ensuremath{\Varid{h_{\mathsf{inc}}}}, revisited again}
\label{sec:hinc-revisit-again}
Using generalized forwarding clauses, we can implement the forwarding
of \ensuremath{\Varid{h_{\mathsf{inc}}}} correctly as follows.
\indentbegin \begin{hscode}\SaveRestoreHook
\column{B}{@{}>{\hspre}l<{\hspost}@{}}%
\column{3}{@{}>{\hspre}l<{\hspost}@{}}%
\column{40}{@{}>{\hspre}l<{\hspost}@{}}%
\column{55}{@{}>{\hspre}l<{\hspost}@{}}%
\column{60}{@{}>{\hspre}l<{\hspost}@{}}%
\column{E}{@{}>{\hspre}l<{\hspost}@{}}%
\>[3]{}\Varid{h_{\mathsf{inc}}}\mathrel{=}{\mathbf{handler}}\;\{\mskip1.5mu \ldots,{\mathbf{fwd}}\;\Varid{f}\;\Varid{p}\;\Varid{k}\mapsto{}\<[40]%
\>[40]{}{\mathbf{return}}\;(\boldsymbol{\lambda}\Varid{s}\, .\,{}\<[55]%
\>[55]{}\Varid{f}\;({}\<[60]%
\>[60]{}\boldsymbol{\lambda}\Varid{y}\, .\,\mathbf{do}\;\Varid{p'}\leftarrow \Varid{p}\;\Varid{y}\, ;\Varid{p'}\;\Varid{s},{}\<[E]%
\\
\>[60]{}\boldsymbol{\lambda}(\Varid{z},\Varid{s'})\, .\,\mathbf{do}\;\Varid{k'}\leftarrow \Varid{k}\;\Varid{z}\, ;\Varid{k'}\;\Varid{s'}))\mskip1.5mu\}{}\<[E]%
\ColumnHook
\end{hscode}\resethooks
\indentend Instead of directly returning the function \ensuremath{\Varid{p'}} like the \ensuremath{{\mathbf{bind}}}-based
forwarding clause, the \ensuremath{{\mathbf{fwd}}}-based forwarding clause applies the
initial state \ensuremath{\Varid{s}} to \ensuremath{\Varid{p'}} to avoid the escaping of the \ensuremath{\mathtt{choose}}
operation captured in \ensuremath{\Varid{p'}}.
Then, to connect the results inside the scope with outside, we apply
the continuation \ensuremath{\Varid{k}} to the result \ensuremath{\Varid{z}} and the updated state \ensuremath{\Varid{s'}}.
The derivation is shown in \Cref{fig:fwd-deriv-fwd}. We only get the first result of
the \ensuremath{\mathtt{choose}} operation.

\begin{figure}\indentbegin \begin{hscode}\SaveRestoreHook
\column{B}{@{}>{\hspre}l<{\hspost}@{}}%
\column{7}{@{}>{\hspre}l<{\hspost}@{}}%
\column{9}{@{}>{\hspre}l<{\hspost}@{}}%
\column{11}{@{}>{\hspre}l<{\hspost}@{}}%
\column{17}{@{}>{\hspre}l<{\hspost}@{}}%
\column{19}{@{}>{\hspre}l<{\hspost}@{}}%
\column{20}{@{}>{\hspre}l<{\hspost}@{}}%
\column{45}{@{}>{\hspre}l<{\hspost}@{}}%
\column{E}{@{}>{\hspre}l<{\hspost}@{}}%
\>[7]{}{\mathbf{with}}\;\Varid{h_{\mathsf{once}}}\;{\mathbf{handle}}{}\<[E]%
\\
\>[7]{}\hsindent{2}{}\<[9]%
\>[9]{}\mathbf{do}\;\Varid{f}\leftarrow \underline{{\mathbf{with}}\;\Varid{h_{\mathsf{inc}}}\;{\mathbf{handle}}}{}\<[E]%
\\
\>[9]{}\hsindent{2}{}\<[11]%
\>[11]{}{\mathbf{sc}}\;\underline{\mathtt{once}}\;()\;(\anonymous \, .\,{\mathbf{op}}\;\mathtt{inc}\;()\;(\anonymous \, .\,{\mathbf{op}}\;\mathtt{choose}\;()\;(\Varid{b}\, .\,{\mathbf{return}}\;\Varid{b})))\;(\Varid{z}\, .\,{\mathbf{return}}\;\Varid{z})\, ;{}\<[E]%
\\
\>[7]{}\hsindent{2}{}\<[9]%
\>[9]{}\Varid{f}\;\mathrm{0}{}\<[E]%
\\
\>[B]{}\leadsto^{\ast}{}\<[7]%
\>[7]{}{\mathbf{with}}\;\Varid{h_{\mathsf{once}}}\;{\mathbf{handle}}{}\<[E]%
\\
\>[7]{}\hsindent{2}{}\<[9]%
\>[9]{}\mathbf{do}\;\underline{\Varid{f}}\leftarrow \underline{{\mathbf{return}}}\;(\boldsymbol{\lambda}\Varid{s}\, .\,{\mathbf{sc}}\;\mathtt{once}\;(){}\<[E]%
\\
\>[9]{}\hsindent{2}{}\<[11]%
\>[11]{}(\anonymous \, .\,{}\<[17]%
\>[17]{}\mathbf{do}\;\Varid{p'}\leftarrow ({\mathbf{with}}\;\Varid{h_{\mathsf{inc}}}\;{\mathbf{handle}}\;({\mathbf{op}}\;\mathtt{inc}\;()\;(\anonymous \, .\,{\mathbf{op}}\;\mathtt{choose}\;()\;(\Varid{b}\, .\,{\mathbf{return}}\;\Varid{b}))))\, ;{}\<[E]%
\\
\>[17]{}\hsindent{3}{}\<[20]%
\>[20]{}\Varid{p'}\;\Varid{s}){}\<[E]%
\\
\>[9]{}\hsindent{2}{}\<[11]%
\>[11]{}(\Varid{z}\, .\,\ldots)){}\<[E]%
\\
\>[7]{}\hsindent{2}{}\<[9]%
\>[9]{}\underline{\Varid{f}\;\mathrm{0}}{}\<[E]%
\\
\>[B]{}\leadsto^{\ast}{}\<[7]%
\>[7]{}\underline{{\mathbf{with}}\;\Varid{h_{\mathsf{once}}}\;{\mathbf{handle}}}\;{\mathbf{sc}}\;\underline{\mathtt{once}}\;()\;{}\<[E]%
\\
\>[7]{}\hsindent{4}{}\<[11]%
\>[11]{}(\anonymous \, .\,{}\<[17]%
\>[17]{}\mathbf{do}\;\Varid{p'}\leftarrow ({\mathbf{with}}\;\Varid{h_{\mathsf{inc}}}\;{\mathbf{handle}}\;({\mathbf{op}}\;\mathtt{inc}\;()\;(\anonymous \, .\,{\mathbf{op}}\;\mathtt{choose}\;()\;(\Varid{b}\, .\,{\mathbf{return}}\;\Varid{b}))))\, ;{}\<[E]%
\\
\>[17]{}\hsindent{3}{}\<[20]%
\>[20]{}\Varid{p'}\;\mathrm{0})\;{}\<[E]%
\\
\>[7]{}\hsindent{4}{}\<[11]%
\>[11]{}(\Varid{z}\, .\,\ldots){}\<[E]%
\\
\>[B]{}\leadsto^{\ast}{}\<[7]%
\>[7]{}\mathbf{do}\;\Varid{ts}\leftarrow {}\<[17]%
\>[17]{}{\mathbf{with}}\;\Varid{h_{\mathsf{once}}}\;{\mathbf{handle}}\;({}\<[E]%
\\
\>[7]{}\hsindent{2}{}\<[9]%
\>[9]{}\mathbf{do}\;\Varid{p'}\leftarrow (\underline{{\mathbf{with}}\;\Varid{h_{\mathsf{inc}}}\;{\mathbf{handle}}}\;({\mathbf{op}}\;\underline{\mathtt{inc}}\;()\;(\anonymous \, .\,{\mathbf{op}}\;\mathtt{choose}\;()\;(\Varid{b}\, .\,{\mathbf{return}}\;\Varid{b}))))\, ;{}\<[E]%
\\
\>[7]{}\hsindent{2}{}\<[9]%
\>[9]{}\Varid{p'}\;\mathrm{0}){}\<[E]%
\\
\>[7]{}\mathbf{if}\;(\Varid{ts}=[\mskip1.5mu \mskip1.5mu])\;\mathbf{then}\;{\mathbf{return}}\;[\mskip1.5mu \mskip1.5mu]\;\mathbf{else}\;\mathbf{do}\;\Varid{t}\leftarrow \mathsf{head}\;\Varid{ts}\, ;\ldots{}\<[E]%
\\
\>[B]{}\leadsto^{\ast}{}\<[7]%
\>[7]{}\mathbf{do}\;\Varid{ts}\leftarrow {}\<[17]%
\>[17]{}{\mathbf{with}}\;\Varid{h_{\mathsf{once}}}\;{\mathbf{handle}}\;({}\<[E]%
\\
\>[7]{}\hsindent{2}{}\<[9]%
\>[9]{}\mathbf{do}\;\underline{\Varid{p'}}\leftarrow \underline{{\mathbf{return}}}\;(\boldsymbol{\lambda}\Varid{s}\, .\,{}\<[45]%
\>[45]{}\mathbf{do}\;\Varid{s'}\leftarrow \Varid{s}\mathbin{+}\mathrm{1}\, ;{}\<[E]%
\\
\>[45]{}\mathbf{do}\;\Varid{k'}\leftarrow {\mathbf{with}}\;\Varid{h_{\mathsf{inc}}}\;{\mathbf{handle}}\;{\mathbf{op}}\;\mathtt{choose}\;()\;(\Varid{b}\, .\,{\mathbf{return}}\;\Varid{b})\, ;{}\<[E]%
\\
\>[45]{}\Varid{k'}\;\Varid{s'}){}\<[E]%
\\
\>[7]{}\hsindent{2}{}\<[9]%
\>[9]{}\underline{\Varid{p'}\;\mathrm{0}}){}\<[E]%
\\
\>[7]{}\mathbf{if}\;(\Varid{ts}=[\mskip1.5mu \mskip1.5mu])\;\mathbf{then}\;{\mathbf{return}}\;[\mskip1.5mu \mskip1.5mu]\;\mathbf{else}\;\mathbf{do}\;\Varid{t}\leftarrow \mathsf{head}\;\Varid{ts}\, ;\ldots{}\<[E]%
\\
\>[B]{}\leadsto^{\ast}{}\<[7]%
\>[7]{}\mathbf{do}\;\Varid{ts}\leftarrow {}\<[17]%
\>[17]{}{\mathbf{with}}\;\Varid{h_{\mathsf{once}}}\;{\mathbf{handle}}\;({}\<[E]%
\\
\>[7]{}\hsindent{2}{}\<[9]%
\>[9]{}\mathbf{do}\;\underline{\Varid{s'}}\leftarrow \underline{\mathrm{0}\mathbin{+}\mathrm{1}}\, ;{}\<[E]%
\\
\>[7]{}\hsindent{2}{}\<[9]%
\>[9]{}\mathbf{do}\;\Varid{k'}\leftarrow {\mathbf{with}}\;\Varid{h_{\mathsf{inc}}}\;{\mathbf{handle}}\;{\mathbf{op}}\;\mathtt{choose}\;()\;(\Varid{b}\, .\,{\mathbf{return}}\;\Varid{b})\, ;{}\<[E]%
\\
\>[7]{}\hsindent{2}{}\<[9]%
\>[9]{}\Varid{k'}\;\underline{\Varid{s'}}){}\<[E]%
\\
\>[7]{}\mathbf{if}\;(\Varid{ts}=[\mskip1.5mu \mskip1.5mu])\;\mathbf{then}\;{\mathbf{return}}\;[\mskip1.5mu \mskip1.5mu]\;\mathbf{else}\;\mathbf{do}\;\Varid{t}\leftarrow \mathsf{head}\;\Varid{ts}\, ;\ldots{}\<[E]%
\\
\>[B]{}\leadsto^{\ast}{}\<[7]%
\>[7]{}\mathbf{do}\;\Varid{ts}\leftarrow {}\<[17]%
\>[17]{}{\mathbf{with}}\;\Varid{h_{\mathsf{once}}}\;{\mathbf{handle}}\;({}\<[E]%
\\
\>[17]{}\hsindent{2}{}\<[19]%
\>[19]{}{\mathbf{op}}\;\mathtt{choose}\;()\;(\Varid{b}\, .\,\mathbf{do}\;\underline{\Varid{k'}}\leftarrow \underline{{\mathbf{return}}}\;(\boldsymbol{\lambda}\Varid{s}\, .\,(\Varid{b},\Varid{s}))\, ;\underline{\Varid{k'}}\;\mathrm{1}))\, ;{}\<[E]%
\\
\>[7]{}\mathbf{if}\;(\Varid{ts}=[\mskip1.5mu \mskip1.5mu])\;\mathbf{then}\;{\mathbf{return}}\;[\mskip1.5mu \mskip1.5mu]\;\mathbf{else}\;\mathbf{do}\;\Varid{t}\leftarrow \mathsf{head}\;\Varid{ts}\, ;\ldots{}\<[E]%
\\
\>[B]{}\leadsto^{\ast}{}\<[7]%
\>[7]{}\mathbf{do}\;\Varid{ts}\leftarrow {}\<[17]%
\>[17]{}\underline{{\mathbf{with}}\;\Varid{h_{\mathsf{once}}}\;{\mathbf{handle}}}\;({}\<[E]%
\\
\>[17]{}\hsindent{2}{}\<[19]%
\>[19]{}{\mathbf{op}}\;\underline{\mathtt{choose}}\;()\;(\Varid{b}\, .\,{\mathbf{return}}\;(\Varid{b},\mathrm{1})))\, ;{}\<[E]%
\\
\>[7]{}\mathbf{if}\;(\Varid{ts}=[\mskip1.5mu \mskip1.5mu])\;\mathbf{then}\;{\mathbf{return}}\;[\mskip1.5mu \mskip1.5mu]\;\mathbf{else}\;\mathbf{do}\;\Varid{t}\leftarrow \mathsf{head}\;\Varid{ts}\, ;\ldots{}\<[E]%
\\
\>[B]{}\leadsto^{\ast}{}\<[7]%
\>[7]{}\mathbf{do}\;\Varid{ts}\leftarrow {}\<[17]%
\>[17]{}{\mathbf{return}}\;[\mskip1.5mu (\mathsf{true},\mathrm{1}),(\mathsf{false},\mathrm{1})\mskip1.5mu]\, ;{}\<[E]%
\\
\>[7]{}\mathbf{if}\;(\Varid{ts}=[\mskip1.5mu \mskip1.5mu])\;\mathbf{then}\;{\mathbf{return}}\;[\mskip1.5mu \mskip1.5mu]\;\mathbf{else}\;\mathbf{do}\;\Varid{t}\leftarrow \mathsf{head}\;\Varid{ts}\, ;\ldots{}\<[E]%
\\
\>[B]{}\leadsto^{\ast}{}\<[7]%
\>[7]{}[\mskip1.5mu (\mathsf{true},\mathrm{1})\mskip1.5mu]{}\<[E]%
\ColumnHook
\end{hscode}\resethooks
\indentend \caption{\ensuremath{\mathtt{once}} staying in its scope forwarded by \ensuremath{{\mathbf{fwd}}}}
\label{fig:fwd-deriv-fwd}
\end{figure}

\section{Metatheory}
\label{sec:metatheory}
\subsection{Syntax-directed version of \ensuremath{\lambda_{\mathit{sc}}}}
\label{sec:syntax-directed}
\Cref{app:syntax-directed} contains a syntax-directed version of \ensuremath{\lambda_{\mathit{sc}}}.
We prove the type safety of \ensuremath{\lambda_{\mathit{sc}}} by connecting the typing judgement
of the declarative version of \ensuremath{\lambda_{\mathit{sc}}} to those of the syntax-directed
version of \ensuremath{\lambda_{\mathit{sc}}} in \Cref{sec:metatheory-appendix}.
The type inference algorithm implemented in our interpreter is also
built on this syntax-directed version.
The syntax-directed version was obtained by the following three transformations.

First we removed rule \textsc{T-EqV}, which re-types expressions to some
equivalent type.
This rule is used to make types line up exactly at the site of applications, for
example by changing the order of the labels in effect rows.
As a consequence, in the syntax directed version we essentially inline
\textsc{T-EqV} wherever it is needed.

Secondly, we removed the rules dealing with generalisation and instantiation
(\textsc{T-Inst}, \textsc{T-InstEff}, \textsc{T-Gen} and \textsc{T-GenEff}).
Instead, whenever rules insist on some kind of polymorphism on some
subderivation, we extend the environment with fresh type variables, and
generalize over them locally, instead of via axillary rules.

Finally, as dealing with higher-kinded polymorphism is orthogonal to our work
(and real programming languages like Haskell and OCaml already have their
solutions for higher-kinded polymorphism \cite{lhkp}), we avoid higher-order
unification by annotating handlers with the type operator they apply (e.g.
\ensuremath{{\mathbf{handler}_M}\;\{\mskip1.5mu \ldots\mskip1.5mu\}} instead of \ensuremath{{\mathbf{handler}}\;\{\mskip1.5mu \ldots\mskip1.5mu\}}).

\subsection{Safety}
\label{sec:safety}
The type-and-effect system of \ensuremath{\lambda_{\mathit{sc}}} is type safe, which we show by proving
Subject Reduction and Progress.
Here we briefly state the theorems; the proofs and used lemmas can be found in
\Cref{sec:metatheory-appendix}.

As values are inert, these theorems range over computations only.
The formulation of Subject Reduction (\Cref{thm:subjectreduction}) is standard.
Furthermore, apart from an additional normal form \ensuremath{{\mathbf{sc}}\;\ell^{\textsf{sc}}\;\Varid{v}\;(\Varid{y}\, .\,\Varid{c}_{1})\;(\Varid{z}\, .\,\Varid{c}_{2})},
Progress (\Cref{thm:progress}) is standard as well.
Note that Progess implies effect safety, which intuitively means that
only operations tracked in the effect type can be invoked.

\begin{restatable}[Subject Reduction]{thm}{subjectreduction}
  \label{thm:subjectreduction}
  If \ensuremath{\Gamma\vdash\Varid{c}\mathbin{:} \underline{C} } and \ensuremath{\Varid{c}\leadsto\Varid{c'}}, then there exists a \ensuremath{ \underline{C} '} such that \ensuremath{ \underline{C} \;\equiv\; \underline{C} '} and \ensuremath{\Gamma\vdash\Varid{c'}\mathbin{:} \underline{C} '}.
\end{restatable}

\begin{restatable}[Progress]{thm}{progress}
  \label{thm:progress}
  If \ensuremath{ \cdot\vdash\Varid{c}\mathbin{:}\Conid{A}\hspace{0.1em}!\hspace{0.1em}\langle\Conid{E}\rangle}, then either:
  \begin{itemize}
  \item there exists a computation \ensuremath{\Varid{c'}} such that \ensuremath{\Varid{c}\leadsto\Varid{c'}}, or
  \item c is in a \emph{normal form}, which means it is in one of the following
    forms: (1) \ensuremath{\Varid{c}\mathrel{=}{\mathbf{return}}\;\Varid{v}}, (2) \ensuremath{\Varid{c}\mathrel{=}{\mathbf{op}}\;\ell^{\textsf{op}}\;\Varid{v}\;(\Varid{y}\, .\,\Varid{c'})} where \ensuremath{\ell^{\textsf{op}}\;\in\;\Conid{E}},
    or (3) \ensuremath{\Varid{c}\mathrel{=}{\mathbf{sc}}\;\ell^{\textsf{sc}}\;\Varid{v}\;(\Varid{y}\, .\,\Varid{c}_{1})\;(\Varid{z}\, .\,\Varid{c}_{2})} where \ensuremath{\ell^{\textsf{sc}}\;\in\;\Conid{E}}.
  \end{itemize}
\end{restatable}

\section{Examples \& Implementation}
\label{sec:examples}
Now that we have formalized the calculus we can cover some examples.
This serves two purposes.
First, we will highlight how the conventional encoding of scoped
effects as handlers, the solution proposed in Plotkin and Power
\cite{Plotkin03}, is not expressive enough, even though it is applied
in the real world \cite{DBLP:journals/pacmpl/ThomsonRWS22}.
We have postponed doing so, because now that we have formally introduced a
calculus, we can immediately show how \ensuremath{\lambda_{\mathit{sc}}} addresses these issues.
The first two examples in this section (exceptions with catch and reader with
local) therefore contain both an attempt at encoding them as an handler, as well
as a proper encoding as a \ensuremath{{\mathbf{sc}}} in \ensuremath{\lambda_{\mathit{sc}}}.
Secondly, the examples exemplify the expressivity of \ensuremath{\lambda_{\mathit{sc}}}.

\paragraph{Syntactic sugar}
To enhance readability, we write the examples in a higher-level syntax following
Eff's conventions: we use top-level definitions, coalesce values and
computations, implicitly sequence steps and insert \ensuremath{{\mathbf{return}}} where needed.
Furthermore, we drop trivial \ensuremath{{\mathbf{return}}} continuations of operations:
\begin{equation*}
  \begin{array}{rcl}
  \ensuremath{{\mathbf{op}}\;\ell^{\textsf{op}}\;\Varid{x}}&\ensuremath{\equiv}&\ensuremath{{\mathbf{op}}\;\ell^{\textsf{op}}\;\Varid{x}\;(\Varid{y}\, .\,{\mathbf{return}}\;\Varid{y})}\\
  \ensuremath{{\mathbf{sc}}\;\ell^{\textsf{sc}}\;\Varid{x}\;(\Varid{y}\, .\,\Varid{c}_{1})}&\ensuremath{\equiv}&\ensuremath{{\mathbf{sc}}\;\ell^{\textsf{sc}}\;\Varid{x}\;(\Varid{y}\, .\,\Varid{c}_{1})\;(\Varid{z}\, .\,{\mathbf{return}}\;\Varid{z})}
  \end{array}
\end{equation*}

\paragraph{Implementation}
A prototype implementation of an interpreter for \ensuremath{\lambda_{\mathit{sc}}} is available
at \url{https://github.com/thwfhk/lambdaSC}.
This implementation contains a Hindley-Milner
\cite{DBLP:journals/jcss/Milner78} type inference algorithm for
\ensuremath{\lambda_{\mathit{sc}}}, which is mostly an extension of the type inference of Koka
\cite{Leijen14}.
There are two main things that we need to additionally deal with for
scoped effects.
One is that we explicit force all handlers to be polymorphic.
Note that there is no requirement for undecidable polymorphic
recursion; deep handlers are implicitly recursive and always given
polymorphic types.
The other is that we require handlers to be annotated with the type
operators that they use for their carriers in order to avoid
higher-order unification.
With explicit annotations for type operators, our unification
algorithm just needs to reduce types first before unifying.
The details can be found in the code.

Each of the examples covered in the rest of this section have been implemented.
See the readme of the implementation's repository for more information.

\subsection{Exceptions}
\label{eg:exceptions}

Wu et al. \cite{wu14} have shown how to catch exceptions with a scoped operation.
Raising an exception is an algebraic operation \ensuremath{\mathtt{raise}\mathbin{:}\mathsf{String}\rightarrowtriangle\mathsf{Empty}}, and
catching an exception is a scoped operation \ensuremath{\mathtt{catch}\mathbin{:}\mathsf{String}\rightarrowtriangle\mathsf{Bool}}.
For example, consider that we are dealing with a counter with a maximum value of \ensuremath{\mathrm{10}}.
The following computation increases the counter by \ensuremath{\mathrm{1}} and raises an exception when
the counter exceeds \ensuremath{\mathrm{10}}:
\indentbegin \begin{hscode}\SaveRestoreHook
\column{B}{@{}>{\hspre}l<{\hspost}@{}}%
\column{3}{@{}>{\hspre}l<{\hspost}@{}}%
\column{11}{@{}>{\hspre}l<{\hspost}@{}}%
\column{E}{@{}>{\hspre}l<{\hspost}@{}}%
\>[3]{}\Varid{incr}\mathrel{=}{}\<[11]%
\>[11]{}\mathbf{do}\;\Varid{x}\leftarrow {\mathbf{op}}\;\mathtt{inc}\;()\, ;{}\<[E]%
\\
\>[11]{}\mathbf{if}\;\Varid{x}\mathbin{>}\mathrm{10}\;\mathbf{then}\;{\mathbf{op}}\;\mathtt{raise}\;\text{\ttfamily \char34 Overflow\char34}\;(\Varid{y}\, .\,\mathsf{absurd}\;\Varid{y})\;\mathbf{else}\;{\mathbf{return}}\;\Varid{x}{}\<[E]%
\ColumnHook
\end{hscode}\resethooks
\indentend %
Clearly, if we start with a state of 8 and call \ensuremath{\mathtt{inc}} thrice, we end up with an
exception.
We want to define a \ensuremath{\mathtt{catch}} operation that executes an alternative computation
when an exception is thrown inside its scope.

\subsubsection{Catch as handler}
One might attempt to write catch as a handler.
However, as we will see, this method does not have the same modularity and
expressivity as our calculus because it cannot achieve the local update
semantics \cite{wu14}.
\indentbegin \begin{hscode}\SaveRestoreHook
\column{B}{@{}>{\hspre}l<{\hspost}@{}}%
\column{3}{@{}>{\hspre}l<{\hspost}@{}}%
\column{17}{@{}>{\hspre}c<{\hspost}@{}}%
\column{17E}{@{}l@{}}%
\column{20}{@{}>{\hspre}l<{\hspost}@{}}%
\column{38}{@{}>{\hspre}l<{\hspost}@{}}%
\column{40}{@{}>{\hspre}c<{\hspost}@{}}%
\column{40E}{@{}l@{}}%
\column{43}{@{}>{\hspre}l<{\hspost}@{}}%
\column{44}{@{}>{\hspre}l<{\hspost}@{}}%
\column{57}{@{}>{\hspre}l<{\hspost}@{}}%
\column{71}{@{}>{\hspre}l<{\hspost}@{}}%
\column{74}{@{}>{\hspre}l<{\hspost}@{}}%
\column{E}{@{}>{\hspre}l<{\hspost}@{}}%
\>[3]{}\Varid{h_{\mathsf{except}}}_{\text{\xmark}}{}\<[17]%
\>[17]{}\mathrel{=}{}\<[17E]%
\>[20]{}{\mathbf{handler}}\;\{\mskip1.5mu {\mathbf{return}}\;\Varid{x}\mapsto{}\<[44]%
\>[44]{}\mathsf{right}\;\Varid{x},{\mathbf{op}}\;\mathtt{raise}\;\Varid{e}\;\anonymous \mapsto{}\<[71]%
\>[71]{}\mathsf{left}\;\Varid{e}\mskip1.5mu\}{}\<[E]%
\\
\>[3]{}\mathtt{catch}_{\text{\xmark}}\;\Varid{c}_{1}\;\Varid{c}_{2}\;\equiv\;{\mathbf{with}}\;{\mathbf{handler}}\;{}\<[40]%
\>[40]{}\{\mskip1.5mu {}\<[40E]%
\>[43]{}{\mathbf{return}}\;\Varid{x}{}\<[57]%
\>[57]{}\mapsto{\mathbf{return}}\;\Varid{x}\;{}\<[74]%
\>[74]{}\null{}\<[E]%
\\
\>[40]{},{}\<[40E]%
\>[43]{}{\mathbf{op}}\;\mathtt{raise}\;\anonymous \;\anonymous {}\<[57]%
\>[57]{}\mapsto\Varid{c}_{2}\;{}\<[74]%
\>[74]{}\null\mskip1.5mu\}\;{\mathbf{handle}}\;\Varid{c}_{1}{}\<[E]%
\\
\>[3]{}\Varid{c_{\mathsf{catch}}}_{\text{\xmark}}\mathrel{=}\mathbf{do}\;\Varid{incr}\, ;\mathtt{catch}_{\text{\xmark}}\;{}\<[38]%
\>[38]{}(\mathbf{do}\;\Varid{incr}\, ;\mathbf{do}\;\Varid{incr}\, ;{\mathbf{return}}\;\text{\ttfamily \char34 success\char34})\;{}\<[E]%
\\
\>[38]{}({\mathbf{return}}\;\text{\ttfamily \char34 fail\char34}){}\<[E]%
\ColumnHook
\end{hscode}\resethooks
\indentend %
The \ensuremath{\mathtt{catch}_{\text{\xmark}}} implements the \ensuremath{\mathtt{catch}} operation as a handler. The
\ensuremath{\Varid{h_{\mathsf{except}}}_{\text{\xmark}}} interprets exceptions using a sum type. It is used to
handle potential exceptions not captured by \ensuremath{\mathtt{catch}_{\text{\xmark}}}.
By handling exceptions before state we obtain global update semantics
where the state updates are not discarded when an error is raised:
\indentbegin \begin{hscode}\SaveRestoreHook
\column{B}{@{}>{\hspre}l<{\hspost}@{}}%
\column{3}{@{}>{\hspre}l<{\hspost}@{}}%
\column{E}{@{}>{\hspre}l<{\hspost}@{}}%
\>[3]{}\Varid{run_{\mathsf{inc}}}\;\mathrm{8}\;({\mathbf{with}}\;\Varid{h_{\mathsf{except}}}_{\text{\xmark}}\;{\mathbf{handle}}\;\Varid{c_{\mathsf{catch}}}_{\text{\xmark}})\leadsto^{\ast}(\mathsf{right}\;\text{\ttfamily \char34 fail\char34},\mathrm{11}){}\<[E]%
\ColumnHook
\end{hscode}\resethooks
\indentend %
When handling exceptions \emph{after} state, we would expect
\emph{local} update semantics, where we discard the state updates when
an error is raised.  Thus, the expected result should be
\ensuremath{\mathsf{right}\;(\text{\ttfamily \char34 fail\char34},\mathrm{9})}.
However, we again get the global update semantics:
\indentbegin \begin{hscode}\SaveRestoreHook
\column{B}{@{}>{\hspre}l<{\hspost}@{}}%
\column{3}{@{}>{\hspre}l<{\hspost}@{}}%
\column{E}{@{}>{\hspre}l<{\hspost}@{}}%
\>[3]{}{\mathbf{with}}\;\Varid{h_{\mathsf{except}}}_{\text{\xmark}}\;{\mathbf{handle}}\;(\Varid{run_{\mathsf{inc}}}\;\mathrm{8}\;\Varid{c_{\mathsf{catch}}}_{\text{\xmark}})\leadsto^{\ast}(\mathsf{right}\;\text{\ttfamily \char34 fail\char34},\mathrm{11}){}\<[E]%
\ColumnHook
\end{hscode}\resethooks
\indentend %
How can this be? By implementing \ensuremath{\mathtt{catch}_{\text{\xmark}}} as a handler, we have lost the
separation between syntax and semantics: \ensuremath{\mathtt{catch}_{\text{\xmark}}} is supposed to denote
syntax, but it contains semantics in the form of a handler.
Since we apply \ensuremath{\mathtt{catch}_{\text{\xmark}}} to a computation (\ensuremath{\Varid{c_{\mathsf{catch}}}_{\text{\xmark}}}), any containing
\ensuremath{\mathtt{raise}} will have already been handled by \ensuremath{\mathtt{catch}_{\text{\xmark}}} before \ensuremath{\Varid{h_{\mathsf{inc}}}} is applied.
In other words, we have lost modular composition because the handler of
\ensuremath{\mathtt{raise}} always come before the handler of \ensuremath{\mathtt{inc}}.
As a result, we cannot change the semantics of the interactions of
effects by swaping the order of their handlers, which is certainly
harmful to expressivity.

\subsubsection{Catch as scoped effect}
Let us implement \ensuremath{\mathtt{catch}} as a scoped operation in \ensuremath{\lambda_{\mathit{sc}}}.
\indentbegin \begin{hscode}\SaveRestoreHook
\column{B}{@{}>{\hspre}l<{\hspost}@{}}%
\column{3}{@{}>{\hspre}l<{\hspost}@{}}%
\column{13}{@{}>{\hspre}l<{\hspost}@{}}%
\column{15}{@{}>{\hspre}l<{\hspost}@{}}%
\column{32}{@{}>{\hspre}l<{\hspost}@{}}%
\column{E}{@{}>{\hspre}l<{\hspost}@{}}%
\>[3]{}\Varid{c_{\mathsf{catch}}}\mathrel{=}{}\<[13]%
\>[13]{}\mathbf{do}\;\Varid{incr}\, ;{\mathbf{sc}}\;\mathtt{catch}\;{}\<[32]%
\>[32]{}\text{\ttfamily \char34 Overflow\char34}\;(\Varid{b}\, .\,{}\<[E]%
\\
\>[13]{}\hsindent{2}{}\<[15]%
\>[15]{}\mathbf{if}\;\Varid{b}\;\mathbf{then}\;(\mathbf{do}\;\Varid{incr}\, ;\mathbf{do}\;\Varid{incr}\, ;{\mathbf{return}}\;\text{\ttfamily \char34 success\char34}){}\<[E]%
\\
\>[13]{}\hsindent{2}{}\<[15]%
\>[15]{}\mathbf{else}\;{\mathbf{return}}\;\text{\ttfamily \char34 fail\char34}){}\<[E]%
\ColumnHook
\end{hscode}\resethooks
\indentend %
The scoped computation's true branch is the program that may \emph{raise} exceptions,
while the false branch \emph{deals with} the exception.
Our handler interprets exceptions in terms of a sum type \ensuremath{\mathbf{data}\;\alpha\mathbin{+}\beta\mathrel{=}\mathsf{left}\;\alpha\ \mid\ \mathsf{right}\;\beta}, where \ensuremath{\mathsf{left}\;\Varid{v}} denotes an exception and \ensuremath{\mathsf{right}\;\Varid{v}} a
result.
\indentbegin \begin{hscode}\SaveRestoreHook
\column{B}{@{}>{\hspre}l<{\hspost}@{}}%
\column{3}{@{}>{\hspre}l<{\hspost}@{}}%
\column{6}{@{}>{\hspre}l<{\hspost}@{}}%
\column{10}{@{}>{\hspre}c<{\hspost}@{}}%
\column{10E}{@{}l@{}}%
\column{13}{@{}>{\hspre}l<{\hspost}@{}}%
\column{22}{@{}>{\hspre}l<{\hspost}@{}}%
\column{27}{@{}>{\hspre}l<{\hspost}@{}}%
\column{31}{@{}>{\hspre}l<{\hspost}@{}}%
\column{38}{@{}>{\hspre}l<{\hspost}@{}}%
\column{57}{@{}>{\hspre}l<{\hspost}@{}}%
\column{E}{@{}>{\hspre}l<{\hspost}@{}}%
\>[B]{}\Varid{h_{\mathsf{except}}}{}\<[10]%
\>[10]{}\mathbin{:}{}\<[10E]%
\>[13]{}\forall\;\alpha\;\mu\, .\,\alpha\mathbin{!}\langle\mathtt{raise}\, ;\mathtt{catch}\, ;\mu\rangle\Rightarrow \mathsf{String}\mathbin{+}\alpha\hspace{0.1em}!\hspace{0.1em}\langle\mu\rangle{}\<[E]%
\\
\>[B]{}\Varid{h_{\mathsf{except}}}{}\<[10]%
\>[10]{}\mathrel{=}{}\<[10E]%
\>[13]{}{\mathbf{handler}}{}\<[E]%
\\
\>[B]{}\hsindent{3}{}\<[3]%
\>[3]{}\{\mskip1.5mu {}\<[6]%
\>[6]{}{\mathbf{return}}\;\Varid{x}{}\<[22]%
\>[22]{}\mapsto{}\<[27]%
\>[27]{}\mathsf{right}\;\Varid{x}{}\<[E]%
\\
\>[B]{}\hsindent{3}{}\<[3]%
\>[3]{},{}\<[6]%
\>[6]{}{\mathbf{op}}\;\mathtt{raise}\;\Varid{e}\;\anonymous {}\<[22]%
\>[22]{}\mapsto{}\<[27]%
\>[27]{}\mathsf{left}\;\Varid{e}{}\<[E]%
\\
\>[B]{}\hsindent{3}{}\<[3]%
\>[3]{},{}\<[6]%
\>[6]{}{\mathbf{sc}}\;\mathtt{catch}\;\Varid{e}\;\Varid{p}\;\Varid{k}{}\<[22]%
\>[22]{}\mapsto{}\<[27]%
\>[27]{}\mathbf{do}\;{}\<[31]%
\>[31]{}\Varid{x}\leftarrow \Varid{p}\;\mathsf{true}\, ;{}\<[E]%
\\
\>[27]{}\mathbf{case}\;\Varid{x}\;\mathbf{of}\;{}\<[38]%
\>[38]{}\mathsf{left}\;\Varid{e'}\mid \Varid{e'}\mathrel{=}\Varid{e}{}\<[57]%
\>[57]{}\to \mathsf{exceptMap}\;(\Varid{p}\;\mathsf{false})\;\Varid{k}{}\<[E]%
\\
\>[38]{}\anonymous {}\<[57]%
\>[57]{}\to \mathsf{exceptMap}\;\Varid{x}\;\Varid{k}{}\<[E]%
\\
\>[B]{}\hsindent{3}{}\<[3]%
\>[3]{},{\mathbf{bind}}\;\Varid{x}\;\Varid{k}{}\<[22]%
\>[22]{}\mapsto\mathsf{exceptMap}\;\Varid{x}\;\Varid{k}\mskip1.5mu\}{}\<[E]%
\ColumnHook
\end{hscode}\resethooks
\indentend   %
The return clause and algebraic operation clause for \ensuremath{\mathtt{raise}} construct a
return value and raise an exception \ensuremath{\Varid{e}} by calling the \ensuremath{\mathsf{right}} and \ensuremath{\mathsf{left}} constructors, respectively.
The scoped operation clause for \ensuremath{\mathtt{catch}} catches an exception \ensuremath{\Varid{e}}.
If the scoped computation in \ensuremath{\Varid{p}\;\mathsf{true}} raises an exception \ensuremath{\Varid{e}}, it is caught by \ensuremath{\mathtt{catch}}
and replaced by the scoped computation \ensuremath{(\Varid{p}\;\mathsf{false})}.
Otherwise, it continues with \ensuremath{\Varid{p}\;\mathsf{true}} and its results are passed to the
continuation \ensuremath{\Varid{k}}.
The forwarding clause follows the same principle as that of how we
write the forwarding clause for \ensuremath{\Varid{h_{\mathsf{inc}}}} in
\Cref{sec:motivation-forwarding}.
We just need to think about how to continue with \ensuremath{\Varid{k}} when we have the
result \ensuremath{\Varid{x}} that potentially fails.
The most intuitive way is to return the exception if \ensuremath{\Varid{x}} fails
(\ensuremath{\mathsf{left}\;\Varid{e}}), and we run the continuation \ensuremath{\Varid{k}} with the result if \ensuremath{\Varid{x}}
succeeds (\ensuremath{\mathsf{right}\;\Varid{y}}).
\indentbegin \begin{hscode}\SaveRestoreHook
\column{B}{@{}>{\hspre}l<{\hspost}@{}}%
\column{9}{@{}>{\hspre}c<{\hspost}@{}}%
\column{9E}{@{}l@{}}%
\column{12}{@{}>{\hspre}l<{\hspost}@{}}%
\column{25}{@{}>{\hspre}l<{\hspost}@{}}%
\column{34}{@{}>{\hspre}l<{\hspost}@{}}%
\column{E}{@{}>{\hspre}l<{\hspost}@{}}%
\>[B]{}\mathsf{exceptMap}{}\<[9]%
\>[9]{}\mathbin{:}{}\<[9E]%
\>[12]{}\forall\;\alpha\;\beta\;\mu\, .\,\mathsf{String}\mathbin{+}\beta\;{\rightarrow}^{\mu}\;(\beta\;{\rightarrow}^{\mu}\;\mathsf{String}\mathbin{+}\alpha)\;{\rightarrow}^{\mu}\;\mathsf{String}\mathbin{+}\alpha{}\<[E]%
\\
\>[B]{}\mathsf{exceptMap}\;\Varid{x}\;\Varid{k}\mathrel{=}\mathbf{case}\;\Varid{x}\;\mathbf{of}\;{}\<[25]%
\>[25]{}\mathsf{left}\;\Varid{e}{}\<[34]%
\>[34]{}\to \mathsf{left}\;\Varid{e}{}\<[E]%
\\
\>[25]{}\mathsf{right}\;\Varid{y}{}\<[34]%
\>[34]{}\to \Varid{k}\;\Varid{y}{}\<[E]%
\ColumnHook
\end{hscode}\resethooks
\indentend %

Given an initial counter value \ensuremath{\mathrm{8}}, we can handle the program \ensuremath{\Varid{c_{\mathsf{catch}}}} with
\ensuremath{\Varid{h_{\mathsf{except}}}} and \ensuremath{\Varid{h_{\mathsf{inc}}}}.
Different orders of the application of handlers give us different semantics of
the interaction of effects \cite{wu14}.
Handling exceptions before increments gives us global updates:
\indentbegin \begin{hscode}\SaveRestoreHook
\column{B}{@{}>{\hspre}l<{\hspost}@{}}%
\column{3}{@{}>{\hspre}l<{\hspost}@{}}%
\column{E}{@{}>{\hspre}l<{\hspost}@{}}%
\>[3]{}\Varid{run_{\mathsf{inc}}}\;\mathrm{8}\;({\mathbf{with}}\;\Varid{h_{\mathsf{except}}}\;{\mathbf{handle}}\;\Varid{c_{\mathsf{catch}}})\leadsto^{\ast}(\mathsf{right}\;\text{\ttfamily \char34 fail\char34},\mathrm{11}){}\<[E]%
\ColumnHook
\end{hscode}\resethooks
\indentend %
Although an exception is raised and caught, the final value is still updated to
\ensuremath{\mathrm{11}} by the two \ensuremath{\mathtt{inc}} operations and exceeds the maximum value of our counter.
When handling exceptions after increments, we obtain the expected local update semantics:
\indentbegin \begin{hscode}\SaveRestoreHook
\column{B}{@{}>{\hspre}l<{\hspost}@{}}%
\column{3}{@{}>{\hspre}l<{\hspost}@{}}%
\column{E}{@{}>{\hspre}l<{\hspost}@{}}%
\>[3]{}{\mathbf{with}}\;\Varid{h_{\mathsf{except}}}\;{\mathbf{handle}}\;(\Varid{run_{\mathsf{inc}}}\;\mathrm{8}\;\Varid{c_{\mathsf{catch}}})\leadsto^{\ast}\mathsf{right}\;(\text{\ttfamily \char34 fail\char34},\mathrm{9}){}\<[E]%
\ColumnHook
\end{hscode}\resethooks
\indentend The state updates introduced by the two invocations of \ensuremath{\mathtt{inc}} inside
the scope of \ensuremath{\mathtt{catch}} are discarded. This correctly reflects our
intuition when the handler of \ensuremath{\mathtt{inc}} comes before the handler of
\ensuremath{\mathtt{raise}}.

\subsection{Reader with Local}
\label{eg:local}
Reader entails an \ensuremath{\mathtt{ask}} operation that lets one read the (integer) state that is
passed around.
The scoped effect \ensuremath{\mathtt{local}} takes a function \ensuremath{\Varid{f}} which alters the state, and a
computation for which the state should be altered, after which the state should
be returned to its original state.\footnote{The operation signature of \ensuremath{\mathtt{local}} requires polymorphic
parameter types like \ensuremath{\mathtt{local}\mathbin{:}(\forall\;\mu\, .\,\mathsf{Int}\;{\rightarrow}^{\mu}\;\mathsf{Int})\rightarrowtriangle\mathsf{()}}, which we do not support. It is easy to extend operations in
\ensuremath{\lambda_{\mathit{sc}}} with prenex polymorphic parameter types without any need of
other mechanism for higher-rank polymorphism.}
For example, in \ensuremath{{\mathbf{sc}}\;\mathtt{local}\;(\boldsymbol{\lambda}\Varid{i}\, .\,\Varid{i}\mathbin{*}\mathrm{2})\;({\mathbf{op}}\;\mathtt{ask}\;())\;({\mathbf{op}}\;\mathtt{ask}\;())}, the first ask
receives a state that is doubled, whereas the second ask receives the original
state.
To exemplify the problems that arise when implementing \ensuremath{\mathtt{local}} as a handler, our
example uses effect \ensuremath{\mathtt{foo}}, which is simply mapped to \ensuremath{\mathtt{ask}} by \ensuremath{\Varid{h_{\mathsf{foo}}}}:
\indentbegin \begin{hscode}\SaveRestoreHook
\column{B}{@{}>{\hspre}l<{\hspost}@{}}%
\column{3}{@{}>{\hspre}l<{\hspost}@{}}%
\column{19}{@{}>{\hspre}l<{\hspost}@{}}%
\column{33}{@{}>{\hspre}l<{\hspost}@{}}%
\column{E}{@{}>{\hspre}l<{\hspost}@{}}%
\>[3]{}\Varid{h_{\mathsf{foo}}}\mathrel{=}{\mathbf{handler}}\;{}\<[19]%
\>[19]{}\{\mskip1.5mu {\mathbf{return}}\;\Varid{x}{}\<[33]%
\>[33]{}\mapsto{\mathbf{return}}\;\Varid{x}{}\<[E]%
\\
\>[19]{},{\mathbf{op}}\;\mathtt{foo}\;\anonymous \;\Varid{k}{}\<[33]%
\>[33]{}\mapsto\mathbf{do}\;\Varid{x}\leftarrow {\mathbf{op}}\;\mathtt{ask}\;\mathsf{()}\;(\Varid{y}\, .\,\Varid{k}\;\Varid{y}){}\<[E]%
\\
\>[19]{},{\mathbf{bind}}\;\Varid{x}\;\Varid{k}{}\<[33]%
\>[33]{}\mapsto\Varid{k}\;\Varid{x}\mskip1.5mu\}{}\<[E]%
\ColumnHook
\end{hscode}\resethooks
\indentend \subsubsection{Local as a handler}
Whereas the lack of effect interaction control in example of \ensuremath{\mathtt{catch}} as a
handler could be described as unfortunate, in the case for ask there is arguably
only one correct interaction, which is not the one that arises from scoped
effects as handlers. Consider \ensuremath{\Varid{c_{\mathsf{local}}}} below, which includes \ensuremath{\mathtt{foo}}, which is
mapped to \ensuremath{\mathtt{ask}} by \ensuremath{\Varid{h_{\mathsf{foo}}}}.
\indentbegin \begin{hscode}\SaveRestoreHook
\column{B}{@{}>{\hspre}l<{\hspost}@{}}%
\column{3}{@{}>{\hspre}l<{\hspost}@{}}%
\column{18}{@{}>{\hspre}l<{\hspost}@{}}%
\column{27}{@{}>{\hspre}c<{\hspost}@{}}%
\column{27E}{@{}l@{}}%
\column{30}{@{}>{\hspre}l<{\hspost}@{}}%
\column{38}{@{}>{\hspre}c<{\hspost}@{}}%
\column{38E}{@{}l@{}}%
\column{41}{@{}>{\hspre}l<{\hspost}@{}}%
\column{44}{@{}>{\hspre}l<{\hspost}@{}}%
\column{51}{@{}>{\hspre}l<{\hspost}@{}}%
\column{52}{@{}>{\hspre}l<{\hspost}@{}}%
\column{70}{@{}>{\hspre}l<{\hspost}@{}}%
\column{77}{@{}>{\hspre}l<{\hspost}@{}}%
\column{84}{@{}>{\hspre}c<{\hspost}@{}}%
\column{84E}{@{}l@{}}%
\column{E}{@{}>{\hspre}l<{\hspost}@{}}%
\>[3]{}\mathtt{local}_{\text{\xmark}}\;\Varid{f}\;\Varid{c}\;\equiv\;{\mathbf{with}}\;{\mathbf{handler}}\;{}\<[38]%
\>[38]{}\{\mskip1.5mu {}\<[38E]%
\>[41]{}{\mathbf{return}}\;\Varid{x}{}\<[52]%
\>[52]{}\mapsto{\mathbf{return}}\;\Varid{x}{}\<[E]%
\\
\>[38]{},{}\<[38E]%
\>[41]{}{\mathbf{op}}\;\mathtt{ask}\;\anonymous {}\<[52]%
\>[52]{}\mapsto\Varid{x}\leftarrow \mathtt{ask}\, ;\Varid{f}\;\Varid{x}\mskip1.5mu\}\;{\mathbf{handle}}\;\Varid{c}{}\<[E]%
\\
\>[3]{}\Varid{h_{\mathsf{read}}}_{\text{\xmark}}\mathrel{=}{\mathbf{handler}}\;{}\<[27]%
\>[27]{}\{\mskip1.5mu {}\<[27E]%
\>[30]{}{\mathbf{return}}\;\Varid{x}\mapsto{}\<[44]%
\>[44]{}\boldsymbol{\lambda}\Varid{s}\, .\,{}\<[51]%
\>[51]{}\Varid{x},{\mathbf{op}}\;\mathtt{ask}\;\anonymous \;\Varid{k}\mapsto{}\<[70]%
\>[70]{}\boldsymbol{\lambda}\Varid{s}\, .\,{}\<[77]%
\>[77]{}\Varid{k}\;\Varid{s}\;\Varid{s}{}\<[84]%
\>[84]{}\mskip1.5mu\}{}\<[84E]%
\\
\>[3]{}\Varid{run_{\mathsf{read}}}_{\text{\xmark}}\;\Varid{s}\;\Varid{c}\;\equiv\;\mathbf{do}\;\Varid{c'}\leftarrow {\mathbf{with}}\;\Varid{h_{\mathsf{read}}}_{\text{\xmark}}\;{\mathbf{handle}}\;\Varid{c}\, ;\Varid{c'}\;\Varid{s}{}\<[E]%
\\
\>[3]{}\Varid{c_{\mathsf{local}}}_{\text{\xmark}}\mathrel{=}{}\<[18]%
\>[18]{}\mathbf{do}\;\Varid{x}\leftarrow {\mathbf{op}}\;\mathtt{ask}\;()\, ;\mathbf{do}\;\Varid{y}\leftarrow {\mathbf{op}}\;\mathtt{foo}\;(){}\<[E]%
\\
\>[18]{}\mathtt{local}_{\text{\xmark}}\;(\boldsymbol{\lambda}\Varid{a}\to \mathrm{2}\mathbin{*}\Varid{a})\;(\Varid{z}\leftarrow {\mathbf{op}}\;\mathtt{ask}\;()\, ;\Varid{u}\leftarrow {\mathbf{op}}\;\mathtt{foo}\;()\, ;{\mathbf{return}}\;(\Varid{x},\Varid{y},\Varid{z},\Varid{u})){}\<[E]%
\ColumnHook
\end{hscode}\resethooks
\indentend Since \ensuremath{\Varid{h_{\mathsf{foo}}}} introduces \ensuremath{\mathtt{ask}}, we must (re)apply \ensuremath{\Varid{h_{\mathsf{read}}}} after applying \ensuremath{\Varid{h_{\mathsf{foo}}}}.
Since \ensuremath{\mathtt{foo}} is mapped to \ensuremath{\mathtt{ask}}, in \ensuremath{\Varid{c_{\mathsf{local}}}_{\text{\xmark}}} we expect \ensuremath{\Varid{x}} to be equal to
\ensuremath{\Varid{y}}, and \ensuremath{\Varid{z}} equal to \ensuremath{\Varid{u}}.
Starting with the reader state set to 1, we expect the result \ensuremath{(\mathrm{1},\mathrm{1},\mathrm{2},\mathrm{2})}.
Instead, we get:
\indentbegin \begin{hscode}\SaveRestoreHook
\column{B}{@{}>{\hspre}l<{\hspost}@{}}%
\column{3}{@{}>{\hspre}l<{\hspost}@{}}%
\column{E}{@{}>{\hspre}l<{\hspost}@{}}%
\>[3]{}\Varid{run_{\mathsf{read}}}\;\mathrm{1}\;({\mathbf{with}}\;\Varid{h_{\mathsf{foo}}}\;{\mathbf{handle}}\;\Varid{c_{\mathsf{local}}}_{\text{\xmark}})\leadsto^{\ast}{\mathbf{return}}\;(\mathrm{1},\mathrm{1},\mathrm{2},\mathrm{1}){}\<[E]%
\ColumnHook
\end{hscode}\resethooks
\indentend 
Again, how can this be?
The cause is the same as the example with \ensuremath{\mathtt{catch}}: since we encode the semantics
of \ensuremath{\mathtt{local}_{\text{\xmark}}} in its definition, we are forced to perform the handling at the
moment of application.
Notice that \ensuremath{\mathtt{foo}} is not caught by \ensuremath{\mathtt{local}_{\text{\xmark}}}!
Therefore, \ensuremath{\Varid{f}} is only applied to \ensuremath{\mathtt{ask}}.
When \ensuremath{\mathtt{foo}} is mapped to \ensuremath{\mathtt{ask}} by \ensuremath{\Varid{h_{\mathsf{foo}}}}, \ensuremath{\mathtt{local}_{\text{\xmark}}}'s effect will already have
been triggered, which is why \ensuremath{\Varid{f}} is not applied to it.

\subsubsection{Local as a scoped effect}
Using a scoped effect we can properly encode \ensuremath{\mathtt{local}}:
\label{eq:hread}
\indentbegin \begin{hscode}\SaveRestoreHook
\column{B}{@{}>{\hspre}l<{\hspost}@{}}%
\column{8}{@{}>{\hspre}c<{\hspost}@{}}%
\column{8E}{@{}l@{}}%
\column{11}{@{}>{\hspre}l<{\hspost}@{}}%
\column{21}{@{}>{\hspre}l<{\hspost}@{}}%
\column{22}{@{}>{\hspre}c<{\hspost}@{}}%
\column{22E}{@{}l@{}}%
\column{25}{@{}>{\hspre}l<{\hspost}@{}}%
\column{34}{@{}>{\hspre}l<{\hspost}@{}}%
\column{53}{@{}>{\hspre}c<{\hspost}@{}}%
\column{53E}{@{}l@{}}%
\column{58}{@{}>{\hspre}l<{\hspost}@{}}%
\column{63}{@{}>{\hspre}l<{\hspost}@{}}%
\column{65}{@{}>{\hspre}l<{\hspost}@{}}%
\column{69}{@{}>{\hspre}l<{\hspost}@{}}%
\column{95}{@{}>{\hspre}c<{\hspost}@{}}%
\column{95E}{@{}l@{}}%
\column{E}{@{}>{\hspre}l<{\hspost}@{}}%
\>[B]{}\Varid{h_{\mathsf{read}}}{}\<[8]%
\>[8]{}\mathbin{:}{}\<[8E]%
\>[11]{}\forall\;\alpha\;\mu\, .\,{}\<[34]%
\>[34]{}\alpha\mathbin{!}\langle\mathtt{ask}\, ;\mathtt{local}\, ;\mu\rangle\Rightarrow {}\<[63]%
\>[63]{}(\mathsf{Int}\;{\rightarrow}^{\mu}\;\alpha)\hspace{0.1em}!\hspace{0.1em}\langle\mu\rangle{}\<[E]%
\\
\>[B]{}\Varid{h_{\mathsf{read}}}{}\<[8]%
\>[8]{}\mathrel{=}{}\<[8E]%
\>[11]{}{\mathbf{handler}}\;{}\<[22]%
\>[22]{}\{\mskip1.5mu {}\<[22E]%
\>[25]{}{\mathbf{return}}\;\Varid{x}{}\<[53]%
\>[53]{}\mapsto{}\<[53E]%
\>[58]{}\boldsymbol{\lambda}\Varid{s}\, .\,{}\<[65]%
\>[65]{}\Varid{x}{}\<[E]%
\\
\>[22]{},{}\<[22E]%
\>[25]{}{\mathbf{op}}\;\mathtt{ask}\;\anonymous \;\Varid{k}{}\<[53]%
\>[53]{}\mapsto{}\<[53E]%
\>[58]{}\boldsymbol{\lambda}\Varid{s}\, .\,{}\<[65]%
\>[65]{}\Varid{k}\;\Varid{s}\;\Varid{s}{}\<[E]%
\\
\>[22]{},{}\<[22E]%
\>[25]{}{\mathbf{sc}}\;\mathtt{local}\;\Varid{f}\;\Varid{p}\;\Varid{k}{}\<[53]%
\>[53]{}\mapsto{}\<[53E]%
\>[58]{}\boldsymbol{\lambda}\Varid{s}\, .\,{}\<[65]%
\>[65]{}\mathbf{do}\;{}\<[69]%
\>[69]{}\Varid{x}\leftarrow \Varid{p}\;\mathsf{()}\;(\Varid{f}\;\Varid{s})\, ;\Varid{k}\;\Varid{x}\;\Varid{s}{}\<[E]%
\\
\>[22]{},{}\<[22E]%
\>[25]{}{\mathbf{fwd}}\;\Varid{f}\;\Varid{p}\;\Varid{k}{}\<[53]%
\>[53]{}\mapsto{}\<[53E]%
\>[58]{}\boldsymbol{\lambda}\Varid{s}\, .\,{}\<[65]%
\>[65]{}\Varid{f}\;(\boldsymbol{\lambda}\Varid{y}\, .\,\Varid{p}\;\Varid{y}\;\Varid{s},\boldsymbol{\lambda}\Varid{z}\, .\,\Varid{k}\;\Varid{z}\;\Varid{s}){}\<[95]%
\>[95]{}\mskip1.5mu\}{}\<[95E]%
\\
\>[B]{}\Varid{run_{\mathsf{read}}}\;\Varid{s}\;\Varid{c}\;\equiv\;\mathbf{do}\;\Varid{c'}\leftarrow {\mathbf{with}}\;\Varid{h_{\mathsf{read}}}\;{\mathbf{handle}}\;\Varid{c}\, ;\Varid{c'}\;\Varid{s}{}\<[E]%
\\
\>[B]{}\Varid{c_{\mathsf{local}}}\mathrel{=}{}\<[11]%
\>[11]{}\mathbf{do}\;\Varid{x}\leftarrow {\mathbf{op}}\;\mathtt{ask}\;()\, ;\mathbf{do}\;\Varid{y}\leftarrow {\mathbf{op}}\;\mathtt{foo}\;()\, ;{}\<[E]%
\\
\>[11]{}{\mathbf{sc}}\;\mathtt{local}\;{}\<[21]%
\>[21]{}(\boldsymbol{\lambda}\Varid{a}\to \mathrm{2}\mathbin{*}\Varid{a})\;{}\<[E]%
\\
\>[21]{}(\mathbf{do}\;\Varid{z}\leftarrow {\mathbf{op}}\;\mathtt{ask}\;()\, ;\mathbf{do}\;\Varid{u}\leftarrow {\mathbf{op}}\;\mathtt{foo}\;()\, ;{\mathbf{return}}\;(\Varid{x},\Varid{y},\Varid{z},\Varid{u})){}\<[E]%
\ColumnHook
\end{hscode}\resethooks
\indentend %
Note that the forwarding clause of \ensuremath{\Varid{h_{\mathsf{read}}}} is similar to that of
\ensuremath{\Varid{h_{\mathsf{inc}}}} in \Cref{sec:hinc-revisit-again} except for passing the same
state to the scoped computation \ensuremath{\Varid{p}} and continuation \ensuremath{\Varid{k}}.
This makes natural sense because we definitely do not want the scopes
of other irrelevant operations to change what can be read.
Since \ensuremath{\mathtt{local}} is now purely syntactic, we can apply \ensuremath{\Varid{h_{\mathsf{foo}}}} before
\ensuremath{\Varid{h_{\mathsf{read}}}}, and have \ensuremath{\Varid{h_{\mathsf{read}}}} handle the \ensuremath{\mathtt{ask}} that \ensuremath{\Varid{h_{\mathsf{foo}}}} outputs:
\indentbegin \begin{hscode}\SaveRestoreHook
\column{B}{@{}>{\hspre}l<{\hspost}@{}}%
\column{3}{@{}>{\hspre}c<{\hspost}@{}}%
\column{3E}{@{}l@{}}%
\column{7}{@{}>{\hspre}l<{\hspost}@{}}%
\column{20}{@{}>{\hspre}l<{\hspost}@{}}%
\column{30}{@{}>{\hspre}l<{\hspost}@{}}%
\column{33}{@{}>{\hspre}l<{\hspost}@{}}%
\column{E}{@{}>{\hspre}l<{\hspost}@{}}%
\>[7]{}\Varid{run_{\mathsf{read}}}\;\mathrm{1}\;({\mathbf{with}}\;\Varid{h_{\mathsf{foo}}}\;{\mathbf{handle}}\;\Varid{c_{\mathsf{local}}}){}\<[E]%
\\
\>[3]{}\leadsto{}\<[3E]%
\>[7]{}\Varid{run_{\mathsf{read}}}\;\mathrm{1}\;({}\<[20]%
\>[20]{}\Varid{x}\leftarrow {\mathbf{op}}\;\mathtt{ask}\;()\, ;\Varid{y}\leftarrow {\mathbf{op}}\;\mathtt{ask}\;()\, ;{}\<[E]%
\\
\>[20]{}{\mathbf{sc}}\;\mathtt{local}\;{}\<[30]%
\>[30]{}(\boldsymbol{\lambda}\Varid{a}\to \mathrm{2}\mathbin{*}\Varid{a})\;{}\<[E]%
\\
\>[30]{}({}\<[33]%
\>[33]{}\mathbf{do}\;\Varid{z}\leftarrow {\mathbf{op}}\;\mathtt{ask}\;()\, ;\mathbf{do}\;\Varid{u}\leftarrow {\mathbf{op}}\;\mathtt{ask}\;()\, ;{\mathbf{return}}\;(\Varid{x},\Varid{y},\Varid{z},\Varid{u}))){}\<[E]%
\\
\>[3]{}\leadsto{}\<[3E]%
\>[7]{}{\mathbf{return}}\;(\mathrm{1},\mathrm{1},\mathrm{2},\mathrm{2}){}\<[E]%
\ColumnHook
\end{hscode}\resethooks
\indentend 

\subsection{Nondeterminism with Cut}
\label{eg:call-cut}

The algebraic operation \ensuremath{\mathtt{cut}\mathbin{:}\mathsf{()}\rightarrowtriangle\mathsf{()}} provides a different flavor of pruning nondeterminism
that has its origin as a Prolog primitive.
The idea is that \ensuremath{\mathtt{cut}} prunes all remaining branches and only allows the current branch to continue.
Typically, we want to keep the effect of \ensuremath{\mathtt{cut}} local. This is achieved
with the scoped operation \ensuremath{\mathtt{call}\mathbin{:}\mathsf{()}\rightarrowtriangle\mathsf{()}}, as proposed by
Wu et al. \cite{wu14}.
To handle \ensuremath{\mathtt{cut}} and \ensuremath{\mathtt{call}}, we use the \ensuremath{\mathsf{CutList}} datatype
\cite{DBLP:journals/jfp/PirogS17}.
\indentbegin \begin{hscode}\SaveRestoreHook
\column{B}{@{}>{\hspre}l<{\hspost}@{}}%
\column{3}{@{}>{\hspre}l<{\hspost}@{}}%
\column{E}{@{}>{\hspre}l<{\hspost}@{}}%
\>[3]{}\mathbf{data}\;\mathsf{CutList}\;\alpha\mathrel{=}\mathsf{opened}\;(\mathsf{List}\;\alpha)\ \mid\ \mathsf{closed}\;(\mathsf{List}\;\alpha){}\<[E]%
\ColumnHook
\end{hscode}\resethooks
\indentend %
We can think of \ensuremath{\mathsf{opened}\;\Varid{v}} as a list that may be extended and \ensuremath{\mathsf{closed}\;\Varid{v}} as a list
that may not be extended with further elements. This intention is captured in
the \ensuremath{\mathsf{append_{CutList}}} function, which discards the second list if the constructor of the first list is \ensuremath{\mathsf{closed}}.
\indentbegin \begin{hscode}\SaveRestoreHook
\column{B}{@{}>{\hspre}l<{\hspost}@{}}%
\column{9}{@{}>{\hspre}l<{\hspost}@{}}%
\column{20}{@{}>{\hspre}l<{\hspost}@{}}%
\column{31}{@{}>{\hspre}l<{\hspost}@{}}%
\column{39}{@{}>{\hspre}l<{\hspost}@{}}%
\column{E}{@{}>{\hspre}l<{\hspost}@{}}%
\>[B]{}\mathsf{append_{CutList}}\mathbin{:}\forall\;\alpha\;\mu\, .\,\mathsf{CutList}\;\alpha\;{\rightarrow}^{\mu}\;\mathsf{CutList}\;\alpha\;{\rightarrow}^{\mu}\;\mathsf{CutList}\;\alpha{}\<[E]%
\\
\>[B]{}\mathsf{append_{CutList}}\;{}\<[9]%
\>[9]{}(\mathsf{opened}\;\Varid{xs})\;{}\<[20]%
\>[20]{}(\mathsf{opened}\;\Varid{ys}){}\<[31]%
\>[31]{}\mathrel{=}\mathsf{opened}\;{}\<[39]%
\>[39]{}(\Varid{xs}+\hspace{-0.4em}+\Varid{ys}){}\<[E]%
\\
\>[B]{}\mathsf{append_{CutList}}\;{}\<[9]%
\>[9]{}(\mathsf{opened}\;\Varid{xs})\;{}\<[20]%
\>[20]{}(\mathsf{closed}\;\Varid{ys}){}\<[31]%
\>[31]{}\mathrel{=}\mathsf{closed}\;{}\<[39]%
\>[39]{}(\Varid{xs}+\hspace{-0.4em}+\Varid{ys}){}\<[E]%
\\
\>[B]{}\mathsf{append_{CutList}}\;{}\<[9]%
\>[9]{}(\mathsf{closed}\;\Varid{xs})\;{}\<[20]%
\>[20]{}\anonymous {}\<[31]%
\>[31]{}\mathrel{=}\mathsf{closed}\;{}\<[39]%
\>[39]{}\Varid{xs}{}\<[E]%
\ColumnHook
\end{hscode}\resethooks
\indentend 
The handler for nondeterminism with cut is defined as follows:
\indentbegin \begin{hscode}\SaveRestoreHook
\column{B}{@{}>{\hspre}l<{\hspost}@{}}%
\column{7}{@{}>{\hspre}c<{\hspost}@{}}%
\column{7E}{@{}l@{}}%
\column{10}{@{}>{\hspre}l<{\hspost}@{}}%
\column{19}{@{}>{\hspre}c<{\hspost}@{}}%
\column{19E}{@{}l@{}}%
\column{22}{@{}>{\hspre}l<{\hspost}@{}}%
\column{37}{@{}>{\hspre}c<{\hspost}@{}}%
\column{37E}{@{}l@{}}%
\column{42}{@{}>{\hspre}l<{\hspost}@{}}%
\column{E}{@{}>{\hspre}l<{\hspost}@{}}%
\>[B]{}\Varid{h_{\mathsf{cut}}}{}\<[7]%
\>[7]{}\mathbin{:}{}\<[7E]%
\>[10]{}\forall\;\alpha\;\mu\, .\,\alpha\mathbin{!}\langle\mathtt{choose}\, ;\mathtt{fail}\, ;\mathtt{cut}\, ;\mathtt{call}\, ;\mu\rangle\Rightarrow \mathsf{CutList}\;\alpha\hspace{0.1em}!\hspace{0.1em}\langle\mu\rangle{}\<[E]%
\\
\>[B]{}\Varid{h_{\mathsf{cut}}}{}\<[7]%
\>[7]{}\mathrel{=}{}\<[7E]%
\>[10]{}{\mathbf{handler}}\;{}\<[19]%
\>[19]{}\{\mskip1.5mu {}\<[19E]%
\>[22]{}{\mathbf{return}}\;\Varid{x}{}\<[37]%
\>[37]{}\mapsto{}\<[37E]%
\>[42]{}\mathsf{opened}\;[\mskip1.5mu \Varid{x}\mskip1.5mu]{}\<[E]%
\\
\>[19]{},{}\<[19E]%
\>[22]{}{\mathbf{op}}\;\mathtt{fail}\;\anonymous \;\anonymous {}\<[37]%
\>[37]{}\mapsto{}\<[37E]%
\>[42]{}\mathsf{opened}\;[\mskip1.5mu \mskip1.5mu]{}\<[E]%
\\
\>[19]{},{}\<[19E]%
\>[22]{}{\mathbf{op}}\;\mathtt{choose}\;\Varid{x}\;\Varid{k}{}\<[37]%
\>[37]{}\mapsto{}\<[37E]%
\>[42]{}\mathbf{do}\;\Varid{xs}\leftarrow \Varid{k}\;\mathsf{true}\, ;{}\<[E]%
\\
\>[42]{}\mathbf{if}\;\mathsf{isclose}\;\Varid{xs}\;\mathbf{then}\;\Varid{xs}\;\mathbf{else}\;\mathsf{append_{CutList}}\;\Varid{xs}\;(\Varid{k}\;\mathsf{false}){}\<[E]%
\\
\>[19]{},{}\<[19E]%
\>[22]{}{\mathbf{op}}\;\mathtt{cut}\;\anonymous \;\Varid{k}{}\<[37]%
\>[37]{}\mapsto{}\<[37E]%
\>[42]{}\mathsf{close}\;(\Varid{k}\;\mathsf{()}){}\<[E]%
\\
\>[19]{},{}\<[19E]%
\>[22]{}{\mathbf{sc}}\;\mathtt{call}\;\anonymous \;\Varid{p}\;\Varid{k}{}\<[37]%
\>[37]{}\mapsto{}\<[37E]%
\>[42]{}\mathsf{concatMap_{CutList}}\;(\mathsf{open}\;(\Varid{p}\;\mathsf{()}))\;\Varid{k}{}\<[E]%
\\
\>[19]{},{}\<[19E]%
\>[22]{}{\mathbf{bind}}\;\Varid{x}\;\Varid{k}{}\<[37]%
\>[37]{}\mapsto{}\<[37E]%
\>[42]{}\mathsf{concatMap_{CutList}}\;\Varid{x}\;\Varid{k}\mskip1.5mu\}{}\<[E]%
\ColumnHook
\end{hscode}\resethooks
\indentend %
The operation clause for \ensuremath{\mathtt{cut}} closes the cutlist and the clause for \ensuremath{\mathtt{call}} (re-)opens it
when coming out of the scope.

\noindent
\begin{tabular}{lr}
\begin{minipage}{0.45\linewidth}
\indentbegin \begin{hscode}\SaveRestoreHook
\column{B}{@{}>{\hspre}l<{\hspost}@{}}%
\column{18}{@{}>{\hspre}l<{\hspost}@{}}%
\column{E}{@{}>{\hspre}l<{\hspost}@{}}%
\>[B]{}\mathsf{close}\mathbin{:}\forall\;\alpha\;\mu\, .\,\mathsf{CutList}\;\alpha\;{\rightarrow}^{\mu}\;\mathsf{CutList}\;\alpha{}\<[E]%
\\
\>[B]{}\mathsf{close}\;(\mathsf{closed}\;\Varid{as}){}\<[18]%
\>[18]{}\mathrel{=}\mathsf{closed}\;\Varid{as}{}\<[E]%
\\
\>[B]{}\mathsf{close}\;(\mathsf{opened}\;\Varid{as}){}\<[18]%
\>[18]{}\mathrel{=}\mathsf{closed}\;\Varid{as}{}\<[E]%
\ColumnHook
\end{hscode}\resethooks
\indentend \end{minipage}
\begin{minipage}{0.45\linewidth}
\indentbegin \begin{hscode}\SaveRestoreHook
\column{B}{@{}>{\hspre}l<{\hspost}@{}}%
\column{17}{@{}>{\hspre}l<{\hspost}@{}}%
\column{E}{@{}>{\hspre}l<{\hspost}@{}}%
\>[B]{}\mathsf{open}\mathbin{:}\forall\;\alpha\;\mu\, .\,\mathsf{CutList}\;\alpha\;{\rightarrow}^{\mu}\;\mathsf{CutList}\;\alpha{}\<[E]%
\\
\>[B]{}\mathsf{open}\;(\mathsf{closed}\;\Varid{as}){}\<[17]%
\>[17]{}\mathrel{=}\mathsf{opened}\;\Varid{as}{}\<[E]%
\\
\>[B]{}\mathsf{open}\;(\mathsf{opened}\;\Varid{as}){}\<[17]%
\>[17]{}\mathrel{=}\mathsf{opened}\;\Varid{as}{}\<[E]%
\ColumnHook
\end{hscode}\resethooks
\indentend \end{minipage}
\end{tabular}

The function \ensuremath{\mathsf{isclose}} checks whether a cutlist is closed. It is used
in the \ensuremath{\mathtt{choose}} clause for efficiency; we do not need to execute \ensuremath{\Varid{k}\;\mathsf{false}} when \ensuremath{\Varid{k}\;\mathsf{true}} returns a closed list.
\indentbegin \begin{hscode}\SaveRestoreHook
\column{B}{@{}>{\hspre}l<{\hspost}@{}}%
\column{15}{@{}>{\hspre}l<{\hspost}@{}}%
\column{19}{@{}>{\hspre}l<{\hspost}@{}}%
\column{E}{@{}>{\hspre}l<{\hspost}@{}}%
\>[B]{}\mathsf{isclose}\mathbin{:}\forall\;\alpha\;\mu\, .\,\mathsf{CutList}\;\alpha\;{\rightarrow}^{\mu}\;\mathsf{CutList}\;\alpha{}\<[E]%
\\
\>[B]{}\mathsf{isclose}\;(\mathsf{closed}\;\anonymous ){}\<[19]%
\>[19]{}\mathrel{=}\mathsf{true}{}\<[E]%
\\
\>[B]{}\mathsf{isclose}\;(\mathsf{opened}\;{}\<[15]%
\>[15]{}\anonymous ){}\<[19]%
\>[19]{}\mathrel{=}\mathsf{false}{}\<[E]%
\ColumnHook
\end{hscode}\resethooks
\indentend %

The forwarding of \ensuremath{\Varid{h_{\mathsf{cut}}}} uses the function \ensuremath{\mathsf{concatMap_{CutList}}}, the cutlist
counterpart of \ensuremath{\mathsf{concatMap}} which takes the extensibility of \ensuremath{\mathsf{CutList}}
(signalled by \ensuremath{\mathsf{opened}} and \ensuremath{\mathsf{closed}}) into account when concatenating.
\indentbegin \begin{hscode}\SaveRestoreHook
\column{B}{@{}>{\hspre}l<{\hspost}@{}}%
\column{11}{@{}>{\hspre}c<{\hspost}@{}}%
\column{11E}{@{}l@{}}%
\column{14}{@{}>{\hspre}l<{\hspost}@{}}%
\column{25}{@{}>{\hspre}l<{\hspost}@{}}%
\column{30}{@{}>{\hspre}l<{\hspost}@{}}%
\column{34}{@{}>{\hspre}l<{\hspost}@{}}%
\column{E}{@{}>{\hspre}l<{\hspost}@{}}%
\>[B]{}\mathsf{concatMap_{CutList}}{}\<[11]%
\>[11]{}\mathbin{:}{}\<[11E]%
\>[14]{}\forall\;\alpha\;\beta\;\mu\, .\,\mathsf{CutList}\;\beta\;{\rightarrow}^{\mu}\;(\beta\;{\rightarrow}^{\mu}\;\mathsf{CutList}\;\alpha)\;{\rightarrow}^{\mu}\;\mathsf{CutList}\;\alpha{}\<[E]%
\\
\>[B]{}\mathsf{concatMap_{CutList}}\;(\mathsf{opened}\;[\mskip1.5mu \mskip1.5mu])\;{}\<[25]%
\>[25]{}\Varid{f}\mathrel{=}{\mathbf{return}}\;(\mathsf{opened}\;[\mskip1.5mu \mskip1.5mu]){}\<[E]%
\\
\>[B]{}\mathsf{concatMap_{CutList}}\;(\mathsf{closed}\;[\mskip1.5mu \mskip1.5mu])\;{}\<[25]%
\>[25]{}\Varid{f}\mathrel{=}{\mathbf{return}}\;(\mathsf{closed}\;[\mskip1.5mu \mskip1.5mu]){}\<[E]%
\\
\>[B]{}\mathsf{concatMap_{CutList}}\;(\mathsf{opened}\;(\Varid{b}\mathbin{:}\Varid{bs}))\;{}\<[25]%
\>[25]{}\Varid{f}\mathrel{=}{}\<[30]%
\>[30]{}\mathbf{do}\;{}\<[34]%
\>[34]{}\Varid{as}\leftarrow \Varid{f}\;\Varid{b}\, ;{}\<[E]%
\\
\>[34]{}\Varid{as'}\leftarrow \mathsf{concatMap_{CutList}}\;(\mathsf{opened}\;\Varid{bs})\;\Varid{f}\, ;{}\<[E]%
\\
\>[34]{}\mathsf{append_{CutList}}\;\Varid{as}\;\Varid{as'}{}\<[E]%
\\
\>[B]{}\mathsf{concatMap_{CutList}}\;(\mathsf{closed}\;(\Varid{b}\mathbin{:}\Varid{bs}))\;{}\<[25]%
\>[25]{}\Varid{f}\mathrel{=}{}\<[30]%
\>[30]{}\mathbf{do}\;{}\<[34]%
\>[34]{}\Varid{as}\leftarrow \Varid{f}\;\Varid{b}\, ;{}\<[E]%
\\
\>[34]{}\Varid{as'}\leftarrow \mathsf{concatMap_{CutList}}\;(\mathsf{closed}\;\Varid{bs})\;\Varid{f}\, ;{}\<[E]%
\\
\>[34]{}\mathsf{append_{CutList}}\;\Varid{as}\;\Varid{as'}{}\<[E]%
\ColumnHook
\end{hscode}\resethooks
\indentend %
In \Cref{eg:parser}, we give an example usage of \ensuremath{\mathtt{cut}} to improve parsers.

\subsection{Depth-Bounded Search}
\label{eg:dbs}

The handler \ensuremath{\Varid{h_{\mathsf{ND}}}} for nondeterminism shown in \Cref{sec:background-motivation} implements the \emph{depth-first search} (DFS) strategy.
However, with scoped effects and handlers we can implement other search strategies, such as \emph{depth-bounded search} (DBS) \cite{DBLP:conf/esop/YangPWBS22}, which uses the scoped operation \ensuremath{\mathtt{depth}\mathbin{:}\mathsf{Int}\rightarrowtriangle\mathsf{()}}
to bound the depth of the branches in the scoped computation.
The handler uses return type \ensuremath{\mathsf{Int}\;{\rightarrow}^{\mu}\;\mathsf{List}\;(\alpha,\mathsf{Int})}. Here, the \ensuremath{\mathsf{Int}}
parameter is the current depth bound, and the result is a list of \ensuremath{(\alpha,\mathsf{Int})}
pairs, where \ensuremath{\alpha} denotes the result and \ensuremath{\mathsf{Int}} reflects the remaining global
depth bound.\footnote{The pair type \ensuremath{(\alpha,\mathsf{Int})} differs from Yang et al.
  \cite{DBLP:conf/esop/YangPWBS22}'s \ensuremath{\alpha} in order to enable forwarding.}
\indentbegin \begin{hscode}\SaveRestoreHook
\column{B}{@{}>{\hspre}l<{\hspost}@{}}%
\column{3}{@{}>{\hspre}c<{\hspost}@{}}%
\column{3E}{@{}l@{}}%
\column{6}{@{}>{\hspre}l<{\hspost}@{}}%
\column{9}{@{}>{\hspre}c<{\hspost}@{}}%
\column{9E}{@{}l@{}}%
\column{12}{@{}>{\hspre}l<{\hspost}@{}}%
\column{23}{@{}>{\hspre}c<{\hspost}@{}}%
\column{23E}{@{}l@{}}%
\column{28}{@{}>{\hspre}l<{\hspost}@{}}%
\column{35}{@{}>{\hspre}l<{\hspost}@{}}%
\column{E}{@{}>{\hspre}l<{\hspost}@{}}%
\>[B]{}\Varid{h_{\mathsf{depth}}}{}\<[9]%
\>[9]{}\mathbin{:}{}\<[9E]%
\>[12]{}\forall\;\alpha\;\mu\, .\,\alpha\mathbin{!}\langle\mathtt{choose}\, ;\mathtt{fail}\, ;\mathtt{depth}\, ;\mu\rangle\Rightarrow (\mathsf{Int}\;{\rightarrow}^{\mu}\;\mathsf{List}\;(\alpha,\mathsf{Int}))\hspace{0.1em}!\hspace{0.1em}\langle\mu\rangle{}\<[E]%
\\
\>[B]{}\Varid{h_{\mathsf{depth}}}{}\<[9]%
\>[9]{}\mathrel{=}{}\<[9E]%
\>[12]{}{\mathbf{handler}}{}\<[E]%
\\
\>[B]{}\hsindent{3}{}\<[3]%
\>[3]{}\{\mskip1.5mu {}\<[3E]%
\>[6]{}{\mathbf{return}}\;\Varid{x}{}\<[23]%
\>[23]{}\mapsto{}\<[23E]%
\>[28]{}\boldsymbol{\lambda}\Varid{d}\, .\,[\mskip1.5mu (\Varid{x},\Varid{d})\mskip1.5mu]{}\<[E]%
\\
\>[B]{}\hsindent{3}{}\<[3]%
\>[3]{},{}\<[3E]%
\>[6]{}{\mathbf{op}}\;\mathtt{fail}\;\anonymous \;\anonymous {}\<[23]%
\>[23]{}\mapsto{}\<[23E]%
\>[28]{}\boldsymbol{\lambda}\anonymous \, .\,[\mskip1.5mu \mskip1.5mu]{}\<[E]%
\\
\>[B]{}\hsindent{3}{}\<[3]%
\>[3]{},{}\<[3E]%
\>[6]{}{\mathbf{op}}\;\mathtt{choose}\;\Varid{x}\;\Varid{k}{}\<[23]%
\>[23]{}\mapsto{}\<[23E]%
\>[28]{}\boldsymbol{\lambda}\Varid{d}\, .\,{}\<[35]%
\>[35]{}\mathbf{if}\;\Varid{d}=\mathrm{0}\;\mathbf{then}\;[\mskip1.5mu \mskip1.5mu]\;\mathbf{else}\;\Varid{k}\;\mathsf{true}\;(\Varid{d}\mathbin{-}\mathrm{1})+\hspace{-0.4em}+\Varid{k}\;\mathsf{false}\;(\Varid{d}\mathbin{-}\mathrm{1}){}\<[E]%
\\
\>[B]{}\hsindent{3}{}\<[3]%
\>[3]{},{}\<[3E]%
\>[6]{}{\mathbf{sc}}\;\mathtt{depth}\;\Varid{d'}\;\Varid{p}\;\Varid{k}{}\<[23]%
\>[23]{}\mapsto{}\<[23E]%
\>[28]{}\boldsymbol{\lambda}\Varid{d}\, .\,\mathsf{concatMap}\;(\Varid{p}\;\mathsf{()}\;\Varid{d'})\;(\boldsymbol{\lambda}(\Varid{v},\anonymous )\, .\,\Varid{k}\;\Varid{v}\;\Varid{d}){}\<[E]%
\\
\>[B]{}\hsindent{3}{}\<[3]%
\>[3]{},{}\<[3E]%
\>[6]{}{\mathbf{fwd}}\;\Varid{f}\;\Varid{p}\;\Varid{k}{}\<[23]%
\>[23]{}\mapsto{}\<[23E]%
\>[28]{}\boldsymbol{\lambda}\Varid{d}\, .\,{}\<[35]%
\>[35]{}\Varid{f}\;(\boldsymbol{\lambda}\Varid{y}\, .\,\Varid{p}\;\Varid{y}\;\Varid{d},\boldsymbol{\lambda}\Varid{vs}\, .\,\mathsf{concatMap}\;\Varid{vs}\;(\boldsymbol{\lambda}(\Varid{v},\Varid{d})\, .\,\Varid{k}\;\Varid{v}\;\Varid{d}))\mskip1.5mu\}{}\<[E]%
\ColumnHook
\end{hscode}\resethooks
\indentend   %
For the \ensuremath{\mathtt{depth}} operation, we locally use the given depth bound \ensuremath{\Varid{d'}} for the scoped computation
\ensuremath{\Varid{p}} and go back to using the global depth bound \ensuremath{\Varid{d}} for the continuation \ensuremath{\Varid{k}}.
In case of an unknown scoped operation, the forwarding clause just threads the depth bound through, first into the scoped computation and from there into the continuation.
It is similar to a combination of the forwarding of \ensuremath{\Varid{h_{\mathsf{once}}}} and \ensuremath{\Varid{h_{\mathsf{inc}}}}.

For example, the following program (\Cref{fig:cdepth}) has a local depth bound of \ensuremath{\mathrm{1}} and a global depth bound
of \ensuremath{\mathrm{2}}.
It discards the results \ensuremath{\mathrm{2}} and \ensuremath{\mathrm{3}} in the scoped computation as they appear after
the second \ensuremath{\mathtt{choose}} operation, and similarly, the results
\ensuremath{\mathrm{5}} and \ensuremath{\mathrm{6}} in the continuation are ignored.

\indentbegin \begin{hscode}\SaveRestoreHook
\column{B}{@{}>{\hspre}l<{\hspost}@{}}%
\column{14}{@{}>{\hspre}l<{\hspost}@{}}%
\column{20}{@{}>{\hspre}l<{\hspost}@{}}%
\column{24}{@{}>{\hspre}l<{\hspost}@{}}%
\column{28}{@{}>{\hspre}l<{\hspost}@{}}%
\column{E}{@{}>{\hspre}l<{\hspost}@{}}%
\>[B]{}\Varid{c_{\mathsf{depth}}}\mathrel{=}{\mathbf{sc}}\;{}\<[14]%
\>[14]{}\mathtt{depth}\;\mathrm{1}\;{}\<[E]%
\\
\>[14]{}(\anonymous \, .\,{}\<[20]%
\>[20]{}\mathbf{do}\;{}\<[24]%
\>[24]{}\Varid{b}_{1}{}\<[28]%
\>[28]{}\leftarrow {\mathbf{op}}\;\mathtt{choose}\;\mathsf{()}\, ;\mathbf{if}\;\Varid{b}_{1}\;\mathbf{then}\;{\mathbf{return}}\;\mathrm{1}\;\mathbf{else}{}\<[E]%
\\
\>[20]{}\mathbf{do}\;{}\<[24]%
\>[24]{}\Varid{b}_{2}{}\<[28]%
\>[28]{}\leftarrow {\mathbf{op}}\;\mathtt{choose}\;\mathsf{()}\, ;\mathbf{if}\;\Varid{b}_{2}\;\mathbf{then}\;{\mathbf{return}}\;\mathrm{2}\;\mathbf{else}\;{\mathbf{return}}\;\mathrm{3})\;{}\<[E]%
\\
\>[14]{}(\Varid{x}\, .\,{}\<[20]%
\>[20]{}\mathbf{do}\;{}\<[24]%
\>[24]{}\Varid{b}_{1}{}\<[28]%
\>[28]{}\leftarrow {\mathbf{op}}\;\mathtt{choose}\;\mathsf{()}\, ;\mathbf{if}\;\Varid{b}_{1}\;\mathbf{then}\;{\mathbf{return}}\;\Varid{x}\;\mathbf{else}{}\<[E]%
\\
\>[20]{}\mathbf{do}\;{}\<[24]%
\>[24]{}\Varid{b}_{2}{}\<[28]%
\>[28]{}\leftarrow {\mathbf{op}}\;\mathtt{choose}\;\mathsf{()}\, ;\mathbf{if}\;\Varid{b}_{2}\;\mathbf{then}\;{\mathbf{return}}\;\mathrm{4}\;\mathbf{else}{}\<[E]%
\\
\>[20]{}\mathbf{do}\;{}\<[24]%
\>[24]{}\Varid{b}_{3}{}\<[28]%
\>[28]{}\leftarrow {\mathbf{op}}\;\mathtt{choose}\;\mathsf{()}\, ;\mathbf{if}\;\Varid{b}_{3}\;\mathbf{then}\;{\mathbf{return}}\;\mathrm{5}\;\mathbf{else}\;{\mathbf{return}}\;\mathrm{6}){}\<[E]%
\ColumnHook
\end{hscode}\resethooks
\indentend \vspace{-0.4cm}\indentbegin \begin{hscode}\SaveRestoreHook
\column{B}{@{}>{\hspre}l<{\hspost}@{}}%
\column{3}{@{}>{\hspre}l<{\hspost}@{}}%
\column{E}{@{}>{\hspre}l<{\hspost}@{}}%
\>[3]{}\texttt{>>>}\;({\mathbf{with}}\;\Varid{h_{\mathsf{depth}}}\;{\mathbf{handle}}\;\Varid{c_{\mathsf{depth}}})\;\mathrm{2}{}\<[E]%
\\
\>[3]{}[\mskip1.5mu (\mathrm{1},\mathrm{1}),(\mathrm{4},\mathrm{0})\mskip1.5mu]{}\<[E]%
\ColumnHook
\end{hscode}\resethooks
\indentend The result is \ensuremath{[\mskip1.5mu (\mathrm{1},\mathrm{1}),(\mathrm{4},\mathrm{0})\mskip1.5mu]}, where the tuple's second parameter represents the
global depth bound.
Notice that \ensuremath{\mathtt{choose}} operations in the scoped computation \ensuremath{\mathtt{depth}} do not consume the global depth bound in the handler.
We also show an alternative implementation of the handler clause for \ensuremath{\mathtt{depth}} in the Supplementary Material.

\begin{figure}
\begin{tikzpicture}[node distance=1.5cm]
\small

\node (sc) [] {\ensuremath{{\mathbf{sc}}\;\mathtt{depth}\;\mathrm{1}}};

\node (top) [right of=sc, xshift=1.5cm] {\ensuremath{(\anonymous \, .\,{\mathbf{op}}\;\mathtt{choose}\;())}};
\node (1) [below of=top, left of=top, xshift=0.5cm, yshift=0.75cm] {\ensuremath{{\mathbf{return}}\;\mathrm{1}}};
\node (choose_sc) [below of=top, right of=top, xshift=-0.5cm, yshift=0.75cm] {\ensuremath{{\mathbf{op}}\;\mathtt{choose}\;()}};
\node (2) [below of=choose_sc, left of=choose_sc, xshift=0.5cm, yshift=0.75cm] {\ensuremath{{\mathbf{return}}\;\mathrm{2}}};
\node (3) [below of=choose_sc, right of=choose_sc, xshift=-0.5cm, yshift=0.75cm] {\ensuremath{{\mathbf{return}}\;\mathrm{3}}};

\draw [] (top) -- (1);
\draw [] (top) -- (choose_sc);
\draw [] (choose_sc) -- (2);
\draw [] (choose_sc) -- (3);

\node[draw=gray,thick,dotted,fit=(1) (choose_sc),label={[xshift=1.8cm,yshift=-0.1cm,font=\scriptsize,text=gray]depth 1},inner sep=-0.5mm] {};
\node[draw=gray,thick,dotted,fit=(2) (3),label={[xshift=1.6cm,yshift=-0.1cm,font=\scriptsize,text=gray]depth 2},inner sep=-0.5mm] {};

\node (cont) [right of=top, xshift=3.5cm] {\ensuremath{(\Varid{x}\, .\,{\mathbf{op}}\;\mathtt{choose}\;())}};
\node (x) [below of=cont, left of=cont, xshift=0.5cm, yshift=0.75cm] {\ensuremath{{\mathbf{return}}\;\Varid{x}}};
\node (choose1) [below of=cont, right of=cont, xshift=-0.5cm, yshift=0.75cm] {\ensuremath{{\mathbf{op}}\;\mathtt{choose}\;()}};
\node (4) [below of=choose1, left of=choose1, xshift=0.5cm, yshift=0.75cm] {\ensuremath{{\mathbf{return}}\;\mathrm{4}}};
\node (choose2) [below of=choose1, right of=choose1, xshift=-0.5cm, yshift=0.75cm] {\ensuremath{{\mathbf{op}}\;\mathtt{choose}\;()}};
\node (5) [below of=choose2, left of=choose2, xshift=0.5cm, yshift=0.75cm] {\ensuremath{{\mathbf{return}}\;\mathrm{5}}};
\node (6) [below of=choose2, right of=choose2, xshift=-0.5cm, yshift=0.75cm] {\ensuremath{{\mathbf{return}}\;\mathrm{6}}};

\draw [] (cont) -- (x);
\draw [] (cont) -- (choose1);
\draw [] (choose1) -- (4);
\draw [] (choose1) -- (choose2);
\draw [] (choose2) -- (5);
\draw [] (choose2) -- (6);

\node[draw=gray,thick,dotted,fit=(x) (choose1),label={[xshift=1.8cm,yshift=-0.1cm,font=\scriptsize,text=gray]depth 1},inner sep=-0.5mm] {};
\node[draw=gray,thick,dotted,fit=(4) (choose2),label={[xshift=1.8cm,yshift=-0.1cm,font=\scriptsize,text=gray]depth 2},inner sep=-0.5mm] {};
\node[draw=gray,thick,dotted,fit=(5) (6),label={[xshift=1.6cm,yshift=-0.1cm,font=\scriptsize,text=gray]depth 3},inner sep=-0.5mm] {};

\end{tikzpicture}
\caption{Visual representation of \ensuremath{\Varid{c_{\mathsf{depth}}}}.}
\label{fig:cdepth}
\end{figure}

\subsection{Parsers}
\label{eg:parser}

A parser effect can be achieved by combining the nondeterminism-with-cut effect and
a token-consuming effect \cite{wu14}.
The latter features the algebraic operation \ensuremath{\mathtt{token}\mathbin{:}\mathsf{Char}\rightarrowtriangle\mathsf{Char}} where
\ensuremath{{\mathbf{op}}\;\mathtt{token}\;\Varid{t}} consumes a single character from the implicit input string;
if it is \ensuremath{\Varid{t}}, it is passed on to the continuation; otherwise the operation fails.
The token handler has result type \ensuremath{\mathsf{String}\;{\rightarrow}^{\langle\mathtt{fail}\, ;\mu\rangle}\;(\alpha,\mathsf{String})}:
it threads through the remaining part of the input string. Observe that
the function type signals it may \ensuremath{\mathtt{fail}}, in case the token does not match.
\indentbegin \begin{hscode}\SaveRestoreHook
\column{B}{@{}>{\hspre}l<{\hspost}@{}}%
\column{3}{@{}>{\hspre}c<{\hspost}@{}}%
\column{3E}{@{}l@{}}%
\column{6}{@{}>{\hspre}l<{\hspost}@{}}%
\column{11}{@{}>{\hspre}c<{\hspost}@{}}%
\column{11E}{@{}l@{}}%
\column{14}{@{}>{\hspre}l<{\hspost}@{}}%
\column{21}{@{}>{\hspre}c<{\hspost}@{}}%
\column{21E}{@{}l@{}}%
\column{26}{@{}>{\hspre}l<{\hspost}@{}}%
\column{33}{@{}>{\hspre}l<{\hspost}@{}}%
\column{44}{@{}>{\hspre}l<{\hspost}@{}}%
\column{53}{@{}>{\hspre}l<{\hspost}@{}}%
\column{E}{@{}>{\hspre}l<{\hspost}@{}}%
\>[B]{}\Varid{h_{\mathsf{token}}}{}\<[11]%
\>[11]{}\mathbin{:}{}\<[11E]%
\>[14]{}\forall\;\alpha\;\mu\, .\,\alpha\mathbin{!}\langle\mathtt{token}\, ;\mathtt{fail}\, ;\mu\rangle\Rightarrow (\mathsf{String}\;{\rightarrow}^{\langle\mathtt{fail}\, ;\mu\rangle}\;(\alpha,\mathsf{String}))\hspace{0.1em}!\hspace{0.1em}\langle\mathtt{fail}\, ;\mu\rangle{}\<[E]%
\\
\>[B]{}\Varid{h_{\mathsf{token}}}{}\<[11]%
\>[11]{}\mathrel{=}{}\<[11E]%
\>[14]{}{\mathbf{handler}}{}\<[E]%
\\
\>[B]{}\hsindent{3}{}\<[3]%
\>[3]{}\{\mskip1.5mu {}\<[3E]%
\>[6]{}{\mathbf{return}}\;\Varid{x}{}\<[21]%
\>[21]{}\mapsto{}\<[21E]%
\>[26]{}\boldsymbol{\lambda}\Varid{s}\, .\,{}\<[33]%
\>[33]{}(\Varid{x},\Varid{s}){}\<[E]%
\\
\>[B]{}\hsindent{3}{}\<[3]%
\>[3]{},{}\<[3E]%
\>[6]{}{\mathbf{op}}\;\mathtt{token}\;\Varid{x}\;\Varid{k}{}\<[21]%
\>[21]{}\mapsto{}\<[21E]%
\>[26]{}\boldsymbol{\lambda}\Varid{s}\, .\,{}\<[33]%
\>[33]{}\mathbf{case}\;\Varid{s}\;\mathbf{of}\;{}\<[44]%
\>[44]{}[\mskip1.5mu \mskip1.5mu]{}\<[53]%
\>[53]{}\to \mathsf{failure}\;(){}\<[E]%
\\
\>[44]{}(\Varid{x'}\mathbin{:}\Varid{xs}){}\<[53]%
\>[53]{}\to \mathbf{if}\;\Varid{x}=\Varid{x'}\;\mathbf{then}\;\Varid{k}\;\Varid{x}\;\Varid{xs}\;\mathbf{else}\;\mathsf{failure}\;(){}\<[E]%
\\
\>[B]{}\hsindent{3}{}\<[3]%
\>[3]{},{}\<[3E]%
\>[6]{}{\mathbf{fwd}}\;\Varid{f}\;\Varid{p}\;\Varid{k}{}\<[21]%
\>[21]{}\mapsto{}\<[21E]%
\>[26]{}\boldsymbol{\lambda}\Varid{s}\, .\,{}\<[33]%
\>[33]{}\Varid{f}\;(\boldsymbol{\lambda}\Varid{y}\, .\,\Varid{p}\;\Varid{y}\;\Varid{s},\boldsymbol{\lambda}(\Varid{t},\Varid{s'})\, .\,\Varid{k}\;\Varid{t}\;\Varid{s'})\mskip1.5mu\}{}\<[E]%
\ColumnHook
\end{hscode}\resethooks
\indentend   %

The forwarding clause of \ensuremath{\Varid{h_{\mathsf{token}}}} is the same as that of \ensuremath{\Varid{h_{\mathsf{inc}}}} in
\Cref{sec:hinc-revisit-again}.
We uses the initial state \ensuremath{\Varid{s}} for the scoped computation \ensuremath{\Varid{p}}, and the
updated state \ensuremath{\Varid{s'}} for the continuation \ensuremath{\Varid{k}}.
The updated state \ensuremath{\Varid{s'}} is passed in from the scoped computation \ensuremath{\Varid{p}}
and reflects the consumption of tokens in \ensuremath{\Varid{p}}.
An alternative forwarding semantics is \ensuremath{{\mathbf{fwd}}\;\Varid{f}\;\Varid{p}\;\Varid{k}\mapsto\boldsymbol{\lambda}\Varid{s}\, .\,\Varid{f}\;(\boldsymbol{\lambda}\Varid{y}\, .\,\Varid{p}\;\Varid{y}\;\Varid{s},\boldsymbol{\lambda}(\Varid{t},\Varid{s'})\, .\,\Varid{k}\;\Varid{t}\;\Varid{s})}, where the continuation \ensuremath{\Varid{k}} also uses the
initial state \ensuremath{\Varid{s}}.
In this way, we actually discard all the changes to the state in the
scoped computation, which could potentially be useful to implement
parser lookahead.
It is up to the programmer to decide which semantics they want.

We give an example parser for a small expression language, in the typical parser
combinator style, built on top of the token-consumer and nondeterminism. For
convenience, it uses the syntactic sugar \ensuremath{\Varid{x}\diamond\Varid{y}\;\equiv\;{\mathbf{op}}\;\mathtt{choose}\;(\Varid{b}\, .\,\mathbf{if}\;\Varid{b}\;\mathbf{then}\;\Varid{x}\;\mathbf{else}\;\Varid{y})}.
\indentbegin \begin{hscode}\SaveRestoreHook
\column{B}{@{}>{\hspre}l<{\hspost}@{}}%
\column{9}{@{}>{\hspre}l<{\hspost}@{}}%
\column{12}{@{}>{\hspre}l<{\hspost}@{}}%
\column{16}{@{}>{\hspre}l<{\hspost}@{}}%
\column{E}{@{}>{\hspre}l<{\hspost}@{}}%
\>[B]{}\mathsf{digit}{}\<[12]%
\>[12]{}\mathbin{:}\forall\;\mu\, .\,\mathsf{()}\to \mathsf{Char}\mathbin{!}\langle\mathtt{token}\, ;\mathtt{choose}\, ;\mu\rangle{}\<[E]%
\\
\>[B]{}\mathsf{digit}\;\anonymous {}\<[12]%
\>[12]{}\mathrel{=}{\mathbf{op}}\;\mathtt{token}\;\text{\ttfamily '0'}\diamond{\mathbf{op}}\;\mathtt{token}\;\text{\ttfamily '1'}\diamond\ldots\diamond{\mathbf{op}}\;\mathtt{token}\;\text{\ttfamily '9'}{}\<[E]%
\\
\>[B]{}\mathsf{many}_{1}{}\<[12]%
\>[12]{}\mathbin{:}\forall\;\alpha\;\mu\, .\,(\mathsf{()}\;{\rightarrow}^{\mu}\;\alpha)\;{\rightarrow}^{\mu}\;\mathsf{List}\;\alpha{}\<[E]%
\\
\>[B]{}\mathsf{many}_{1}\;\Varid{p}{}\<[12]%
\>[12]{}\mathrel{=}\mathbf{do}\;\Varid{a}\leftarrow \Varid{p}\;()\, ;\mathbf{do}\;\Varid{as}\leftarrow \mathsf{many}_{1}\;\Varid{p}\diamond{\mathbf{return}}\;[\mskip1.5mu \mskip1.5mu]\, ;{\mathbf{return}}\;(\Varid{a}\mathbin{:}\Varid{as}){}\<[E]%
\\
\>[B]{}\mathsf{expr}'{}\<[12]%
\>[12]{}\mathbin{:}\forall\;\mu\, .\,\mathsf{()}\to \mathsf{Int}\mathbin{!}\langle\mathtt{token}\, ;\mathtt{choose}\, ;\mu\rangle{}\<[E]%
\\
\>[B]{}\mathsf{expr}'\;\anonymous {}\<[12]%
\>[12]{}\mathrel{=}{}\<[16]%
\>[16]{}(\mathbf{do}\;\Varid{i}\leftarrow \mathsf{term}'\;\mathsf{()}\, ;\mathbf{do}\;{\mathbf{op}}\;\mathtt{token}\;\text{\ttfamily '+'}\, ;\mathbf{do}\;\Varid{j}\leftarrow \mathsf{expr}'\;\mathsf{()}\, ;{\mathbf{return}}\;(\Varid{i}\mathbin{+}\Varid{j})){}\<[E]%
\\
\>[12]{}\diamond{}\<[16]%
\>[16]{}(\mathbf{do}\;\Varid{i}\leftarrow \mathsf{term}'\;\mathsf{()}\, ;{\mathbf{return}}\;\Varid{i}){}\<[E]%
\\
\>[B]{}\mathsf{term}'{}\<[12]%
\>[12]{}\mathbin{:}\forall\;\mu\, .\,\mathsf{()}\to \mathsf{Int}\mathbin{!}\langle\mathtt{token}\, ;\mathtt{choose}\, ;\mu\rangle{}\<[E]%
\\
\>[B]{}\mathsf{term}'\;\anonymous {}\<[12]%
\>[12]{}\mathrel{=}{}\<[16]%
\>[16]{}(\mathbf{do}\;\Varid{i}\leftarrow \mathsf{factor}\;\mathsf{()}\, ;\mathbf{do}\;{\mathbf{op}}\;\mathtt{token}\;\text{\ttfamily '*'}\, ;\mathbf{do}\;\Varid{j}\leftarrow \mathsf{term}\;\mathsf{()}\, ;{\mathbf{return}}\;(\Varid{i}\mathbin{*}\Varid{j})){}\<[E]%
\\
\>[12]{}\diamond{}\<[16]%
\>[16]{}(\mathbf{do}\;\Varid{i}\leftarrow \mathsf{factor}\;\mathsf{()}\, ;{\mathbf{return}}\;\Varid{i}){}\<[E]%
\\
\>[B]{}\mathsf{factor}{}\<[12]%
\>[12]{}\mathbin{:}\forall\;\mu\, .\,\mathsf{()}\to \mathsf{Int}\mathbin{!}\langle\mathtt{token}\, ;\mathtt{choose}\, ;\mu\rangle{}\<[E]%
\\
\>[B]{}\mathsf{factor}\;{}\<[9]%
\>[9]{}\anonymous {}\<[12]%
\>[12]{}\mathrel{=}{}\<[16]%
\>[16]{}(\mathbf{do}\;\Varid{ds}\leftarrow \mathsf{many}_{1}\;\mathsf{digit}\, ;{\mathbf{return}}\;(\mathsf{read}\;\Varid{ds})){}\<[E]%
\\
\>[12]{}\diamond{}\<[16]%
\>[16]{}(\mathbf{do}\;{\mathbf{op}}\;\mathtt{token}\;\text{\ttfamily '('}\, ;\mathbf{do}\;\Varid{i}\leftarrow \mathsf{expr}'\;\mathsf{()}\, ;\mathbf{do}\;{\mathbf{op}}\;\mathtt{token}\;\text{\ttfamily ')'}\, ;{\mathbf{return}}\;\Varid{i}){}\<[E]%
\ColumnHook
\end{hscode}\resethooks
\indentend 
The \ensuremath{\mathsf{expr}'} and \ensuremath{\mathsf{term}'} parsers are naive and can be improved by two
steps of refactoring: (1) factoring out the common prefix in the two
branches, and (2) pruning the second branch when the first branch
successfully consumes a \ensuremath{\mathbin{+}} or \ensuremath{\mathbin{*}}, respectively.
The optimised versions of \ensuremath{\mathsf{expr}}, \ensuremath{\mathsf{term}}, and \ensuremath{\mathsf{factor}} are shown as
follows using the \ensuremath{\mathtt{cut}} operation defined in \Cref{eg:call-cut}.
\indentbegin \begin{hscode}\SaveRestoreHook
\column{B}{@{}>{\hspre}l<{\hspost}@{}}%
\column{11}{@{}>{\hspre}l<{\hspost}@{}}%
\column{17}{@{}>{\hspre}l<{\hspost}@{}}%
\column{20}{@{}>{\hspre}l<{\hspost}@{}}%
\column{37}{@{}>{\hspre}l<{\hspost}@{}}%
\column{42}{@{}>{\hspre}l<{\hspost}@{}}%
\column{79}{@{}>{\hspre}l<{\hspost}@{}}%
\column{E}{@{}>{\hspre}l<{\hspost}@{}}%
\>[B]{}\mathsf{expr}{}\<[11]%
\>[11]{}\mathbin{:}\forall\;\mu\, .\,\mathsf{()}\to \mathsf{Int}\mathbin{!}\langle\mathtt{token}\, ;\mathtt{choose}\, ;\mathtt{cut}\, ;\mu\rangle{}\<[E]%
\\
\>[B]{}\mathsf{expr}\;\anonymous {}\<[11]%
\>[11]{}\mathrel{=}\mathbf{do}\;{}\<[17]%
\>[17]{}\Varid{i}{}\<[20]%
\>[20]{}\leftarrow \mathsf{term}\;\mathsf{()}\, ;{}\<[E]%
\\
\>[20]{}{\mathbf{sc}}\;\mathtt{call}\;\mathsf{()}\;(\anonymous {}\<[37]%
\>[37]{}\, .\,({}\<[42]%
\>[42]{}\mathbf{do}\;{\mathbf{op}}\;\mathtt{token}\;\text{\ttfamily '+'}\, ;\mathbf{do}\;{\mathbf{op}}\;\mathtt{cut}\;()\, ;{}\<[E]%
\\
\>[42]{}\mathbf{do}\;\Varid{j}\leftarrow \mathsf{expr}\;\mathsf{()}\, ;{\mathbf{return}}\;(\Varid{i}\mathbin{+}\Varid{j}))\diamond{}\<[79]%
\>[79]{}\Varid{i}){}\<[E]%
\\
\>[B]{}\mathsf{term}{}\<[11]%
\>[11]{}\mathbin{:}\forall\;\mu\, .\,\mathsf{()}\to \mathsf{Int}\mathbin{!}\langle\mathtt{token}\, ;\mathtt{choose}\, ;\mathtt{cut}\, ;\mu\rangle{}\<[E]%
\\
\>[B]{}\mathsf{term}\;\anonymous {}\<[11]%
\>[11]{}\mathrel{=}\mathbf{do}\;{}\<[17]%
\>[17]{}\Varid{i}{}\<[20]%
\>[20]{}\leftarrow \mathsf{factor}\;\mathsf{()}\, ;{}\<[E]%
\\
\>[20]{}{\mathbf{sc}}\;\mathtt{call}\;\mathsf{()}\;(\anonymous {}\<[37]%
\>[37]{}\, .\,({}\<[42]%
\>[42]{}\mathbf{do}\;{\mathbf{op}}\;\mathtt{token}\;\text{\ttfamily '*'}\, ;\mathbf{do}\;{\mathbf{op}}\;\mathtt{cut}\;()\, ;{}\<[E]%
\\
\>[42]{}\mathbf{do}\;\Varid{j}\leftarrow \mathsf{term}\;\mathsf{()}\, ;{\mathbf{return}}\;(\Varid{i}\mathbin{*}\Varid{j}))\diamond{}\<[79]%
\>[79]{}\Varid{i}){}\<[E]%
\ColumnHook
\end{hscode}\resethooks
\indentend 
Here is how we invoke the optimised parser on an example input.\indentbegin \begin{hscode}\SaveRestoreHook
\column{B}{@{}>{\hspre}l<{\hspost}@{}}%
\column{3}{@{}>{\hspre}l<{\hspost}@{}}%
\column{E}{@{}>{\hspre}l<{\hspost}@{}}%
\>[3]{}\texttt{>>>}\;{\mathbf{with}}\;\Varid{h_{\mathsf{cut}}}\;{\mathbf{handle}}\;({\mathbf{with}}\;\Varid{h_{\mathsf{token}}}\;{\mathbf{handle}}\;\mathsf{expr}\;\mathsf{()})\;\text{\ttfamily \char34 (2+5)*8\char34}{}\<[E]%
\\
\>[3]{}\mathsf{opened}\;[\mskip1.5mu (\mathrm{56},\text{\ttfamily \char34 \char34})\mskip1.5mu]{}\<[E]%
\ColumnHook
\end{hscode}\resethooks
\indentend %
Note that only the fully parsed result is returned because in \ensuremath{\mathsf{expr}}
and \ensuremath{\mathsf{term}} we prune other branches when the first branch succeeds.
If we parse the same input using the unoptimised parser \ensuremath{\mathsf{expr}'}, we
also get partially parsed results.\indentbegin \begin{hscode}\SaveRestoreHook
\column{B}{@{}>{\hspre}l<{\hspost}@{}}%
\column{3}{@{}>{\hspre}l<{\hspost}@{}}%
\column{E}{@{}>{\hspre}l<{\hspost}@{}}%
\>[3]{}\texttt{>>>}\;{\mathbf{with}}\;\Varid{h_{\mathsf{cut}}}\;{\mathbf{handle}}\;({\mathbf{with}}\;\Varid{h_{\mathsf{token}}}\;{\mathbf{handle}}\;\mathsf{expr}'\;\mathsf{()})\;\text{\ttfamily \char34 (2+5)*8\char34}{}\<[E]%
\\
\>[3]{}\mathsf{opened}\;[\mskip1.5mu (\mathrm{56},\text{\ttfamily \char34 \char34}),(\mathrm{7},\text{\ttfamily \char34 *8\char34})\mskip1.5mu]{}\<[E]%
\ColumnHook
\end{hscode}\resethooks
\indentend %
\section{Related Work}
\label{sec:related-work}

In this section, we discuss related work on algebraic effects, scoped effects, and effect systems.

\subsection{Algebraic Effects \& Handlers}
Many research languages for algebraic effects and handlers have been
proposed, including Eff
\cite{DBLP:conf/calco/BauerP13, pretnar15}, Frank \cite{frank}, Effekt
\cite{effekt}, or have
been extended to include them, such as Links \cite{Hillerstrom16} and
Koka \cite{Leijen17}. OCaml~\cite{multicore} is the first industrial
language supporting algebraic effects and handlers.
Although it is possible to use handlers or write algebraic operations
with function parameters to simulate the behaviours of scoped effects
in these languages, none of them can achieve the expressiveness of
real scoped effects \& handlers as we have compared in
\Cref{eg:exceptions} and \Cref{eg:local}.

There are also many packages for writing effect handlers in general purpose
languages like Haskell and OCaml \cite{extensible-effects, eff-ocaml,
fused-effects, polysemy, eff}.
Yet, as far as we know, \ensuremath{\lambda_{\mathit{sc}}} is the first \emph{calculus} that supports scoped
effects \& handlers.

\subsection{Effect Systems}

Most languages with support for algebraic effects are equipped with an effect system to keep track of the effects that are used in the programs.
There is already much work on different approaches to effect systems for algebraic effects.

Eff \cite{DBLP:conf/calco/BauerP13, pretnar15} uses an effect system
based on subtyping relations.  Each type of computation is decorated
with an effect type $\Delta$ to represent the set of operations that
might be invoked.
The subtyping relations are used to extend the effect type $\Delta$
with other effects, which makes it possible to compose programs in a
modular way.
The implementation of such a subtyping system is orthogonal to the
implementation of scoped effects, which is why we have opted for Koka-style row
polymorphisms, which has yielded a simpler type system.
Furthermore, supporting both polymorphism and non-trivial subtyping would
complicate the type inference a lot as shown in \cite{pretnar13} and
\cite{TangHLM24} which use type inference with constraints and qualified types,
respectively. Moreover, since \ensuremath{\lambda_{\mathit{sc}}} has type operators, we need to
additionally distinguish between the positive and negative positions in type
operators and propagate subtyping relations through these positions properly.
This would further complicate the formalisation of type system and type
inference.

Row polymorphism is another mainstream approach to effect systems.
Links \cite{Hillerstrom16} uses the R{\'e}my-style row polymorphism \cite{remy1994type}, where the row types are able to represent the absence of labels and each label is restricted to appear at most once.
Koka \cite{Leijen17} uses row polymorphism based on scoped labels \cite{DBLP:conf/sfp/Leijen05}, which allows duplicated labels and as a result is easier to implement.
We can use row polymorphism to write handlers that handle particular effects and forward other effects represented by a row variable.
In \ensuremath{\lambda_{\mathit{sc}}}, we opted for an effect system similar to Koka's, mainly because of its brevity.
We believe that the Links-style effect system should also work well with scoped effects.

\subsection{Scoped Effects \& Handlers}

Wu et al. \cite{wu14} first introduced the idea of scoped effects \& handlers to solve the problem of separating syntax from semantics in programming with effects that delimit the scope.
They proposed a syntax based on higher-order functors, which impose fewer restrictions on the shape of the signatures of scoped operations than \ensuremath{\lambda_{\mathit{sc}}}.
Their work also considers the problem of modular composition of handlers, and
presents a solution--called ``weaving''--based on threading a handler state through
unknown operations. This approach is rather ad-hoc; it is not a generalization
of the forwarding approach for algebraic effects.
This approach has been adopted by several Haskell packages
\cite{fused-effects, polysemy, eff}.

Pir{\'{o}}g et al. \cite{pirog18} and Yang et al. \cite{DBLP:conf/esop/YangPWBS22} have developed
denotational semantic domains of scoped effects, backed by category theoretical
models.
The key idea is to generalize the denotational approach of algebraic effects \&
handlers that is based on free monads and their unique homomorphisms. Indeed,
the underlying category can be seen as a parameter.
Then, by shifting from the base category of types and functions to a different
(indexed or functor) category, scoped operations and their handlers turn out to
be ``just'' an instance of the generalized notion of algebraic operations and
handlers with the same structure and properties.
We focus on a calculus for scoped effects instead of the denotational semantics
of scoped effects. We make a simplification with respect to Yang et al.
\cite{DBLP:conf/esop/YangPWBS22} where we avoid duplication of the base algebra and endoalgebra
(for the outer and inner scoped respectively), and thus duplication of the
scoped effect clauses in our handlers.
With respect to Pir{\'{o}}g et al. \cite{pirog18}, we specialize the generic
endofunctor \ensuremath{\Gamma} with signatures \ensuremath{{A_\ell}\rightarrowtriangle{B_\ell}} of endofunctors of the form \ensuremath{\Conid{A}\;\times\;(\Conid{B}\to \mathbin{-})}.
Our \ensuremath{\lambda_{\mathit{sc}}} calculus uses a similar idea to the `explicit substitution' monad of
Pir{\'{o}}g et al. \cite{pirog18}, a generalization of Ghani and Uustalu's
\cite{Ghani03} monad of explicit substitutions where each operation is
associated with two computations representing the computation in scope and out
of the scope (continuation) respectively.
Neither Pir{\'{o}}g et al. \cite{pirog18} nor Yang et al. \cite{DBLP:conf/esop/YangPWBS22} considered the modular
composition of scoped effects or the forwarding of operations in their
categorical models.

Yang and Wu~\cite{YangW23} develop a framework for (generalized) monoids with
operations, of which scope effects are an instance. They study the problem of
semantic modularity in this framework, and have some generic results which can
be applied to scoped effects. As far as we know, no language design
or library implementation of scoped effects has resulted from this yet.

Lindley et al.~\cite{LindleyMMSWY24} developed an equational reasoning
framework for scoped effects based on parameterised algebraic
theories~\cite{Staton13}. They do not consider handlers or forwarding.

\section{Conclusion and Future Work}
\label{sec:conclusion}
In this work, we have presented \ensuremath{\lambda_{\mathit{sc}}}, a novel calculus in which scoped
effects \& handlers are
built-in. %
We started from the core calculus of Eff, extended it with a row-based
effect system in the style of Koka, and added primitive support for
scoped operations and their handlers.
We introduced novel forwarding clauses as a means to obtain the
modular composition of handlers in the presence of scoped effects.
Finally, we have demonstrated the usability of \ensuremath{\lambda_{\mathit{sc}}} by implementing
a range of examples.
We believe that the features to support scoped effect in \ensuremath{\lambda_{\mathit{sc}}} are orthogonal to
other language features and can be added to any programming
language with algebraic effects, polymorphism and type operators.

Scoped effects require \emph{every} handler in \ensuremath{\lambda_{\mathit{sc}}} to be polymorphic and
equipped with an explicit forwarding clause.
This breaks backwards compatibility: calculi that support only algebraic
effects, such as Eff, miss an explicit forwarding clause for scoped operations
and allow monomorphic handlers.
This problem can be easily mitigated by
kinds and kind polymorphism. The core idea is that we extend \ensuremath{\lambda_{\mathit{sc}}} with
two kinds \ensuremath{{\mathsf{op}}} and \ensuremath{{\mathsf{sc}}} for effect types, such that \ensuremath{\Gamma\vdash\Conid{E}\mathbin{:}{\mathsf{op}}} means
effect type \ensuremath{\Conid{E}} only contains algebraic operations, and \ensuremath{\Gamma\vdash\Conid{E}\mathbin{:}{\mathsf{sc}}}
means effect type \ensuremath{\Conid{E}} may contain some scoped operations. Then, for handlers
of type \ensuremath{\Conid{A}\hspace{0.1em}!\hspace{0.1em}\langle\Conid{E}\rangle\Rightarrow \Conid{M}\;\Conid{A}\hspace{0.1em}!\hspace{0.1em}\langle\Conid{F}\rangle} which lack forwarding clauses, we can
just add the condition \ensuremath{\Gamma\vdash\Conid{E}\mathbin{:}{\mathsf{op}}} to their typing rules.
We leave the full formalisation and implementation of it to future work.

Making scoped effects and handlers into practical languages is a
pretty new research area full of potential.
There is a lot of future work to do.
Directions for future work include: extending \ensuremath{\lambda_{\mathit{sc}}} with shallow
handlers~\cite{KammarLO13,HillerstromL18}, named
handlers~\cite{BiernackiPPS20,XieCIL22,effekt}, and other forms of higher-order
effects~\cite{VANDENBERG2024103086} including latent effects
\cite{DBLP:conf/aplas/BergSPW21} and parallel effects \cite{parallel-effects};
exploring the notion of named scoped effects where scoped effects
themselves generate fresh names for the operations in their scope in a
similar style to named handlers; defining a CPS translation for scoped
effects and handlers~\cite{HillerstromLA20,Leijen17}; developing
control-flow linearity~\cite{TangHLM24} for scoped effects to soundly
extend \ensuremath{\lambda_{\mathit{sc}}} with linear types; exploring guiding principles and
equational theories for writing and reasoning about forwarding
clauses.

\section*{Acknowledgment}
\noindent Part of this work was funded by FWO project G0A9423N.

\bibliography{reference}

\newcommand{\etalchar}[1]{$^{#1}$}
\begin{thebibliography}{YPW{\etalchar{+}}22}

\bibitem[BP13]{DBLP:conf/calco/BauerP13}
Andrej Bauer and Matija Pretnar.
\newblock An effect system for algebraic effects and handlers.
\newblock In Reiko Heckel and Stefan Milius, editors, {\em Algebra and Coalgebra in Computer Science - 5th International Conference, {CALCO} 2013, Warsaw, Poland, September 3-6, 2013. Proceedings}, volume 8089 of {\em Lecture Notes in Computer Science}, pages 1--16. Springer, 2013.
\newblock \href {https://doi.org/10.1007/978-3-642-40206-7\_1} {\path{doi:10.1007/978-3-642-40206-7\_1}}.

\bibitem[BP15]{bauer15}
Andrej Bauer and Matija Pretnar.
\newblock Programming with algebraic effects and handlers.
\newblock {\em Journal of Logical and Algebraic Methods in Programming}, 84(1):108--123, 2015.
\newblock Special Issue: The 23rd Nordic Workshop on Programming Theory (NWPT 2011) Special Issue: Domains X, International workshop on Domain Theory and applications, Swansea, 5-7 September, 2011.
\newblock URL: \url{https://www.sciencedirect.com/science/article/pii/S2352220814000194}, \href {https://doi.org/10.1016/j.jlamp.2014.02.001} {\path{doi:10.1016/j.jlamp.2014.02.001}}.

\bibitem[BPPS20]{BiernackiPPS20}
Dariusz Biernacki, Maciej Pir{\'{o}}g, Piotr Polesiuk, and Filip Sieczkowski.
\newblock Binders by day, labels by night: effect instances via lexically scoped handlers.
\newblock {\em Proc. {ACM} Program. Lang.}, 4({POPL}):48:1--48:29, 2020.
\newblock \href {https://doi.org/10.1145/3371116} {\path{doi:10.1145/3371116}}.

\bibitem[BS24]{VANDENBERG2024103086}
Birthe van~den Berg and Tom Schrijvers.
\newblock A framework for higher-order effects \& handlers.
\newblock {\em Science of Computer Programming}, 234:103086, 2024.
\newblock URL: \url{https://www.sciencedirect.com/science/article/pii/S0167642324000091}, \href {https://doi.org/10.1016/j.scico.2024.103086} {\path{doi:10.1016/j.scico.2024.103086}}.

\bibitem[BSO20]{effekt}
Jonathan~Immanuel Brachth\"{a}user, Philipp Schuster, and Klaus Ostermann.
\newblock Effects as capabilities: Effect handlers and lightweight effect polymorphism.
\newblock {\em Proc. ACM Program. Lang.}, 4(OOPSLA), November 2020.
\newblock \href {https://doi.org/10.1145/3428194} {\path{doi:10.1145/3428194}}.

\bibitem[BSPW21]{DBLP:conf/aplas/BergSPW21}
Birthe van~den Berg, Tom Schrijvers, Casper~Bach Poulsen, and Nicolas Wu.
\newblock Latent effects for reusable language components.
\newblock In Hakjoo Oh, editor, {\em Programming Languages and Systems - 19th Asian Symposium, {APLAS} 2021, Chicago, IL, USA, October 17-18, 2021, Proceedings}, volume 13008 of {\em Lecture Notes in Computer Science}, pages 182--201. Springer, 2021.
\newblock \href {https://doi.org/10.1007/978-3-030-89051-3\_11} {\path{doi:10.1007/978-3-030-89051-3\_11}}.

\bibitem[GU03]{Ghani03}
Neil Ghani and Tarmo Uustalu.
\newblock Explicit substitutions and higher-order syntax.
\newblock In {\em Proceedings of the 2003 ACM SIGPLAN Workshop on Mechanized Reasoning about Languages with Variable Binding}, MERLIN '03, page 1–7, New York, NY, USA, 2003. Association for Computing Machinery.
\newblock \href {https://doi.org/10.1145/976571.976580} {\path{doi:10.1145/976571.976580}}.

\bibitem[HL16]{Hillerstrom16}
Daniel Hillerstr\"{o}m and Sam Lindley.
\newblock Liberating effects with rows and handlers.
\newblock In {\em Proceedings of the 1st International Workshop on Type-Driven Development}, TyDe 2016, page 15–27, New York, NY, USA, 2016. Association for Computing Machinery.
\newblock \href {https://doi.org/10.1145/2976022.2976033} {\path{doi:10.1145/2976022.2976033}}.

\bibitem[HL18a]{DBLP:conf/aplas/HillerstromL18}
Daniel Hillerstr{\"{o}}m and Sam Lindley.
\newblock Shallow effect handlers.
\newblock In Sukyoung Ryu, editor, {\em Programming Languages and Systems - 16th Asian Symposium, {APLAS} 2018, Wellington, New Zealand, December 2-6, 2018, Proceedings}, volume 11275 of {\em Lecture Notes in Computer Science}, pages 415--435. Springer, 2018.
\newblock \href {https://doi.org/10.1007/978-3-030-02768-1\_22} {\path{doi:10.1007/978-3-030-02768-1\_22}}.

\bibitem[HL18b]{HillerstromL18}
Daniel Hillerstr{\"{o}}m and Sam Lindley.
\newblock Shallow effect handlers.
\newblock In Sukyoung Ryu, editor, {\em Programming Languages and Systems - 16th Asian Symposium, {APLAS} 2018, Wellington, New Zealand, December 2-6, 2018, Proceedings}, volume 11275 of {\em Lecture Notes in Computer Science}, pages 415--435. Springer, 2018.
\newblock \href {https://doi.org/10.1007/978-3-030-02768-1\_22} {\path{doi:10.1007/978-3-030-02768-1\_22}}.

\bibitem[HLA20]{HillerstromLA20}
Daniel Hillerstr{\"{o}}m, Sam Lindley, and Robert Atkey.
\newblock Effect handlers via generalised continuations.
\newblock {\em J. Funct. Program.}, 30:e5, 2020.
\newblock \href {https://doi.org/10.1017/S0956796820000040} {\path{doi:10.1017/S0956796820000040}}.

\bibitem[Kin19]{eff}
Alexis King.
\newblock eff -- screaming fast extensible effects for less, 2019.
\newblock \url{https://github.com/hasura/eff}.

\bibitem[KLO13]{KammarLO13}
Ohad Kammar, Sam Lindley, and Nicolas Oury.
\newblock Handlers in action.
\newblock In Greg Morrisett and Tarmo Uustalu, editors, {\em {ACM} {SIGPLAN} International Conference on Functional Programming, ICFP'13, Boston, MA, {USA} - September 25 - 27, 2013}, pages 145--158. {ACM}, 2013.
\newblock \href {https://doi.org/10.1145/2500365.2500590} {\path{doi:10.1145/2500365.2500590}}.

\bibitem[KS18]{eff-ocaml}
Oleg Kiselyov and KC~Sivaramakrishnan.
\newblock Eff directly in {O}{C}aml.
\newblock {\em Electronic Proceedings in Theoretical Computer Science}, 285:23–58, Dec 2018.
\newblock URL: \url{http://dx.doi.org/10.4204/EPTCS.285.2}, \href {https://doi.org/10.4204/eptcs.285.2} {\path{doi:10.4204/eptcs.285.2}}.

\bibitem[KSSF19]{extensible-effects}
Oleg Kiselyov, Amr Sabry, Cameron Swords, and Ben Foppa.
\newblock extensible-effects: An alternative to monad transformers, 2019.
\newblock \url{https://hackage.haskell.org/package/extensible-effects}.

\bibitem[Lei05]{DBLP:conf/sfp/Leijen05}
Daan Leijen.
\newblock Extensible records with scoped labels.
\newblock In Marko C. J.~D. van Eekelen, editor, {\em Revised Selected Papers from the Sixth Symposium on Trends in Functional Programming, {TFP} 2005, Tallinn, Estonia, 23-24 September 2005}, volume~6 of {\em Trends in Functional Programming}, pages 179--194. Intellect, 2005.

\bibitem[Lei14]{Leijen14}
Daan Leijen.
\newblock Koka: Programming with row polymorphic effect types.
\newblock {\em Electronic Proceedings in Theoretical Computer Science}, 153, 06 2014.
\newblock \href {https://doi.org/10.4204/EPTCS.153.8} {\path{doi:10.4204/EPTCS.153.8}}.

\bibitem[Lei17]{Leijen17}
Daan Leijen.
\newblock Type directed compilation of row-typed algebraic effects.
\newblock In {\em Proceedings of the 44th ACM SIGPLAN Symposium on Principles of Programming Languages}, POPL 2017, page 486–499, New York, NY, USA, 2017. Association for Computing Machinery.
\newblock \href {https://doi.org/10.1145/3009837.3009872} {\path{doi:10.1145/3009837.3009872}}.

\bibitem[LMM17]{frank}
Sam Lindley, Conor McBride, and Craig McLaughlin.
\newblock Do be do be do.
\newblock In {\em Proceedings of the 44th ACM SIGPLAN Symposium on Principles of Programming Languages}, POPL 2017, page 500–514, New York, NY, USA, 2017. Association for Computing Machinery.
\newblock \href {https://doi.org/10.1145/3009837.3009897} {\path{doi:10.1145/3009837.3009897}}.

\bibitem[LMM{\etalchar{+}}23]{LindleyMMSWY24}
Sam Lindley, Cristina Matache, Sean Moss, Sam Staton, Nicolas Wu, and Zhixuan Yang.
\newblock Scoped effects as parameterized algebraic theories.
\newblock In {\em European Symposium on Programming 2024}, Lecture Notes in Computer Science. Springer, December 2023.
\newblock URL: \url{https://etaps.org/2024/conferences/esop/}.

\bibitem[LPT03]{DBLP:journals/iandc/LevyPT03}
Paul~Blain Levy, John Power, and Hayo Thielecke.
\newblock Modelling environments in call-by-value programming languages.
\newblock {\em Inf. Comput.}, 185(2):182--210, 2003.
\newblock \href {https://doi.org/10.1016/S0890-5401(03)00088-9} {\path{doi:10.1016/S0890-5401(03)00088-9}}.

\bibitem[Mag19]{polysemy}
Sandy Maguire.
\newblock polysemy: Higher-order, low-boilerplate free monads, 2019.
\newblock \url{https://hackage.haskell.org/package/polysemy}.

\bibitem[Mil78]{DBLP:journals/jcss/Milner78}
Robin Milner.
\newblock A theory of type polymorphism in programming.
\newblock {\em J. Comput. Syst. Sci.}, 17(3):348--375, 1978.
\newblock \href {https://doi.org/10.1016/0022-0000(78)90014-4} {\path{doi:10.1016/0022-0000(78)90014-4}}.

\bibitem[Mog89]{Moggi89}
Eugenio Moggi.
\newblock An abstract view of programming languages.
\newblock Technical Report ECS-LFCS-90-113, Edinburgh University, Department of Computer Science, June 1989.

\bibitem[Mog91]{Moggi95}
Eugenio Moggi.
\newblock Notions of computation and monads.
\newblock {\em Information and Computation}, 93(1):55 -- 92, 1991.
\newblock Selections from 1989 IEEE Symposium on Logic in Computer Science.
\newblock \href {https://doi.org/10.1016/0890-5401(91)90052-4} {\path{doi:10.1016/0890-5401(91)90052-4}}.

\bibitem[PP03]{Plotkin03}
Gordon~D. Plotkin and John Power.
\newblock Algebraic operations and generic effects.
\newblock {\em Appl. Categorical Struct.}, 11(1):69--94, 2003.
\newblock \href {https://doi.org/10.1023/A:1023064908962} {\path{doi:10.1023/A:1023064908962}}.

\bibitem[PP09]{Plotkin09}
Gordon Plotkin and Matija Pretnar.
\newblock Handlers of algebraic effects.
\newblock In Giuseppe Castagna, editor, {\em Programming Languages and Systems}, pages 80--94, Berlin, Heidelberg, 2009. Springer Berlin Heidelberg.
\newblock \href {https://doi.org/10.1007/978-3-642-00590-9\_7} {\path{doi:10.1007/978-3-642-00590-9\_7}}.

\bibitem[Pre14]{pretnar13}
Matija Pretnar.
\newblock Inferring algebraic effects.
\newblock {\em Log. Methods Comput. Sci.}, 10(3), 2014.
\newblock \href {https://doi.org/10.2168/LMCS-10(3:21)2014} {\path{doi:10.2168/LMCS-10(3:21)2014}}.

\bibitem[Pre15]{pretnar15}
Matija Pretnar.
\newblock An introduction to algebraic effects and handlers. invited tutorial paper.
\newblock {\em Electronic Notes in Theoretical Computer Science}, 319:19--35, 2015.
\newblock The 31st Conference on the Mathematical Foundations of Programming Semantics (MFPS XXXI).
\newblock URL: \url{https://www.sciencedirect.com/science/article/pii/S1571066115000705}, \href {https://doi.org/10.1016/j.entcs.2015.12.003} {\path{doi:10.1016/j.entcs.2015.12.003}}.

\bibitem[PS17]{DBLP:journals/jfp/PirogS17}
Maciej Pir{\'{o}}g and Sam Staton.
\newblock Backtracking with cut via a distributive law and left-zero monoids.
\newblock {\em J. Funct. Program.}, 27:e17, 2017.
\newblock \href {https://doi.org/10.1017/S0956796817000077} {\path{doi:10.1017/S0956796817000077}}.

\bibitem[PSWJ18]{pirog18}
Maciej Pir\'{o}g, Tom Schrijvers, Nicolas Wu, and Mauro Jaskelioff.
\newblock Syntax and semantics for operations with scopes.
\newblock In {\em Proceedings of the 33rd Annual ACM/IEEE Symposium on Logic in Computer Science}, LICS '18, page 809–818, New York, NY, USA, 2018. Association for Computing Machinery.
\newblock \href {https://doi.org/10.1145/3209108.3209166} {\path{doi:10.1145/3209108.3209166}}.

\bibitem[R{\'e}m94]{remy1994type}
Didier R{\'e}my.
\newblock Type inference for records in a natural extension of ml.
\newblock In {\em Theoretical Aspects of Object-Oriented Programming: Types, Semantics, and Language Design}. Citeseer, 1994.

\bibitem[RTWS18]{fused-effects}
Rob Rix, Patrick Thomson, Nicolas Wu, and Tom Schrijvers.
\newblock fused-effects: A fast, flexible, fused effect system, 2018.
\newblock \url{https://hackage.haskell.org/package/fused-effects}.

\bibitem[SDW{\etalchar{+}}21]{multicore}
KC~Sivaramakrishnan, Stephen Dolan, Leo White, Tom Kelly, Sadiq Jaffer, and Anil Madhavapeddy.
\newblock Retrofitting effect handlers onto {O}{C}aml.
\newblock In {\em Proceedings of the 42nd ACM SIGPLAN International Conference on Programming Language Design and Implementation}, PLDI 2021, page 206–221, New York, NY, USA, 2021. Association for Computing Machinery.
\newblock \href {https://doi.org/10.1145/3453483.3454039} {\path{doi:10.1145/3453483.3454039}}.

\bibitem[SPWJ19]{DBLP:conf/haskell/SchrijversPWJ19}
Tom Schrijvers, Maciej Pir{\'{o}}g, Nicolas Wu, and Mauro Jaskelioff.
\newblock Monad transformers and modular algebraic effects: what binds them together.
\newblock In Richard~A. Eisenberg, editor, {\em Proceedings of the 12th {ACM} {SIGPLAN} International Symposium on Haskell, Haskell@ICFP 2019, Berlin, Germany, August 18-23, 2019}, pages 98--113. {ACM}, 2019.
\newblock \href {https://doi.org/10.1145/3331545.3342595} {\path{doi:10.1145/3331545.3342595}}.

\bibitem[Sta13]{Staton13}
Sam Staton.
\newblock Instances of computational effects: An algebraic perspective.
\newblock In {\em 28th Annual {ACM/IEEE} Symposium on Logic in Computer Science, {LICS} 2013, New Orleans, LA, USA, June 25-28, 2013}, page 519. {IEEE} Computer Society, 2013.
\newblock \href {https://doi.org/10.1109/LICS.2013.58} {\path{doi:10.1109/LICS.2013.58}}.

\bibitem[THLM24]{TangHLM24}
Wenhao Tang, Daniel Hillerstr{\"{o}}m, Sam Lindley, and J.~Garrett Morris.
\newblock Soundly handling linearity.
\newblock {\em Proc. {ACM} Program. Lang.}, 8({POPL}):1600--1628, 2024.
\newblock \href {https://doi.org/10.1145/3632896} {\path{doi:10.1145/3632896}}.

\bibitem[TRWS22]{DBLP:journals/pacmpl/ThomsonRWS22}
Patrick Thomson, Rob Rix, Nicolas Wu, and Tom Schrijvers.
\newblock Fusing industry and academia at {G}it{H}ub (experience report).
\newblock {\em Proc. {ACM} Program. Lang.}, 6({ICFP}):496--511, 2022.
\newblock \href {https://doi.org/10.1145/3547639} {\path{doi:10.1145/3547639}}.

\bibitem[Wad95]{DBLP:conf/afp/Wadler95}
Philip Wadler.
\newblock Monads for functional programming.
\newblock In Johan Jeuring and Erik Meijer, editors, {\em Advanced Functional Programming, First International Spring School on Advanced Functional Programming Techniques, B{\aa}stad, Sweden, May 24-30, 1995, Tutorial Text}, volume 925 of {\em Lecture Notes in Computer Science}, pages 24--52. Springer, 1995.
\newblock \href {https://doi.org/10.1007/3-540-59451-5\_2} {\path{doi:10.1007/3-540-59451-5\_2}}.

\bibitem[WSH14]{wu14}
Nicolas Wu, Tom Schrijvers, and Ralf Hinze.
\newblock Effect handlers in scope.
\newblock In {\em Proceedings of the 2014 ACM SIGPLAN Symposium on Haskell}, Haskell '14, page 1–12, New York, NY, USA, 2014. Association for Computing Machinery.
\newblock \href {https://doi.org/10.1145/2633357.2633358} {\path{doi:10.1145/2633357.2633358}}.

\bibitem[XCIL22]{XieCIL22}
Ningning Xie, Youyou Cong, Kazuki Ikemori, and Daan Leijen.
\newblock First-class names for effect handlers.
\newblock {\em Proc. {ACM} Program. Lang.}, 6({OOPSLA2}):30--59, 2022.
\newblock \href {https://doi.org/10.1145/3563289} {\path{doi:10.1145/3563289}}.

\bibitem[XJMP21]{parallel-effects}
Ningning Xie, Daniel~D. Johnson, Dougal Maclaurin, and Adam Paszke.
\newblock Parallel algebraic effect handlers.
\newblock {\em CoRR}, abs/2110.07493, 2021.
\newblock URL: \url{https://arxiv.org/abs/2110.07493}, \href {https://arxiv.org/abs/2110.07493} {\path{arXiv:2110.07493}}.

\bibitem[YPW{\etalchar{+}}22]{DBLP:conf/esop/YangPWBS22}
Zhixuan Yang, Marco Paviotti, Nicolas Wu, Birthe van~den {Berg}, and Tom Schrijvers.
\newblock Structured handling of scoped effects.
\newblock In Ilya Sergey, editor, {\em Programming Languages and Systems - 31st European Symposium on Programming, {ESOP} 2022, Held as Part of the European Joint Conferences on Theory and Practice of Software, {ETAPS} 2022, Munich, Germany, April 2-7, 2022, Proceedings}, volume 13240 of {\em Lecture Notes in Computer Science}, pages 462--491. Springer, 2022.
\newblock \href {https://doi.org/10.1007/978-3-030-99336-8\_17} {\path{doi:10.1007/978-3-030-99336-8\_17}}.

\bibitem[YW14]{lhkp}
Jeremy Yallop and Leo White.
\newblock Lightweight higher-kinded polymorphism.
\newblock In Michael Codish and Eijiro Sumii, editors, {\em Functional and Logic Programming}, pages 119--135, Cham, 2014. Springer International Publishing.

\bibitem[YW23]{YangW23}
Zhixuan Yang and Nicolas Wu.
\newblock Modular models of monoids with operations.
\newblock {\em Proc. {ACM} Program. Lang.}, 7({ICFP}):566--603, 2023.
\newblock \href {https://doi.org/10.1145/3607850} {\path{doi:10.1145/3607850}}.

\end{thebibliography}
\bibliographystyle{alphaurl}

\clearpage
\appendix

\section*{Appendix Contents}

\begin{itemize}
\item \Cref{app:opsem-derivations}: Semantic Derivations \dotfill\ \pageref{app:opsem-derivations}
\item \Cref{app:type-equiv}: Type Equivalence Rules \dotfill\ \pageref{app:type-equiv}
\item \Cref{app:well-scopedness}: Well-scopedness Rules \dotfill\ \pageref{app:well-scopedness}
\item \Cref{app:syntax-directed}: Syntax-directed version of \ensuremath{\lambda_{\mathit{sc}}} \dotfill\ \pageref{app:syntax-directed}
\item \Cref{sec:metatheory-appendix}: Metatheory \dotfill\ \pageref{sec:metatheory-appendix}
\end{itemize}

\section{Semantic Derivations}
\label{app:opsem-derivations}

This Appendix contains semantic derivations of different handler applications
that are used in the examples throughout this paper.

\subsection{Nondeterminism}
\label{subsec:nd-derivation}
\indentbegin \begin{hscode}\SaveRestoreHook
\column{B}{@{}>{\hspre}l<{\hspost}@{}}%
\column{3}{@{}>{\hspre}l<{\hspost}@{}}%
\column{10}{@{}>{\hspre}l<{\hspost}@{}}%
\column{11}{@{}>{\hspre}l<{\hspost}@{}}%
\column{15}{@{}>{\hspre}l<{\hspost}@{}}%
\column{46}{@{}>{\hspre}l<{\hspost}@{}}%
\column{48}{@{}>{\hspre}l<{\hspost}@{}}%
\column{51}{@{}>{\hspre}l<{\hspost}@{}}%
\column{52}{@{}>{\hspre}l<{\hspost}@{}}%
\column{54}{@{}>{\hspre}l<{\hspost}@{}}%
\column{57}{@{}>{\hspre}l<{\hspost}@{}}%
\column{67}{@{}>{\hspre}l<{\hspost}@{}}%
\column{69}{@{}>{\hspre}l<{\hspost}@{}}%
\column{72}{@{}>{\hspre}l<{\hspost}@{}}%
\column{E}{@{}>{\hspre}l<{\hspost}@{}}%
\>[10]{}{\mathbf{with}}\;\Varid{h_{\mathsf{ND}}}\;{\mathbf{handle}}\;\Varid{c_{\mathsf{ND1}}}{}\<[E]%
\\
\>[3]{}\equiv\;{}\<[10]%
\>[10]{}{\mathbf{with}}\;\Varid{h_{\mathsf{ND}}}\;{\mathbf{handle}}\;{\mathbf{op}}\;\mathtt{choose}\;\mathsf{()}\;(\Varid{b}\, .\,\mathbf{if}\;\Varid{b}\;{}\<[51]%
\>[51]{}\mathbf{then}\;{}\<[57]%
\>[57]{}{\mathbf{return}}\;\mathrm{1}\;\mathbf{else}\;{}\<[72]%
\>[72]{}{\mathbf{return}}\;\mathrm{2}){}\<[E]%
\\
\>[3]{}\leadsto{}\<[10]%
\>[10]{}\mbox{\commentbegin ~  \textsc{E-HandOp}   \commentend}{}\<[E]%
\\
\>[10]{}\hsindent{1}{}\<[11]%
\>[11]{}\mathbf{do}\;{}\<[15]%
\>[15]{}\Varid{xs}\leftarrow (\boldsymbol{\lambda}\Varid{y}\, .\,{\mathbf{with}}\;\Varid{h_{\mathsf{ND}}}\;{\mathbf{handle}}\;\mathbf{if}\;\Varid{b}\;{}\<[48]%
\>[48]{}\mathbf{then}\;{}\<[54]%
\>[54]{}{\mathbf{return}}\;\mathrm{1}\;\mathbf{else}\;{}\<[69]%
\>[69]{}{\mathbf{return}}\;\mathrm{2})\;\mathsf{true}{}\<[E]%
\\
\>[10]{}\hsindent{1}{}\<[11]%
\>[11]{}\mathbf{do}\;{}\<[15]%
\>[15]{}\Varid{ys}\leftarrow (\boldsymbol{\lambda}\Varid{y}\, .\,{\mathbf{with}}\;\Varid{h_{\mathsf{ND}}}\;{\mathbf{handle}}\;\mathbf{if}\;\Varid{b}\;{}\<[48]%
\>[48]{}\mathbf{then}\;{}\<[54]%
\>[54]{}{\mathbf{return}}\;\mathrm{1}\;\mathbf{else}\;{}\<[69]%
\>[69]{}{\mathbf{return}}\;\mathrm{2})\;\mathsf{false}{}\<[E]%
\\
\>[15]{}\Varid{xs}+\hspace{-0.4em}+\Varid{ys}{}\<[E]%
\\
\>[3]{}\leadsto{}\<[10]%
\>[10]{}\mbox{\commentbegin ~  \textsc{E-AppAbs}   \commentend}{}\<[E]%
\\
\>[10]{}\hsindent{1}{}\<[11]%
\>[11]{}\mathbf{do}\;{}\<[15]%
\>[15]{}\Varid{xs}\leftarrow {\mathbf{with}}\;\Varid{h_{\mathsf{ND}}}\;{\mathbf{handle}}\;\mathbf{if}\;\mathsf{true}\;{}\<[46]%
\>[46]{}\mathbf{then}\;{}\<[52]%
\>[52]{}{\mathbf{return}}\;\mathrm{1}\;\mathbf{else}\;{}\<[67]%
\>[67]{}{\mathbf{return}}\;\mathrm{2}{}\<[E]%
\\
\>[10]{}\hsindent{1}{}\<[11]%
\>[11]{}\mathbf{do}\;{}\<[15]%
\>[15]{}\Varid{ys}\leftarrow (\boldsymbol{\lambda}\Varid{y}\, .\,{\mathbf{with}}\;\Varid{h_{\mathsf{ND}}}\;{\mathbf{handle}}\;\mathbf{if}\;\Varid{b}\;{}\<[48]%
\>[48]{}\mathbf{then}\;{}\<[54]%
\>[54]{}{\mathbf{return}}\;\mathrm{1}\;\mathbf{else}\;{}\<[69]%
\>[69]{}{\mathbf{return}}\;\mathrm{2})\;\mathsf{false}{}\<[E]%
\\
\>[15]{}\Varid{xs}+\hspace{-0.4em}+\Varid{ys}{}\<[E]%
\\
\>[3]{}\leadsto{}\<[10]%
\>[10]{}\mbox{\commentbegin ~  reducing \ensuremath{\mathbf{if}}   \commentend}{}\<[E]%
\\
\>[10]{}\hsindent{1}{}\<[11]%
\>[11]{}\mathbf{do}\;{}\<[15]%
\>[15]{}\Varid{xs}\leftarrow {\mathbf{with}}\;\Varid{h_{\mathsf{ND}}}\;{\mathbf{handle}}\;{\mathbf{return}}\;\mathrm{1}{}\<[E]%
\\
\>[10]{}\hsindent{1}{}\<[11]%
\>[11]{}\mathbf{do}\;{}\<[15]%
\>[15]{}\Varid{ys}\leftarrow (\boldsymbol{\lambda}\Varid{y}\, .\,{\mathbf{with}}\;\Varid{h_{\mathsf{ND}}}\;{\mathbf{handle}}\;\mathbf{if}\;\Varid{b}\;{}\<[48]%
\>[48]{}\mathbf{then}\;{}\<[54]%
\>[54]{}{\mathbf{return}}\;\mathrm{1}\;\mathbf{else}\;{}\<[69]%
\>[69]{}{\mathbf{return}}\;\mathrm{2})\;\mathsf{false}{}\<[E]%
\\
\>[15]{}\Varid{xs}+\hspace{-0.4em}+\Varid{ys}{}\<[E]%
\\
\>[3]{}\leadsto{}\<[10]%
\>[10]{}\mbox{\commentbegin ~  \textsc{E-HandRet}   \commentend}{}\<[E]%
\\
\>[10]{}\hsindent{1}{}\<[11]%
\>[11]{}\mathbf{do}\;{}\<[15]%
\>[15]{}\Varid{ys}\leftarrow (\boldsymbol{\lambda}\Varid{y}\, .\,{\mathbf{with}}\;\Varid{h_{\mathsf{ND}}}\;{\mathbf{handle}}\;\mathbf{if}\;\Varid{b}\;{}\<[48]%
\>[48]{}\mathbf{then}\;{}\<[54]%
\>[54]{}{\mathbf{return}}\;\mathrm{1}\;\mathbf{else}\;{}\<[69]%
\>[69]{}{\mathbf{return}}\;\mathrm{2})\;\mathsf{false}{}\<[E]%
\\
\>[15]{}[\mskip1.5mu \mathrm{1}\mskip1.5mu]+\hspace{-0.4em}+\Varid{ys}{}\<[E]%
\\
\>[3]{}\leadsto{}\<[10]%
\>[10]{}\mbox{\commentbegin ~  similar to above (the first branch of \ensuremath{\mathbf{if}})   \commentend}{}\<[E]%
\\
\>[10]{}\hsindent{5}{}\<[15]%
\>[15]{}[\mskip1.5mu \mathrm{1}\mskip1.5mu]+\hspace{-0.4em}+[\mskip1.5mu \mathrm{2}\mskip1.5mu]{}\<[E]%
\\
\>[3]{}\leadsto{}\<[10]%
\>[10]{}\mbox{\commentbegin ~  reducing \ensuremath{+\hspace{-0.4em}+}   \commentend}{}\<[E]%
\\
\>[10]{}\hsindent{5}{}\<[15]%
\>[15]{}{\mathbf{return}}\;[\mskip1.5mu \mathrm{1},\mathrm{2}\mskip1.5mu]{}\<[E]%
\ColumnHook
\end{hscode}\resethooks
\indentend 
\subsection{Increment}

\indentbegin \begin{hscode}\SaveRestoreHook
\column{B}{@{}>{\hspre}l<{\hspost}@{}}%
\column{3}{@{}>{\hspre}l<{\hspost}@{}}%
\column{10}{@{}>{\hspre}l<{\hspost}@{}}%
\column{11}{@{}>{\hspre}l<{\hspost}@{}}%
\column{12}{@{}>{\hspre}l<{\hspost}@{}}%
\column{E}{@{}>{\hspre}l<{\hspost}@{}}%
\>[10]{}{\mathbf{with}}\;\Varid{h_{\mathsf{ND}}}\;{\mathbf{handle}}\;\Varid{run_{\mathsf{inc}}}\;\mathrm{0}\;\Varid{c_{\mathsf{inc}}}\;\equiv\;{\mathbf{with}}\;\Varid{h_{\mathsf{ND}}}\;{\mathbf{handle}}\;(\boldsymbol{\lambda}\Varid{c}\;\Varid{p}\, .\,\mathbf{do}\;\Varid{p'}\leftarrow {}\<[E]%
\\
\>[10]{}\hsindent{1}{}\<[11]%
\>[11]{}{\mathbf{with}}\;\Varid{h_{\mathsf{inc}}}\;{\mathbf{handle}}\;\Varid{p}\, ;\Varid{p'}\;\Varid{c})\;\mathrm{0}\;\Varid{c_{\mathsf{inc}}}{}\<[E]%
\\
\>[3]{}\leadsto{}\<[10]%
\>[10]{}\mbox{\commentbegin ~  \textsc{E-AppAbs}   \commentend}{}\<[E]%
\\
\>[10]{}{\mathbf{with}}\;\Varid{h_{\mathsf{ND}}}\;{\mathbf{handle}}\;\mathbf{do}\;\Varid{p'}\leftarrow {\mathbf{with}}\;\Varid{h_{\mathsf{inc}}}\;{\mathbf{handle}}\;\Varid{c_{\mathsf{inc}}}\, ;\Varid{p'}\;\mathrm{0}{}\<[E]%
\\
\>[3]{}\equiv\;{}\<[10]%
\>[10]{}\mbox{\commentbegin ~  definition of \ensuremath{\Varid{c_{\mathsf{inc}}}}   \commentend}{}\<[E]%
\\
\>[10]{}{\mathbf{with}}\;\Varid{h_{\mathsf{ND}}}\;{\mathbf{handle}}\;\mathbf{do}\;\Varid{p'}\leftarrow {\mathbf{with}}\;\Varid{h_{\mathsf{inc}}}\;{\mathbf{handle}}\;{\mathbf{op}}\;\mathtt{choose}\;()\;(\Varid{b}\, .\,\mathbf{if}\;\Varid{b}\;\mathbf{then}{}\<[E]%
\\
\>[10]{}\hsindent{1}{}\<[11]%
\>[11]{}{\mathbf{op}}\;\mathtt{inc}\;()\;(\Varid{x}\, .\,\Varid{x}\mathbin{+}\mathrm{5})\;\mathbf{else}\;{\mathbf{op}}\;\mathtt{inc}\;()\;(\Varid{y}\, .\,\Varid{y}\mathbin{+}\mathrm{2}))\, ;\Varid{p'}\;\mathrm{0}{}\<[E]%
\\
\>[3]{}\leadsto{}\<[10]%
\>[10]{}\mbox{\commentbegin ~  \textsc{E-FwdOp}   \commentend}{}\<[E]%
\\
\>[10]{}{\mathbf{with}}\;\Varid{h_{\mathsf{ND}}}\;{\mathbf{handle}}\;\mathbf{do}\;\Varid{p'}\leftarrow {\mathbf{op}}\;\mathtt{choose}\;()\;(\Varid{b}\, .\,{\mathbf{with}}\;\Varid{h_{\mathsf{inc}}}\;{\mathbf{handle}}\;\mathbf{if}\;\Varid{b}\;\mathbf{then}{}\<[E]%
\\
\>[10]{}\hsindent{1}{}\<[11]%
\>[11]{}{\mathbf{op}}\;\mathtt{inc}\;()\;(\Varid{x}\, .\,\Varid{x}\mathbin{+}\mathrm{5})\;\mathbf{else}\;{\mathbf{op}}\;\mathtt{inc}\;()\;(\Varid{y}\, .\,\Varid{y}\mathbin{+}\mathrm{2}))\, ;\Varid{p'}\;\mathrm{0}{}\<[E]%
\\
\>[3]{}\leadsto{}\<[10]%
\>[10]{}\mbox{\commentbegin ~  \textsc{E-Hand} and \textsc{E-DoOp}   \commentend}{}\<[E]%
\\
\>[10]{}{\mathbf{with}}\;\Varid{h_{\mathsf{ND}}}\;{\mathbf{handle}}\;{\mathbf{op}}\;\mathtt{choose}\;()\;(\Varid{b}\, .\,\mathbf{do}\;\Varid{p'}\leftarrow {\mathbf{with}}\;\Varid{h_{\mathsf{inc}}}\;{\mathbf{handle}}\;\mathbf{if}\;\Varid{b}\;\mathbf{then}{}\<[E]%
\\
\>[10]{}\hsindent{1}{}\<[11]%
\>[11]{}{\mathbf{op}}\;\mathtt{inc}\;()\;(\Varid{x}\, .\,\Varid{x}\mathbin{+}\mathrm{5})\;\mathbf{else}\;{\mathbf{op}}\;\mathtt{inc}\;()\;(\Varid{y}\, .\,\Varid{y}\mathbin{+}\mathrm{2})\, ;\Varid{p'}\;\mathrm{0}){}\<[E]%
\\
\>[3]{}\leadsto{}\<[10]%
\>[10]{}\mbox{\commentbegin ~  \textsc{E-HandOp}   \commentend}{}\<[E]%
\\
\>[10]{}\mathbf{do}\;\Varid{xs}\leftarrow (\boldsymbol{\lambda}\Varid{b}\, .\,{\mathbf{with}}\;\Varid{h_{\mathsf{ND}}}\;{\mathbf{handle}}\;\mathbf{do}\;\Varid{p'}\leftarrow {\mathbf{with}}\;\Varid{h_{\mathsf{inc}}}\;{\mathbf{handle}}\;\mathbf{if}\;\Varid{b}\;\mathbf{then}\;{\mathbf{op}}{}\<[E]%
\\
\>[10]{}\hsindent{1}{}\<[11]%
\>[11]{}\mathtt{inc}\;()\;(\Varid{x}\, .\,\Varid{x}\mathbin{+}\mathrm{5})\;\mathbf{else}\;{\mathbf{op}}\;\mathtt{inc}\;()\;(\Varid{y}\, .\,\Varid{y}\mathbin{+}\mathrm{2})\, ;\Varid{p'}\;\mathrm{0})\;\mathsf{true}\, ;{}\<[E]%
\\
\>[10]{}\mathbf{do}\;\Varid{ys}\leftarrow (\boldsymbol{\lambda}\Varid{b}\, .\,{\mathbf{with}}\;\Varid{h_{\mathsf{ND}}}\;{\mathbf{handle}}\;\mathbf{do}\;\Varid{p'}\leftarrow {\mathbf{with}}\;\Varid{h_{\mathsf{inc}}}\;{\mathbf{handle}}\;\mathbf{if}\;\Varid{b}\;\mathbf{then}\;{\mathbf{op}}{}\<[E]%
\\
\>[10]{}\hsindent{1}{}\<[11]%
\>[11]{}\mathtt{inc}\;()\;(\Varid{x}\, .\,\Varid{x}\mathbin{+}\mathrm{5})\;\mathbf{else}\;{\mathbf{op}}\;\mathtt{inc}\;()\;(\Varid{y}\, .\,\Varid{y}\mathbin{+}\mathrm{2})\, ;\Varid{p'}\;\mathrm{0})\;\mathsf{false}\, ;\Varid{xs}+\hspace{-0.4em}+\Varid{ys}{}\<[E]%
\\
\>[3]{}\leadsto{}\<[10]%
\>[10]{}\mbox{\commentbegin ~  \textsc{E-AppAbs}   \commentend}{}\<[E]%
\\
\>[10]{}\mathbf{do}\;\Varid{xs}\leftarrow {\mathbf{with}}\;\Varid{h_{\mathsf{ND}}}\;{\mathbf{handle}}\;\mathbf{do}\;\Varid{p'}\leftarrow {\mathbf{with}}\;\Varid{h_{\mathsf{inc}}}\;{\mathbf{handle}}\;\mathbf{if}\;\mathsf{true}\;\mathbf{then}\;{\mathbf{op}}{}\<[E]%
\\
\>[10]{}\hsindent{1}{}\<[11]%
\>[11]{}\mathtt{inc}\;()\;(\Varid{x}\, .\,\Varid{x}\mathbin{+}\mathrm{5})\;\mathbf{else}\;{\mathbf{op}}\;\mathtt{inc}\;()\;(\Varid{y}\, .\,\Varid{y}\mathbin{+}\mathrm{2})\, ;\Varid{p'}\;\mathrm{0}{}\<[E]%
\\
\>[10]{}\mathbf{do}\;\Varid{ys}\leftarrow (\boldsymbol{\lambda}\Varid{b}\, .\,{\mathbf{with}}\;\Varid{h_{\mathsf{ND}}}\;{\mathbf{handle}}\;\mathbf{do}\;\Varid{p'}\leftarrow {\mathbf{with}}\;\Varid{h_{\mathsf{inc}}}\;{\mathbf{handle}}\;\mathbf{if}\;\Varid{b}\;\mathbf{then}\;{\mathbf{op}}{}\<[E]%
\\
\>[10]{}\hsindent{1}{}\<[11]%
\>[11]{}\mathtt{inc}\;()\;(\Varid{x}\, .\,\Varid{x}\mathbin{+}\mathrm{5})\;\mathbf{else}\;{\mathbf{op}}\;\mathtt{inc}\;()\;(\Varid{y}\, .\,\Varid{y}\mathbin{+}\mathrm{2})\, ;\Varid{p'}\;\mathrm{0})\;\mathsf{false}\, ;\Varid{xs}+\hspace{-0.4em}+\Varid{ys}{}\<[E]%
\\
\>[3]{}\leadsto{}\<[10]%
\>[10]{}\mbox{\commentbegin ~  reducing \ensuremath{\mathbf{if}}   \commentend}{}\<[E]%
\\
\>[10]{}\mathbf{do}\;\Varid{xs}\leftarrow {\mathbf{with}}\;\Varid{h_{\mathsf{ND}}}\;{\mathbf{handle}}\;(\mathbf{do}\;\Varid{p'}\leftarrow {\mathbf{with}}\;\Varid{h_{\mathsf{inc}}}\;{\mathbf{handle}}\;{\mathbf{op}}\;\mathtt{inc}\;(){}\<[E]%
\\
\>[10]{}\hsindent{1}{}\<[11]%
\>[11]{}(\Varid{x}\, .\,\Varid{x}\mathbin{+}\mathrm{5})\, ;\Varid{p'}\;\mathrm{0}){}\<[E]%
\\
\>[10]{}\mathbf{do}\;\Varid{ys}\leftarrow (\boldsymbol{\lambda}\Varid{b}\, .\,{\mathbf{with}}\;\Varid{h_{\mathsf{ND}}}\;{\mathbf{handle}}\;\mathbf{do}\;\Varid{p'}\leftarrow {\mathbf{with}}\;\Varid{h_{\mathsf{inc}}}\;{\mathbf{handle}}\;\mathbf{if}\;\Varid{b}\;\mathbf{then}\;{\mathbf{op}}{}\<[E]%
\\
\>[10]{}\hsindent{1}{}\<[11]%
\>[11]{}\mathtt{inc}\;()\;(\Varid{x}\, .\,\Varid{x}\mathbin{+}\mathrm{5})\;\mathbf{else}\;{\mathbf{op}}\;\mathtt{inc}\;()\;(\Varid{y}\, .\,\Varid{y}\mathbin{+}\mathrm{2})\, ;\Varid{p'}\;\mathrm{0})\;\mathsf{false}\, ;\Varid{xs}+\hspace{-0.4em}+\Varid{ys}{}\<[E]%
\\
\>[3]{}\leadsto{}\<[10]%
\>[10]{}\mbox{\commentbegin ~  \textsc{E-HandOp}   \commentend}{}\<[E]%
\\
\>[10]{}\mathbf{do}\;\Varid{xs}\leftarrow {\mathbf{with}}\;\Varid{h_{\mathsf{ND}}}\;{\mathbf{handle}}\;(\mathbf{do}\;\Varid{p'}\leftarrow {\mathbf{return}}\;(\boldsymbol{\lambda}\Varid{s}\, .\,\mathbf{do}\;\Varid{s'}\leftarrow \Varid{s}\mathbin{+}\mathrm{1}\, ;\mathbf{do}\;\Varid{k'}\leftarrow {}\<[E]%
\\
\>[10]{}\hsindent{1}{}\<[11]%
\>[11]{}(\boldsymbol{\lambda}\Varid{x}\, .\,{\mathbf{with}}\;\Varid{h_{\mathsf{inc}}}\;{\mathbf{handle}}\;(\Varid{x}\mathbin{+}\mathrm{5}))\;\Varid{s'}\, ;\Varid{k'}\;\Varid{s'})\, ;\Varid{p'}\;\mathrm{0}){}\<[E]%
\\
\>[10]{}\mathbf{do}\;\Varid{ys}\leftarrow (\boldsymbol{\lambda}\Varid{b}\, .\,{\mathbf{with}}\;\Varid{h_{\mathsf{ND}}}\;{\mathbf{handle}}\;\mathbf{do}\;\Varid{p'}\leftarrow {\mathbf{with}}\;\Varid{h_{\mathsf{inc}}}\;{\mathbf{handle}}\;\mathbf{if}\;\Varid{b}\;\mathbf{then}\;{\mathbf{op}}{}\<[E]%
\\
\>[10]{}\hsindent{1}{}\<[11]%
\>[11]{}\mathtt{inc}\;()\;(\Varid{x}\, .\,\Varid{x}\mathbin{+}\mathrm{5})\;\mathbf{else}\;{\mathbf{op}}\;\mathtt{inc}\;()\;(\Varid{y}\, .\,\Varid{y}\mathbin{+}\mathrm{2})\, ;\Varid{p'}\;\mathrm{0})\;\mathsf{false}\, ;\Varid{xs}+\hspace{-0.4em}+\Varid{ys}{}\<[E]%
\\
\>[3]{}\leadsto{}\<[10]%
\>[10]{}\mbox{\commentbegin ~  \textsc{E-DoRet}   \commentend}{}\<[E]%
\\
\>[10]{}\mathbf{do}\;\Varid{xs}\leftarrow {\mathbf{with}}\;\Varid{h_{\mathsf{ND}}}\;{\mathbf{handle}}\;(\boldsymbol{\lambda}\Varid{s}\, .\,\mathbf{do}\;\Varid{s'}\leftarrow \Varid{s}\mathbin{+}\mathrm{1}\, ;\Varid{k'}\leftarrow (\boldsymbol{\lambda}\Varid{x}\, .\,{\mathbf{with}}\;\Varid{h_{\mathsf{inc}}}{}\<[E]%
\\
\>[10]{}\hsindent{1}{}\<[11]%
\>[11]{}{\mathbf{handle}}\;(\Varid{x}\mathbin{+}\mathrm{5}))\Varid{s'}\, ;\Varid{k'}\;\Varid{s'})\;\mathrm{0}{}\<[E]%
\\
\>[10]{}\mathbf{do}\;\Varid{ys}\leftarrow (\boldsymbol{\lambda}\Varid{b}\, .\,{\mathbf{with}}\;\Varid{h_{\mathsf{ND}}}\;{\mathbf{handle}}\;\mathbf{do}\;\Varid{p'}\leftarrow {\mathbf{with}}\;\Varid{h_{\mathsf{inc}}}\;{\mathbf{handle}}\;\mathbf{if}\;\Varid{b}\;\mathbf{then}\;{\mathbf{op}}{}\<[E]%
\\
\>[10]{}\hsindent{1}{}\<[11]%
\>[11]{}\mathtt{inc}\;()\;(\Varid{x}\, .\,\Varid{x}\mathbin{+}\mathrm{5})\;\mathbf{else}\;{\mathbf{op}}\;\mathtt{inc}\;()\;(\Varid{y}\, .\,\Varid{y}\mathbin{+}\mathrm{2})\, ;\Varid{p'}\;\mathrm{0})\;\mathsf{false}\, ;\Varid{xs}+\hspace{-0.4em}+\Varid{ys}{}\<[E]%
\\
\>[3]{}\leadsto^{\ast}{}\<[11]%
\>[11]{}\mbox{\commentbegin ~  \textsc{E-AppAbs} and reducing \ensuremath{\mathbin{+}}  \commentend}{}\<[E]%
\\
\>[3]{}\hsindent{7}{}\<[10]%
\>[10]{}\mathbf{do}\;\Varid{xs}\leftarrow {\mathbf{with}}\;\Varid{h_{\mathsf{ND}}}\;{\mathbf{handle}}\;({\mathbf{with}}\;\Varid{h_{\mathsf{inc}}}\;{\mathbf{handle}}\;({\mathbf{return}}\;\mathrm{6}))\;\mathrm{1}{}\<[E]%
\\
\>[3]{}\hsindent{7}{}\<[10]%
\>[10]{}\mathbf{do}\;\Varid{ys}\leftarrow (\boldsymbol{\lambda}\Varid{b}\, .\,{\mathbf{with}}\;\Varid{h_{\mathsf{ND}}}\;{\mathbf{handle}}\;\mathbf{do}\;\Varid{p'}\leftarrow {\mathbf{with}}\;\Varid{h_{\mathsf{inc}}}\;{\mathbf{handle}}\;\mathbf{if}\;\Varid{b}\;\mathbf{then}\;{\mathbf{op}}{}\<[E]%
\\
\>[10]{}\hsindent{1}{}\<[11]%
\>[11]{}\mathtt{inc}\;()\;(\Varid{x}\, .\,\Varid{x}\mathbin{+}\mathrm{5})\;\mathbf{else}\;{\mathbf{op}}\;\mathtt{inc}\;()\;(\Varid{y}\, .\,\Varid{y}\mathbin{+}\mathrm{2})\, ;\Varid{p'}\;\mathrm{0})\;\mathsf{false}\, ;\Varid{xs}+\hspace{-0.4em}+\Varid{ys}{}\<[E]%
\\
\>[3]{}\leadsto{}\<[10]%
\>[10]{}\mbox{\commentbegin ~  \textsc{E-HandRet}  \commentend}{}\<[E]%
\\
\>[10]{}\mathbf{do}\;\Varid{xs}\leftarrow (\boldsymbol{\lambda}\Varid{s}\, .\,{\mathbf{return}}\;(\mathrm{6},\Varid{s}))\;\mathrm{1}{}\<[E]%
\\
\>[10]{}\mathbf{do}\;\Varid{ys}\leftarrow (\boldsymbol{\lambda}\Varid{b}\, .\,{\mathbf{with}}\;\Varid{h_{\mathsf{ND}}}\;{\mathbf{handle}}\;\mathbf{do}\;\Varid{p'}\leftarrow {\mathbf{with}}\;\Varid{h_{\mathsf{inc}}}\;{\mathbf{handle}}\;\mathbf{if}\;\Varid{b}\;\mathbf{then}\;{\mathbf{op}}{}\<[E]%
\\
\>[10]{}\hsindent{1}{}\<[11]%
\>[11]{}\mathtt{inc}\;()\;(\Varid{x}\, .\,\Varid{x}\mathbin{+}\mathrm{5})\;\mathbf{else}\;{\mathbf{op}}\;\mathtt{inc}\;()\;(\Varid{y}\, .\,\Varid{y}\mathbin{+}\mathrm{2})\, ;\Varid{p'}\;\mathrm{0})\;\mathsf{false}\, ;\Varid{xs}+\hspace{-0.4em}+\Varid{ys}{}\<[E]%
\\
\>[3]{}\leadsto{}\<[10]%
\>[10]{}\mbox{\commentbegin ~  \textsc{E-AppAbs}  \commentend}{}\<[E]%
\\
\>[10]{}\mathbf{do}\;\Varid{xs}\leftarrow {\mathbf{with}}\;\Varid{h_{\mathsf{ND}}}\;{\mathbf{handle}}\;{\mathbf{return}}\;(\mathrm{6},\mathrm{1}){}\<[E]%
\\
\>[10]{}\mathbf{do}\;\Varid{ys}\leftarrow (\boldsymbol{\lambda}\Varid{b}\, .\,{\mathbf{with}}\;\Varid{h_{\mathsf{ND}}}\;{\mathbf{handle}}\;\mathbf{do}\;\Varid{p'}\leftarrow {\mathbf{with}}\;\Varid{h_{\mathsf{inc}}}\;{\mathbf{handle}}\;\mathbf{if}\;\Varid{b}\;\mathbf{then}\;{\mathbf{op}}{}\<[E]%
\\
\>[10]{}\hsindent{1}{}\<[11]%
\>[11]{}\mathtt{inc}\;()\;(\Varid{x}\, .\,\Varid{x}\mathbin{+}\mathrm{5})\;\mathbf{else}\;{\mathbf{op}}\;\mathtt{inc}\;()\;(\Varid{y}\, .\,\Varid{y}\mathbin{+}\mathrm{2})\, ;\Varid{p'}\;\mathrm{0})\;\mathsf{false}\, ;\Varid{xs}+\hspace{-0.4em}+\Varid{ys}{}\<[E]%
\\
\>[3]{}\leadsto{}\<[10]%
\>[10]{}\mbox{\commentbegin ~  \textsc{E-HandOp}  \commentend}{}\<[E]%
\\
\>[10]{}\mathbf{do}\;\Varid{xs}\leftarrow {\mathbf{return}}\;[\mskip1.5mu (\mathrm{6},\mathrm{1})\mskip1.5mu]{}\<[E]%
\\
\>[10]{}\mathbf{do}\;\Varid{ys}\leftarrow (\boldsymbol{\lambda}\Varid{b}\, .\,{\mathbf{with}}\;\Varid{h_{\mathsf{ND}}}\;{\mathbf{handle}}\;\mathbf{do}\;\Varid{p'}\leftarrow {\mathbf{with}}\;\Varid{h_{\mathsf{inc}}}\;{\mathbf{handle}}\;\mathbf{if}\;\Varid{b}\;\mathbf{then}\;{\mathbf{op}}{}\<[E]%
\\
\>[10]{}\hsindent{1}{}\<[11]%
\>[11]{}\mathtt{inc}\;()\;(\Varid{x}\, .\,\Varid{x}\mathbin{+}\mathrm{5})\;\mathbf{else}\;{\mathbf{op}}\;\mathtt{inc}\;()\;(\Varid{y}\, .\,\Varid{y}\mathbin{+}\mathrm{2})\, ;\Varid{p'}\;\mathrm{0})\;\mathsf{false}\, ;\Varid{xs}+\hspace{-0.4em}+\Varid{ys}{}\<[E]%
\\
\>[3]{}\leadsto{}\<[10]%
\>[10]{}\mbox{\commentbegin ~  \textsc{E-DoRet}  \commentend}{}\<[E]%
\\
\>[10]{}\mathbf{do}\;\Varid{ys}\leftarrow (\boldsymbol{\lambda}\Varid{b}\, .\,{\mathbf{with}}\;\Varid{h_{\mathsf{ND}}}\;{\mathbf{handle}}\;\mathbf{do}\;\Varid{p'}\leftarrow \mathbf{if}\;\Varid{b}\;\mathbf{then}\;{\mathbf{with}}\;\Varid{h_{\mathsf{inc}}}\;{\mathbf{handle}}\;{\mathbf{op}}{}\<[E]%
\\
\>[10]{}\hsindent{1}{}\<[11]%
\>[11]{}\mathtt{inc}\;()\;(\Varid{x}\, .\,\Varid{x}\mathbin{+}\mathrm{5})\;\mathbf{else}\;{\mathbf{with}}\;\Varid{h_{\mathsf{inc}}}\;{\mathbf{handle}}\;{\mathbf{op}}\;\mathtt{inc}\;()\;(\Varid{y}\, .\,\Varid{y}\mathbin{+}\mathrm{2})\, ;\Varid{p'}\;\mathrm{0})\;\mathsf{false}\, ;{}\<[E]%
\\
\>[11]{}\hsindent{1}{}\<[12]%
\>[12]{}[\mskip1.5mu (\mathrm{6},\mathrm{1})\mskip1.5mu]+\hspace{-0.4em}+\Varid{ys}{}\<[E]%
\\
\>[3]{}\leadsto^{\ast}{}\<[11]%
\>[11]{}\mbox{\commentbegin ~  similar to above (the first branch of \ensuremath{\mathbf{if}})  \commentend}{}\<[E]%
\\
\>[3]{}\hsindent{7}{}\<[10]%
\>[10]{}\mathbf{do}\;\Varid{ys}\leftarrow {\mathbf{return}}\;[\mskip1.5mu (\mathrm{3},\mathrm{1})\mskip1.5mu]\;[\mskip1.5mu (\mathrm{6},\mathrm{1})\mskip1.5mu]+\hspace{-0.4em}+\Varid{ys}{}\<[E]%
\\
\>[3]{}\leadsto^{\ast}{}\<[11]%
\>[11]{}\mbox{\commentbegin ~  \textsc{E-DoRet}  \commentend}{}\<[E]%
\\
\>[3]{}\hsindent{7}{}\<[10]%
\>[10]{}[\mskip1.5mu (\mathrm{6},\mathrm{1})\mskip1.5mu]+\hspace{-0.4em}+[\mskip1.5mu (\mathrm{3},\mathrm{1})\mskip1.5mu]{}\<[E]%
\\
\>[3]{}\leadsto^{\ast}{}\<[11]%
\>[11]{}\mbox{\commentbegin ~  reducing \ensuremath{+\hspace{-0.4em}+}  \commentend}{}\<[E]%
\\
\>[3]{}\hsindent{7}{}\<[10]%
\>[10]{}{\mathbf{return}}\;[\mskip1.5mu (\mathrm{6},\mathrm{1}),(\mathrm{3},\mathrm{1})\mskip1.5mu]{}\<[E]%
\ColumnHook
\end{hscode}\resethooks
\indentend %
\subsection{Once}
\indentbegin \begin{hscode}\SaveRestoreHook
\column{B}{@{}>{\hspre}l<{\hspost}@{}}%
\column{3}{@{}>{\hspre}l<{\hspost}@{}}%
\column{7}{@{}>{\hspre}l<{\hspost}@{}}%
\column{8}{@{}>{\hspre}l<{\hspost}@{}}%
\column{11}{@{}>{\hspre}l<{\hspost}@{}}%
\column{15}{@{}>{\hspre}l<{\hspost}@{}}%
\column{19}{@{}>{\hspre}l<{\hspost}@{}}%
\column{E}{@{}>{\hspre}l<{\hspost}@{}}%
\>[3]{}{\mathbf{with}}\;\Varid{h_{\mathsf{once}}}\;{\mathbf{handle}}\;\Varid{c_{\mathsf{once}}}{}\<[E]%
\\
\>[3]{}\leadsto{}\<[7]%
\>[7]{}\mbox{\commentbegin ~  \textsc{E-HandSc}   \commentend}{}\<[E]%
\\
\>[7]{}\mathbf{do}\;{}\<[11]%
\>[11]{}\Varid{ts}{}\<[15]%
\>[15]{}\leftarrow (\boldsymbol{\lambda}\Varid{y}\, .\,{\mathbf{with}}\;\Varid{h_{\mathsf{once}}}\;{\mathbf{handle}}\;{\mathbf{op}}\;\mathtt{choose}\;()\;(\Varid{x}\, .\,{\mathbf{return}}\;\Varid{x}))\;()\, ;{}\<[E]%
\\
\>[7]{}\mathbf{do}\;{}\<[11]%
\>[11]{}\Varid{b}{}\<[15]%
\>[15]{}\leftarrow \Varid{ts}=[\mskip1.5mu \mskip1.5mu]\, ;\mathbf{if}\;\Varid{b}\;\mathbf{then}\;{\mathbf{return}}\;[\mskip1.5mu \mskip1.5mu]\;\mathbf{else}\;\mathbf{do}\;\Varid{t}\leftarrow \mathsf{head}\;\Varid{ts}\, ;(\boldsymbol{\lambda}\Varid{p}\, .\,{\mathbf{with}}{}\<[E]%
\\
\>[7]{}\hsindent{1}{}\<[8]%
\>[8]{}\Varid{h_{\mathsf{once}}}\;{\mathbf{handle}}\;(\mathbf{do}\;\Varid{q}\leftarrow {\mathbf{op}}\;\mathtt{choose}\;(\Varid{b}\, .\,{\mathbf{return}}\;\Varid{b})\, ;{\mathbf{return}}\;(\Varid{p},\Varid{q})))\;\Varid{t}{}\<[E]%
\\
\>[3]{}\leadsto{}\<[7]%
\>[7]{}\mbox{\commentbegin ~  \textsc{E-Do} and \textsc{E-AppAbs}   \commentend}{}\<[E]%
\\
\>[7]{}\mathbf{do}\;{}\<[11]%
\>[11]{}\Varid{ts}{}\<[15]%
\>[15]{}\leftarrow {\mathbf{with}}\;\Varid{h_{\mathsf{once}}}\;{\mathbf{handle}}\;{\mathbf{op}}\;\mathtt{choose}\;()\;(\Varid{x}\, .\,{\mathbf{return}}\;\Varid{x})\, ;{}\<[E]%
\\
\>[7]{}\mathbf{do}\;{}\<[11]%
\>[11]{}\Varid{b}{}\<[15]%
\>[15]{}\leftarrow \Varid{ts}=[\mskip1.5mu \mskip1.5mu]\, ;\mathbf{if}\;\Varid{b}\;\mathbf{then}\;{\mathbf{return}}\;[\mskip1.5mu \mskip1.5mu]\;\mathbf{else}\;\mathbf{do}\;\Varid{t}\leftarrow \mathsf{head}\;\Varid{ts}\, ;(\boldsymbol{\lambda}\Varid{p}\, .\,{\mathbf{with}}{}\<[E]%
\\
\>[7]{}\hsindent{1}{}\<[8]%
\>[8]{}\Varid{h_{\mathsf{once}}}\;{\mathbf{handle}}\;(\mathbf{do}\;\Varid{q}\leftarrow {\mathbf{op}}\;\mathtt{choose}\;(\Varid{b}\, .\,{\mathbf{return}}\;\Varid{b})\, ;{\mathbf{return}}\;(\Varid{p},\Varid{q})))\;\Varid{t}{}\<[E]%
\\
\>[3]{}\leadsto{}\<[7]%
\>[7]{}\mbox{\commentbegin ~  \textsc{E-Do} and \textsc{E-HandOp}   \commentend}{}\<[E]%
\\
\>[7]{}\mathbf{do}\;{}\<[11]%
\>[11]{}\Varid{ts}{}\<[15]%
\>[15]{}\leftarrow {}\<[19]%
\>[19]{}\mathbf{do}\;\Varid{xs}\leftarrow (\boldsymbol{\lambda}\Varid{x}\, .\,{\mathbf{with}}\;\Varid{h_{\mathsf{once}}}\;{\mathbf{handle}}\;{\mathbf{return}}\;\Varid{x})\;\mathsf{true}\, ;{}\<[E]%
\\
\>[19]{}\mathbf{do}\;\Varid{ys}\leftarrow (\boldsymbol{\lambda}\Varid{x}\, .\,{\mathbf{with}}\;\Varid{h_{\mathsf{once}}}\;{\mathbf{handle}}\;{\mathbf{return}}\;\Varid{x})\;\mathsf{false}\, ;\Varid{xs}+\hspace{-0.4em}+\Varid{ys}\, ;{}\<[E]%
\\
\>[7]{}\mathbf{do}\;{}\<[11]%
\>[11]{}\Varid{b}{}\<[15]%
\>[15]{}\leftarrow \Varid{ts}=[\mskip1.5mu \mskip1.5mu]\, ;\mathbf{if}\;\Varid{b}\;\mathbf{then}\;{\mathbf{return}}\;[\mskip1.5mu \mskip1.5mu]\;\mathbf{else}\;\mathbf{do}\;\Varid{t}\leftarrow \mathsf{head}\;\Varid{ts}\, ;(\boldsymbol{\lambda}\Varid{p}\, .\,{\mathbf{with}}{}\<[E]%
\\
\>[7]{}\hsindent{1}{}\<[8]%
\>[8]{}\Varid{h_{\mathsf{once}}}\;{\mathbf{handle}}\;(\mathbf{do}\;\Varid{q}\leftarrow {\mathbf{op}}\;\mathtt{choose}\;(\Varid{b}\, .\,{\mathbf{return}}\;\Varid{b})\, ;{\mathbf{return}}\;(\Varid{p},\Varid{q})))\;\Varid{t}{}\<[E]%
\\
\>[3]{}\leadsto{}\<[7]%
\>[7]{}\mbox{\commentbegin ~  \textsc{E-Do} and \textsc{E-AppAbs}   \commentend}{}\<[E]%
\\
\>[7]{}\mathbf{do}\;{}\<[11]%
\>[11]{}\Varid{ts}{}\<[15]%
\>[15]{}\leftarrow {}\<[19]%
\>[19]{}\mathbf{do}\;\Varid{xs}\leftarrow {\mathbf{with}}\;\Varid{h_{\mathsf{once}}}\;{\mathbf{handle}}\;{\mathbf{return}}\;\mathsf{true}\, ;{}\<[E]%
\\
\>[19]{}\mathbf{do}\;\Varid{ys}\leftarrow {\mathbf{with}}\;\Varid{h_{\mathsf{once}}}\;{\mathbf{handle}}\;{\mathbf{return}}\;\mathsf{false}\, ;\Varid{xs}+\hspace{-0.4em}+\Varid{ys}\, ;{}\<[E]%
\\
\>[7]{}\mathbf{do}\;{}\<[11]%
\>[11]{}\Varid{b}{}\<[15]%
\>[15]{}\leftarrow \Varid{ts}=[\mskip1.5mu \mskip1.5mu]\, ;\mathbf{if}\;\Varid{b}\;\mathbf{then}\;{\mathbf{return}}\;[\mskip1.5mu \mskip1.5mu]\;\mathbf{else}\;\mathbf{do}\;\Varid{t}\leftarrow \mathsf{head}\;\Varid{ts}\, ;(\boldsymbol{\lambda}\Varid{p}\, .\,{\mathbf{with}}{}\<[E]%
\\
\>[7]{}\hsindent{1}{}\<[8]%
\>[8]{}\Varid{h_{\mathsf{once}}}\;{\mathbf{handle}}\;(\mathbf{do}\;\Varid{q}\leftarrow {\mathbf{op}}\;\mathtt{choose}\;(\Varid{b}\, .\,{\mathbf{return}}\;\Varid{b})\, ;{\mathbf{return}}\;(\Varid{p},\Varid{q})))\;\Varid{t}{}\<[E]%
\\
\>[3]{}\leadsto{}\<[7]%
\>[7]{}\mbox{\commentbegin ~  \textsc{E-Do} and \textsc{E-HandRet}   \commentend}{}\<[E]%
\\
\>[7]{}\mathbf{do}\;{}\<[11]%
\>[11]{}\Varid{ts}{}\<[15]%
\>[15]{}\leftarrow \mathbf{do}\;\Varid{xs}\leftarrow {\mathbf{return}}\;[\mskip1.5mu \mathsf{true}\mskip1.5mu]\, ;\mathbf{do}\;\Varid{ys}\leftarrow {\mathbf{return}}\;[\mskip1.5mu \mathsf{false}\mskip1.5mu]\, ;\Varid{xs}+\hspace{-0.4em}+\Varid{ys}\, ;{}\<[E]%
\\
\>[7]{}\mathbf{do}\;{}\<[11]%
\>[11]{}\Varid{b}{}\<[15]%
\>[15]{}\leftarrow \Varid{ts}=[\mskip1.5mu \mskip1.5mu]\, ;\mathbf{if}\;\Varid{b}\;\mathbf{then}\;{\mathbf{return}}\;[\mskip1.5mu \mskip1.5mu]\;\mathbf{else}\;\mathbf{do}\;\Varid{t}\leftarrow \mathsf{head}\;\Varid{ts}\, ;(\boldsymbol{\lambda}\Varid{p}\, .\,{\mathbf{with}}{}\<[E]%
\\
\>[7]{}\hsindent{1}{}\<[8]%
\>[8]{}\Varid{h_{\mathsf{once}}}\;{\mathbf{handle}}\;(\mathbf{do}\;\Varid{q}\leftarrow {\mathbf{op}}\;\mathtt{choose}\;(\Varid{b}\, .\,{\mathbf{return}}\;\Varid{b})\, ;{\mathbf{return}}\;(\Varid{p},\Varid{q})))\;\Varid{t}{}\<[E]%
\\
\>[3]{}\leadsto^{\ast}\mbox{\commentbegin ~  \textsc{E-DoRet}   \commentend}{}\<[E]%
\\
\>[3]{}\hsindent{4}{}\<[7]%
\>[7]{}\mathbf{do}\;{}\<[11]%
\>[11]{}\Varid{ts}{}\<[15]%
\>[15]{}\leftarrow [\mskip1.5mu \mathsf{true},\mathsf{false}\mskip1.5mu]\, ;{}\<[E]%
\\
\>[3]{}\hsindent{4}{}\<[7]%
\>[7]{}\mathbf{do}\;{}\<[11]%
\>[11]{}\Varid{b}{}\<[15]%
\>[15]{}\leftarrow \Varid{ts}=[\mskip1.5mu \mskip1.5mu]\, ;\mathbf{if}\;\Varid{b}\;\mathbf{then}\;{\mathbf{return}}\;[\mskip1.5mu \mskip1.5mu]\;\mathbf{else}\;\mathbf{do}\;\Varid{t}\leftarrow \mathsf{head}\;\Varid{ts}\, ;(\boldsymbol{\lambda}\Varid{p}\, .\,{\mathbf{with}}{}\<[E]%
\\
\>[7]{}\hsindent{1}{}\<[8]%
\>[8]{}\Varid{h_{\mathsf{once}}}\;{\mathbf{handle}}\;(\mathbf{do}\;\Varid{q}\leftarrow {\mathbf{op}}\;\mathtt{choose}\;(\Varid{b}\, .\,{\mathbf{return}}\;\Varid{b})\, ;{\mathbf{return}}\;(\Varid{p},\Varid{q})))\;\Varid{t}{}\<[E]%
\\
\>[3]{}\leadsto^{\ast}\mbox{\commentbegin ~  \textsc{E-DoRet}   \commentend}{}\<[E]%
\\
\>[3]{}\hsindent{4}{}\<[7]%
\>[7]{}(\boldsymbol{\lambda}\Varid{p}\, .\,{\mathbf{with}}\;\Varid{h_{\mathsf{once}}}\;{\mathbf{handle}}\;(\mathbf{do}\;\Varid{q}\leftarrow {\mathbf{op}}\;\mathtt{choose}\;(\Varid{b}\, .\,{\mathbf{return}}\;\Varid{b})\, ;{\mathbf{return}}\;(\Varid{p},\Varid{q})))\;{}\<[E]%
\\
\>[7]{}\hsindent{1}{}\<[8]%
\>[8]{}\mathsf{true}{}\<[E]%
\\
\>[3]{}\leadsto{}\<[7]%
\>[7]{}\mbox{\commentbegin ~  \textsc{E-AppAbs}   \commentend}{}\<[E]%
\\
\>[7]{}{\mathbf{with}}\;\Varid{h_{\mathsf{once}}}\;{\mathbf{handle}}\;(\mathbf{do}\;\Varid{q}\leftarrow {\mathbf{op}}\;\mathtt{choose}\;(\Varid{b}\, .\,{\mathbf{return}}\;\Varid{b})\, ;{\mathbf{return}}\;(\mathsf{true},\Varid{q})){}\<[E]%
\\
\>[3]{}\leadsto^{\ast}\mbox{\commentbegin ~  similar to \ref{subsec:nd-derivation} (handling of \ensuremath{\mathtt{choose}})   \commentend}{}\<[E]%
\\
\>[3]{}\hsindent{4}{}\<[7]%
\>[7]{}{\mathbf{return}}\;[\mskip1.5mu (\mathsf{true},\mathsf{true}),(\mathsf{true},\mathsf{false})\mskip1.5mu]{}\<[E]%
\ColumnHook
\end{hscode}\resethooks
\indentend 

\clearpage
\section{Type Equivalence Rules}
\label{app:type-equiv}

This appendix shows the type equivalence rules of \ensuremath{\lambda_{\mathit{sc}}}.
\Cref{fig:typed-equiv,fig:row-equiv} contains the rules. Rules \textsc{Q-AppAbs}
and \textsc{Q-Swap} deserve special attention. The other rules are
straightforward.

\begin{figure}[h]
    \ruleform{\ensuremath{\sigma_{1}\;\equiv\;\sigma_{2}}}
    \quad\text{Type equivalence} \hspace*{\fill}
\begin{mathpar}

  \inferrule*[right=Q-Refl]
  { }
  { \ensuremath{\sigma\;\equiv\;\sigma} }

  \inferrule*[right=Q-Symm]
  { \ensuremath{\sigma_{1}\;\equiv\;\sigma_{2}}
  }
  { \ensuremath{\sigma_{2}\;\equiv\;\sigma_{1}} }

  \inferrule*[right=Q-Trans]
  { \ensuremath{\sigma_{1}\;\equiv\;\sigma_{2}}
    \\ \ensuremath{\sigma_{2}\;\equiv\;\sigma_{3}}
  }
  { \ensuremath{\sigma_{1}\;\equiv\;\sigma_{3}} }

  \inferrule*[right=Q-Pair]
  { \ensuremath{\Conid{A}_{1}\;\equiv\;\Conid{A}_{2}}
    \\ \ensuremath{\Conid{B}_{1}\;\equiv\;\Conid{B}_{2}}
  }
  { \ensuremath{(\Conid{A}_{1},\Conid{B}_{1})\;\equiv\;(\Conid{A}_{2},\Conid{B}_{2})} }

  \inferrule*[right=Q-Fun]
  { \ensuremath{\Conid{A}\;\equiv\;\Conid{B}}
    \\ \ensuremath{ \underline{C} \;\equiv\; \underline{D} }
  }
  { \ensuremath{\Conid{A}\to  \underline{C} \;\equiv\;\Conid{B}\to  \underline{D} } }

  \inferrule*[right=Q-Hand]
  { \ensuremath{\underline{\Conid{C}_{1}}\;\equiv\;\underline{\Conid{D}_{1}}}
    \\ \ensuremath{\underline{\Conid{C}_{2}}\;\equiv\;\underline{\Conid{D}_{2}}}
  }
  { \ensuremath{\underline{\Conid{C}_{1}}\Rightarrow \underline{\Conid{C}_{2}}\;\equiv\;\underline{\Conid{D}_{1}}\Rightarrow \underline{\Conid{D}_{2}}} }

  \inferrule*[right=Q-AllTy]
  { \ensuremath{\sigma_{1}\;\equiv\;\sigma_{2}}
  }
  { \ensuremath{\forall\;\alpha\, .\,\sigma_{1}\;\equiv\;\forall\;\alpha\, .\,\sigma_{2}} }

  \inferrule*[right=Q-AllRow]
  { \ensuremath{\sigma_{1}\;\equiv\;\sigma_{2}}
  }
  { \ensuremath{\forall\;\mu\, .\,\sigma_{1}\;\equiv\;\forall\;\mu\, .\,\sigma_{2}} }

  \inferrule*[right=Q-Abs]
  { \ensuremath{\Conid{A}\;\equiv\;\Conid{B}}
  }
  { \ensuremath{\lambda\;\alpha\, .\,\Conid{A}\;\equiv\;\lambda\;\alpha\, .\,\Conid{B}} }

  \inferrule*[right=Q-App]
  { \ensuremath{\Conid{M}_{1}\;\equiv\;\Conid{M}_{2}}
    \\ \ensuremath{\Conid{A}\;\equiv\;\Conid{B}}
  }
  { \ensuremath{\Conid{M}_{1}\;\Conid{A}\;\equiv\;\Conid{M}_{2}\;\Conid{B}} }

  \inferrule*[right=Q-AppAbs]
  { }
  { \ensuremath{(\lambda\;\alpha\, .\,\Conid{A})\;\Conid{B}\;\equiv\;\Conid{A}\;[\mskip1.5mu \Conid{B}\mathbin{/}\alpha\mskip1.5mu]} }

  \inferrule*[right=Q-Comp]
  { \ensuremath{\Conid{A}\;\equiv\;\Conid{B}}
    \\ \ensuremath{\Conid{E}\;\equiv_{\langle\rangle}\;\Conid{F}}
  }
  { \ensuremath{\Conid{A}\mathbin{!}\langle\Conid{E}\rangle\equiv\;\Conid{B}\mathbin{!}\langle\Conid{F}\rangle} }

\end{mathpar}
\caption{Type equivalence of \ensuremath{\lambda_{\mathit{sc}}}.}
\label{fig:typed-equiv}
\end{figure}

\begin{figure}[h]
    \ruleform{\ensuremath{\Conid{E}\;\equiv_{\langle\rangle}\;\Conid{F}}}
    \quad\text{Row equivalence} \hspace*{\fill}
\begin{mathpar}

  \inferrule*[right=R-Refl]
  { }
  { \ensuremath{\Conid{E}\;\equiv_{\langle\rangle}\;\Conid{E}} }

  \inferrule*[right=R-Symm]
  { \ensuremath{\Conid{E}\;\equiv_{\langle\rangle}\;\Conid{F}}
  }
  { \ensuremath{\Conid{F}\;\equiv_{\langle\rangle}\;\Conid{E}} }

  \inferrule*[right=R-Trans]
  { \ensuremath{\Conid{E}_{1}\;\equiv_{\langle\rangle}\;\Conid{E}_{2}}
    \\ \ensuremath{\Conid{E}_{2}\;\equiv_{\langle\rangle}\;\Conid{E}_{3}}
  }
  { \ensuremath{\Conid{E}_{1}\;\equiv_{\langle\rangle}\;\Conid{E}_{3}} }

  \inferrule*[right=R-Head]
  { \ensuremath{\Conid{E}\;\equiv_{\langle\rangle}\;\Conid{F}}
  }
  { \ensuremath{\ell\, ;\Conid{E}\;\equiv_{\langle\rangle}\;\ell\, ;\Conid{F}} }

  \inferrule*[right=R-Swap]
  { \ensuremath{\ell_{1}\;\not =\;\ell_{2}}
  }
  { \ensuremath{\ell_{1}\, ;\ell_{2}\, ;\Conid{E}\;\equiv_{\langle\rangle}\;\ell_{2}\, ;\ell_{1}\, ;\Conid{E}} }

\end{mathpar}
\caption{Row equivalence of \ensuremath{\lambda_{\mathit{sc}}}.}
\label{fig:row-equiv}
\end{figure}
\clearpage
\section{Well-scopedness Rules}
\label{app:well-scopedness}

This appendix shows the well-scopedness rules of \ensuremath{\lambda_{\mathit{sc}}}.
\Cref{fig:typed-well-scopedness} contains the rules.

\begin{figure}[h]
  \ruleform{ \makebox[0ex][l]{\phantom{\underline{C}}} \ensuremath{\Gamma\vdash\sigma}}
  \ruleform{ \makebox[0ex][l]{\phantom{\underline{C}}} \ensuremath{\Gamma\vdash\Conid{M}}}
  \ruleform{ \makebox[0ex][l]{\phantom{\underline{C}}} \ensuremath{\Gamma\vdash\Conid{E}}}
  \ruleform{ \makebox[0ex][l]{\phantom{\underline{C}}} \ensuremath{\Gamma\vdash \underline{C} }}
  \quad\text{Type well-scopedness} \hspace*{\fill}
\begin{mathpar}

  \inferrule*[right=W-Unit]
  { }
  { \ensuremath{\Gamma\vdash\mathsf{()}} }

  \inferrule*[right=W-Pair]
  { \ensuremath{\Gamma\vdash\Conid{A}}
    \\ \ensuremath{\Gamma\vdash\Conid{B}}
  }
  { \ensuremath{\Gamma\vdash(\Conid{A},\Conid{B})} }

  \inferrule*[right=W-Var]
  { \ensuremath{\alpha\;\in\;\Gamma}
  }
  { \ensuremath{\Gamma\vdash\alpha} }

  \inferrule*[right=W-All]
  { \ensuremath{\Gamma,\alpha\vdash\Conid{A}}
  }
  { \ensuremath{\Gamma\vdash\forall\;\alpha\, .\,\Conid{A}} }

  \inferrule*[right=W-Comp]
  { \ensuremath{\Gamma\vdash\Conid{A}}
    \\ \ensuremath{\Gamma\vdash\Conid{E}}
  }
  { \ensuremath{\Gamma\vdash\Conid{A}\mathbin{!}\langle\Conid{E}\rangle} }

  \inferrule*[right=W-Abs]
  { \ensuremath{\Gamma,\alpha\vdash\Conid{A}}
  }
  { \ensuremath{\Gamma\vdash\lambda\;\alpha\, .\,\Conid{A}} }

  \inferrule*[right=W-App]
  { \ensuremath{\Gamma\vdash\Conid{M}}
    \\ \ensuremath{\Gamma\vdash\Conid{A}}
  }
  { \ensuremath{\Gamma\vdash\Conid{M}\;\Conid{A}} }

  \inferrule*[right=W-Fun]
  { \ensuremath{\Gamma\vdash\Conid{A}}
    \\ \ensuremath{\Gamma\vdash \underline{C} }
  }
  { \ensuremath{\Gamma\vdash\Conid{A}\to  \underline{C} } }

  \inferrule*[right=W-Hand]
  { \ensuremath{\Gamma\vdash \underline{C} }
    \\ \ensuremath{\Gamma\vdash \underline{D} }
  }
  { \ensuremath{\Gamma\vdash \underline{C} \Rightarrow  \underline{D} } }

  \inferrule*[right=W-RowVar]
  { \ensuremath{\mu\;\in\;\Gamma} }
  { \ensuremath{\Gamma\vdash\mu} }

  \inferrule*[right=W-EmptyRow]
  { }
  { \ensuremath{\Gamma\vdash \cdot} }

  \inferrule*[right=W-Extension]
  { \ensuremath{\Gamma\vdash\Conid{E}}
  }
  { \ensuremath{\Gamma\vdash\ell\, ;\Conid{E}} }
\end{mathpar}
\caption{Well-scopedness rules of \ensuremath{\lambda_{\mathit{sc}}}.}
\label{fig:typed-well-scopedness}
\end{figure}
\clearpage
\section{Syntax-directed version of \ensuremath{\lambda_{\mathit{sc}}}}
\label{app:syntax-directed}
This section describes the syntax-direction version of \ensuremath{\lambda_{\mathit{sc}}}.

The syntax-directed rules can be found in
\Cref{fig:syntax-directed-value-typing} for value typing,
\Cref{fig:syntax-directed-computation-typing} for computation and
\Cref{fig:syntax-directed-handler-typing} for handler typing.
\begin{figure}[h]
\ruleform{\ensuremath{\Gamma\vdash_{\mathsf{SD}}\Varid{v}\mathbin{:}\Conid{A}}} \quad\text{Value Typing} \hspace*{\fill}
\begin{mathpar}
  \inferrule*[right=SD-Var]
  { \ensuremath{(\Varid{x}\mathbin{:}\sigma)\;\in\;\Gamma}
    \\ \ensuremath{\colorbox{lightgray}{$\sigma\leq \Conid{A}$}}
    \\ \ensuremath{\Gamma\vdash_{\mathsf{SD}}\Conid{A}}
  }
  { \ensuremath{\Gamma\vdash_{\mathsf{SD}}\Varid{x}\mathbin{:}\colorbox{lightgray}{$\Conid{A}$}} }

  \inferrule*[right=SD-Unit]
  { }
  { \ensuremath{\Gamma\vdash_{\mathsf{SD}}\mathsf{()}\mathbin{:}\mathsf{()}} }

  \inferrule*[right=SD-Pair]
  { \ensuremath{\Gamma\vdash_{\mathsf{SD}}\Varid{v}_{1}\mathbin{:}\Conid{A}}
    \\ \ensuremath{\Gamma\vdash_{\mathsf{SD}}\Varid{v}_{2}\mathbin{:}\Conid{B}}
  }
  { \ensuremath{\Gamma\vdash_{\mathsf{SD}}(\Varid{v}_{1},\Varid{v}_{2})\mathbin{:}(\Conid{A},\Conid{B})} }

  \inferrule*[right=SD-Abs]
  {
    \ensuremath{\Gamma,\Varid{x}\mathbin{:}\Conid{A}\vdash_{\mathsf{SD}}\Varid{c}\mathbin{:} \underline{C} }
  }
  { \ensuremath{\Gamma\vdash_{\mathsf{SD}}\boldsymbol{\lambda}\Varid{x}\, .\,\Varid{c}\mathbin{:}\Conid{A}\to  \underline{C} } }

  \inferrule*[lab=SD-Handler]
  { \ensuremath{\Conid{F}\;\equiv_{\langle\rangle}\;\mathit{labels\!}\;(\Varid{oprs})\, ;\Conid{E}}
    \\ \ensuremath{\alpha\;\notin\;\Gamma}
    \\ \ensuremath{\Gamma,\alpha\;\mathbin{\vdash_{\Conid{M}}}\;{\mathbf{return}}\;\Varid{x}\mapsto\Varid{c}_{\Varid{r}}\mathbin{:}\Conid{M}\;\alpha\hspace{0.1em}!\hspace{0.1em}\langle\Conid{E}\rangle}
    \\ \ensuremath{\Gamma,\alpha\;\mathbin{\vdash_{\Conid{M}}}\;\Varid{oprs}\mathbin{:}\Conid{M}\;\alpha\hspace{0.1em}!\hspace{0.1em}\langle\Conid{E}\rangle}
    \\ \ensuremath{\Gamma,\alpha\;\mathbin{\vdash_{\Conid{M}}}\;{\mathbf{fwd}}\;\Varid{f}\;\Varid{p}\;\Varid{k}\mapsto\Varid{c}_{\Varid{f}}\mathbin{:}\Conid{M}\;\alpha\hspace{0.1em}!\hspace{0.1em}\langle\Conid{E}\rangle}
    \\ \ensuremath{\Gamma\vdash_{\mathsf{SD}}\Conid{A}}
  }
  { \ensuremath{\Gamma\vdash_{\mathsf{SD}}{\mathbf{handler}_M}\;\{\mskip1.5mu {\mathbf{return}}\;\Varid{x}\mapsto\Varid{c}_{\Varid{r}},\Varid{oprs},{\mathbf{fwd}}\;\Varid{f}\;\Varid{p}\;\Varid{k}\mapsto\Varid{c}_{\Varid{f}}\mskip1.5mu\}\mathbin{:}\Conid{A}\hspace{0.1em}!\hspace{0.1em}\langle\Conid{F}\rangle\Rightarrow \Conid{M}\;\Conid{A}\hspace{0.1em}!\hspace{0.1em}\langle\Conid{E}\rangle} }
\end{mathpar}
\caption{Syntax-directed value typing.}
\label{fig:syntax-directed-value-typing}
\end{figure}

\begin{figure}[tp]
\ruleform{\ensuremath{\Gamma\vdash_{\mathsf{SD}}\Varid{c}\mathbin{:} \underline{C} }} \quad\text{Computation Typing} \hspace*{\fill}
\begin{mathpar}
  \inferrule*[right=SD-App]
  { \ensuremath{\Gamma\vdash_{\mathsf{SD}}\Varid{v}_{1}\mathbin{:}\Conid{A}_{1}\to  \underline{C} }
    \\ \ensuremath{\Gamma\vdash_{\mathsf{SD}}\Varid{v}_{2}\mathbin{:}\Conid{A}_{2}}
    \\ \ensuremath{\Conid{A}_{1}\;\equiv\;\Conid{A}_{2}}
  }
  { \ensuremath{\Gamma\vdash_{\mathsf{SD}}\Varid{v}_{1}\;\Varid{v}_{2}\mathbin{:} \underline{C} } }

  \inferrule*[right=SD-Do]
  { \ensuremath{\Gamma\vdash_{\mathsf{SD}}\Varid{c}_{1}\mathbin{:}\Conid{A}\hspace{0.1em}!\hspace{0.1em}\langle\Conid{E}_{1}\rangle}
    \\ \ensuremath{\Gamma,\Varid{x}\mathbin{:}\Conid{A}\vdash_{\mathsf{SD}}\Varid{c}_{2}\mathbin{:}\Conid{B}\hspace{0.1em}!\hspace{0.1em}\langle\Conid{E}_{2}\rangle}
    \\ \ensuremath{\Conid{E}_{1}\;\equiv\;\Conid{E}_{2}}
  }
  { \ensuremath{\Gamma\vdash_{\mathsf{SD}}\mathbf{do}\;\Varid{x}\leftarrow \Varid{c}_{1}\, ;\Varid{c}_{2}\mathbin{:}\Conid{B}\hspace{0.1em}!\hspace{0.1em}\langle\Conid{E}_{2}\rangle} }

  \inferrule*[right=SD-Let]
  { \ensuremath{\Gamma,\overline{\alpha},\overline{\mu}\vdash_{\mathsf{SD}}\Varid{v}\mathbin{:}\colorbox{lightgray}{$\Conid{A}$}}
    \\ \ensuremath{\colorbox{lightgray}{$(\overline{\alpha}\;\notin\;\Gamma)$}}
    \\ \ensuremath{\colorbox{lightgray}{$(\overline{\mu}\;\notin\;\Gamma)$}}
    \\ \ensuremath{\Gamma,\Varid{x}\mathbin{:}\overline{\forall\;\alpha}\, .\,\overline{\forall\;\mu}\, .\,\Conid{A}\vdash_{\mathsf{SD}}\Varid{c}\mathbin{:} \underline{C} }
  }
  { \ensuremath{\Gamma\vdash_{\mathsf{SD}}\mathbf{let}\;\Varid{x}\mathrel{=}\Varid{v}\;\mathbf{in}\;\Varid{c}\mathbin{:} \underline{C} } }

  \inferrule*[right=SD-Ret]
  { \ensuremath{\Gamma\vdash_{\mathsf{SD}}\Varid{v}\mathbin{:}\Conid{A}}
  }
  { \ensuremath{\Gamma\vdash_{\mathsf{SD}}{\mathbf{return}}\;\Varid{v}\mathbin{:}\Conid{A}\hspace{0.1em}!\hspace{0.1em}\langle\Conid{E}\rangle} }

  \inferrule*[right=SD-Hand]
  {    \ensuremath{\Gamma\vdash_{\mathsf{SD}}\Varid{v}\mathbin{:} \underline{C} _{1}\Rightarrow  \underline{D} _{1}}
    \\ \ensuremath{\Gamma\vdash_{\mathsf{SD}}\Varid{c}\mathbin{:} \underline{C} _{2}}
    \\ \ensuremath{ \underline{C} _{1}\;\equiv\; \underline{C} _{2}}
    \\ \ensuremath{ \underline{D} _{1}\;\equiv\; \underline{D} _{2}}
  }
  { \ensuremath{\Gamma\vdash_{\mathsf{SD}}{\mathbf{with}}\;\Varid{v}\;{\mathbf{handle}}\;\Varid{c}\mathbin{:} \underline{D} _{2}} }

  \inferrule*[right=SD-Op]
  { \ensuremath{(\ell^{\textsf{op}}\mathbin{:}{A_{\textsf{op}}}\rightarrowtriangle{B_{\textsf{op}}})\;\in\;\Sigma}
    \\ \ensuremath{\Gamma\vdash_{\mathsf{SD}}\Varid{v}\mathbin{:}\Conid{A}_{1}}
    \\ \ensuremath{{A_{\textsf{op}}}\;\equiv\;\Conid{A}_{1}}
    \\ \ensuremath{\Gamma,\Varid{y}\mathbin{:}{B_{\textsf{op}}}\vdash_{\mathsf{SD}}\Varid{c}\mathbin{:}\Conid{A}\hspace{0.1em}!\hspace{0.1em}\langle\Conid{E}\rangle}
    \\ \ensuremath{\ell^{\textsf{op}}\;\in\;\Conid{E}}
  }
  { \ensuremath{\Gamma\vdash_{\mathsf{SD}}{\mathbf{op}}\;\ell^{\textsf{op}}\;\Varid{v}\;(\Varid{y}\, .\,\Varid{c})\mathbin{:}\Conid{A}\hspace{0.1em}!\hspace{0.1em}\langle\Conid{E}\rangle} }

  \inferrule*[right=SD-Sc]
  { \ensuremath{(\ell^{\textsf{sc}}\mathbin{:}{A_{\textsf{sc}}}\rightarrowtriangle{B_{\textsf{sc}}})\;\in\;\Sigma}
    \\ \ensuremath{\Gamma\vdash_{\mathsf{SD}}\Varid{v}\mathbin{:}\Conid{A}_{1}}
    \\ \ensuremath{{A_{\textsf{sc}}}\;\equiv\;\Conid{A}_{1}}
    \\ \ensuremath{\Gamma,\Varid{y}\mathbin{:}{B_{\textsf{sc}}}\vdash_{\mathsf{SD}}\Varid{c}_{1}\mathbin{:}\Conid{B}\hspace{0.1em}!\hspace{0.1em}\langle\Conid{E}_{1}\rangle}
    \\ \ensuremath{\Gamma,\Varid{z}\mathbin{:}\Conid{B}\vdash_{\mathsf{SD}}\Varid{c}_{2}\mathbin{:}\Conid{A}\hspace{0.1em}!\hspace{0.1em}\langle\Conid{E}_{2}\rangle}
    \\ \ensuremath{\Conid{E}_{1}\;\equiv\;\Conid{E}_{2}}
    \\ \ensuremath{\ell^{\textsf{sc}}\;\in\;\Conid{E}_{2}}
  }
  { \ensuremath{\Gamma\vdash_{\mathsf{SD}}{\mathbf{sc}}\;\ell^{\textsf{sc}}\;\Varid{v}\;(\Varid{y}\, .\,\Varid{c}_{1})\;(\Varid{z}\, .\,\Varid{c}_{2})\mathbin{:}\Conid{A}\hspace{0.1em}!\hspace{0.1em}\langle\Conid{E}_{2}\rangle} }

\end{mathpar}
\caption{Syntax-directed computation typing.}
\label{fig:syntax-directed-computation-typing}
\end{figure}

\begin{figure}[tp]
  \raggedright
    \ruleform{\ensuremath{\Gamma\;\mathbin{\vdash_{\Conid{M}}}\;{\mathbf{return}}\;\Varid{x}\mapsto\Varid{c}_{\Varid{r}}\mathbin{:}\Conid{M}\;\Conid{A}\hspace{0.1em}!\hspace{0.1em}\langle\Conid{E}\rangle}}
    \ruleform{\ensuremath{\Gamma\;\mathbin{\vdash_{\Conid{M}}}\;\Varid{oprs}\mathbin{:}\Conid{M}\;\Conid{A}\hspace{0.1em}!\hspace{0.1em}\langle\Conid{E}\rangle}} \\
    \ruleform{\ensuremath{\Gamma\;\mathbin{\vdash_{\Conid{M}}}\;{\mathbf{fwd}}\;\Varid{f}\;\Varid{p}\;\Varid{k}\mapsto\Varid{c}_{\Varid{f}}\mathbin{:}\Conid{M}\;\Conid{A}\hspace{0.1em}!\hspace{0.1em}\langle\Conid{E}\rangle}} \\[0.07cm]
    \text{Return-, operation-, and forwarding-clause typing} \\
  \begin{mathpar}
    \inferrule*[right=SD-Return]
    {  \ensuremath{\Gamma,\Varid{x}\mathbin{:}\Conid{A}_{1}\vdash_{\mathsf{SD}}\Varid{c}_{\Varid{r}}\mathbin{:}\Conid{M}\;\Conid{A}_{2}\hspace{0.1em}!\hspace{0.1em}\langle\Conid{E}\rangle}
    \\ \ensuremath{\Conid{A}_{1}\;\equiv\;\Conid{A}_{2}}
    }
    { \ensuremath{\Gamma\;\mathbin{\vdash_{\Conid{M}}}\;{\mathbf{return}}\;\Varid{x}\mapsto\Varid{c}_{\Varid{r}}\mathbin{:}\Conid{M}\;\Conid{A}\hspace{0.1em}!\hspace{0.1em}\langle\Conid{E}_{2}\rangle} }

    \inferrule*[right=SD-Empty]
    {
    }
    { \ensuremath{\Gamma\;\mathbin{\vdash_{\Conid{M}}}\; \cdot\mathbin{:}\Conid{M}\;\Conid{A}\hspace{0.1em}!\hspace{0.1em}\langle\Conid{E}\rangle} }

    \inferrule*[right=SD-OprOp]
    { \ensuremath{\Gamma\;\mathbin{\vdash_{\Conid{M}}}\;\Varid{oprs}\mathbin{:}\Conid{M}\;\Conid{A}_{1}\hspace{0.1em}!\hspace{0.1em}\langle\Conid{E}_{1}\rangle}
      \\ \ensuremath{(\ell^{\textsf{op}}\mathbin{:}{A_{\textsf{op}}}\rightarrowtriangle{B_{\textsf{op}}})\;\in\;\Sigma}
      \\\\ \ensuremath{\Gamma,\Varid{x}\mathbin{:}{A_{\textsf{op}}},\Varid{k}\mathbin{:}{B_{\textsf{op}}}\to \Conid{M}\;\Conid{A}_{1}\hspace{0.1em}!\hspace{0.1em}\langle\Conid{E}_{1}\rangle\vdash_{\mathsf{SD}}\Varid{c}\mathbin{:}\Conid{M}\;\Conid{A}_{2}\hspace{0.1em}!\hspace{0.1em}\langle\Conid{E}_{2}\rangle}
      \\ \ensuremath{\Conid{M}\;\Conid{A}_{1}\hspace{0.1em}!\hspace{0.1em}\langle\Conid{E}_{1}\rangle\;\equiv\;\Conid{M}\;\Conid{A}_{2}\hspace{0.1em}!\hspace{0.1em}\langle\Conid{E}_{2}\rangle}
    }
    { \ensuremath{\Gamma\;\mathbin{\vdash_{\Conid{M}}}\;{\mathbf{op}}\;\ell^{\textsf{op}}\;\Varid{x}\;\Varid{k}\mapsto\Varid{c},\Varid{oprs}\mathbin{:}\Conid{M}\;\Conid{A}_{2}\hspace{0.1em}!\hspace{0.1em}\langle\Conid{E}_{2}\rangle} }

    \inferrule*[right=SD-OprSc]
    { \ensuremath{\Gamma\;\mathbin{\vdash_{\Conid{M}}}\;\Varid{oprs}\mathbin{:}\Conid{M}\;\Conid{A}_{1}\hspace{0.1em}!\hspace{0.1em}\langle\Conid{E}_{1}\rangle}
      \\ \ensuremath{(\ell^{\textsf{sc}}\mathbin{:}{A_{\textsf{sc}}}\rightarrowtriangle{B_{\textsf{sc}}})\;\in\;\Sigma}
      \\ \ensuremath{\beta\;\notin\;\Gamma}
      \\ \ensuremath{\Gamma,\beta,\Varid{x}\mathbin{:}{A_{\textsf{sc}}},\Varid{p}\mathbin{:}{B_{\textsf{sc}}}\to \Conid{M}\;\beta\hspace{0.1em}!\hspace{0.1em}\langle\Conid{E}_{1}\rangle,\Varid{k}\mathbin{:}\beta\to \Conid{M}\;\Conid{A}_{1}\hspace{0.1em}!\hspace{0.1em}\langle\Conid{E}_{1}\rangle\vdash_{\mathsf{SD}}\Varid{c}\mathbin{:}\Conid{M}\;\Conid{A}_{2}\hspace{0.1em}!\hspace{0.1em}\langle\Conid{E}_{2}\rangle}
      \\ \ensuremath{\Conid{M}\;\Conid{A}_{1}\hspace{0.1em}!\hspace{0.1em}\langle\Conid{E}_{1}\rangle\;\equiv\;\Conid{M}\;\Conid{A}_{2}\hspace{0.1em}!\hspace{0.1em}\langle\Conid{E}_{2}\rangle}
    }
    { \ensuremath{\Gamma\;\mathbin{\vdash_{\Conid{M}}}\;{\mathbf{sc}}\;\ell^{\textsf{sc}}\;\Varid{x}\;\Varid{p}\;\Varid{k}\mapsto\Varid{c},\Varid{oprs}\mathbin{:}\Conid{M}\;\Conid{A}_{2}\hspace{0.1em}!\hspace{0.1em}\langle\Conid{E}_{2}\rangle} }

  \inferrule*[right=SD-Fwd]
  { \ensuremath{\alpha,\beta,\gamma,\delta\;\notin\;\Gamma}
    \\ \ensuremath{\Conid{A}_{\Varid{p}}\mathrel{=}\alpha\to \Conid{M}\;\beta\hspace{0.1em}!\hspace{0.1em}\langle\Conid{E}_{1}\rangle}
    \\ \ensuremath{\Conid{A}_{\Varid{p}}'\mathrel{=}\alpha\to \gamma\hspace{0.1em}!\hspace{0.1em}\langle\Conid{E}_{1}\rangle}
    \\ \ensuremath{\Conid{A}_{\Varid{k}}\mathrel{=}\beta\to \Conid{M}\;\Conid{A}_{1}\hspace{0.1em}!\hspace{0.1em}\langle\Conid{E}_{1}\rangle}
    \\ \ensuremath{\Conid{A}_{\Varid{k}}'\mathrel{=}\gamma\to \delta\hspace{0.1em}!\hspace{0.1em}\langle\Conid{E}_{1}\rangle}
    \\ \ensuremath{\Gamma,\alpha,\beta,\Varid{p}\mathbin{:}\Conid{A}_{\Varid{p}},\Varid{k}\mathbin{:}\Conid{A}_{\Varid{k}},\Varid{f}\mathbin{:}\forall\;\gamma\;\delta\, .\,(\Conid{A}_{\Varid{p}}',\Conid{A}_{\Varid{k}}')\to \delta\hspace{0.1em}!\hspace{0.1em}\langle\Conid{E}_{1}\rangle\vdash_{\mathsf{SD}}\Varid{c}_{\Varid{f}}\mathbin{:}\Conid{M}\;\Conid{A}_{2}\hspace{0.1em}!\hspace{0.1em}\langle\Conid{E}_{2}\rangle}
    \\ \ensuremath{\Conid{M}\;\Conid{A}_{1}\hspace{0.1em}!\hspace{0.1em}\langle\Conid{E}_{1}\rangle\;\equiv\;\Conid{M}\;\Conid{A}_{2}\hspace{0.1em}!\hspace{0.1em}\langle\Conid{E}_{2}\rangle}
    \\ \ensuremath{\Gamma\vdash_{\mathsf{SD}}\Conid{A}_{2}}
  }
  { \ensuremath{\Gamma\;\mathbin{\vdash_{\Conid{M}}}\;{\mathbf{fwd}}\;\Varid{f}\;\Varid{p}\;\Varid{k}\mapsto\Varid{c}_{\Varid{f}}\mathbin{:}\Conid{M}\;\Conid{A}_{2}\hspace{0.1em}!\hspace{0.1em}\langle\Conid{E}_{2}\rangle} }

  \end{mathpar}
  \caption{Syntax-directed handler typing.}
  \label{fig:syntax-directed-handler-typing}
\end{figure}

The syntax-directed system is obtained by incorporating the non-syntax-directed
rules into the syntax-directed-ones where needed.
In particular, we inline the non-syntax-directed rules for equivalence
(\textsc{T-EqV} and \textsc{T-EqC}) into the syntax-directed rules that mention the same
type or row twice in their assumptions (e.g., \textsc{SD-App}, \textsc{SD-Do}).
Similarly, we inline the rules \textsc{T-Inst}, \textsc{T-InstEff},
\textsc{T-Gen} and \textsc{T-GenEff} for instantiating and generalizing type and
row variables. The generalization is incorporated into the rule for let-bindings
(\textsc{T-Let}). Instantiation is incorporated into the variable rule
(\textsc{T-Var}) using \ensuremath{\sigma\leq \Conid{A}} defined in \Cref{fig:sigma-instantiation}.

\begin{figure}[tp]
\ruleform{\ensuremath{\sigma\leq \Conid{A}}} \ensuremath{\sigma}-instantiation \hspace*{\fill}
\begin{mathpar}
  \inferrule*[right=\ensuremath{\sigma}-Inst-Base]
  { }
  { \ensuremath{\Conid{A}\leq \Conid{A}} }

  \inferrule*[right=\ensuremath{\sigma}-Inst-\ensuremath{\alpha}]
  {    [\ensuremath{\Conid{B}\mathbin{/}\alpha}] \ensuremath{\sigma\leq \Conid{A}}}
  { \ensuremath{\forall\;\alpha\, .\,\sigma\leq \Conid{A}} }

  \inferrule*[right=\ensuremath{\sigma}-Inst-\ensuremath{\mu}]
  {    \ensuremath{[\mskip1.5mu \Conid{E}\mathbin{/}\mu\mskip1.5mu]\;\sigma\leq \Conid{A}}}
  { \ensuremath{\forall\;\mu\, .\,\sigma\leq \Conid{A}} }
\end{mathpar}
\caption{\ensuremath{\sigma}-instantiation.}
\label{fig:sigma-instantiation}
\end{figure}

Instantiation is also incorporated into the handler rule: we implicitly
instantiate \ensuremath{\alpha} with an arbitrary type \ensuremath{\Conid{A}}, which results in a
monomorphically typed handler.
However, since \textsc{SD-Handler} insists on sufficiency polymorphic handler
clauses, we can still handle scoped effects by polymorphic recursion.

\Cref{fig:handler-gen-and-inst} displays declarative and syntax-directed typing
derivations for both inline handler application (\ensuremath{{\mathbf{with}}\;\Varid{h}\;{\mathbf{handle}}\;\Varid{c}}) as well as let-bound
handlers.
As can be seen in the first derivation, in the case of inline handler
application, the declarative system derives a polymorphically typed handler,
which is instantiated.
The syntax-directed system essentially combines these steps, as can be seen in
the second derivation.
In the case of a let-bound handler, the declarative system keeps the polymorphic
handler type as-is (third derivation).
The syntax-directed system however instantiates and then immediately generalizes
handlers, as can be seen in the fourth derivation.

The other rules of the declarative system are syntax-directed and remain unchanged.

\begin{figure}[h]
  \begin{mathpar}
    \inferrule*[right=T-Hand]
    { \inferrule*[right=T-Inst]
      { \inferrule*[right=T-Handler]{
          \ensuremath{\Gamma,\alpha\vdash_{\mathsf{SD}}\Varid{oprs}\mathbin{:}\Conid{M}\;\alpha\hspace{0.1em}!\hspace{0.1em}\langle\Conid{E}\rangle} \\
        }{\ensuremath{\Gamma\vdash_{\mathsf{SD}}\Varid{h}\mathbin{:}\forall\;\alpha\, .\,\alpha\hspace{0.1em}!\hspace{0.1em}\langle\Conid{F}\rangle\Rightarrow \Conid{M}\;\alpha\hspace{0.1em}!\hspace{0.1em}\langle\Conid{E}\rangle}}
        \ensuremath{\Gamma\vdash_{\mathsf{SD}}\Conid{A}}
      }
      { \ensuremath{\Gamma\vdash_{\mathsf{SD}}\Varid{h}\mathbin{:}\Conid{A}\hspace{0.1em}!\hspace{0.1em}\langle\Conid{F}\rangle\Rightarrow \Conid{M}\;\Conid{A}\hspace{0.1em}!\hspace{0.1em}\langle\Conid{E}\rangle} }
      \ensuremath{\Gamma\vdash_{\mathsf{SD}}\Varid{c}\mathbin{:}\Conid{A}\hspace{0.1em}!\hspace{0.1em}\langle\Conid{F}\rangle}
    }
    { \ensuremath{\Gamma\vdash_{\mathsf{SD}}{\mathbf{with}}\;\Varid{h}\;{\mathbf{handle}}\;\Varid{c}\mathbin{:}\Conid{M}\;\Conid{A}\hspace{0.1em}!\hspace{0.1em}\langle\Conid{E}\rangle} }

    \inferrule*[right=SD-Hand]
    { \inferrule*[right=SD-Handler]{
        \ensuremath{\Gamma,\alpha\;\mathbin{\vdash_{\Conid{M}}}\;\Varid{oprs}\mathbin{:}\Conid{M}\;\alpha\hspace{0.1em}!\hspace{0.1em}\langle\Conid{E}\rangle} \\
        \ensuremath{\Gamma\vdash_{\mathsf{SD}}\Conid{A}}
      }{\ensuremath{\Gamma\vdash_{\mathsf{SD}}{h}_{\Conid{M}}\mathbin{:}\Conid{A}\hspace{0.1em}!\hspace{0.1em}\langle\Conid{F}\rangle\Rightarrow \Conid{M}\;\Conid{A}\hspace{0.1em}!\hspace{0.1em}\langle\Conid{E}\rangle}}
      \\ \ensuremath{\Gamma\vdash_{\mathsf{SD}}\Varid{c}\mathbin{:}\Conid{A}\hspace{0.1em}!\hspace{0.1em}\langle\Conid{F}\rangle}
    }
    { \ensuremath{\Gamma\vdash_{\mathsf{SD}}{\mathbf{with}}\;{h}_{\Conid{M}}\;{\mathbf{handle}}\;\Varid{c}\mathbin{:}\Conid{M}\;\Conid{A}\hspace{0.1em}!\hspace{0.1em}\langle\Conid{E}\rangle} }

    \inferrule*[right=T-Let]
    { \inferrule*[right=T-Handler]{
        \ensuremath{\Gamma,\alpha\vdash_{\mathsf{SD}}\Varid{oprs}\mathbin{:}\Conid{M}\;\alpha\hspace{0.1em}!\hspace{0.1em}\langle\Conid{E}\rangle} \\
      }{\ensuremath{\Gamma\vdash_{\mathsf{SD}}\Varid{h}\mathbin{:}\forall\;\alpha\, .\,\alpha\hspace{0.1em}!\hspace{0.1em}\langle\Conid{F}\rangle\Rightarrow \Conid{M}\;\alpha\hspace{0.1em}!\hspace{0.1em}\langle\Conid{E}\rangle}}
      \ensuremath{\Gamma,\Varid{x}\mathbin{:}\forall\;\alpha\, .\,\alpha\hspace{0.1em}!\hspace{0.1em}\langle\Conid{F}\rangle\Rightarrow \Conid{M}\;\alpha\hspace{0.1em}!\hspace{0.1em}\langle\Conid{E}\rangle\vdash_{\mathsf{SD}}\Varid{c}\mathbin{:} \underline{C} }
    }
    { \ensuremath{\Gamma\vdash_{\mathsf{SD}}\mathbf{let}\;\Varid{x}\mathrel{=}\Varid{h}\;\mathbf{in}\;\Varid{c}\mathbin{:} \underline{C} } }

    \inferrule*[right=SD-Let]
    { \inferrule*[right=SD-Handler]{
        \ensuremath{\Gamma,\alpha,\beta\;\mathbin{\vdash_{\Conid{M}}}\;\Varid{oprs}\mathbin{:}\Conid{M}\;\beta\hspace{0.1em}!\hspace{0.1em}\langle\Conid{E}\rangle} \\\\
        \ensuremath{\Gamma,\alpha\vdash_{\mathsf{SD}}\alpha} \\
      }{\ensuremath{\Gamma\vdash_{\mathsf{SD}}{h}_{\Conid{M}}\mathbin{:}\alpha\hspace{0.1em}!\hspace{0.1em}\langle\Conid{F}\rangle\Rightarrow \Conid{M}\;\alpha\hspace{0.1em}!\hspace{0.1em}\langle\Conid{E}\rangle}}
      \ensuremath{\Gamma,\Varid{x}\mathbin{:}\forall\;\alpha\, .\,\alpha\hspace{0.1em}!\hspace{0.1em}\langle\Conid{F}\rangle\Rightarrow \Conid{M}\;\alpha\hspace{0.1em}!\hspace{0.1em}\langle\Conid{E}\rangle\vdash_{\mathsf{SD}}\Varid{c}\mathbin{:} \underline{C} }
    }
    { \ensuremath{\Gamma\vdash_{\mathsf{SD}}\mathbf{let}\;\Varid{x}\mathrel{=}{h}_{\Conid{M}}\;\mathbf{in}\;\Varid{c}\mathbin{:} \underline{C} } }
    \end{mathpar}
  \caption{Handler generalisation and instantiation}
  \label{fig:handler-gen-and-inst}
\end{figure}
\clearpage
\section{Metatheory}
\label{sec:metatheory-appendix}

\subsection{Lemmas}

\begin{lem}[Canonical forms]
  \label{lem:canonicalforms}
  \begin{itemize}
  \item[]
  \item If \ensuremath{ \cdot\vdash_{\mathsf{SD}}\Varid{v}\mathbin{:}\Conid{A}\to  \underline{C} } then \ensuremath{\Varid{v}} is of shape \ensuremath{\boldsymbol{\lambda}\Varid{x}\, .\,\Varid{c}}.
  \item If \ensuremath{ \cdot\vdash_{\mathsf{SD}}\Varid{v}\mathbin{:} \underline{C} \Rightarrow  \underline{D} } then \ensuremath{\Varid{v}} is of shape \ensuremath{\Varid{h}}.
  \end{itemize}
\end{lem}

\begin{lem}[Generalisation-equivalence]
  \label{lem:geneq}
  If \ensuremath{\sigma_{1}\leq \Conid{A}_{1}} and \ensuremath{\sigma_{1}\;\equiv\;\sigma_{2}}, then there exists a \ensuremath{\Conid{A}_{2}} such
  that \ensuremath{\Conid{A}_{1}\;\equiv\;\Conid{A}_{2}} and \ensuremath{\sigma_{2}\leq \Conid{A}_{2}}.
\end{lem}

\begin{lem}[Generalisation-instantiation]
  \label{lem:geninst}
  If \ensuremath{\Gamma,\overline{\alpha},\overline{\mu}\vdash_{\mathsf{SD}}\Varid{v}\mathbin{:}\Conid{A}} and \ensuremath{\overline{\forall\;\alpha}\;\overline{\forall\;\mu}\, .\,\Conid{A}\leq \Conid{B}}, then \ensuremath{\Gamma\vdash_{\mathsf{SD}}\Varid{v}\mathbin{:}\Conid{B}}.
\end{lem}

\begin{lem}[Preservation of types under term substitution]
  \label{lem:termsubstitutioneliminationgeneq}
  Given \ensuremath{\Gamma_{1},\overline{\alpha},\overline{\mu}\vdash_{\mathsf{SD}}\Varid{v}\mathbin{:}\Conid{A}_{1}} and \ensuremath{\Conid{A}_{1}\;\equiv\;\Conid{A}_{2}} we have that:
  \begin{itemize}
  \item If \ensuremath{\Gamma_{1},\Varid{x}\mathbin{:}\overline{\forall\;\alpha}\;\overline{\forall\;\mu}\, .\,\Conid{A}_{2},\Gamma_{2}\vdash_{\mathsf{SD}}\Varid{c}\mathbin{:} \underline{C} _{1}}, then there exists a \ensuremath{ \underline{C} _{2}} such that \ensuremath{ \underline{C} _{1}\;\equiv\; \underline{C} _{2}} and \ensuremath{\Gamma_{1},\Gamma_{2}\vdash_{\mathsf{SD}}[\mskip1.5mu \Varid{v}\mathbin{/}\Varid{x}\mskip1.5mu]\;\Varid{c}\mathbin{:} \underline{C} _{2}}.
  \item If \ensuremath{\Gamma_{1},\Varid{x}\mathbin{:}\overline{\forall\;\alpha}\;\overline{\forall\;\mu}\, .\,\Conid{A}_{2},\Gamma_{2}\vdash_{\mathsf{SD}}\Varid{v}\mathbin{:}\Conid{B}_{1}},
    then there exists a \ensuremath{\Conid{B}_{2}} such that \ensuremath{\Conid{B}_{1}\;\equiv\;\Conid{B}_{2}} and \ensuremath{\Gamma_{1},\Gamma_{2}\vdash_{\mathsf{SD}}[\mskip1.5mu \Varid{v}\mathbin{/}\Varid{x}\mskip1.5mu]\;\Varid{v}\mathbin{:}\Conid{B}_{2}}.
  \end{itemize}
\end{lem}
\begin{proof}
  By mutual induction on the typing derivations. The only interesting case,
  \textsc{SD-Var}, requires us to show that, given \ensuremath{\Gamma_{1},\Varid{x}\mathbin{:}\overline{\forall\;\alpha}\;\overline{\forall\;\mu}\, .\,\Conid{A}_{2},\Gamma_{2}\vdash_{\mathsf{SD}}\Varid{y}\mathbin{:}\Conid{B}_{1}}, there exists a \ensuremath{\Conid{B}_{2}} such that
  \ensuremath{\Conid{B}_{1}\;\equiv\;\Conid{B}_{2}} and \ensuremath{\Gamma_{1},\Gamma_{2}\vdash_{\mathsf{SD}}[\mskip1.5mu \Varid{v}\mathbin{/}\Varid{x}\mskip1.5mu]\;\Varid{y}\mathbin{:}\Conid{B}_{2}}. If \ensuremath{\Varid{x}\;\not =\;\Varid{y}}, it is
  trivial. If \ensuremath{\Varid{x}\mathrel{=}\Varid{y}}, then \ensuremath{\overline{\forall\;\alpha}\;\overline{\forall\;\mu}\, .\,\Conid{A}_{2}\leq \Conid{B}_{1}}, which
  means by \Cref{lem:geneq} there exists a \ensuremath{\Conid{B}_{2}} such that \ensuremath{\Conid{B}_{1}\;\equiv\;\Conid{B}_{2}} and
  \ensuremath{\overline{\forall\;\alpha}\;\overline{\forall\;\mu}\, .\,\Conid{A}_{1}\leq \Conid{B}_{2}}, which means the result follows
  from \Cref{lem:geninst}.
\end{proof}

\begin{lem}[Preservation of types under type substitution]
  \label{lem:typesubstitutionelimination}
  If \ensuremath{\Gamma_{1},\alpha,\Gamma_{2}\vdash_{\mathsf{SD}}\Varid{c}\mathbin{:} \underline{C} } and \ensuremath{\Gamma_{1}\vdash_{\mathsf{SD}}\Conid{B}}, then \ensuremath{\Gamma_{1},[\mskip1.5mu \Conid{B}\mathbin{/}\alpha\mskip1.5mu]\;\Gamma_{2}\vdash_{\mathsf{SD}}\Varid{c}\mathbin{:}[\mskip1.5mu \Conid{B}\mathbin{/}\alpha\mskip1.5mu]\; \underline{C} }.
\end{lem}

\begin{lem}[Unused binding insertion]
  \label{lem:termweakening}
  If \ensuremath{\Gamma_{1},\Gamma_{2}\vdash_{\mathsf{SD}}\Varid{c}\mathbin{:} \underline{C} } and \ensuremath{\Varid{x}\;\notin\;\Varid{c}} then \ensuremath{\Gamma_{1},\Varid{x}\mathbin{:}\Conid{A},\Gamma_{2}\vdash_{\mathsf{SD}}\Varid{c}\mathbin{:} \underline{C} }.
\end{lem}

\begin{lem}[Handlers are polymorphic]
  \label{lem:handlersarepolymorphic}
  If \ensuremath{\Gamma\vdash_{\mathsf{SD}}\Varid{h}\mathbin{:}\Conid{A}\hspace{0.1em}!\hspace{0.1em}\langle\Conid{F}\rangle\Rightarrow \Conid{M}\;\Conid{A}\hspace{0.1em}!\hspace{0.1em}\langle\Conid{E}\rangle} and \ensuremath{\Gamma\vdash_{\mathsf{SD}}\Conid{B}}, then \ensuremath{\Gamma\vdash_{\mathsf{SD}}\Varid{h}\mathbin{:}\Conid{B}\hspace{0.1em}!\hspace{0.1em}\langle\Conid{F}\rangle\Rightarrow \Conid{M}\;\Conid{B}\hspace{0.1em}!\hspace{0.1em}\langle\Conid{E}\rangle}.
\end{lem}

\begin{lem}[Op membership]
  \label{lem:opmembership}
  If \ensuremath{\Gamma\vdash_{\mathsf{SD}}\Varid{oprs}\mathbin{:} \underline{C} } and \ensuremath{{\mathbf{op}}\;\ell^{\textsf{op}}\;\Varid{x}\;\Varid{k}\mapsto\Varid{c}\;\in\;\Varid{oprs}}, then there exists
  \ensuremath{\Varid{oprs}_{1}} and \ensuremath{\Varid{oprs}_{2}} such that \ensuremath{\Varid{oprs}\mathrel{=}\Varid{oprs}_{1},{\mathbf{op}}\;\ell^{\textsf{op}}\;\Varid{k}\vdash_{\mathsf{SD}}\Varid{c},\Varid{oprs}_{2}} and \ensuremath{\Gamma\vdash_{\mathsf{SD}}{\mathbf{op}}\;\ell^{\textsf{op}}\;\Varid{x}\;\Varid{k}\mapsto\Varid{c},\Varid{oprs}_{2}\mathbin{:} \underline{C} }.
\end{lem}

\begin{lem}[Sc membership]
  \label{lem:scmembership}
  If \ensuremath{\Gamma\vdash_{\mathsf{SD}}\Varid{oprs}\mathbin{:} \underline{C} } and \ensuremath{({\mathbf{sc}}\;\ell^{\textsf{sc}}\;\Varid{x}\;\Varid{p}\;\Varid{k}\mapsto\Varid{c})\;\in\;\Varid{h}}, then there exists
  \ensuremath{\Varid{oprs}_{1}} and \ensuremath{\Varid{oprs}_{2}} such that \ensuremath{\Varid{oprs}\mathrel{=}\Varid{oprs}_{1},{\mathbf{sc}}\;\ell^{\textsf{sc}}\;\Varid{x}\;\Varid{p}\;\Varid{k}\mapsto\Varid{c},\Varid{oprs}_{2}} and
  \ensuremath{\Gamma\vdash_{\mathsf{SD}}{\mathbf{sc}}\;\ell^{\textsf{sc}}\;\Varid{x}\;\Varid{p}\;\Varid{k}\mapsto\Varid{c},\Varid{oprs}_{2}\mathbin{:} \underline{C} }.
\end{lem}

\begin{lem}[Syntax-directed to Non-syntax-directed value typing]
  If \ensuremath{\Gamma\vdash\Varid{v}\mathbin{:}\overline{\forall\;\alpha}\, .\,\overline{\forall\;\mu}\, .\,\Conid{A}} then \ensuremath{\Gamma,\overline{\alpha},\overline{\mu}\vdash_{\mathsf{SD}}\Varid{v}\mathbin{:}\Conid{A}}.
\end{lem}

\begin{lem}[Non-syntax-directed iff Syntax-directed]
  \label{lem:nsdiffsd}
  \ensuremath{\Gamma\vdash\Varid{c}\mathbin{:} \underline{C} \;\iff\;\Gamma\vdash_{\mathsf{SD}}\Varid{c}\mathbin{:} \underline{C} }.
\end{lem}

\subsection{Subject reduction}
\subjectreduction*
\begin{proof}
  By \Cref{lem:nsdiffsd} (above) and \Cref{thm:sdsubjectreduction} (below).
\end{proof}

\begin{thm}[Syntax-directed Subject Reduction]
  \label{thm:sdsubjectreduction}
  If \ensuremath{\Gamma\vdash_{\mathsf{SD}}\Varid{c}\mathbin{:} \underline{C} } and \ensuremath{\Varid{c}\leadsto\Varid{c'}}, then there exists a \ensuremath{ \underline{C} '} such that \ensuremath{ \underline{C} \;\equiv\; \underline{C} '} and \ensuremath{\Gamma\vdash_{\mathsf{SD}}\Varid{c'}\mathbin{:} \underline{C} '}.
\end{thm}

\begin{proof}
  Assume, without loss of generality, that \ensuremath{ \underline{C} \mathrel{=}\Conid{B}\hspace{0.1em}!\hspace{0.1em}\langle\Conid{F}\rangle} for some \ensuremath{\Conid{B}}, \ensuremath{\Conid{F}}.
  Proceed by induction on the derivation \ensuremath{\Varid{c}\leadsto\Varid{c'}}.
  \begin{itemize}
  \item \textsc{E-AppAbs}: Inversion on \ensuremath{\Gamma\vdash_{\mathsf{SD}}(\boldsymbol{\lambda}\Varid{x}\, .\,\Varid{c})\;\Varid{v}\mathbin{:}\Conid{B}\hspace{0.1em}!\hspace{0.1em}\langle\Conid{F}\rangle}
    (\textsc{SD-App}) gives \ensuremath{\Gamma\vdash_{\mathsf{SD}}\boldsymbol{\lambda}\Varid{x}\, .\,\Varid{c}\mathbin{:}\Conid{A}_{1}\to \Conid{B}\hspace{0.1em}!\hspace{0.1em}\langle\Conid{F}\rangle} (1), \ensuremath{\Gamma\vdash_{\mathsf{SD}}\Varid{v}\mathbin{:}\Conid{A}_{2}} (2), and \ensuremath{\Conid{A}_{1}\;\equiv\;\Conid{A}_{2}} (3). Inversion on fact 1 (\textsc{SD-Abs})
    gives \ensuremath{\Gamma,\Varid{x}\mathbin{:}\Conid{A}_{1}\vdash_{\mathsf{SD}}\Varid{c}\mathbin{:}\Conid{B}\hspace{0.1em}!\hspace{0.1em}\langle\Conid{F}\rangle} (4), which means the goal follows
    from facts 2 and 4 and \Cref{lem:termsubstitutioneliminationgeneq}.
  \item \textsc{E-Let}: Inversion on \ensuremath{\Gamma\vdash_{\mathsf{SD}}\mathbf{let}\;\Varid{x}\mathrel{=}\Varid{v}\;\mathbf{in}\;\Varid{c}\mathbin{:}\Conid{B}\hspace{0.1em}!\hspace{0.1em}\langle\Conid{F}\rangle}
    (\textsc{SD-Let}) gives \ensuremath{\Gamma\vdash_{\mathsf{SD}}\Varid{v}\mathbin{:}\Conid{A}} (1), \ensuremath{\sigma\mathrel{=}\mathit{gen}\;(\Conid{A},\Gamma)} (2),
    and \ensuremath{\Gamma,\Varid{x}\mathbin{:}\sigma\vdash_{\mathsf{SD}}\Varid{c}\mathbin{:}\Conid{B}\hspace{0.1em}!\hspace{0.1em}\langle\Conid{F}\rangle} (3), which means the goal follows
    from facts 1 and 3 and \Cref{lem:termsubstitutioneliminationgeneq}.
  \item \textsc{E-Do}: Follows from the IH.
  \item \textsc{E-DoRet}: Inversion on \ensuremath{\Gamma\vdash_{\mathsf{SD}}\mathbf{do}\;\Varid{x}\leftarrow {\mathbf{return}}\;\Varid{v}\;\mathbf{in}\;\Varid{c}\mathbin{:}\Conid{B}\hspace{0.1em}!\hspace{0.1em}\langle\Conid{F}_{2}\rangle} (\textsc{SD-Do}) gives \ensuremath{\Gamma\vdash_{\mathsf{SD}}{\mathbf{return}}\;\Varid{v}\mathbin{:}\Conid{A}\hspace{0.1em}!\hspace{0.1em}\langle\Conid{F}_{1}\rangle} (1)
    and \ensuremath{\Gamma,\Varid{x}\mathbin{:}\Conid{A}\vdash_{\mathsf{SD}}\Varid{c}\mathbin{:}\Conid{B}\hspace{0.1em}!\hspace{0.1em}\langle\Conid{F}_{2}\rangle} (2). Inversion on (1)
    (\textsc{SD-Ret}) gives \ensuremath{\Gamma\vdash_{\mathsf{SD}}\Varid{v}\mathbin{:}\Conid{A}} (3). The case follows from facts 2
    and 4 and \Cref{lem:termsubstitutioneliminationgeneq}.
  \item \textsc{E-DoOp}: Similar to \textsc{E-DoSc}. By inversion on \ensuremath{\Gamma\vdash_{\mathsf{SD}}\mathbf{do}\;\Varid{x}\leftarrow {\mathbf{op}}\;\ell^{\textsf{op}}\;\Varid{v}\;(\Varid{y}\, .\,\Varid{c}_{1})\;\mathbf{in}\;\Varid{c}_{2}\mathbin{:}\Conid{B}\hspace{0.1em}!\hspace{0.1em}\langle\Conid{F}_{2}\rangle} (\textsc{SD-Do}) we have that
    \ensuremath{\Gamma\vdash_{\mathsf{SD}}{\mathbf{op}}\;\ell^{\textsf{op}}\;\Varid{v}\;(\Varid{y}\, .\,\Varid{c}_{1})\mathbin{:}\Conid{A}\hspace{0.1em}!\hspace{0.1em}\langle\Conid{F}_{1}\rangle} (1), \ensuremath{\Gamma,\Varid{x}\mathbin{:}\Conid{A}\vdash_{\mathsf{SD}}\Varid{c}_{2}\mathbin{:}\Conid{B}\hspace{0.1em}!\hspace{0.1em}\langle\Conid{F}_{2}\rangle} (2), and \ensuremath{\Conid{F}_{1}\;\equiv\;\Conid{F}_{2}} (3). From inversion on fact 1
    (\textsc{SD-Op}) it follows that \ensuremath{\ell^{\textsf{op}}\mathbin{:}{A_{\textsf{op}}}\rightarrowtriangle{B_{\textsf{op}}}\;\in\;\Sigma} (4), \ensuremath{\Gamma\vdash_{\mathsf{SD}}\Varid{v}\mathbin{:}\Conid{A}_{1}} (5), \ensuremath{{A_{\textsf{op}}}\;\equiv\;\Conid{A}_{1}} (6), \ensuremath{\Gamma,\Varid{y}\mathbin{:}{B_{\textsf{op}}}\vdash_{\mathsf{SD}}\Varid{c}_{1}\mathbin{:}\Conid{A}\hspace{0.1em}!\hspace{0.1em}\langle\Conid{F}_{1}\rangle}
    (7), and \ensuremath{\ell^{\textsf{op}}\;\in\;\Conid{F}_{1}} (8). \Cref{lem:termweakening} on (2) gives us \ensuremath{\Gamma,\Varid{y}\mathbin{:}{B_{\textsf{op}}},\Varid{x}\mathbin{:}\Conid{A}\vdash_{\mathsf{SD}}\Varid{c}_{2}\mathbin{:}\Conid{B}\hspace{0.1em}!\hspace{0.1em}\langle\Conid{F}_{2}\rangle} (9). Facts 3, 7 and 9 and rule
    \textsc{SD-Do} give us \ensuremath{\Gamma,\Varid{y}\mathbin{:}{B_{\textsf{op}}}\vdash_{\mathsf{SD}}\mathbf{do}\;\Varid{x}\leftarrow \Varid{c}_{1}\;\mathbf{in}\;\Varid{c}_{2}\mathbin{:}\Conid{B}\hspace{0.1em}!\hspace{0.1em}\langle\Conid{F}_{2}\rangle}
    (10). Our goal then follows from facts 4, 5, 6, 8, and 10 and rule
    \textsc{SD-Op}.
  \item \textsc{E-DoSc}: Similar to \textsc{E-DoOp}. By inversion on \ensuremath{\Gamma\vdash_{\mathsf{SD}}\mathbf{do}\;\Varid{x}\leftarrow {\mathbf{sc}}\;\ell^{\textsf{sc}}\;\Varid{v}\;(\Varid{y}\, .\,\Varid{c}_{1})\;(\Varid{z}\, .\,\Varid{c}_{2})\;\mathbf{in}\;\Varid{c}_{3}\mathbin{:}\Conid{B}\hspace{0.1em}!\hspace{0.1em}\langle\Conid{F}_{3}\rangle} (\textsc{SD-Do}) we
    have that \ensuremath{\Gamma\vdash_{\mathsf{SD}}{\mathbf{sc}}\;\ell^{\textsf{sc}}\;\Varid{v}\;(\Varid{y}\, .\,\Varid{c}_{1})\;(\Varid{z}\, .\,\Varid{c}_{2})\mathbin{:}\Conid{A}\hspace{0.1em}!\hspace{0.1em}\langle\Conid{F}_{2}\rangle} (1), \ensuremath{\Gamma,\Varid{x}\mathbin{:}\Conid{A}\vdash_{\mathsf{SD}}\Varid{c}_{3}\mathbin{:}\Conid{B}\hspace{0.1em}!\hspace{0.1em}\langle\Conid{F}_{3}\rangle} (2), and \ensuremath{\Conid{F}_{2}\;\equiv\;\Conid{F}_{3}} (2.1). From inversion on fact
    1 (\textsc{SD-Sc}) it follows that \ensuremath{\ell^{\textsf{sc}}\mathbin{:}{A_{\textsf{sc}}}\rightarrowtriangle{B_{\textsf{sc}}}\;\in\;\Sigma} (3),
    \ensuremath{\Gamma\vdash_{\mathsf{SD}}\Varid{v}\mathbin{:}\Conid{A}_{1}} (4), \ensuremath{{A_{\textsf{sc}}}\;\equiv\;\Conid{A}_{1}} (5), \ensuremath{\Gamma,\Varid{y}\mathbin{:}{B_{\textsf{sc}}}\vdash_{\mathsf{SD}}\Varid{c}_{1}\mathbin{:}\Conid{B'}\hspace{0.1em}!\hspace{0.1em}\langle\Conid{F}_{1}\rangle} (6), \ensuremath{\Gamma,\Varid{z}\mathbin{:}\Conid{B'}\vdash_{\mathsf{SD}}\Varid{c}_{2}\mathbin{:}\Conid{A}\hspace{0.1em}!\hspace{0.1em}\langle\Conid{F}_{2}\rangle} (7), \ensuremath{\Conid{F}_{1}\;\equiv\;\Conid{F}_{2}}
    (8), and \ensuremath{\ell^{\textsf{sc}}\;\in\;\Conid{F}_{2}} (9). \Cref{lem:termweakening} on (2) gives us \ensuremath{\Gamma,\Varid{z}\mathbin{:}\Conid{B'},\Varid{x}\mathbin{:}\Conid{A}\vdash_{\mathsf{SD}}\Varid{c}_{3}\mathbin{:}\Conid{B}\hspace{0.1em}!\hspace{0.1em}\langle\Conid{F}_{3}\rangle} (10), which means facts 2.1, 7 and 10
    and rule \textsc{SD-Do} give us \ensuremath{\Gamma,\Varid{z}\mathbin{:}\Conid{B'}\vdash_{\mathsf{SD}}\mathbf{do}\;\Varid{x}\leftarrow \Varid{c}_{2}\;\mathbf{in}\;\Varid{c}_{3}\mathbin{:}\Conid{B}\hspace{0.1em}!\hspace{0.1em}\langle\Conid{F}_{3}\rangle} (11). Our goal then follows from facts 3, 4, 5, 6, 8 9, and
    11 and rule \textsc{SD-Sc}.
  \item \textsc{E-Hand}: Follows from the IH.
  \item \textsc{E-HandRet}: By inversion on \ensuremath{\Gamma\vdash_{\mathsf{SD}}{\mathbf{with}}\;\Varid{h}\;{\mathbf{handle}}\;{\mathbf{return}}\;\Varid{v}\mathbin{:}\Conid{B}\hspace{0.1em}!\hspace{0.1em}\langle\Conid{F}_{2}\rangle} (\textsc{SD-Hand}) we have that \ensuremath{\Gamma\vdash_{\mathsf{SD}}\Varid{h}\mathbin{:} \underline{C} _{1}\Rightarrow \Conid{B}\hspace{0.1em}!\hspace{0.1em}\langle\Conid{F}_{2}\rangle} (1), \ensuremath{\Gamma\vdash_{\mathsf{SD}}{\mathbf{return}}\;\Varid{v}\mathbin{:} \underline{C} _{2}} (2), and \ensuremath{ \underline{C} _{1}\;\equiv\; \underline{C} _{2}} (3).
    Inversion on fact 1 (\textsc{SD-Handler}) gives \ensuremath{\Conid{B}\mathrel{=}\Conid{M}\;\Conid{A}_{2}}, \ensuremath{ \underline{C} _{1}\mathrel{=}\Conid{A}_{2}\hspace{0.1em}!\hspace{0.1em}\langle\Conid{E}\rangle}, and \ensuremath{\Gamma,\alpha\;\mathbin{\vdash_{\Conid{M}}}\;{\mathbf{return}}\;\Varid{x}\mapsto\Varid{c}_{\Varid{r}}\mathbin{:}\Conid{M}\;\alpha\hspace{0.1em}!\hspace{0.1em}\langle\Conid{F}_{2}\rangle} (4). Based on fact (3) we get that \ensuremath{ \underline{C} _{2}\mathrel{=}\Conid{A2'}\hspace{0.1em}!\hspace{0.1em}\langle\Conid{E'}\rangle}, \ensuremath{\Conid{A}_{2}\;\equiv\;\Conid{A2'}} (4), and \ensuremath{\Conid{E}\;\equiv\;\Conid{E'}} (5). Inversion on fact 4
    (\textsc{SD-Return}) gives \ensuremath{\Gamma,\alpha,\Varid{x}\mathbin{:}\Conid{A}_{1}\vdash_{\mathsf{SD}}\Varid{c}_{\Varid{r}}\mathbin{:}\Conid{M}\;\alpha\hspace{0.1em}!\hspace{0.1em}\langle\Conid{F}_{2}\rangle} (5) and \ensuremath{\Conid{A}_{1}\;\equiv\;\Conid{A}_{2}} (6). Inversion on fact 2 (\textsc{SD-Ret})
    gives \ensuremath{\Gamma\vdash_{\mathsf{SD}}\Varid{v}\mathbin{:}\Conid{A2'}} (7). From facts 4-8 and
    \Cref{lem:termsubstitutioneliminationgeneq}, we get that \ensuremath{\Gamma,\alpha\vdash_{\mathsf{SD}}[\mskip1.5mu \Varid{v}\mathbin{/}\Varid{x}\mskip1.5mu]\;\Varid{c}_{\Varid{r}}\mathbin{:}\Conid{M}\;\alpha\hspace{0.1em}!\hspace{0.1em}\langle\Conid{F}_{2}\rangle} (8). We obtain our goal from fact 8 and
    \Cref{lem:typesubstitutionelimination}.
  \item \textsc{E-HandOp}: By inversion on \ensuremath{\Gamma\vdash_{\mathsf{SD}}{\mathbf{with}}\;\Varid{h}\;{\mathbf{handle}}\;{\mathbf{op}}\;\ell^{\textsf{op}}\;\Varid{v}\;(\Varid{y}\, .\,\Varid{c}_{1})\mathbin{:}\Conid{B}\hspace{0.1em}!\hspace{0.1em}\langle\Conid{F}_{2}\rangle} (\textsc{SD-Hand}) we have that \ensuremath{\Gamma\vdash_{\mathsf{SD}}\Varid{h}\mathbin{:} \underline{C} _{1}\Rightarrow \Conid{B}\hspace{0.1em}!\hspace{0.1em}\langle\Conid{F}_{2}\rangle} (1), \ensuremath{\Gamma\vdash_{\mathsf{SD}}{\mathbf{op}}\;\ell^{\textsf{op}}\;\Varid{v}\;(\Varid{y}\, .\,\Varid{c}_{1})\mathbin{:} \underline{C} _{2}} (2), and \ensuremath{ \underline{C} _{1}\;\equiv\; \underline{C} _{2}}
    (3). Inversion on fact 1 (\textsc{SD-Handler}) gives \ensuremath{\Conid{B}\mathrel{=}\Conid{M}\;\Conid{A}_{2}}, \ensuremath{ \underline{C} _{1}\mathrel{=}\Conid{A}_{2}\hspace{0.1em}!\hspace{0.1em}\langle\Conid{E}\rangle}, and \ensuremath{\Gamma,\alpha\;\mathbin{\vdash_{\Conid{M}}}\;\Varid{oprs}\mathbin{:}\Conid{M}\;\alpha\hspace{0.1em}!\hspace{0.1em}\langle\Conid{F}_{2}\rangle} (4).
    Based on fact (3) we get that \ensuremath{ \underline{C} _{2}\mathrel{=}\Conid{A2'}\hspace{0.1em}!\hspace{0.1em}\langle\Conid{E'}\rangle}, \ensuremath{\Conid{A}_{2}\;\equiv\;\Conid{A2'}} (4),
    and \ensuremath{\Conid{E}\;\equiv\;\Conid{E'}} (5). Inversion on fact 2 (\textsc{SD-Op}) gives us \ensuremath{\ell^{\textsf{op}}\mathbin{:}{A_{\textsf{op}}}\rightarrowtriangle{B_{\textsf{op}}}\;\in\;\Sigma} (6), \ensuremath{\Gamma\vdash_{\mathsf{SD}}\Varid{v}\mathbin{:}\Conid{A}_{1}} (7), \ensuremath{{A_{\textsf{op}}}\;\equiv\;\Conid{A}_{1}} (8),
    \ensuremath{\Gamma,\Varid{y}\mathbin{:}{B_{\textsf{op}}}\vdash_{\mathsf{SD}}\Varid{c}_{1}\mathbin{:}\Conid{A2'}\hspace{0.1em}!\hspace{0.1em}\langle\Conid{E'}\rangle} (9), and \ensuremath{\ell^{\textsf{op}}\;\in\;\Conid{E'}} (10). By
    \Cref{lem:opmembership} we get that \ensuremath{\Gamma,\alpha\;\mathbin{\vdash_{\Conid{M}}}\;{\mathbf{op}}\;\ell^{\textsf{op}}\;\Varid{x}\;\Varid{k}\mapsto\Varid{c},\Varid{oprs}_{2}\mathbin{:}\Conid{M}\;\alpha\hspace{0.1em}!\hspace{0.1em}\langle\Conid{F}_{2}\rangle} (11). Inversion on fact 11
    (\textsc{SD-OprOp}) gives that \ensuremath{\Gamma,\alpha\;\mathbin{\vdash_{\Conid{M}}}\;\Varid{oprs}\mathbin{:}\Conid{M}\;\alpha\hspace{0.1em}!\hspace{0.1em}\langle\Conid{F}_{1}\rangle} (12), \ensuremath{(\ell^{\textsf{op}}\mathbin{:}{A_{\textsf{op}}}\rightarrowtriangle{B_{\textsf{op}}})\;\in\;\Sigma} (13), \ensuremath{\Gamma,\alpha,\Varid{x}\mathbin{:}{A_{\textsf{op}}},\Varid{k}\mathbin{:}{B_{\textsf{op}}}\to \Conid{M}\;\alpha\hspace{0.1em}!\hspace{0.1em}\langle\Conid{F}_{1}\rangle\vdash_{\mathsf{SD}}\Varid{c}\mathbin{:}\Conid{M}\;\alpha\hspace{0.1em}!\hspace{0.1em}\langle\Conid{F}_{2}\rangle} (14), and \ensuremath{\Conid{F}_{1}\;\equiv\;\Conid{F}_{2}} (15). Facts 1, 4 and 9 in combination with constructors
    \textsc{SD-Abs} and \textsc{ST-Hand} gives us that \ensuremath{\Gamma\vdash_{\mathsf{SD}}\boldsymbol{\lambda}\Varid{y}\, .\,{\mathbf{with}}\;\Varid{h}\;{\mathbf{handle}}\;\Varid{c}_{1}\mathbin{:}{B_{\textsf{op}}}\to \Conid{M}\;\Conid{A}_{2}\hspace{0.1em}!\hspace{0.1em}\langle\Conid{F}_{2}\rangle} (16). The goal follows from facts 7, 8,
    14 and 16 and lemmas
    \Cref{lem:termsubstitutioneliminationgeneq,lem:typesubstitutionelimination}.
  \item \textsc{E-FwdOp} By inversion on \ensuremath{\Gamma\vdash_{\mathsf{SD}}{\mathbf{with}}\;\Varid{h}\;{\mathbf{handle}}\;{\mathbf{op}}\;\ell^{\textsf{op}}\;\Varid{v}\;(\Varid{y}\, .\,\Varid{c}_{1})\mathbin{:}\Conid{B}\hspace{0.1em}!\hspace{0.1em}\langle\Conid{F}_{2}\rangle} (\textsc{SD-Hand}) we have that \ensuremath{\Gamma\vdash_{\mathsf{SD}}\Varid{h}\mathbin{:} \underline{C} _{1}\Rightarrow \Conid{B}\hspace{0.1em}!\hspace{0.1em}\langle\Conid{F}_{2}\rangle} (1), \ensuremath{\Gamma\vdash_{\mathsf{SD}}{\mathbf{op}}\;\ell^{\textsf{op}}\;\Varid{v}\;(\Varid{y}\, .\,\Varid{c}_{1})\mathbin{:} \underline{C} _{2}} (2), and \ensuremath{ \underline{C} _{1}\;\equiv\; \underline{C} _{2}}
    (3). Inversion on fact 1 (\textsc{SD-Handler}) gives \ensuremath{\Conid{B}\mathrel{=}\Conid{M}\;\Conid{A}_{2}}, and \ensuremath{ \underline{C} _{1}\mathrel{=}\Conid{A}_{2}\hspace{0.1em}!\hspace{0.1em}\langle\Conid{E}\rangle}. Based on fact (3) we get that \ensuremath{ \underline{C} _{2}\mathrel{=}\Conid{A2'}\hspace{0.1em}!\hspace{0.1em}\langle\Conid{E'}\rangle}, \ensuremath{\Conid{A}_{2}\;\equiv\;\Conid{A2'}} (4), and \ensuremath{\Conid{E}\;\equiv\;\Conid{E'}} (5). Inversion on fact 2 (\textsc{SD-Op})
    gives us \ensuremath{\ell^{\textsf{op}}\mathbin{:}{A_{\textsf{op}}}\rightarrowtriangle{B_{\textsf{op}}}\;\in\;\Sigma} (6), \ensuremath{\Gamma\vdash_{\mathsf{SD}}\Varid{v}\mathbin{:}\Conid{A}_{1}} (7), \ensuremath{{A_{\textsf{op}}}\;\equiv\;\Conid{A}_{1}} (8), \ensuremath{\Gamma,\Varid{y}\mathbin{:}{B_{\textsf{op}}}\vdash_{\mathsf{SD}}\Varid{c}_{1}\mathbin{:}\Conid{A2'}\hspace{0.1em}!\hspace{0.1em}\langle\Conid{E'}\rangle} (9), and \ensuremath{\ell^{\textsf{op}}\;\in\;\Conid{E'}} (10). The goal follows from facts 1, 4, 5, 6, 7, 8, 10, constructors
    \textsc{SD-Hand} and \textsc{SD-Op}, and \Cref{lem:termweakening}.
  \item \textsc{E-HandSc}: By inversion on \ensuremath{\Gamma\vdash_{\mathsf{SD}}{\mathbf{with}}\;\Varid{h}\;{\mathbf{handle}}\;{\mathbf{op}}\;\ell^{\textsf{op}}\;\Varid{v}\;(\Varid{y}\, .\,\Varid{c}_{1})\mathbin{:}\Conid{B}\hspace{0.1em}!\hspace{0.1em}\langle\Conid{F}_{2}\rangle} (\textsc{SD-Hand}) we have that \ensuremath{\Gamma\vdash_{\mathsf{SD}}\Varid{h}\mathbin{:} \underline{C} _{1}\Rightarrow \Conid{B}\hspace{0.1em}!\hspace{0.1em}\langle\Conid{F}_{2}\rangle} (1), \ensuremath{\Gamma\vdash_{\mathsf{SD}}\ell^{\textsf{sc}}\;\Varid{v}\;(\Varid{y}\, .\,\Varid{c}_{1})\;(\Varid{z}\, .\,\Varid{c}_{2})\mathbin{:} \underline{C} _{2}} (2), and \ensuremath{ \underline{C} _{1}\;\equiv\; \underline{C} _{2}} (3). Inversion on fact 1 (\textsc{SD-Handler}) gives \ensuremath{\Conid{B}\mathrel{=}\Conid{M}\;\Conid{A}_{2}}, \ensuremath{ \underline{C} _{1}\mathrel{=}\Conid{A}_{2}\hspace{0.1em}!\hspace{0.1em}\langle\Conid{E}_{1}\rangle}, and \ensuremath{\Gamma,\alpha\;\mathbin{\vdash_{\Conid{M}}}\;\Varid{oprs}\mathbin{:}\Conid{M}\;\alpha\hspace{0.1em}!\hspace{0.1em}\langle\Conid{F}_{2}\rangle}
    (4). Based on fact (3) we get that \ensuremath{ \underline{C} _{2}\mathrel{=}\Conid{A2'}\hspace{0.1em}!\hspace{0.1em}\langle\Conid{E}_{2}\rangle}, \ensuremath{\Conid{A}_{2}\;\equiv\;\Conid{A2'}}
    (5), and \ensuremath{\Conid{E}_{1}\;\equiv\;\Conid{E}_{2}} (6). Inversion on fact 2 (\textsc{SD-Sc}) gives us
    \ensuremath{\ell^{\textsf{sc}}\mathbin{:}{A_{\textsf{sc}}}\rightarrowtriangle{B_{\textsf{sc}}}\;\in\;\Sigma} (7), \ensuremath{\Gamma\vdash_{\mathsf{SD}}\Varid{v}\mathbin{:}\Conid{A}_{1}} (8), \ensuremath{{A_{\textsf{sc}}}\;\equiv\;\Conid{A}_{1}}
    (9), \ensuremath{\Gamma,\Varid{y}\mathbin{:}{B_{\textsf{sc}}}\vdash_{\mathsf{SD}}\Varid{c}_{1}\mathbin{:}\Conid{A}_{3}\hspace{0.1em}!\hspace{0.1em}\langle\Conid{E}_{3}\rangle} (10), \ensuremath{\Gamma,\Varid{z}\mathbin{:}\Conid{A}_{3}\vdash_{\mathsf{SD}}\Varid{c}_{2}\mathbin{:}\Conid{A2'}\hspace{0.1em}!\hspace{0.1em}\langle\Conid{E}_{2}\rangle} (11), \ensuremath{\Conid{E}_{3}\;\equiv\;\Conid{E}_{2}} (12), and \ensuremath{\ell^{\textsf{sc}}\;\in\;\Conid{E}_{2}} (13).
    \Cref{lem:opmembership} we get that \ensuremath{\Gamma,\alpha\;\mathbin{\vdash_{\Conid{M}}}\;{\mathbf{op}}\;\ell^{\textsf{op}}\;\Varid{x}\;\Varid{k}\mapsto\Varid{c},\Varid{oprs}_{2}\mathbin{:}\Conid{M}\;\alpha\hspace{0.1em}!\hspace{0.1em}\langle\Conid{F}_{2}\rangle} (13.1). Inversion on fact 13.1
    (\textsc{SD-OprSc} gives \ensuremath{\ell^{\textsf{sc}}\;\in\;\Sigma} (14), \ensuremath{\beta\;\Varid{fresh}} (15), \ensuremath{\Gamma,\alpha,\beta,\Varid{x}\mathbin{:}{A_{\textsf{sc}}},\Varid{p}\mathbin{:}{B_{\textsf{sc}}}\to \Conid{M}\;\beta\hspace{0.1em}!\hspace{0.1em}\langle\Conid{F}_{3}\rangle,\Varid{k}\mathbin{:}\beta\to \Conid{M}\;\alpha\hspace{0.1em}!\hspace{0.1em}\langle\Conid{F}_{3}\rangle\vdash_{\mathsf{SD}}\Varid{c}\mathbin{:}\Conid{M}\;\alpha\hspace{0.1em}!\hspace{0.1em}\langle\Conid{F}_{2}\rangle} (16), and \ensuremath{\Conid{F}_{2}\;\equiv\;\Conid{F}_{3}} (17).
    Facts 1, 5, 6, 10 and 12, constructors \textsc{SD-Abs} and \textsc{SD-Hand}
    and \Cref{lem:termweakening,lem:handlersarepolymorphic} give us that \ensuremath{\Gamma,\beta\vdash_{\mathsf{SD}}\boldsymbol{\lambda}\Varid{y}\, .\,{\mathbf{with}}\;\Varid{h}\;{\mathbf{handle}}\;\Varid{c}_{1}\mathbin{:}\Varid{p}\mathbin{:}{B_{\textsf{sc}}}\to \Conid{M}\;\beta\hspace{0.1em}!\hspace{0.1em}\langle\Conid{F}_{2}\rangle} (18).
    Facts 1, 5, 6 and 11 and constructors \textsc{SD-Abs} and \textsc{SD-Hand}
    and \Cref{lem:termweakening} give us that \ensuremath{\Gamma,\beta\vdash_{\mathsf{SD}}\boldsymbol{\lambda}\Varid{z}\, .\,{\mathbf{with}}\;\Varid{h}\;{\mathbf{handle}}\;\Varid{c}_{2}\mathbin{:}\beta\to \Conid{M}\;\Conid{A}_{2}\hspace{0.1em}!\hspace{0.1em}\langle\Conid{F}_{2}\rangle} (19).
    The goal now follows from facts 8, 9, 16, 17 and 18 and
    \Cref{lem:termsubstitutioneliminationgeneq,lem:termsubstitutioneliminationgeneq}.
  \item \textsc{E-FwdSc} By inversion on \ensuremath{\Gamma\vdash_{\mathsf{SD}}{\mathbf{with}}\;\Varid{h}\;{\mathbf{handle}}\;{\mathbf{op}}\;\ell^{\textsf{op}}\;\Varid{v}\;(\Varid{y}\, .\,\Varid{c}_{1})\mathbin{:}\Conid{B}\hspace{0.1em}!\hspace{0.1em}\langle\Conid{F}_{2}\rangle} (\textsc{SD-Hand}) we have that \ensuremath{\Gamma\vdash_{\mathsf{SD}}\Varid{h}\mathbin{:} \underline{C} _{1}\Rightarrow \Conid{B}\hspace{0.1em}!\hspace{0.1em}\langle\Conid{F}_{2}\rangle} (1), \ensuremath{\Gamma\vdash_{\mathsf{SD}}\ell^{\textsf{sc}}\;\Varid{v}\;(\Varid{y}\, .\,\Varid{c}_{1})\;(\Varid{z}\, .\,\Varid{c}_{2})\mathbin{:} \underline{C} _{2}} (2), and \ensuremath{ \underline{C} _{1}\;\equiv\; \underline{C} _{2}} (3). Inversion on fact 1 (\textsc{SD-Handler}) gives \ensuremath{\Conid{B}\mathrel{=}\Conid{M}\;\Conid{A}_{2}}, \ensuremath{ \underline{C} _{1}\mathrel{=}\Conid{A}_{2}\hspace{0.1em}!\hspace{0.1em}\langle\Conid{E}_{1}\rangle}, and \ensuremath{\Gamma,\alpha\;\mathbin{\vdash_{\Conid{M}}}\;{\mathbf{fwd}}\;\Varid{f}\;\Varid{p}\;\Varid{k}\mapsto\Varid{c}_{\Varid{f}}\mathbin{:}\Conid{M}\;\alpha\hspace{0.1em}!\hspace{0.1em}\langle\Conid{F}_{2}\rangle} (4). Based on fact (3) we get that \ensuremath{ \underline{C} _{2}\mathrel{=}\Conid{A2'}\hspace{0.1em}!\hspace{0.1em}\langle\Conid{E}_{2}\rangle}, \ensuremath{\Conid{A}_{2}\;\equiv\;\Conid{A2'}} (5), and \ensuremath{\Conid{E}_{1}\;\equiv\;\Conid{E}_{2}} (6). Inversion on fact 2 (\textsc{SD-Sc})
    gives us \ensuremath{\ell^{\textsf{sc}}\mathbin{:}{A_{\textsf{sc}}}\rightarrowtriangle{B_{\textsf{sc}}}\;\in\;\Sigma} (7), \ensuremath{\Gamma\vdash_{\mathsf{SD}}\Varid{v}\mathbin{:}\Conid{A}_{1}} (8), \ensuremath{{A_{\textsf{sc}}}\;\equiv\;\Conid{A}_{1}} (9), \ensuremath{\Gamma,\Varid{y}\mathbin{:}{B_{\textsf{sc}}}\vdash_{\mathsf{SD}}\Varid{c}_{1}\mathbin{:}\Conid{A}_{3}\hspace{0.1em}!\hspace{0.1em}\langle\Conid{E}_{3}\rangle} (10), \ensuremath{\Gamma,\Varid{z}\mathbin{:}\Conid{A}_{3}\vdash_{\mathsf{SD}}\Varid{c}_{2}\mathbin{:}\Conid{A2'}\hspace{0.1em}!\hspace{0.1em}\langle\Conid{E}_{2}\rangle} (11), \ensuremath{\Conid{E}_{3}\;\equiv\;\Conid{E}_{2}} (12), and \ensuremath{\ell^{\textsf{sc}}\;\in\;\Conid{E}_{2}} (13).
    Inversion on fact 4 (\textsc{SD-Fwd}) gives
    \ensuremath{\Conid{A}_{\Varid{p}}\mathrel{=}\alpha'\to \Conid{M}\;\beta\hspace{0.1em}!\hspace{0.1em}\langle\Conid{F}_{1}\rangle},
    \ensuremath{\Conid{A}_{\Varid{p}}'\mathrel{=}\alpha'\to \gamma\hspace{0.1em}!\hspace{0.1em}\langle\Conid{F}_{1}\rangle},
    \ensuremath{\Conid{A}_{\Varid{k}}\mathrel{=}\beta\to \Conid{M}\;\Conid{A}_{4}\hspace{0.1em}!\hspace{0.1em}\langle\Conid{F}_{1}\rangle},
    \ensuremath{\Conid{A}_{\Varid{k}}'\mathrel{=}\gamma\to \delta\hspace{0.1em}!\hspace{0.1em}\langle\Conid{F}_{1}\rangle},
    \ensuremath{\Gamma,\alpha,\alpha',\beta,\Varid{p}\mathbin{:}\Conid{A}_{\Varid{p}},\Varid{k}\mathbin{:}\Conid{A}_{\Varid{k}},\Varid{f}\mathbin{:}\forall\;\gamma\;\delta\, .\,(\Conid{A}_{\Varid{p}}',\Conid{A}_{\Varid{k}}')\to \delta\hspace{0.1em}!\hspace{0.1em}\langle\Conid{F}_{1}\rangle\vdash_{\mathsf{SD}}\Varid{c}_{\Varid{f}}\mathbin{:}\Conid{M}\;\alpha\hspace{0.1em}!\hspace{0.1em}\langle\Conid{F}_{2}\rangle} (14) and
    \ensuremath{\Conid{M}\;\Conid{A}_{1}\hspace{0.1em}!\hspace{0.1em}\langle\Conid{F}_{1}\rangle\;\equiv\;\Conid{M}\;\alpha\hspace{0.1em}!\hspace{0.1em}\langle\Conid{F}_{2}\rangle} (15).
    Facts 1, 6, 10 and 12, constructors \textsc{SD-Abs} and \textsc{SD-Hand} and
    \Cref{lem:termweakening,lem:handlersarepolymorphic} give us that \ensuremath{\Gamma,\alpha,\Varid{y}\mathbin{:}{B_{\textsf{sc}}}\vdash_{\mathsf{SD}}{\mathbf{with}}\;\Varid{h}\;{\mathbf{handle}}\;\Varid{c}_{1}\mathbin{:}\Conid{M}\;\Conid{A}_{3}\hspace{0.1em}!\hspace{0.1em}\langle\Conid{F}_{2}\rangle} (16)
    Facts 1, 6, 10 and 12, constructors \textsc{SD-Abs} and \textsc{SD-Hand} and
    \Cref{lem:termweakening} give us that \ensuremath{\Gamma,\alpha\vdash_{\mathsf{SD}}\boldsymbol{\lambda}\Varid{z}\, .\,{\mathbf{with}}\;\Varid{h}\;{\mathbf{handle}}\;\Varid{c}_{2}\mathbin{:}\Conid{A}_{3}\to \Conid{M}\;\Conid{A}_{4}\hspace{0.1em}!\hspace{0.1em}\langle\Conid{F}_{2}\rangle} (17)
    Facts 7, 8, 9, 13, the fact that \ensuremath{\ell^{\textsf{sc}}\;\notin\;\mathit{labels\!}\;(\Varid{oprs})}, constructors
    \textsc{SD-Abs}, \textsc{SD-App} and \textsc{SD-Var} and
    \Cref{lem:termweakening} give us that \ensuremath{\Gamma,\alpha,(\Varid{p'},\Varid{k'})\mathbin{:}\forall\;\gamma\;\delta\, .\,({B_{\textsf{sc}}}\to \gamma\hspace{0.1em}!\hspace{0.1em}\langle\Conid{F}_{1}\rangle,\gamma\to \delta\hspace{0.1em}!\hspace{0.1em}\langle\Conid{F}_{1}\rangle)\vdash_{\mathsf{SD}}{\mathbf{sc}}\;\ell^{\textsf{sc}}\;\Varid{v}\;(\Varid{y}\, .\,\Varid{p'}\;\Varid{y})\;(\Varid{z}\, .\,\Varid{k'}\;\Varid{z})\mathbin{:}\delta\hspace{0.1em}!\hspace{0.1em}\langle\Conid{F}_{1}\rangle} (18).
    Our goal then follows from facts 14, 15, 16, 17, 18 and
    \Cref{lem:termsubstitutioneliminationgeneq}.
  \end{itemize}
\end{proof}

\subsection{Progress}
\progress*
\begin{proof}
  By \Cref{lem:nsdiffsd} (above) and \Cref{thm:sdprogress} (below).
\end{proof}

\begin{thm}[Syntax-directed Progress]
  \label{thm:sdprogress}
  If \ensuremath{ \cdot\vdash_{\mathsf{SD}}\Varid{c}\mathbin{:}\Conid{A}\hspace{0.1em}!\hspace{0.1em}\langle\Conid{E}\rangle}, then either:
  \begin{itemize}
  \item there exists a computation \ensuremath{\Varid{c'}} such that \ensuremath{\Varid{c}\leadsto\Varid{c'}}, or
  \item c is in a \emph{normal form}, which means it is in one of the following
    forms: (1) \ensuremath{\Varid{c}\mathrel{=}{\mathbf{return}}\;\Varid{v}}, (2) \ensuremath{\Varid{c}\mathrel{=}{\mathbf{op}}\;\ell^{\textsf{op}}\;\Varid{v}\;(\Varid{y}\, .\,\Varid{c'})} where \ensuremath{\ell^{\textsf{op}}\;\in\;\Conid{E}},
    or (3) \ensuremath{\Varid{c}\mathrel{=}{\mathbf{sc}}\;\ell^{\textsf{sc}}\;\Varid{v}\;(\Varid{y}\, .\,\Varid{c}_{1})\;(\Varid{z}\, .\,\Varid{c}_{2})} where \ensuremath{\ell^{\textsf{sc}}\;\in\;\Conid{E}}.
  \end{itemize}
\end{thm}

\begin{proof}
  By induction on the typing derivation \ensuremath{ \cdot\vdash_{\mathsf{SD}}\Varid{c}\mathbin{:} \underline{C} }.

  \begin{itemize}
  \item \textsc{SD-App}: Here, \ensuremath{ \cdot\vdash_{\mathsf{SD}}\Varid{v}_{1}\;\Varid{v}_{2}}. Since \ensuremath{\Varid{v}_{1}} has type \ensuremath{\Conid{A}\to \Conid{B}\hspace{0.1em}!\hspace{0.1em}\langle\Conid{F}\rangle}, by \Cref{lem:canonicalforms} it must be of shape \ensuremath{\boldsymbol{\lambda}\Varid{x}\, .\,\Varid{c}},
    which means we can step by rule \textsc{E-AppAbs}.
  \item \textsc{SD-Do}: Here, \ensuremath{ \cdot\vdash_{\mathsf{SD}}\mathbf{do}\;\Varid{x}\leftarrow \Varid{c}_{1}\;\mathbf{in}\;\Varid{c}_{2}\mathbin{:} \underline{C} }. By the induction
    hypothesis, \ensuremath{\Varid{c}_{1}} can either step (in which case we can step by
    \textsc{E-Do}), or it is a computation result. Every possible form has a
    corresponding reduction: if \ensuremath{\Varid{c}_{1}\mathrel{=}{\mathbf{return}}\;\Varid{v}} we can step by \textsc{E-DoRet},
    if \ensuremath{\Varid{c}_{1}\mathrel{=}{\mathbf{op}}\;\ell^{\textsf{op}}\;\Varid{v}\;(\Varid{y}\, .\,\Varid{c})} we can step by \textsc{E-DoOp}, and if \ensuremath{{\mathbf{sc}}\;\ell^{\textsf{op}}\;\Varid{v}\;(\Varid{y}\, .\,\Varid{c}_{1}')\;(\Varid{z}\, .\,\Varid{c}_{2}')} we can step by \textsc{E-DoSc}.
  \item \textsc{SD-Let}: Here, \ensuremath{ \cdot\vdash_{\mathsf{SD}}\mathbf{let}\;\Varid{x}\mathrel{=}\Varid{v}\;\mathbf{in}\;\Varid{c}\mathbin{:} \underline{C} }, which means we
    can step by \textsc{E-Let}.
  \item \textsc{SD-Ret}, \textsc{SD-Op}, and \textsc{SD-Sc}: all of these are
    computation results (forms (1), (2), and (3), resp.). Since they
    are well-typed, the scoping condition of \ensuremath{\ell^{\textsf{op}}} or \ensuremath{\ell^{\textsf{sc}}} must be satisfied.
  \item \textsc{SD-Hand}: Here \ensuremath{ \cdot\vdash_{\mathsf{SD}}{\mathbf{with}}\;\Varid{v}\;{\mathbf{handle}}\;\Varid{c}\mathbin{:}\Conid{M}\;\Conid{A}\hspace{0.1em}!\hspace{0.1em}\langle\Conid{F}\rangle}. By
    \Cref{lem:canonicalforms}, \ensuremath{\Varid{v}} is of shape \ensuremath{\Varid{h}}. By the induction hypothesis,
    \ensuremath{\Varid{c}} can either step (in which case we can step by \textsc{E-Hand}), or it is
    in a normal form. Proceed by case split on the three forms.
    \begin{enumerate}
    \item Case \ensuremath{\Varid{c}\mathrel{=}{\mathbf{return}}\;\Varid{v}}. Since \ensuremath{ \cdot\vdash_{\mathsf{SD}}\Varid{h}\mathbin{:} \underline{C} \Rightarrow  \underline{D} }, there must be
      some \ensuremath{({\mathbf{return}}\;\Varid{x}\mapsto\Varid{c}_{\Varid{r}})\;\in\;\Varid{h}} which means we can step by rule
      \textsc{E-HandRet}.
    \item Case \ensuremath{\Varid{c}\mathrel{=}{\mathbf{op}}\;\ell^{\textsf{op}}\;\Varid{v}\;(\Varid{y}\, .\,\Varid{c'})}. Depending on \ensuremath{({\mathbf{op}}\;\ell^{\textsf{op}}\;\Varid{x}\;\Varid{k}\mapsto\Varid{c})\;\in\;\Varid{h}} we
      can step by \textsc{E-HandOp} or \textsc{E-FwdOp}.
    \item Case \ensuremath{\Varid{c}\mathrel{=}{\mathbf{sc}}\;\ell^{\textsf{sc}}\;\Varid{v}\;(\Varid{y}\, .\,\Varid{c}_{1})\;(\Varid{z}\, .\,\Varid{c}_{2})}. If \ensuremath{({\mathbf{sc}}\;\ell^{\textsf{sc}}\;\Varid{x}\;\Varid{p}\;\Varid{k}\mapsto\Varid{c})\;\in\;\Varid{h}}, we
      can step by \textsc{E-HandSc}. If not, since \ensuremath{ \cdot\vdash_{\mathsf{SD}}\Varid{h}\mathbin{:} \underline{C} \Rightarrow  \underline{D} },
      there must be some \ensuremath{({\mathbf{fwd}}\;\Varid{f}\;\Varid{p}\;\Varid{k}\mapsto\Varid{c}_{\Varid{f}})\;\in\;\Varid{h}} which means we can step by
      rule \textsc{E-FwdSc}. \qedhere
    \end{enumerate}
  \end{itemize}
\end{proof}
\clearpage

\end{document}